\documentclass[10pt]{amsart}

\usepackage{bbold}

\usepackage{epigraph}
\usepackage{enumerate}
\usepackage{amsmath,amsthm,amscd,amssymb}

\usepackage{latexsym}
\usepackage{upref}
\usepackage{verbatim}

\usepackage[mathscr]{eucal}
\usepackage{dsfont}

\usepackage{graphicx}

\usepackage[colorlinks,hyperindex,pagebackref,hypertex]{hyperref}

\usepackage[OT2,OT1]{fontenc}
\newcommand\cyr
{
\renewcommand\rmdefault{wncyr}
\renewcommand\sfdefault{wncyss}
\renewcommand\encodingdefault{OT2}
\normalfont
\selectfont
}
\DeclareTextFontCommand{\textcyr}{\cyr}
\def\cprime{\char"7E }

\def\Eoborotnoye{\char'003}

\usepackage{pdfpages}

\usepackage{color}

\usepackage{pstricks}

\usepackage{graphicx}
\usepackage{skull}
\usepackage{wasysym}

\usepackage[colorlinks,hyperindex,pagebackref,hypertex]{hyperref}
\newtheorem{theorem}{Theorem}[section]
\newtheorem{lemma}[theorem]{Lemma}

\theoremstyle{definition}
\newtheorem{definition}[theorem]{Definition}

\newtheorem{hypothesis}[theorem]{Hypothesis}
\newtheorem{corollary}[theorem]{Corollary}
\newtheorem{proposition}[theorem]{Proposition}

\theoremstyle{remark}
\newtheorem{remark}[theorem]{Remark}

\numberwithin{equation}{section}

\newcommand{\bbK}{\mathbb K}

\newcommand{\bbR}{\mathbb R}
\newcommand{\bbT}{\mathbb T}

\newcommand{\bbZ}{\mathbb Z}

\newcommand{\bbD}{\mathbb D}

\newcommand{\bbC}{\mathbb C}

\newcommand{\bbH}{\mathbb H}

\newcommand{\bbY}{\mathbb Y}
\newcommand{\bbX}{\mathbb X}

\renewcommand{\epsilon}{\varepsilon}

\newcommand{\be}{\begin{equation}}
\newcommand{\ee}{\end{equation}}

\newcommand{\spec}{\mathrm{spec}}
\newcommand{\Span}{\mathrm{span}}

\renewcommand{\L}{\mathbb{L}}

\newcommand{\EE}{\mathsf{E}}

\newcommand{\cA}{{\mathcal A}}
\newcommand{\cB}{{\mathcal B}}
\newcommand{\cC}{{\mathcal C}}
\newcommand{\cD}{{\mathcal D}}
\newcommand{\cE}{{\mathcal E}}

\newcommand{\cH}{{\mathcal H}}

\newcommand{\cK}{{\mathcal K}}
\newcommand{\cL}{{\mathcal L}}
\newcommand{\cM}{{\mathcal M}}
\newcommand{\cN}{{\mathcal N}}

\newcommand{\cP}{{\mathcal P}}
\newcommand{\cQ}{{\mathcal Q}}
\newcommand{\cR}{{\mathcal R}}

\newcommand{\cU}{{\mathcal U}}

\newcommand{\cW}{{\mathcal W}}

\renewcommand{\Im}{{\ensuremath{\mathrm{Im}}}}
\renewcommand{\Re}{{\ensuremath{\mathrm{Re}}}}

\newcommand{\tr}{\mathrm{tr}}

\newcommand{\sign}{\mathrm{sgn}}

\newcommand{\fA}{\mathfrak{A}}
\newcommand{\fB}{\mathfrak{B}}

\newcommand{\fD}{\mathfrak{D}}

\newcommand{\fH}{\mathfrak{H}}

\newcommand{\fM}{\mathfrak{M}}

\DeclareMathOperator*{\slim}{s-lim}
\DeclareMathOperator{\Ran}{\mathrm{Ran}}
\DeclareMathOperator{\Ker}{\mathrm{Ker}}

\newcommand{\linspan}{\mathrm{lin\ span}}

\newcommand{\Dom}{\mathrm{Dom}}
\newcommand{\dom}{\mathrm{Dom}}

\DeclareFontFamily{U}{rcjhbltx}{}
\DeclareFontShape{U}{rcjhbltx}{m}{n}{<->rcjhbltx}{}
\DeclareSymbolFont{hebrewletters}{U}{rcjhbltx}{m}{n}

\let\aleph\relax\let\beth\relax
\let\gimel\relax\let\daleth\relax

\DeclareMathSymbol{\aleph}{\mathord}{hebrewletters}{39}
\DeclareMathSymbol{\beth}{\mathord}{hebrewletters}{98}
\DeclareMathSymbol{\gimel}{\mathord}{hebrewletters}{103}
\DeclareMathSymbol{\daleth}{\mathord}{hebrewletters}{100}

\DeclareMathSymbol{\lamed}{\mathord}{hebrewletters}{108}
\DeclareMathSymbol{\mem}{\mathord}{hebrewletters}{109}
\DeclareMathSymbol{\ayin}{\mathord}{hebrewletters}{96}
\DeclareMathSymbol{\tsadi}{\mathord}{hebrewletters}{118}
\DeclareMathSymbol{\qof}{\mathord}{hebrewletters}{113}
\DeclareMathSymbol{\shin}{\mathord}{hebrewletters}{152}

\begin{document}

\title [Dissipative and Non-unitary Representations ]
{ Representations of commutation relations in Dissipative Quantum Mechanics}

\author{K. A. Makarov}
\address{Department of Mathematics, University of Missouri, Columbia, MO 63211, USA}
\email{makarovk@missouri.edu}

\author{E. Tsekanovski\u{i} }
\address{ Department of Mathematics, Niagara University, Lewiston,
NY  14109, USA } \email{tsekanov@niagara.edu}


\keywords{Weyl commutation relations, affine group, deficiency indices,
  quasi-selfadjoint extensions, continuous monitoring, Quantum Zeno effect, exponential decay, stable laws.\\
  The first author was partially  supported by the Simons collaboration grant 00061759 while preparing this monograph.}

\begin{abstract}  We prove the uniqueness theorem for the solutions to the restricted Weyl commutation relations braiding unitary groups and semi-groups of contractions
  that are    close to unitaries. We also discuss related mathematical problems of continuous monitoring of quantum systems and provide rigorous foundations for the
  exponential decay phenomenon of a resonant state  in quantum mechanics.

\end{abstract}

\maketitle

\setlength\epigraphwidth{.8\textwidth}

\vspace*{\fill}
\epigraph{One might still like to ask: ``How does it work? What is the machinery behind the law?'' No one has found any machinery behind the law. No one can ``explain'' any more than we have just ``explained.'' No one will give you any deeper representation of the situation. We have no ideas about a more basic mechanism from which these results can be deduced.}{Richard Feynman, \textit{The Feynman Lectures on Physics, Volume III}}

\tableofcontents

\dedicatory{\it In  respectful memory of our beloved teachers M.  Liv\v{s}ic and  B. Pavlov.}



\part{\Large  Representations of operator  commutation relations}

\section{Introduction}\label{s1}

{\it This book is dedicated to the memory of the remarkable Human Beings and Mathematicians Michail Samoilovich  Liv\v{s}ic (M.S.) and Boris Sergeevich Pavlov (B.S.). Their pioneering research in the theory of non-selfadjoint operators and applications to  scattering  problems and system theory have  attracted many researchers and made  the present book possible.}

{\it The  Apostolic service to students, colleagues and the mathematical community provided by M.S. and B.S. was enormous and their good deeds will never be forgotten. The light of scientific accomplishments of M.S. and B.S.  shines brightly   and is succinctly described by the poetic word of Galina Volchek: }
$$
$$
\begin{center}
{\cyr \textsf{Uhodya ostav\cprime te Svet! {\Eoborotnoye}to bol\cprime she, chem ostat\cprime sya...}}

{\cyr \qquad\textsf{{\Eoborotnoye}to luchshe, chem proshchat\cprime sya i vazhne{\u i}, chem dat\cprime \  sovet...}}

{\cyr \textsf{Uhodya ostav\cprime te Svet - pered nim ot{s}tupit holod!}}

{\cyr \qquad\textsf{Svet sobo{\u i} zapolnit gorod... Dazhe esli vas tam net...}}

\end{center}
$$
$$

The classical Stone-von Neumann theorem \cite{EM,Mac,Si,S} states
that  the unitary  representations of the canonical
commutation relations  (CCR) of Quantum Mechanics in the Weyl form \cite{W}
\begin{equation}\label{weylccr}
U_tV_s=e^{ist}V_sU_t
\end{equation}
for  strongly  continuous, one-parameter groups of unitary operators  $U_t$ and $V_s$ in a
separable Hilbert space $\cH$ are unitarily equivalent to a direct sum
of copies of the unique
 irreducible system in the Hilbert space
$\cH=L^2(\bbR)$ with
$$
(U_tf)(x)=\exp (ixt)f(x) \quad \text{ and } \quad (V_sf)(x)=f(x-s).
$$

For the history of the subject we refer to \cite{Rosen} where  one can  find a thorough discussion of the further generalizations initiated by G. W. Mackey
in  his ground-breaking paper \cite{Mac},  and  the subsequent development  in number theory due to A. Weil \cite{Weil}.
We  refer also to  the series of publications
\cite{BL,Dixmier,Foash,Kato62,Lax,LJ,Putnam,Rellich}
where  the interested reader can find a  truly extensive body of information on the subject.





The CCR \eqref{weylccr}  can  be reformulated  in an equivalent
infinitesimal form  (see \cite{EM}, \cite{Si})
as a relation for   the self-adjoint generator $A$  of the group $V_s=e^{isA}$
\begin{equation}\label{ifccr}
U_tAU_t^{*}=A+tI\quad \text{ on } \Dom(A), \quad  t\in \bbR,
\end{equation}
or as the equality invoking the spectral measure   $E(d\lambda)$  of the self-adjoint operator $A$
$$
U_tE(\delta)U_t^{*}=E(\delta-t), \,\, t\in \bbR, \quad \delta \text{ a Borel set}.
$$

It is worth mentioning that rewriting  relations  \eqref{weylccr}  in its infinitesimal
(semi-Weyl) form \eqref{ifccr}
opens a way for the further developments  and generalizations   in various  directions.

In this book we choose a presentation line  taking into account the following observation:
the Stone-von Neumann uniqueness result \cite{JvN} implies that  if a self-adjoint operator $A$   satisfies \eqref{ifccr}, then
$A$  always admits a symmetric restriction $\dot A\subset A$ with deficiency indices $(1,1)$ such that
 the same commutation relations
\begin{equation}\label{1/2Wdot}
U_t\dot AU_t^*=\dot A+t I \quad \text{on } \quad \Dom(\dot A), \quad t\in \bbR,
\end{equation}
hold.

Given the commutation relations for a symmetric operator \eqref{1/2Wdot},  see Hypothesis \ref{muhly}, the following natural problems can be posed:

 (I) $a)$ {\it Characterize  such  symmetric operator solutions    $\dot A$  up to unitary equivalence;}

$\hskip .6cm b)$ {\it Provide an intrinsic characterization  of those solutions. }

 (II) {\it  Find the maximal dissipative solutions $\widehat A$ to  the  infinitesimal Weyl
 relations of the form
 \begin{equation}\label{mmm}
U_t \widehat AU_t^*=\widehat A+tI\quad \text{on } \quad \Dom( \widehat A)
\end{equation}
  such that $\dot A\subset  \widehat A\subset (\dot A)^*$. }

Problem (I) b) was posed  in \cite{J80} and  we will refer to it as the   {\it J$\clock$rgensen-Muhly  problem}.

Notice that in this situation  the semigroup $V_s=e^{is \widehat A}$, $s\ge 0$, generated by the dissipative operator $\widehat A$ and the unitary group $U_t=e^{itB}$, generated by a self-adjoint operator $B=B^*$, satisfy the restricted Weyl commutation relations
\begin{equation}\label{resweyl}
U_tV_s=e^{ist}V_s U_t, \quad t\in \bbR,\quad  s\ge 0.
\end{equation}

More generally,  one can ask to  provide the  complete classification (up to mutual unitary equivalence) of the pairs of corresponding generators $(\widehat A,B)$  under the solely assumption   that the generator  $\widehat A$
is an  extension of a symmetric operator with arbitrary deficiency indices  $(m, n)$.

Much progress has been achieved  in this area of research (see  \cite{BP2007,Jorg1,J79,J80,J80a,J81,JM,Sch1, Sch2},  also see \cite{TMF}).
For instance,
it is known that   a semi-group satisfies the  restricted Weyl  relations if and only if the    characteristic function of its generator has a particularly simple form \cite[Theorem 20]{JM}.
Moreover, in this case,
 the restricted Weyl system can be dilated to  a canonical Weyl system in an extended Hilbert space \cite[Theorem 15]{JM}.
 However, to the best of our knowledge, the  {\it complete classification } of irreducible representations of the restricted commutation
relations \eqref{resweyl},  even in the simplest case  of deficiency indices $(1,1)$, has not  been  obtained yet.

In this book  we give    the {\it complete solution to  problem } (II) under the assumption that
 the generator $\widehat A$ of the semi-group $V_s$ is a dissipative quasi-selfadjoint extension \cite{AkG} of a prime symmetric operator $\dot A$  with deficiency indices $(1,1)$.
 As in the Stone-von Neumann theorem,  we show that the pair  $(\widehat A,B)$  of  generators  is   mutually unitarily  equivalent
 to  the ``canonical''  pair $(\widehat \cP, \cQ)$  on  a metric graph $\bbY$, finite or infinite.
Here  $\widehat \cP$ stands for   a dissipative differentiation  ({\it momentum})  operator
 on  $\bbY$ with appropriate vertex boundary conditions
  and $\cQ$ is  the self-adjoint multiplication  ({\it position}) operator on the graph  $\bbY$.
  In contrast to the Stone-von Neumann  uniqueness theorem,
  where the corresponding graph is just the  real axis (with no reference vertices), the graph geometry of  $\bbY$ is more varied  (see Definition \ref{canonical}
  for the classification).
Moreover, the knowledge of the complete set of  unitary invariants of the solutions  to the commutation relations determines not only the geometry of the metric graph $\bbY$ but also  the location of the central vertex of the graph.
For instance,  given a solution of the commutation relations on a metric graph, one   obtains new series of unitarily inequivalent solutions
 by shifting  the  graph.

Our approach  is based on the detailed study of unitary invariants of operators
  such as  the   Liv\v{s}ic and/or Weyl-Titchmarsh functions   associated with the pair of  a symmetric operator and its self-adjoint (reference) extension
 as well as  the characteristic function of a dissipative triple of operators.
 A comprehensive  study of the concept of a characteristic function associated with
various classes of non-selfadjoint operators, in particular,
with applications to scattering theory, system theory and
boundary value problems one can find in \cite{BrRov,Lax,L46,Lv1,LOOW,LP,Nagy,PavDrog,Pav96,Pav2010} 
as well as in 
  \cite{AA,ABT,BT,BL60,Brod,Davies,DM,DM17,Dub1,Dub2,Kuzhel,KK,LJETP,LVin,LJ,Naboko,Nik1,Nik2,Pav75,PF, Peller,PT,Sakh,
  Str60,St68,TS,Zol,Zolbook}.

The departure point  for our study of commutation relations is   structure  Theorem \ref{lem:diss}. This result states  that in the situation in question
 the characteristic function of a dissipative triple is either  (i) a constant, or   (ii)  a singular inner function  in the upper half-plane with ``mass at infinity",
 or (iii)  the product of those  two.
The examples of   differentiation  operators (more precisely, a  triple of those)  with  either a constant or  entire characteristic function are known (see, e.g., \cite[Ch. IX]{AkG}). The construction of  the model differentiation operator/triple in the general case (iii)
can be achieved in the framework of operator coupling  theory  \cite{MT-MAT}. Notice that  those examples of differentiation operators are {\it the building blocks}  in our approach.
In particular, addressing Problem (I) a),  we obtain the complete classification   of the symmetric operators with deficiency indices $(1,1)$ that solve  the infinitesimal  relations  \eqref{1/2Wdot} up to unitary equivalence.
  We also  provide an intrinsic characterization of the corresponding  symmetric operator solutions to the commutation relations, thus giving a {\it comprehensive answer  to   the
 J$\clock$rgensen-Muhly  problem  } (I) $b)$ \cite{J80} (see Remark \ref{muhlyresolution}).

In  the second part of the book we address  the quantum  continuous monitoring theory  and discuss in detail complementarity  of the {\sc                                                                                                                                                                                                                                                                                                                                                                                                                                                                                                                                                                                                         Quantum Zeno
and {\sc Exponential Decay} scenarios in frequent quantum measurements experiments.}
 For  background material,
we refer to \cite{Atm,Kh,Misra} and \cite{Fock,KMPY,Weiskopf}, respectively, and the  references therein.
In this context, we also want to  mention the revolutionary paper  by Gamow
\cite{Gamov} who  was the first to introduce quantum states with ``complex" energies  and, based on this concept, gave an explanation for the decay law for   a quasi-stationary state.

In the framework of our formalism we give a justification for
the exponential decay scenario (under continuous monitoring) by recognizing  the phenomenon as a
variant of   the Gnedenko-Kolmogorov  $1$-stable limit theorem. Having this link in mind, we
 obtain several principal  results in quantum measurement theory.
In particular,
we show  that a  ``typical  smooth" state of a material $($massive$)$ particle  under continuous monitoring is either a  {\it Zeno state or an  anti-Zeno state} (see Theorems \ref{1/2thm} and \ref{3/2thm}).
In contrast to that, for the systems of massless particles (fields)   the situation is quite different: if  the Hamiltonian of the system  is given by
 the first order differentiation operator, then  the {\it  quantum Zeno
and  exponential decay scenarios are complementary} instead. In addition, it turns out that for that  kind of
systems, the decay rate   is rather  sensitive   to the choice of a self-adjoint realization of the  Hamiltonian  on the metric  graph, especially if the graph is not simply connected.
From the point of view of physics,
this phenomenon is a manifestation of    the  Aharonov-Bohm effect. That is,  in the absence of the magnetic field,
the magnetic potential by itself   affects the magnitude of the decay rate in this case (see Theorem \ref{decr} and Section 18, eq. \eqref{intell}).
We also notice that the existence of  states the decay   rate  of which is independent of the Aharonov-Bohm field  is closely related to   the search for dissipative operator  solutions of  the infinitesimal  Weyl commutation relations  \eqref{mmm}.

As an illustration, within  continuous monitoring paradigm we discuss a {\it Gedankenexperiment}
where the renowned  {\it exclusive} and {\it interference} measurement  alternatives in quantum theory can be rigorously analyzed. In addition,  on the basis of  an explicitly soluble model, we  present a variant of  the celebrated ``double-slit experiment"  in a way that is accessible for mathematicians (see Theorem \ref{quanmon},
 eq. \eqref{QED0}, and Theorem   \ref{quanmondis}, eq. \eqref{clastau1}).

We conclude our treatise by the discussion   of  limit theorems  in the framework of operator coupling theory. More specifically, we introduce a new  mode of convergence for dissipative operators (in distribution) and
 show that basic dissipative solutions to the commutation relations  can be considered  analogs of stable distributions in the orthodox probability theory. This observation,  in our opinion,  sheds some light on the foundations of dissipative quantum theory of open
 systems  which our dear teachers M.S. and B.S. dedicated their scientific life to.

  The book is organized as follows.

$\bullet$ In Section \ref{s2}, following our work \cite{MT-S},  we recall  the concept of the Weyl-Titchmarsh and  the Liv\v{s}ic functions associated with the  pair of operators ($\dot A, A)$ where $\dot A$ is a symmetric operator with deficiency indices $(1,1)$ and $A$ its self-adjoint extension.
Given  a dissipative quasi-selfadjoint extension   $\widehat A$ of $\dot A$, we also introduce a  characteristic function associated with the triple
 $(\dot A, \widehat A, A)$ and
provide a characterization of  the projection of the  deficiency subspace  $\Ker ((\dot A)^*-zI)$  onto
the subspace $\Ker ((\dot A)^*-\overline{z}I)$
along the domain of the dissipative operator $\widehat A$  in terms of the characteristic function of the triple (see Proposition \ref{NeuLif}).


 $\bullet$ In Section \ref{s3} we  study  symmetric operators with deficiency indices $(1,1)$ that satisfy the semi-Weyl commutation relations \eqref{1/2Wdot} (see Hypothesis \ref{muhly}).
 We  show  that if
$\dot A$ admits a self-adjoint extension $A$ that solves  the same commutation relations as the operator $\dot A$ does, then
  the Liv\v{s}ic function associated with the pair $(\dot A, A)$ is identically zero in the upper half-plane (see Theorem \ref{sacr}). In this case the corresponding Weyl-Titchmarsh  function is $z$-independent and coincides with $i=\sqrt{-1}$ on the whole upper half-plane.
On the other hand, if $\dot A$
  has no self-adjoint extension  that solves  the same commutation relations but does have  a  dissipative extension solving \eqref{mmm}, then Theorem  \ref{lem:diss} asserts that the  characteristic function
  of the corresponding triple is periodic with a real period and has  a particularly simple form. Notice that in the case where $\dot A$ admits both a self-adjoint $A$ and a dissipative extension $\widehat A $
that  solve \eqref{mmm}, the characteristic function of the corresponding triple $(\dot A, \widehat A, A)$  is a constant from the open unit disk.

  As a corollary to Theorems    \ref{sacr} and \ref{lem:diss}  one gets an explicit representation
   for the    Liv\v{s}ic function
 of  an arbitrary symmetric operator satisfying  Hypothesis \ref{muhly} (see Corollary \ref{ssss} for a precise statement).

 $\bullet$ In Section \ref{s4intro} we discuss first order symmetric differential operators on a metric graph $\bbY$ assuming that the graph  is in one of the following three cases:   the graph  (i)
is the real axis with a reference point,  (ii)
 is a finite interval,  and finally, (iii)
is  obtained by attaching a finite interval to the real axis.
Notice that the symmetric operators in question
 solve  the semi-Weyl commutation relations  \eqref{1/2Wdot} (see Remark \ref{remdotd}).

$\bullet$  In Section \ref{s5} we give a useful  parameterization  for the family  of all self-adjoint  as well as quasi-selfadjoint extensions of a symmetric differentiation discussed in Section 4. We also  show   that any  self-adjoint realization $D_\Theta$  of the symmetric differentiation operator  $\dot D$ on the graph $\bbY$ in Case (iii) serves as a minimal self-adjoint dilation of an appropriate  quasi-selfadjoint dissipative differentiation operator  on the metric graph in Case (ii) (see Theorem \ref{dilthm}).

 $\bullet$ In Section \ref{s6}  we  examine   the   Liv\v{s}ic functions  associated with the pair $(\dot D, D_\Theta)$ where
 $\dot D$ is the symmetric operator on the metric  graph $\bbY$ (in Cases (i)-(iii)) and $  D_\Theta$ is its  arbitrary reference self-adjoint extension.
In particular, we provide a complete   solution   of Problem (I) $a)$ (see Corollary \ref{muhly1}).

 $\bullet$  Section \ref{s7} is devoted to the comprehensive  study of the corresponding Weyl-Titchmarsh functions.

  $\bullet$ In  Section \ref{secmodop}, given  the graph $\bbY$   in one of the  Cases (i)-(iii),  we introduce maximal  dissipative differentiation (model) operators
  the
     domains  of which are  invariant with respect to the group  of gauge transformations.
Those are  the prototypes of general  dissipative solutions to the commutation relations \eqref{mmm}. Notice that in  Cases (i) and (iii) the corresponding boundary conditions are determined not only by the geometry of the metric  graph $\bbY$ but also by its peculiar  ``conductivity" exponent $k$, $0\le k<1$, which we call the {\it quantum gate coefficient}.
     We also  give  an explicit informal description of  the associated  contraction  semi-groups generated by these operators. A visualization of the corresponding
     dissipative dynamics is shown on  Figures {\bf 1--3}.

  $\bullet$  In Section \ref{s9} we evaluate the characteristic function of the triples   associated with the model dissipative operators
discussed in Section \ref{secmodop}
  and  prove the converse of  important structure Theorem \ref{lem:diss} (see Theorem \ref{thmconv}).

 $\bullet$   In Section \ref{transmission}, following the line of research initiated  by  Liv\v{s}ic
   \cite{LJETP}, who was the first to discover the connection between the Heisenberg scattering matrix and the characteristic function of a dissipative operator (also see  \cite{AA,BL60,Lax}),  we focus on the discussion  of general quasi-selfadjoint dissipative differentiation operators  $\widehat d$ on the metric graph $\bbY$ in Case (ii).
In particular, we relate  the characteristic function of the corresponding triple   with the reciprocal of the transmission coefficient in a scattering problem for a self-adjoint dilation
of $\widehat d$, the magnetic Hamiltonian. Using the fundamental relation \eqref{chchch} in  Appendix C between the characteristic  and  Liv\v{s}ic  functions,  we also get
  a representation for the transmission coefficient via  the Liv\v{s}ic function   and the von  Neumann parameter of the triple (see Corollary \ref{transs} combined with eq. \eqref{slif}). Notice that both the differentiation operator  $\widehat d$ and the magnetic Hamiltonian solve commutation relations \ref{mmm} with respect to a discrete subgroup of the unitary group $U_t$.

      $\bullet$  In Section 11  we discuss  uniqueness results for symmetric operators that commute  with a unitary. In particular,  we show that the unitary group $U_t$ from \eqref{mmm}
 is uniquely determined (up to a character) unless the dissipative solution to the commutation relations has point spectrum filling in the whole upper half-plane. In the latter (exceptional)  case the representation $t \mapsto U_t$ is reducible  and  splits  into the orthogonal sum of  two irreducible  representations uniquely  determined up to a unitary character.

      $\bullet$  In Section 12  we   characterize  maximal dissipative solutions to the  semi-Weyl commutation relations \eqref{mmm}
        that extend a prime symmetric operator satisfying Hypothesis \ref{muhly}  (see Theorem \ref{premain}). As a by-product of these considerations, we provide
        an intrinsic characterization of these symmetric operators thus solving  the J$\clock$rgensen-Muhly problem in the particular case of the deficiency indices $(1,1)$.

    $\bullet$    In Section 13  we  study  solutions to  the restricted Weyl commutation relations \eqref{resweyl}
for     a unitary group $U_t=e^{iBt}$ and  a  semi-group $\widehat V_s=e^{i\widehat A s}$ of contractions. Under the assumption that
  the generator $\widehat A$ of   $\widehat V_s$ is a quasi-selfadjoint extension of a prime symmetric operator with deficiency indices $(1,1)$
   we  characterize  the pairs of generators $(\widehat A, B)$ up to mutual unitary equivalence. In particular, we show that the generators $\widehat A$  and $B$ can be realized
as the dissipative differentiation operator $\widehat \cP$ on the metric graph $\bbY$  in one of the Cases I$^*$, I--III (see Definition \ref{canonical} for the classification) with appropriate boundary conditions at the vertex  of the graph
    and  the multiplication operator $\cQ$ on the graph, respectively (see Theorem \ref{vNvN}).  In contrast to the classic  Stone-von Neumann uniqueness result, the pairs
     $(\widehat \cP, \cQ)$  are not unitarily equivalent for different choices of  the center of the graph. However, the
         uniqueness theorem in the self-adjoint case can be adopted to fit the format of its non-selfadjoint counterpart (see Theorem \ref{vNvNad}). In particular, we get
          the {\it full version} of the Stone-von Neumann Theorem (see Corollary \ref{bell}).

$\bullet$ In Section 14   we consider a family of  unitary
        solutions  to the canonical  Weyl commutation relations  \eqref{weylccr} on the full metric graph
      $\bbX$  obtained by composing two identical copies of the metric graph $(-\infty, \mu)\sqcup(\mu, \infty)$, $\mu\in \bbR $ is a parameter.
       The corresponding
         self-adjoint momentum differentiation operator $\cP$, the generator of the group $V_s$ of shifts,  is determined by the boundary conditions       at the vertex $\mu$ with  the  bond $S$-matrix   given by
            $$
        S=\begin{pmatrix}
        k&-\sqrt{1-k^2}\\ \sqrt{1-k^2}&k
        \end{pmatrix} \quad \text{  for some $0\le k< 1$,}
        $$
        and the generator $\cQ$  of the second group $U_t$ is just the position operator on the graph $\bbX$.
         As long as the generator $\widehat A$ of   $\widehat V_s$ is a quasi-selfadjoint extension of a prime symmetric operator with deficiency indices $(1,1)$,    the   structure  Theorem \ref{cbcb} shows that up to mutual unitary equivalence any  solution to the restricted Weyl  commutation relations \eqref{resweyl} can be obtained by an appropriate  compression of the unitary groups $U_t$ and $V_s$  onto  some subspace $\cK$ of  the Hilbert space $L^2(\bbX)$.
    In fact, the  subspace $\cK$ coincides with  the Hilbert space $L^2(\bbY)$ where  $\bbY\subset \bbX$ is a subgraph of $\bbX$ in one of the three canonical cases discussed above. Notice that the subspace $\cK$ reduces the multiplication group $U_t=e^{it\cQ}$ and is   coinvariant for the group  of shifts $V_s=e^{is\cP}$.
         For a pictorial description of the corresponding  unitary dynamics on the full metric graph $\bbX$ we refer to Figure {\bf  4}.

    In the  second part of the book we discuss applications   to  decay phenomena  in quantum systems theory.

  $\bullet$  In Section 15 we recall the concept of continuous monitoring of a quantum system and describe  possible   {\it ex-post}   monitoring scenarios:
 the Quantum Zeno  and  Anti-Zeno effects as well as  the Exponential Decay phenomenon  in frequent measurements theory.

$\bullet$ In Section 16 we discuss the quantum Zeno versus  Anti-Zeno effect alternative for massive particles (see Theorems \ref{1/2thm}  and \ref{3/2thm}).

 $\bullet$  In Section 17 we examine  the exponential decay phenomenon for ``zero-mass"  systems where  the exponential decay typically alternates with the  quantum Zeno scenario (see Theorem \ref{decr}).

 $\bullet$  In Section 18 we recall main concepts
  for the  ``exclusive'' versus ``interference" alternatives theory going  back
to  the celebrated two slit experiment  in quantum mechanics.

 $\bullet$ In Section 19 we consider a model of a quantum system on a ring that describes a motion of a (relativistic) massless particle and set the stage for monitoring such systems within the continuous observation paradigm.

 $\bullet$ In Section 20 we show that  in some cases continuous monitoring of  the model quantum system
triggers emission of particle, see Theorem \ref{quanmon}. For instance, this phenomenon occurs   if the initial state has a unique jump-point discontinuity on the ring.
The magnitude of the state  decay in this case can be theoretically predicted as if the quantum particle were a wave.  That  is, the particle interferes with itself  at the point of observations (where the wave functions has a jump) and  the results of this interference
can informally be explained  in the framework of  the ``interference" alternatives theory.
We  also
  illustrate some features   in  the state decay  under continuous monitoring of a massless particle moving along  the   Aharonov-Bohm ring (see Corollary \ref{bohmcor} and eq. \eqref{intell}, cf. \cite{KosSch,Kur}).

 $\bullet$ In Section 21 we  discuss results of continuous monitoring of  open quantum systems on a ring under the hypothesis that the  time evolution of the system is governed by a semi-group of contractions. If the initial state of the system  satisfies the radiation condition \eqref{instr}, that is, it belongs to the domain of the evolution generator,
then the decay rate  can easily be computed using purely classical considerations, as if the quantum particle were  a classical particle (see Theorem \ref{quanmondis}, Corollary \ref{smes0} and the related discussion).

 $\bullet$ In Section 22 we explicitly describe the self-adjoint dilation of the dissipative generator  of the open system on the ring.

   $\bullet$  In Section 23 we consider more general states of the open quantum system on the ring. The main result (see Theorem \ref{general})
states that the emission rate splits into the sum of two terms.
 One of the terms  is due to the interference of the particle with itself at the point of observations. The second source of emission  is caused by inelastic collision of the quantum particle
  with the point ``defect'' (membrane) at the observation point where
  the corpuscular nature of the quantum particle is fully manifested.

 $\bullet$  In Section 24 we introduce the concept of convergence of dissipative operators in distribution and  prove several limit theorems with respect to multiple  coupling    of an operator with itself. We also  show that  the generator of the nilpotent semi-group,  one of the building blocks  of the restricted  commutation relations theory,
 can be  obtained as the limit of  appropriately normalized $n$-fold couplings of almost  arbitrary  dissipative operator (see Theorem \ref{limthm}).

 The book has eight  appendices where,  for the reader's convenience, one can find
relevant  information  scattered in the literature. Some of the results  presented there are  new.

In Appendix A we  recall the notion of the Weyl-Titchmarsh functions as well as a characteristic function for rank-one dissipative perturbations of a self-adjoint operator. We also provide  the corresponding uniqueness result (see Theorem \ref{lifunique}).

In Appendix B  we collect necessary background material from the theory of symmetric  operators,

In Appendix C  we present  a functional model for
triples of operators following our work \cite{MT-S} (also see  
\cite[ Ch. 10, Sec. 10.4, p. 357]{ABT} and \cite{Kuzhel,MT,Ryzh2007,Ryzh2008,T88}).

Appendix D contains a discussion aimed at  the spectral analysis of  model dissipative  triple of operators.

 In Appendix E we study
the dependence of the Weyl-Titchmarsh, Liv\v{s}ic, and characteristic function under  affine transformation of the operators.

In Appendix F  we discuss the invariance principle for affine transformations of a dissipative operator.

In Appendix G  we recall the concept of an operator coupling of two  dissipative  operators  and
discuss
the corresponding multiplication theorem.

In Appendix H one can find a brief discussion of stable laws in probability theory and the formulation of the general Gnedenko-Kolmogorov limit theorem.

Some words about notation:

The domain of a linear operator $K$ is denoted by $\Dom(K)$, its range by $\Ran(K)$, and its kernel by $\Ker(K)$. The restriction
of $K$ to a given subset $\cC$ of $\Dom(K)$ is written as $K|_{\cC}$.
We write $\rho(K)$ for the resolvent set of a closed operator $K$ on a Hilbert space, and $K^*$ stands for the adjoint operator of
$K$ if $K$ is densely defined.

\textbf{Acknowledgment.}  We are very grateful to M. Ashbaugh and S. Belyi for the invaluable help in preparation of this paper for publication.
K.A.M. is indebted to A.~B.~Plachenov for stimulating discussions.

\section{Preliminaries and basic definitions}\label{s2}
Let  $\dot A$ be  a densely defined symmetric operator with deficiency indices $(1,1)$ and   $ A$ its self-adjoint (reference) extension.

Following \cite{D,GT,L46,MT-S}  recall the concept of the Weyl-Titchmarsh and   Liv\v{s}ic functions associated with the pair $(\dot A, A)$.

Suppose that (normalized) deficiency elements $g_\pm\in \Ker( (\dot A)^*\mp iI)$,
 $\|g_\pm\|=1$, are chosen in such a way that
 \begin{equation}\label{rss}
g_+-g_-\in \dom (A).
\end{equation}

Consider   the {\it Weyl-Titchmarsh function }
\begin{equation}\label{WTF}
M(z)=
((Az+I)(A-zI)^{-1}g_+,g_+), \quad z\in \bbC_+,
\end{equation}
and     the  {\it  Liv\v{s}ic function}
\begin{equation}\label{charsum}
s(z)=\frac{z-i}{z+i}\cdot \frac{(g_z, g_-)}{(g_z, g_+)}, \quad
z\in \bbC_+,
\end{equation}
 associated with the pair $(\dot A, A)$.
Here $g_z\in \Ker( (\dot A)^*-z I)$, $z\in \bbC_+$.

 Clearly,   the Weyl-Titchmarsh function $M(z)$  does not depend on the concrete choice of the normalized deficiency element $g_+$.
 We also remark  that if
 $$g_\pm'\in \Ker(( \dot A)^*\mp iI), \quad
 \|g_\pm'\|=1,$$ is any deficiency elements  such that
 $$
g_+'-g_-'\in \dom (A),
$$
then necessarily $g_\pm'=\Theta g_\pm$ for some unimodular factor $\Theta$. Therefore, from \eqref{charsum} it follows that
 the  Liv\v{s}ic function does not depend on the choice of the deficiency elements  $g_\pm$ (whenever
 \eqref{rss} holds). However it may and in most of the cases does depend on the reference operator $A$.
As far as the  Weyl-Titchmarsh functions are  concerned we also refer to the related concept of a $Q$-function introduced in \cite{KL}
and discussed in \cite{KO}.

 Recall the  important  relationship between    the Weyl-Titchmarsh and  Liv\v{s}ic functions
  \cite {MT-S}
 \begin{equation}\label{blog}
s(z)=\frac{M(z)-i}{M(z)+i},\quad z\in \bbC_+.
\end{equation}

Next,  suppose  that $\widehat A \ne (\widehat A )^*$ is a  maximal dissipative extension of $\dot A$,
$$
\Im(\widehat A f,f)\ge 0, \quad f\in \dom(\widehat A ).
$$
Since $\dot A$ is symmetric, its dissipative extension $\widehat A$ is automatically quasi-selfadjoint  \cite{Phil,St68} (also see \cite{ABT,MT-S}),
that  is,
$$
\dot A \subset \widehat  A \subset (\dot A)^*,
$$
and hence,
\begin{equation}\label{parpar}
g_+-\varkappa g_-\in \dom
 (\widehat A  )\quad \text{for some }
|\varkappa|<1.
\end{equation}

 By definition, we call $\varkappa$ {\it the von Neumann parameter of the triple}  $(\dot A,  \widehat A,A)$.

\begin{remark}\label{vonbas}
Likewise, one can think of  the von Neumann parameter $\varkappa$ being determined by the dissipative operator $  \widehat A$ and the pair 
$\{g_+,g_-\}$ of normalized deficiency elements $g_\pm\in \Ker( (\dot A)^*\mp iI)$.
Indeed, given $  \widehat A$ and  $\{g_+,g_-\}$ there are a unique $\varkappa$ satisfying \eqref{parpar}
and  a unique self-adjoint reference extension $A$ of $\dot A$ such that \eqref{rss} holds.
Therefore,   the triple  $( \dot A, \widehat A , A)$  is uniquely  determined by the knowledge of  $  \widehat A$ and  $\{g_+,g_-\}$.
 Cleary, $\varkappa$ coincides with 
the von Neumann parameter of the triple $\varkappa_{( \dot A, \widehat A , A)}$ which proves the claim.

\end{remark}

Given  \eqref{rss} and \eqref{parpar},  consider
{\it the characteristic function} $S(z)=S_{(\dot A, \widehat A, A)}(z)$ associated with   the triple  $(\dot A,  \widehat A,A)$  (see \cite{MT-S}, cf. \cite{Lv1})
\begin{equation}\label{ch12}
S(z)=\frac{s(z)-\varkappa} {\overline{ \varkappa }\,s(z)-1}, \quad z\in \bbC_+,
\end{equation}
where  $s(z)=s_{(\dot A, A)}(z)$ is the Liv\v{s}ic function associated with the pair
$(\dot A, A)$.

By  \eqref{blog} and  \eqref{ch12}, one also gets   the  representation for the characteristic function via the Weyl-Titchmarch function as
\begin{equation}\label{charweyl}
S (z )=-\frac{1-\varkappa}{1-\overline{\varkappa}}\cdot
\frac{M\left (z\right )-i\frac{1+\varkappa}{1-\varkappa}}{M\left (z\right )+i\frac{1+\overline{\varkappa}}
{1-\overline{\varkappa}}}.
\end{equation}

We remark that given  a triple  $( \dot A, \widehat A , A)$,
one can always find a
basis $g_\pm$ in the subspace
 $\Ker (\dot A^*-iI)\dot +\Ker (\dot A^*+iI)$ such that
$$\|g_\pm\|=1, \quad g_\pm\in \Ker((\dot A)^*\mp iI),
$$
$$
g_+-g_-\in \dom (A)
$$
and
$$
g_+-\varkappa g_-\in \dom (\widehat A ) \quad \text{for some }
|\varkappa|<1.
$$

In this case, the von Neumann parameter $\varkappa$ can explicitly be evaluated in terms of the characteristic function of the triple  $( \dot A, \widehat A , A)$ as
\begin{equation}\label{sac}
\varkappa =S_{( \dot A, \widehat A , A)}(i).
\end{equation}
Hence, as it follows from \eqref{ch12},  the Liv\v{s}ic function associated with the pair
$(\dot A, A)$  admits the representation
\begin{equation}\label{obr}
s_{(\dot A, A)}(z)=\frac{S(z)-\varkappa} {\overline{ \varkappa }\,S(z)-1}, \quad z\in \bbC_+.
\end{equation}

In particular,
$$
S_{( \dot A, \widehat A , A)}(z)=-s_{(\dot A, A)}(z), \quad z\in \bbC_+,
$$
whenever the von Neumann parameter $\varkappa$ of  the triple $( \dot A, \widehat A , A)$ vanishes.

The following proposition provides a curious characterization for  the projection of the  deficiency subspace  $\Ker ((\dot A)^*-zI)$  onto
the subspace $\Ker ((\dot A)^*-\overline{z}I)$
along  $\Dom(\widehat A)$ in terms of the characteristic function of the triple.

\begin{proposition}\label{NeuLif}{\it  Let  $( \dot A, \widehat A , A)$ be a triple. Suppose that the deficiency elements $g_z\in \Ker((\dot A)^*-zI)$ and  $g_{\overline{z}}\in \Ker((\dot A)^*-\overline{z}I)$,
$\|g_z\|=\|g_{\overline{z}}\|\ne 0$ are chosen in such a way that
\begin{equation}\label{gammalog}
g_z-\gamma(z)g_{\overline{z}}\in \Dom (\widehat A), \quad z\in \bbC_+.
\end{equation}
Then
$$
|\gamma(z)|=|S(z)|, \quad z\in \bbC_+,
$$
where $S(z)$ is the characteristic function of the triple $(\dot A, \widehat A, A)$.}
\end{proposition}
\begin{proof} By Theorem \ref{unitar} in Appendix C,  the triple  $( \dot A, \widehat A , A)$ is mutually unitarily equivalent to  the model triple $(\dot \cB, \widehat \cB,\cB)$ in the Hilbert space $L^2(\bbR; d\mu)$  given by \eqref{nacha1}, \eqref{nacha2} and \eqref{nacha3} in Appendix C. Here $\mu(d \lambda)$ is the measure from the representations
\begin{equation}\label{defi}
M(z) =\int_\bbR\left (\frac{1}{\lambda-z}-\frac{\lambda}{\lambda^2+1} \right) d\mu(\lambda), \quad z\in \bbC_+,
\end{equation}
for the Weyl-Titchmarsh function associated with the pair $(\dot A, A)$.

In view of Remark \ref{snoska} in Appendix C
it suffices to show that
if
\begin{equation}\label{requir}
\frac{1}{\lambda-z}-\alpha\frac{1}{\lambda-\overline{z}}\in \Dom (\widehat \cB),
\end{equation}
then
\begin{equation}\label{zaver}
\alpha=|S(z)|, \quad z\in \bbC_+.
\end{equation}

We claim that
$$
\frac{1}{\lambda-z}=\frac{\lambda}{\lambda^2+1}+\frac{M(z)}{\lambda^2+1}+f(\lambda),
$$
where $f\in \Dom (\dot \cB)$.

Indeed,
\begin{align*}
\int_\bbR f(\lambda)d\mu(\lambda)&=\int_\bbR\left (\frac{1}{\lambda-z}-\frac{\lambda}{\lambda^2+1}\right )d\mu(\lambda)-M(z)\int_\bbR\frac{d\mu(\lambda)}{\lambda^2+1}
\\&=
M(z)-M(z)=0, \quad z\in \bbC_+.
\end{align*}
Here we have used \eqref{defi} and the normalization condition \eqref{normmu} in Appendix C. From the characteriszation
 of the domain of the symmetric operator $\dot \cB$ given by \eqref{nacha2} in Appendix C, it follows that  $f\in \dom (\dot \cB)$, proving the claim.

Next, we have
\begin{equation}\label{ininin}
\frac{1}{\lambda-z}-\alpha\frac{1}{\lambda-\overline{z}}=\frac{(1-\alpha)\lambda}{\lambda^2+1}+\frac{M(z)-\alpha M(\overline{z})}{\lambda^2+1}+h(\lambda),
\end{equation}
where $h\in \Dom (\dot \cB)$.

Recall that by \eqref{nacha3} in Appendix C,
$$
\frac{1}{\lambda-i}-\varkappa\frac{1}{\lambda+i}\in \Dom (\widehat \cB),
$$
 where $\varkappa $ is the von Neumann parameter of the triple   $(\dot \cB, \widehat \cB, \cB)$. Since the triples $(\dot A, \widehat A, A)$
 and $(\dot \cB, \widehat \cB, \cB)$ are mutually unitarily equivalent,  $\varkappa $ coincides with  von Neumann parameter of the triple
 $(\dot A, \widehat A, A)$ as well.

Since
\begin{align*}
\frac{(1-\alpha)\lambda}{\lambda^2+1}+\frac{M(z)-\alpha M(\overline{z})}{\lambda^2+1}&=\frac12 \left[ (1-\alpha)-i(M(z)-\alpha M(\overline{z})\right ]\frac1{\lambda-i}
\\&+\frac12 \left[ (1-\alpha)-i(M(z)-\alpha M(\overline{z})\right ]\frac1{\lambda+i},
\end{align*}
in view of \eqref{ininin}, the requirement \eqref{requir} connecting the characteristic functions $S(z)$ of the triple and the Weyl-Tichmarsch function yields the relation
$$
-\varkappa =\frac{(1-\alpha)+i(M(z)-\alpha M(\overline{z})}{(1-\alpha)-i(M(z)-\alpha M(\overline{z})}, \quad z\in \bbC_+.
$$
Hence,
$$
\alpha=\frac{F(z)}{F(\overline{z})},
$$
where
$$F(z)=1+\varkappa +iM(z)(1-\varkappa).$$

On the other hand, from the relation  \eqref{charweyl} it follows that
$$
S(z)=\frac{F(z)}{\overline{F(\overline{z})}}
$$
and therefore
$$|\alpha|=|S(z)|, \quad z\in \bbC_+.
$$
which proves \eqref{zaver} and completes the proof of the proposition.
\end{proof}

\begin{remark}Notice that given the deficiency elements $\|g_z\|=\|g_{\overline{z}}\|\ne 0$, $z\in \bbC_+$, one can always find $\gamma(z)$ such that
\eqref{gammalog} holds. Indeed, from the definition of the quasi-selfadjoint extension $\widehat A$ it follows that one can find $\alpha$ and $\beta$ such that
$$
0\ne \alpha g_z+\beta g_{\overline{z}}\in \dom (\widehat A).
$$
Therefore, it suffices to show that $\alpha\ne 0$. Otherwise,  $ g_{\overline{z}}\in \dom (\widehat A)$ and hence
$$
\Im (\widehat A g_{\overline{z}},g_{\overline{z}})=\Im ((\dot A )^*g_{\overline{z}},g_{\overline{z}})=
\Im (\overline{z}g_{\overline{z}},g_{\overline{z}})
=
\Im(\overline{z})\|g_{\overline{z}}\|^2<0,
$$
which conradicts the requirement that $\widehat A$ is a dissipative operator.
\end{remark}

\begin{definition}\label{logpotdef}
We call  the  harmonic function
\begin{equation}\label{logpot}
\Gamma_{\widehat A}(z)=\log |\gamma(z)|=\log |S(z)|,  \quad  z\in \bbC_+\setminus\{ z_k\, |\, S(z_k)=0\},
\end{equation}
the {\it von Neumann $($logarithmic$)$ potential } of the dissipative operator $\widehat A$. Here $\gamma(z) $ is the function referred to in Proposition \ref{NeuLif}.
\end{definition}

\begin{remark}\label{mnogopros}
We remark  that under the assumption  that the symmetric operator $\dot A $ is prime,
both
 the Weyl-Titchmarsh function $M(z)$ and  the Liv\v{s}ic function $s(z)$ are
 complete unitary invariants  of the  pair $(\dot A, A)$.
Moreover, in this case,
the characteristic function $S(z)$   is a complete unitary invariant of the triple
$(\dot A, \widehat A, A )$  (see  \cite{MT-S}, also see Theorem \ref{unitar} in Appendix C).
Notice that if the  symmetric operator  $\dot A$ from a triple
$(\dot A, \widehat A, A )$
 in the Hilbert space $\cH$ is not prime and  $\dot A'$ is  its prime part  in a subspace $\cH'\subset \cH$, then the triples
$(\dot A, \widehat A , A)$ and $(\dot A|_{\cH'}, \widehat A|_{\cH'} ,  A|_{\cH'})$ have the same characteristic function (see Theorem \ref{prostota0} in Appendix B for details).

As it follows from    Lemma \ref{ind0}  and Proposition \ref{concept} in Appendix E, the  absolute value of the   Liv\v{s}ic function $|s(z)|$
is a complete unitary invarinat of the symmetric operator $\dot A$ while
$|S(z)|$ is a (complete)  unitary invariant
of the maximal dissipative  operator $\widehat A$.
In particular, in view of \eqref{sac},
the von Neumann parameter $\varkappa_{(\dot A, \widehat A,A )}$ is a unitary invariant of the  triple
$(\dot A, \widehat A,A )$, while  its absolute value
 \begin{equation}\label{defmoduli}\widehat \kappa (\widehat A)=|\varkappa_{(\dot A, \widehat A,A )}|
\end{equation}
is a well defined unitary invariant of the dissipative operator $\widehat A$.

We  also refer  to \cite{L46} where it was shown that the knowledge
or $s(z)$ and $S(z)$ (up to a unimodular constant factor), equivalently, $|s(z)|$ and $|S(z)|$ characterizes  the symmetric operator $\dot A$ and is  maximal dissipative extension $\widehat A$, respectively,
up to  unitary equivalence (whenever  $\dot A$ is a prime operator).

We also notice that the knowledge of  the  von Neumann logarithmic  potential $\Gamma_{\widehat A}(z)$  determines the dissipative operator $\widehat A$.
up to unitary equivalence.
\end{remark}


\section{The commutation relations   and character-automorphic functions }\label{s3}

Throughout this section we assume the following hypothesis.

\begin{hypothesis}\label{muhly} Assume that $ \dot A$ is  a symmetric operator  with deficiency indices $(1,1)$  and  $\bbR\ni t \mapsto U_t$
a strongly continuous unitary group.  Suppose, in addition,   that  $\dom (\dot A)$ is $U_t$-invariant and
the commutation relations
\begin{equation}\label{ccrmuh}
U_t^*\dot AU_t=\dot A+tI\quad \text{on }\quad \Dom (\dot A), \quad t\in \bbR,
\end{equation}
hold.

\end{hypothesis}

It is well known that under Hypothesis \ref{muhly} the symmetric operator $\dot A$  has  either a) a self-adjoint $A$ or/and b)  a  non-selfadjoint maximal dissipative  extension
$\widehat A$
satisfying the same commutation relations (cf. \cite[Theorem 15]{J80}, also see \cite[Theorem 5.4]{TMF}).

It is worth mentioning that  the existence of a quasi-selfadjoint extension of $\dot A$ that solves \eqref{ccrmuh}
required in Hypothesis \ref{muhly} is  a consequence  of   the  Lefschetz fixed point theorem for flows on manifolds:

\begin{proposition}[{see, e.g., \cite[Theorem 6.28]{Vick}}]\label{Lef}
{\it If $\cM$ is a closed oriented manifold such that the Euler characteristic $\chi (\cM)$ of $\cM$ is not zero, then any flow on $\cM$ has a fixed point.
}\end{proposition}

Indeed,  if $\dot A$ satisfies Hypothesis \ref{muhly} and $\widehat A$ is a maximal dissipative extension of $\dot A$, then $\widehat A_t=U_t\widehat AU_t^*$  also extends $\dot A$.  Since the set of maximal dissipative extensions of $\dot A$ is in one-to-one correspondence  with the closed unit disk $\overline{\bbD}$, and
 \begin{equation}\label{farfor}
\widehat A\mapsto \widehat A_t
\end{equation}
determines a flow $\varphi(t, \cdot)$ on $\overline{\bbD}$ (the continuity of the flow can easily be established (see,  e.g., \cite{TMF})).
 By the  Lefschetz theorem the flow   $\varphi(t, \cdot)$  has a fixed point either
on the boundary of the unit disk or in its interior.

First, consider the case where  the flow   $\varphi(t, \cdot)$  has a fixed point
on the boundary of the unit disk
and therefore the commutation relations \eqref{ccrmuh} have a self-adjoint solution.

\begin{theorem}\label{sacr} Assume Hypothesis \ref{muhly}.
Suppose that  $ A$ is a self-adjoint   extension of the symmetric operator   $\dot A$  such that
 \begin{equation}\label{nucom}
U_t^* AU_t=A+tI\quad \text{on }\quad \Dom ( A).
\end{equation}

Then the Weyl-Titchmarsh function $M(z)$ of the pair $(\dot A, A)$  has the form
$$
M(z)=i, \quad z\in \bbC_+.
$$

Equivalently, the Liv\v{s}ic  function  $s(z)$   associated with the pair   $(\dot A, A)$  vanishes identically in the upper half-plane,$$
s(z)=0, \quad z\in \bbC_+.
$$

\end{theorem}
\begin{proof} Introducing  the family of bounded operators
$$
B_t=U_t(A-iI)(A-iI+tI)^{-1}, \quad t\in \bbR,
$$
it is easy to see that the family $\bbR\ni t\mapsto B_t$ forms a  strongly continuous (commutative) group. Indeed, using the commutation relation \eqref{nucom} for the self-adjoint operator $A$ one obtains
\begin{align*}
B_{t}B_s&=U_t(A-iI)(A+(-i+t)I)^{-1}U_s(A-iI)(A+(-i+s)I)^{-1}\\
&=U_tU_sU_s^*(A-iI)(A-iI+tI)^{-1}U_s(A-iI)(A+(-i+s)I)^{-1}\\
&=U_{t+s}(A-iI+sI)(A-iI+(t+s)I)^{-1}(A-iI)(A-iI+sI)^{-1}\\
&=U_{t+s}(A-iI+(t+s)I)^{-1}(A-iI)
\\&=B_{t+s}, \quad s,t\in \bbR.
\end{align*}

Let $g_+\in \Ker((\dot A)^*-iI)$,  $\|g_+\|=1$, be  a normalized deficiency element of $\dot A$.
Since
$$
(A-iI)(A-(i-t)I)^{-1}g_+\in \Ker((\dot A)^*-(i-t)I)
$$
and
$$
U_t \Ker((\dot A)^*-(i-t)I)= \Ker((\dot A)^*-iI),
$$
one concludes that
the deficiency subspace   $\Ker((\dot A)^*-iI)$ is invariant for $B_t$, $t\in \bbR$. Therefore,
the restriction $B_t$ on the deficiency subspace is a continuous one-dimensional representation
(a one-dimensional representation of a strongly continuous group is continuous).
Hence,
$$
B_tg_+=b^tg_+\quad \text{ for some }\quad b\in \bbC.
$$
From the definition of the Weyl-Titchmarsh function \eqref{WTF} it follows that
$$
M'(z):=\frac{d}{dz}M(z)=((A^2+I)(A-zI)^{-2}g_+,g_+), \quad z\in \bbC_+.
$$

One computes
\begin{align*}
|b|^{2t}&M'(i-t)=|b|^{2t}((A^2+I)(A-iI+tI)^{-2}g_+, g_+)
\\&=((A^2+I)(A-iI+tI)^{-2}B_tg_+, B_tg_+)
\\&=(B_t^*(A^2+I)(A-iI+tI)^{-2}B_tg_+, g_+)
\\&=((A+iI)(A+iI+tI)^{-1}U_t^*(A^2+I)(A-iI-tI)^{-2}B_tg_+, g_+)
\\&=((A+iI)(A+iI+tI)^{-1}U_t^*(A^2+I)(A-iI-tI)^{-2}U_t
\\&\quad\times(A-iI)(A-iI+tI)^{-1}g_+, g_+)
\\&=((A+iI)(A+iI+tI)^{-1}((A+tI)^2+I)(A-iI)^{-2}(A-iI)
\\&\quad\times(A-iI+tI)^{-1}g_+, g_+)
\\&=((A+iI)(A-iI)^{-1}g_+, g_+)
\\&=(g_+, (A-iI)(A+iI)^{-1}g_+).
\end{align*}
Hence,
\begin{equation}\label{prima}
M'(i-t)=|b|^{-2t}(g_+,g_-), \quad t\in \bbR,
\end{equation}
where
$$
g_-=(A-iI)(A+iI)^{-1}g_+ \in  \Ker ((\dot A)^*+iI).
$$

Denote by $\mu(d\lambda)$  the spectral measure of the element $g_+$,
that is,
$$
\mu(d\lambda)=(E_A(d\lambda) g_+,g_+),$$
where $E_A(d\lambda)$ is the projection-valued spectral measure of the self-adjoint operator $A$
from the spectral decomposition
$$
A=\int_\bbR  \lambda dE_A(\lambda).
$$

Since
$$
M(z)=\int_\bbR  \frac{\lambda z+1}{\lambda -z}d\mu(\lambda),
$$
and therefore
$$
M'(i-t)=\int_\bbR \frac{\lambda^2+1}{(\lambda -t-i)^{2}}d\mu(\lambda),
$$
 one gets the estimate
\begin{align*}
|M'(i-t)|&\le \int_\bbR \frac{\lambda^2+1}{(\lambda +t)^2+1}d\mu(\lambda)\\
&=\int_{\{|\lambda|\le2 |t|\}}\frac{\lambda^2+1}{(\lambda +t)^2+1}d\mu(\lambda)
+\int _{\{|\lambda|> 2|t|\}}\frac{\lambda^2+1}{(\lambda +t)^2+1}d\mu(\lambda)
\\&\le (4 t^2+1)\mu\{|\lambda|\le |t|\}+\left (\frac{1}{(1-\frac12)^2}+1\right )\mu\{|\lambda |> |t|\}.
\end{align*}
Therefore,
 \begin{equation}\label{veill}
 M'(i-t)=O(t^2)\quad \text{as}\quad  t\to \infty.
 \end{equation}
Combining \eqref{veill} with \eqref{prima} shows that either  $(g_+,g_-)=0$ or
$|b|=1$.

In the first  case,  i.e. $(g_+,g_-)=0$, $M(z)$ is a constant function and hence
$$
M(z)=M(i)=i, \quad z\in \bbC.
$$

If $|b|=1$, we have $$
M'(i-t)=(g_+,g_-), \quad t\in \bbR.
$$
In particular, $M'(z)=(g_+,g_-)$ for all $z\in \bbC_+$ and hence
$$
M(z)=(g_+,g_-)z+C, \quad z\in \bbC_+,
$$
for some constant $C$.
We have, see \cite{KacKr}, $$
\lim_{y\to \infty}\frac{M(iy)}{y}=0,
$$
which implies
$
(g_+,g_-)=0
$
and  again shows that $M(z)$  is a constant function  in the upper half-plane and
$$
M(z)=C=M(i)=i, \quad z\in \bbC_+,
$$
which completes the proof.

\end{proof}

\begin{remark} In connection with our main hypothesis of this section,
we remark that if a self-adjoint operator  $A$ solves  commutation relations  \eqref{nucom}, then one can always find
 a symmetric restriction $\dot A$ that satisfies Hypothesis \ref{muhly}. Indeed, the commutation relations \eqref{nucom} imply that the one-parameter group $V_s=e^{isA}$ generated by $A$
satisfies  commutation  relations in the Weyl form (see, e.g., \cite[Ch. 3, Sect. 1, Theorem 5]{EM} or \cite{Si})
and then the existence of such a symmetric operator $\dot A$  is an immediate corollary of the Stone-von Neumann uniqueness theorem.
\end{remark}

Next, we threat the case where the  flow   $\varphi(t, \cdot)$  associated with the transformation \eqref{farfor} has a fixed point  in the  interior of the unit disk.
That is,  $\dot A$ admits  a maximal dissipative ``invariant'' extension that is not self-adjoint (cf. \cite[Theorem 20]{JM}).
We present the corresponding result in a slightly stronger  form. In particular, in this case the requirement \eqref{ccrmuh} can be relaxed.

\begin{theorem} \label{lem:diss}
Suppose that   $\widehat A$ is a  quasi-selfadjoint dissipative extension  of a closed symmetric operator $\dot A$ with deficiency indices $(1,1)$
 and $A$ is a $($reference$)$ self-adjoint   extension of $\dot A$.

Suppose that the commutation relation
\begin{equation}\label{order}
U_t^* \widehat AU_t=\widehat A+tI\quad \text{on }\quad \Dom ( \widehat A)
\end{equation}
hold.

Then
the characteristic function $S(z)$  associated with  the triple  $(\dot A, \widehat A, A)$ admits the representation
\begin{equation}\label{savto}
S(z)=k e^{i\ell z}, \quad z\in \bbC_+,
\end{equation}
for some  $|k|\le1$ and $\ell \ge 0$. Furthermore, if $\ell=0$, then necessarily $|k|<1$ and if $|k|=1$, then $\ell>0.$

In particular, the von Neumann parameter $\varkappa$ of the triple $(\dot A, \widehat A, A)  $ is given by
$$\varkappa= k e^{-\ell}.$$

\end{theorem}
\begin{proof}

Since $\widehat A$ is a quasi-selfadjoint extension of $\dot A$, we have
 \begin{equation}\label{ongin}
\dot A=\widehat A|_{ \dom (\widehat A)\cap \dom ((\widehat A)^*)}.
\end{equation}

 We claim that
 \begin{equation}\label{orderwidehat}
U_t^*(\widehat  A)^*U_t=(\widehat  A)^*+tI\quad \text{on }\quad \Dom ( (\widehat   A)^*).
\end{equation}
To see that assume that  $g\in\dom (\widehat A)$ and $f\in \dom((\widehat A)^*)$.
Then
\begin{align}
(\widehat A g,U_tf)&=(U_t^*\widehat A g,f)=(U_t^*\widehat AU_tU_t^* g,f)=((\widehat A+tI)U_t^* g,f)\label{forallg}
\\&=(U_t^* g, ((\widehat A)^*+tI)f)=( g, U_t((\widehat A)^*+tI)f)
.\nonumber
\end{align}
Here we have used that the domain $\dom (\widehat A)$ is $U_t$-invariant for  all $t\in \bbR$ and therefore $U_t^* g\in \dom (\widehat A)$.
Since \eqref{forallg} holds for all $g\in\dom (\widehat A)$, one ensures that  $\dom((\widehat A)^*)$ is $U_t$-invariant and
$$
(\widehat A)^*U_tf=U_t ((\widehat A)^*+tI)f,\quad f\in  \Dom ( (\widehat   A)^*),
$$
which proves \eqref{orderwidehat}.

Since $U_t( \Dom ( (\widehat   A)^*))= \Dom ( (\widehat   A)^*)$, from \eqref{ongin} one concludes that the commutation relations
 \begin{equation}\label{orderdot}
U_t^*(\dot  A)^*U_t=(\dot A)^*+tI\quad \text{on }\quad \Dom ( (\dot   A)^*).
\end{equation}
hold.

Taking into account that
$$\widehat A=U_t(\widehat A+tI)U_t^*
$$
and
$$
\dot  A=U_t(\dot  A+tI)U_t^* ,
$$
and also  observing  that the operator $A_t$ given by
$$A_t=U_t(A+tI)U_t^*$$ is a self-adjoint extension of $\dot A$,
we see that the triples $(\dot A, \widehat A, A_t)$ and $(\dot A+tI, \widehat A+tI, A+tI)$ are mutually unitarily equivalent.
In particular,
$$
S_{(\dot A+tI, \widehat  A+tI, A+tI)}(z)=S_{(\dot A,\widehat  A, A_t)}(z),\quad z\in \bbC_+.
$$
Since $A_t$ is a self-adjoint extension of $\dot A$,  by Lemma \ref{ind0}  (see \eqref{vagS}) in Appendix E, we have
$$
S_{(\dot A,\widehat  A, A_t)}(z)=\Theta_t^{(1)}S_{(\dot A,\widehat  A, A)}(z),\quad z\in \bbC_+,
$$
for some  unimodular constant $\Theta_t^{(1)}$, $| \Theta_t^{(1)}|=1$, which is a continuous function of the parameter $t$ (see \cite{TMF} for the proof of continuity).

By Theorem \ref{coinvv} in Appendix F,
$$
S_{(\dot A+tI, \widehat  A+tI, A+tI)}(z)= \Theta_t^{(2)} S_{(\dot A,\widehat  A, A)}(z-t),\quad z\in \bbC_+,
$$
where $\Theta_t^{(2)}$ is another  continuous unimodular function in $t$.
Therefore, the functional equation
$$
S_{(\dot A, \widehat  A, A)}(z-t)=\Theta_t S_{(\dot A, \widehat  A, A)}(z), \quad t\in \bbR,
$$
holds, where
$$
\Theta_t=\Theta_t^{(1)}\overline{\Theta_t^{(2)}}.
$$

 From the  functional equation it also follows that  $ \Theta_t$ is a continuous  unimodular solution of the equation
 $$
  \Theta_{t+s}= \Theta_t \Theta_s, \quad s, t \in \bbR,
 $$
 and therefore (see, e.g., \cite[XVII, 6]{F1})
 $$
 \Theta_t=e^{i\ell t} \quad \text{for some}\quad\ell  \in \bbR.
 $$

In particular, this proves that the characteristic function  $S_{(\dot A, \widehat  A, A)} $ is a character-automorphic function
with respect to the shifts, that is
\begin{equation}\label{funcav}
S_{(\dot A, \widehat  A, A)}(z+t)=e^{i\ell t}S_{(\dot A, \widehat  A, A)}(z),\quad t\in \bbR.
\end{equation}

Since $S_{(\dot A, \widehat  A, A)}(z)$ is a contractive analytic function on
$\bbC_+$, it admits the representation
$$
S_{(\dot A, \widehat  A, A)}(z)= \theta B(z)e^{iM(z)},
$$
where $|\theta|\le 1$,
$B$ is the Blaschke product associated with  the  (possible) zeros of $S_{(\dot A, \widehat  A, A)}(z)$ in the upper half-plane,
and $M(z)$ is a Herglotz-Nevanlinna function.

Suppose that the characteristic  function  $S_{(\dot A, \widehat  A, A)}(z)$ is not identically zero and thus $\theta\ne 0$. Then, from the functional equation \eqref{funcav} it follows that  $S_{(\dot A, \widehat  A, A)}(z)$ has no zeros in $\bbC_+$, and hence
$$
S_{(\dot A,\widehat  A, A)}(z)= \theta e^{iM(z)}, \quad z\in \bbC_+.
$$
Since $S_{(\dot A,\widehat  A, A)}(z)$ is character-automorphic, one concludes that the functional equation
$$
M(z+t)=\ell t+M(z)
$$ holds.

Next, we have  $$\frac{d}{dt} M(z+t)|_{t=0}=M'(z)=\ell.$$
Taking into account that   $M(z) $ maps the upper half-plane into itself, we obtain that
$$
M(z)=\ell z+b,
$$
where $\ell \ge 0$ and
 $\Im (b)\ge 0$.
 Therefore,
$$
S_{(\dot A,\widehat  A, A)}(z)= \theta e^{ib} e^{i\ell z},
$$
which proves \eqref{savto}  with
$
k=\theta e^{ib}.
$
\end{proof}

\begin{remark}\label{metmet}  Notice that in the situation of Therem \ref{lem:diss} we have that
$$
\dot A=\widehat A|_{ \dom (\widehat A)\cap \dom (\widehat A^*)}.
$$
Therefore
 the symmetric operator $\dot A$ is uniquely determined by the generator $\widehat A$.
In this case we will call $\dot A$ the symmetric part of $\widehat A$.

We also remark that if under the hypothesis of Theorem \ref{lem:diss} there is no self-adjoint extension satisfying the commutation relations \eqref{nucom},
the corresponding maximal dissipative extension $\widehat A$ is unique
 (see, e.g.,  \cite[Theorem 6.3]{TMF}).

\end{remark}

The following corollary is the first step towards a complete classification  up to  unitary equivalence of ``invariant" symmetric operators from Hypothesis \ref{muhly}
(see Problem (I) a) in Introduction).

\begin{corollary}\label{ssss} {\it Assume Hypothesis \ref{muhly}. Then there exists a self-adjoint extension $A$ of $\dot A$  such that  the Liv\v{s}ic function associated with the pair $(\dot A, A)$ admits the representation
$$
s_{(\dot A, A)}(z)=k \frac{ e^{iz\ell}-e^{-\ell}}{k^2e^{-\ell} e^{iz\ell}-1}, \quad z\in \bbC_+,
$$
for some $0\le k\le 1 $ and  $\ell> 0$.
}
\end{corollary}
\begin{proof}
If $\dot A$   admits a self-adjoint   extension  $A$ that satisfies the same commutation relations as $\dot A$ does,  then  by Theorem \ref{sacr}
$$
s_{(\dot A, A)}(z)=0, \quad z\in \bbC_+,
$$
which proves the claimed representation with $k=0$.

 If $\dot A$ admits a maximal (non-selfadjoint) dissipative extension  $\widehat A$ that satisfies the same commutation relations, then
by Theorem
 \ref{lem:diss} there exists a (reference) self-adjoint extension $A'$ of $\dot A$ such that
the characteristic function of the triple $(\dot A, \widehat A, A')$ is of the form
$$
S(z)=ke^{i\ell z},  \quad z\in \bbC_+.
$$
Therefore, by \eqref{obr}
$$
s_{(\dot A, A')}(z)=\frac{S(z)-S(i)} {\overline{ S(i) }\,S(z)-1}=k \frac{ e^{iz\ell}-e^{-\ell}}{|k|^2e^{-\ell} e^{iz\ell}-1}.
$$
By Lemma \ref{ind0} in Appendix E, one can always find  a possibly different self-adjoint extension $A$  of $\dot A$ such that
$$
s_{(\dot A,  A)}(z)=|k |\frac{ e^{iz\ell}-e^{-\ell}}{|k|^2e^{-\ell} e^{iz\ell}-1},  \quad z\in \bbC_+,
$$
and the claim follows.
\end{proof}


\section{The differentiation operator on  metric graphs}\label{s4intro}

Let $\bbY$ be a directed   metric graph (see, e.g., \cite{BK,KosSch}). We will distinguish   the following three cases.
\begin{itemize}

 \item[Case (i):] $$ \bbY=(-\infty, 0)\sqcup(0,\infty),$$ with  $(-\infty, 0)$ the incoming and  $(0,\infty)$ outgoing bonds;

\item[Case (ii):] $$ \bbY= (0,\ell), \quad \text{ the outgoing bond};$$

 \item[Case (iii):]  $$\bbY =(-\infty, 0)\sqcup(0,\infty)\sqcup(0,\ell),$$ with  $ (-\infty, 0)$ the incoming and  both  $ (0,\infty)$ and $(0,\ell)$ the outgoing bonds.
\end{itemize}




Denote by   $ \dot D =i\frac{d}{dx}$   the differentiation operator    on the metric graph $\bbY$ in Cases (i)-(iii)
 defined on the domain  $\Dom(\dot D)$ of  functions $f\in W_2^1(\bbY)$ with the following
 boundary conditions on the vertices of the graph,

 in Case (i):
 \begin{align}
 \quad f_\infty(0+)=f_\infty(0-)=0; \label{dotos}
\end{align}

 in Case (ii):
 \begin{align}
 \quad f_\ell(0)=f_\ell(\ell)=0; \label{dotint}
\end{align}

 in Case (iii):
 \begin{align} \begin{cases}
f_\infty(0+)&=k f_\infty(0-)
\\
f_\ell(0)&=\sqrt{1-k^2} f_\infty(0-)\\
 f_\ell(\ell)&=0
\end{cases} \quad \text{for some}\quad  0< k<1.  \label{dotdom0}
\end{align}

Here we have  used  the following notation.

If the graph $\bbY$ is in Cases (i) and (ii), the functions from the Hilbert space $L^2(\bbY)$
are  denoted by $f_\infty$
 and $f_\ell$,  respectively.

If the metric graph $\bbY$ is in Case (iii),  in view of the  natural identification of $L^2(\bbY)$ with the orthogonal  sum $L^2(\bbR)\oplus L^2((0, \ell))$,
 it is convenient to represent  an arbitrary element   $f\in L^2(\bbY)$  as
the two-component vector-function
$$
f=\begin{pmatrix} f_\infty\\
f_\ell
\end{pmatrix}
.$$
(Here $L^2(\bbY)$  denotes the Hilbert  space of square-integrable functions with respect to  Lebesgue measure  on the edges of the metric graph $
 \bbY$.)

Notice that if the graph $\bbY$ is  in  Case (iii)
and $k=0$  in \eqref{dotdom0},
then  the boundary conditions \eqref{dotdom0} can be rewritten as
 \begin{align*}
 \begin{cases}
f_\infty(0+)=0
\\
f_\ell(0)=f_\infty(0-)\\
 f_\ell(\ell)=0
\end{cases}.
\end{align*}
In this case, the operator
$\dot D$ splits into the orthogonal  sum of the symmetric  differentiation operators on the semi-axes $(-\infty, \ell)$ and $(0, \infty)$ with the Dirichlet boundary conditions at the
end-points, respectively.  Therefore, if $k=0$, then the operator  $\dot D$ is unitarily equivalent to the symmetric differentiation in Case (i).

\begin{remark}\label{betti}
In Cases (i) and (iii) the metric graph $\bbY$ is not finite. However,    one can assign two additional vertices to the external edges at $\pm \infty$.
Under this hypothesis,
 in all Cases (i)-(iii) the Euler characteristics $\chi (\bbY) $ of the graph $\bbY$, the number of vertices minus the number of edges, equals one.
Therefore, the corresponding first Betti number $\beta(\bbY)=-\chi (\bbY) +1$  of the graph $\bbY$,
 the number   of edges that have to be removed  to turn the graph into a connected tree,
vanishes.
\end{remark}

\begin{lemma}\label{domsopr} The operator $\dot D$  on a metric graph $\bbY$ is Cases $(i)-(iii)$ is symmetric.
Moreover,
$$
\Dom((\dot D)^*)=\begin{cases}
W_2^1((-\infty, 0))\oplus  W_2^1((0, \infty)), &\text{in Case $(i)$}\\
 W_2^1((0, \ell)), &\text{in Case $(ii)$}\\
\end{cases}.
$$
In Case (iii), the domain $\Dom((\dot D)^*)$
 consists of the vector-functions
$$h=(h_\infty, h_\ell)^T \in ( W_2^1(\bbR_-)\oplus W_2^1(\bbR_+))\oplus W_2^1((0, \ell))
$$
satisfying the ``boundary condition''
\begin{equation}\label{domsop}
h_\infty(0-)-k\, h_\infty(0+)
-\sqrt{1-k^2}\, h_\ell(0)=0.
\end{equation}

\end{lemma}
\begin{proof} The corresponding result in Cases (i) and (ii) is well known.

In Case (iii),  from \eqref{dotdom0} it follows that for
$f=(f_\infty, f_\ell)^T\in \Dom(\dot D)$
 the ``quantum Kirchhoff rule''
$$|f_\infty(0-)|^2=|f_\infty(0+)|^2+|f_\ell(0)|^2
$$
holds. Since also $f_\ell(\ell)=0$, integration by parts
\begin{align*}
(\dot Df, f)&=\int_{-\infty}^0i f_\infty'(x)\overline{f_\infty(x)}dx+\int_0^{\infty}if_\infty'(x)\overline{f_\infty(x)}dx
+\int_0^\ell if_\ell'(x)\overline{f_\ell(x)}dx
\\&
=-\int_{-\infty}^0if_\infty(x)\overline{f_\infty'(x)}dx-
\int_0^{\infty}if_\infty(x)\overline{f_\infty'(x)}dx
-\int_0^\ell if_\ell(x)\overline{f_\ell'(x)}dx
\\&
\,\,\,\,\,\,+i\left (|f_\infty(0-)|^2-|f_\infty(0+)|^2-|f_\ell(0)|^2\right ),\quad f\in \dom (\dot D),
\end{align*}
shows that the quadratic form $(\dot Df,f)$
 is real and therefore the operator $\dot D$ is indeed symmetric.

Similar computations show that
\begin{align*}
(h,\dot D f)&=\int_{-\infty}^0i h_\infty'(x)\overline{f_-(x)}dx+\int_0^{\infty}i h_\infty'(x)\overline{f_+(x)}dx
+\int_0^\ell i h_\ell'(x)\overline{f_\ell(x)}dx
\\&
=-\int_{-\infty}^0i h_\infty(x)\overline{f_\infty'(x)}dx-
\int_0^{\infty}i h_\infty(x)\overline{f_\infty'(x)}dx
-\int_0^\ell i h_\ell(x)\overline{f_\ell'(x)}dx
\\&
\,\,\,\,\,\,+i\left ( h_\infty(0-)\overline{f_\infty(0-)}-h_\infty(0+)\overline{f_\infty(0+)}-h_\ell(0)\overline{f_\ell(0)}\right ),\quad f\in \dom (\dot D).
\end{align*}
Therefore, $h\in \Dom ((\dot D)^*)$ if and only if
$$
 h_\infty(0-)\overline{f_\infty(0-)}-h_\infty(0+)\overline{f_\infty(0+)}-h_\ell(0)\overline{f_\ell(0)}=0 \quad \text{for all }\quad f\in \Dom(\dot D).
$$
Taking into account the boundary conditions \eqref{dotdom0}, we have
\begin{equation}\label{kukin}
\left ( h_\infty(0-)-h_\infty(0+)k -h_\ell(0)\sqrt{1-k^2}\right )\overline{f_\infty(0-)}=0
\end{equation}
for all $f\in \Dom(\dot D)$. Since
 $f_\infty(0-)$ may be chosen arbitrarily,  \eqref{domsop} follows from \eqref{kukin}.

\end{proof}

The following lemma introduces  a natural (standard) basis in the  subspace  $$\cN=\Ker ((\dot D)^*-iI)\dot+\Ker ((\dot D)^*+iI ).$$
\begin{lemma}\label{defeff}
The deficiency subspaces $\Ker ((\dot D)^*\mp i I)$ of the symmetric operator $\dot D$  on the  metric graph $\bbY$ is Cases $(i)$-$(iii)$ are spanned by the following normalized deficiency elements $g_\pm$.
Here,

in Case $(i)$,
\begin{equation}\label{deftip1}
g_+(x)=\sqrt{2}
e^{x} \chi_{(-\infty, 0)}(x)
\quad \text{and}
\quad
g_-(x)=\sqrt{2}
e^{-x}\chi_{(0,\infty)}(x), \quad x\in \bbR,
\end{equation}

in Case $(ii)$,
\begin{equation}\label{deftip2}
g_+(x)=\frac{\sqrt{2}}{\sqrt{e^{2\ell}-1}}e^{ x}
\quad \text{ and } \quad
g_-(x)=\frac{\sqrt{2}}{\sqrt{e^{2\ell}-1}}e^{\ell  -x}, \quad x\in [0, \ell].
\end{equation}

Finally, in Case $(iii)$,
\begin{equation}\label{hyddef1}
g_+(x)=\frac{\sqrt{2}}{\sqrt{e^{2\ell}-k^2}}e^x\begin{cases}
\sqrt{1-k^2} \chi_{(-\infty,0)}(x), & x \in (-\infty, 0)\sqcup (0, \infty) \\
1, &x\in [0,\ell]
\end{cases},
\end{equation}
\begin{equation}\label{hyddef2}
g_-(x)=\frac{\sqrt{2}}{\sqrt{e^{2\ell}-k^2}} e^{\ell-x}
\begin{cases}
-\sqrt{1-k^2} \chi_{(0,\infty)}(x), & x \in (-\infty, 0)\sqcup (0, \infty)\\
k , &  x\in [0,\ell]
\end{cases}.
\end{equation}

In particular, $\dot D$ is a symmetric operator with deficiency indices $(1,1)$.

\end{lemma}

\begin{proof}

The deficiency subspaces of the symmetric operator  $\dot D$ in Cases (i) and (ii)
can  be easily calculated.

Indeed, in Case (i), we have
 $$
\Ker ((\dot D)^*-z I)=\text{span}\{ g_z\},
$$
where
\begin{equation}\label{defsub1}
g_z(x)=\begin{cases}
e^{-izx}, & x<0\\
 0, & x \ge 0\\
\end{cases}
\quad \text{
and} \quad
g_z(x)=\begin{cases} 0, & x<0\\
e^{-izx}, & x \ge 0
\end{cases}
\end{equation}
for $z\in \bbC_+$ and $z\in
\bbC_-$, respectively, which proves
\eqref{deftip1}.

In Case (ii),
$$
\Ker ((\dot D)^*-z I)=\text{span}\{ g_z\},
$$
where
\begin{equation}\label{defsub2}
g_z(x)=e^{-izx}, \quad x\in [0, \ell], \quad z\in \bbC\setminus \bbR,
\end{equation}
and
\eqref{deftip2} follows.

In Case (iii), from  the description of $\Dom((\dot D)^*)$ provided by Lemma \ref{domsopr} it follows that
   the deficiency subspace
$
\Ker ((\dot D)^*-z I)=\text{span}\{ g_z\}
$, $z\in \bbC\setminus \bbR$, is  generated by the functions
\begin{equation}\label{defsub01}
g_z(x)=\begin{cases}  \sqrt{1-k^2} e^{-izx}\chi_{(-\infty,0)}(x), & x \in (-\infty, 0)\sqcup (0, \infty)\\
e^{-izx}, &x\in [0,\ell]
\end{cases}, \quad  z\in \bbC_+,
\end{equation}
 and
 \begin{equation}\label{defsub02}
g_z(x)=
\begin{cases}
-\sqrt{1-k^2}e^{-izx} \chi_{(0,\infty)}(x), & x \in (-\infty, 0)\sqcup (0, \infty)\\
k e^{-izx}, &  x\in [0,\ell]
\end{cases},\quad z\in \bbC_-,
\end{equation}
proving
\eqref{hyddef1} and \eqref{hyddef2}.

\end{proof}

 Recall (see Appendix \ref{A1}) that a symmetric operator
 $\dot A$ is called  a prime operator
if there is no   (non-trivial) subspace invariant under $\dot A$
such that the restriction of $\dot A$ to this subspace is self-adjoint.

\begin{lemma}\label{primeD} The symmetric differentiation operator $\dot D$ on the  metric graph $\bbY$ in Cases $(i)$-$(iii)$ is a
prime operator.
\end{lemma}

\begin{proof}

In Cases (i) and (ii) the corresponding result  is  known (see \cite{AkG}).

Suppose therefore  that $\dot D$ is in Case (iii).

First, we  show that if $f\in L^2(\bbY)= L^2(\bbR)\oplus L^2 ((0, \ell))$ and
\begin{equation}\label{zerozero}
(f, g_z)=0, \quad \text{ for all }
 z\in \bbC\setminus \bbR,
\end{equation}
then necessarily $f=0$.

Indeed, suppose that  $f=(f_\infty ,f_\ell)^T$, with  $f_\infty \in L^2(\bbR)$  and $f_\ell\in L^2 ((0, \ell))$ and let \eqref{zerozero} hold.
In particular,  for all   $0\ne s\in  \bbR$
$$
(f, g_{is})=0.$$
Then, given
 the description of the deficiency subspaces \eqref{defsub01} and  \eqref{defsub02},
one gets that
\begin{equation}\label{vtorr}
\sqrt{1-k^2} \int_{-\infty}^0 f_\infty(x) e^{sx}dx+
\int_0^\ell f_\ell(x)e^{sx}dx=0 \quad \text{ for all } s>0,
\end{equation}
and
\begin{equation}\label{perv}
-\sqrt{1-k^2} \int_0^\infty f_\infty(x) e^{-sx}dx+
k\int_0^\ell f_\ell(x)e^{-sx}dx=0 \quad \text{ for all } s>0.
\end{equation}
Therefore, from \eqref{vtorr} it follows that
$$
\int_{-\infty}^\ell h(x)e^{sx}dx=0 \quad \text{ for all } s>0,
$$
where
$$
h(x)=\begin{cases}
\sqrt{1-k^2} f_\infty(x),&x<0\\
f_\ell(x),&x\in [0, \ell]\\
\end{cases}.
$$

By the
 uniqueness theorem for  the Laplace transformation
(see, e.g., \cite[Theorem 5.1]{Doet}, we have that
 $h(x)=0$ almost everywhere $x\in  \bbR$. Since $k\ne 1$ (see \eqref{dotdom0}), we have
$$f_\infty(x)=0,
\quad \text{ a.e. } x\in(-\infty, 0)
$$
and
$$f_\ell(x)=0, \quad \text{ a.e. } x\in [0,\ell].
$$
Then,
from \eqref{perv} it follows that
$$
\int_0^\infty f_\infty(x) e^{-sx}dx=0\quad \text{ for all } s>0.
$$
By the uniqueness theorem $ f_\infty(x)=0$ for $x\ge 0$ as well.
That is,  $$f=(f_\infty, f_\ell)^T=0.$$

Thus, \eqref{zerozero} implies $f=0$ and therefore, by Theorem \ref{prostota0} in Appendix B,  the differentiation operator  $\dot D$ in Case (iii)  is a prime symmetric operator
as well.

\end{proof}
\begin{remark}
 We remark that the  symmetric operator $\dot D$ in Case (iii)  determined by the boundary conditions \eqref{dotdom0} with $k=0$
is also a prime operator:  in this case  $\dot D$ is unitarily equivalent  to the symmetric differential operator in Case (i),
which is a prime operator by Lemma \ref{primeD}.
 \end{remark}

\begin{remark}\label{remdotd}
It is easy to see that   the   prime symmetric differentiation operator $\dot D$  on the metric graph $\bbY$ in Cases (i)-(iii) satisfies the semi-Weyl commutation relations in the form (cf. Hypothesis \ref{muhly})
\begin{equation}\label{ccrsym}
U_t^*\dot DU_t=\dot D+tI\quad \text{on }\quad \Dom (\dot D), \quad t\in \bbR,
\end{equation}
where $U_t=e^{-it \cQ}$ is the  unitary group generated
by
the operator $\cQ$ of multiplication by independent variable  on the graph $\bbY$.

To show  that the commutation relations \eqref{ccrsym} hold we proceed as follows.
Let   $\cA(x)$ denote a real-valued piecewise continuous function on $\bbY$.  We remark that  the operators
$\dot D$ and $\dot D+\cA(x)$ are unitarily equivalent.
Indeed,  let
$\phi(x)$ be  any solution to the differential equation
\begin{equation}\label{gauget}
\phi'(x)=\cA(x)
\end{equation}
on the edges of the graph. Since the graph  $\bbY$ is a connected tree, the function $\phi(x)$ is determined  up to a constant, and we may without loss require
that   $\phi$ vanishes at the origin of the graph $\bbY$, that is,
\begin{equation}\label{mihmih}
\phi(0)=0.
\end{equation}

Denote by $V$ the unitary local gauge transformation
\begin{equation}\label{gaugege}
(Vf)(x)=e^{i\phi(x)}f(x),  \quad f \in L^2(\bbY).
\end{equation}
Taking into account the boundary conditions \eqref{dotos},
\eqref{dotint}  and
\eqref{dotdom0}
one concludes that the domain of $\dot D$ is $V$-invariant,
that is,
 $$
V(\dom (\dot D)) =\dom (\dot D).
$$
Next, a simple computation shows that
\begin{equation}\label{dopol}
\dot D=V^*(\dot D+\mathcal{A}(x))V.
\end{equation}

In the particular case  of a constant (magnetic) potential
$
\cA(x)\equiv t,
$
 $t\in \bbR$,
 solving \eqref{gauget} with the boundary condition
\eqref{mihmih}
 on the graph $\bbY$,  one immediately concludes that  the unitary operator $V$ from \eqref{gaugege} is given by
$$
V=e^{it \cQ},
$$
and therefore \eqref{dopol} implies the commutation relations \eqref{ccrsym}.
\end{remark}

\section{The  magnetic Hamiltonian}\label{s5}

In  this section we explicitly describe the set of all self-adjoint (reference) and, more generally, quasi-selfadjoint extensions of the differentiation symmetric operators  $\dot D$
on the metric graph $\bbY$  in  Cases (i)-(iii) introduced in Section \ref{s4intro}.

\begin{theorem}\label{gengen} Suppose that the metric graph $\bbY$ is in one of the Cases $(i)$-$(iii)$ and let $\dot D$ be the symmetric differentiation operator given by
\eqref{dotos}, \eqref{dotint} and \eqref{dotdom0}, respectively.

Then the one-parameter family of  differentiation operators $D_{\Theta}$, $|\Theta|=1$ on the graph $\bbY$  in Cases $(i)$-$(iii)$
 with boundary conditions
\begin{align}
f_\infty (0+)&=-\Theta f_\infty (0-),\label{bcI}
 \\
 f_\ell(0)&=-\Theta f_\ell (\ell),\label{bcII}
\\
\begin{pmatrix}
f_\infty(0+)\\
f_\ell(0)
\end{pmatrix}&=\begin{pmatrix}
k& \sqrt{1-k^2} \Theta\\
 \sqrt{1-k^2}&-k\Theta
\end{pmatrix}
 \begin{pmatrix}
f_\infty(0-)\\
f_\ell(\ell)
\end{pmatrix}, \label{bcIII}
\end{align}
 respectively,  coincides with the set of all self-adjoint extensions of the symmetric operator $\dot D$.

Moreover, let $g_\pm$ be the deficiency elements of $\dot D$ referred to in Lemma \ref{defeff}.

Then
$$g=g_+-\varkappa g_-\in \dom (D_\Theta), \quad |\varkappa|=1,
$$
if and only if
$$
\Theta= F(\varkappa),
$$
where
\begin{equation}\label{kappatheta}
F(\varkappa)=
\begin{cases}\varkappa, & in\text{ Case } (i)\\
-\frac{\varkappa-e^{-\ell}}{e^{-\ell}\varkappa-1},& in\text{ Case } (ii)\\
-\frac{\varkappa-ke^{-\ell}}{ke^{-\ell}\varkappa-1},& in\text{ Case } (iii)
\end{cases}.
\end{equation}
\end{theorem}
\begin{proof} If $\bbY$ is in Cases (i)-(ii), the first assertion of the theorem   is well known (see, e.g., \cite {AkG}).

 If the graph $\bbY$
is
in Case (iii) and
$$Y=(-\infty, 0)\sqcup (0,\infty)\sqcup (0,\ell), $$ one can identify the right endpoint of the edge $[0,\ell]$  of the graph $\bbY$  with its  origin  thus making the  number of incoming   and outgoing bonds equal. Since the incoming  $ \begin{pmatrix}
f_\infty(0-)\\
f_\ell(\ell)
\end{pmatrix} $ and outgoing $\begin{pmatrix}
f_\infty(0+)\\
f_\ell(0)
\end{pmatrix} $ data
are related by the unitary matrix $\sigma$ with
$$\sigma=\begin{pmatrix}
k& \sqrt{1-k^2} \Theta\\
 \sqrt{1-k^2}&-k\Theta
\end{pmatrix},
$$
from  \cite[Theorem 2.2.1]{BK} it follows that the operator $D$ is self-adjoint.
Recall that this theorem states that the differentiation operator $D=i\frac{d}{dx}$  on an oriented graph  is self-adjoint if and only if for each (finite) vertex $v$ the numbers of incoming and outgoing bonds are equal and the vectors $F^{\text{in}}(v)$ and $F^{\text{out}}(v)$ composed from the values of $f\in \dom (D)$ attained by $f$ from the incoming and outgoing bonds satisfy the condition
$$
F^{\text{out}}(v)=\sigma (v)F^{\text{in}}(v),
$$
where $\sigma$ is a unitary matrix.
Next, if $f\in \dom(\dot D)$,  the boundary conditions \eqref{dotdom0}  imply that the boundary conditions \eqref{bcIII} also hold, and therefore the self-adjoint operator $D$ extends $\dot D$.

Conversely, if $D$ is a self-adjoint extension of $\dot D$,  by  \cite[Theorem 2.2.1]{BK} the boundary conditions
 \begin{align}
\begin{pmatrix}
f_\infty(0+)\\
f_\ell(0)
\end{pmatrix}&=\sigma' \begin{pmatrix}
f_\infty(0-)\\
f_\ell(\ell)
\end{pmatrix}, \label{bcV}
\end{align}
hold, where $\sigma'$ is a unitary matrix.
The requirement that the self-adjoint operator $D$ extends $\dot D$ shows that $\sigma'$ has to be of  the form
$$\sigma'=\begin{pmatrix}
k&\alpha\\
 \sqrt{1-k^2}&\beta
\end{pmatrix}
$$
for some $\alpha$ and $\beta$. Since $\sigma'$ is unitary, we have that
$$
\sigma'=\begin{pmatrix}
k& \sqrt{1-k^2} \Theta\\
 \sqrt{1-k^2}&-k\Theta
\end{pmatrix}
$$
for some $\Theta$, $|\Theta|=1$, which completes the proof of the first assertion of the theorem.

To  prove \eqref{kappatheta} we argue as follows.

We use the following notation $g=g_\infty$ in Case (i),  $g=g_\ell$ in Case (ii),  and finally, $g=(g_\infty, g_\ell)^T$ in Case (iii) as introduced in Section 4.

From the representation  \eqref{deftip1} we get that
$$
g_\infty(x)=
\sqrt{2}
e^{x} \chi_{(-\infty, 0)}(x)
-\varkappa\sqrt{2}
e^{-x}\chi_{(0,\infty)}(x), \quad x\in \bbR,
$$
so that
$$
g_\infty(0+)=-\varkappa \sqrt{2} \quad \text{while}\quad  g_\infty(0-)=\sqrt{2}.
$$
That is, the element $ g$ satisfied boundary condition \eqref{bcI} with $\Theta=\varkappa$ which proves \eqref{kappatheta}  in Case (i).

In Case (ii), we use \eqref{deftip2} to  see that
$$
g_\ell(0)=\frac{\sqrt{2}}{\sqrt{e^{2\ell}-1}}(1-\varkappa e^{\ell})
$$
and
$$
g_\ell(\ell)=\frac{\sqrt{2}}{\sqrt{e^{2\ell}-1}}(e^{\ell}-\varkappa),
$$
which shows that the requirement $g_+-\varkappa g_-\in \dom (D_\Theta)$ means that
$$\Theta=-\frac{1-\varkappa e^{\ell}}{e^{\ell}-\varkappa}=-\frac{\varkappa- e^{-\ell}}{\varkappa e^{-\ell}-1}.
$$

 In Case (iii), from the representations for the deficiency elements $g_\pm$ \eqref{hyddef1} and  \eqref{hyddef2} it follows
 \begin{align*}
g_\infty (0+)&=a \sqrt{1-k^2} e^{\ell}\varkappa,
 \\
 g_\infty (0-)&=a \sqrt{1-k^2},
 \\
 g_\ell(\ell)&=a(1-ke^{\ell}\varkappa),
\end{align*}
 where
 $$a=\frac{\sqrt{2}}{\sqrt{e^{2\ell}-k^2}} .
 $$
 Since $g\in \Dom (D_\Theta)$, by  \eqref{bcIII} we have
$$
g_\infty(0+)=
kg_\infty(0-)+ \sqrt{1-k^2} \Theta g_\ell(\ell),
 $$
 which implies
 $$
 \sqrt{1-k^2} e^{\ell}\varkappa=k\sqrt{1-k^2} + \Theta \sqrt{1-k^2}(1- ke^{\ell}\varkappa).
 $$
Therefore
 $$\Theta= \frac{e^{\ell}\varkappa-k}{1- ke^{\ell}\varkappa}=-\frac{\varkappa-ke^{-\ell}}{ke^{-\ell}\varkappa-1}.
 $$

\end{proof}

\begin{remark}\label{remsamd}
Notice that in Case (i)
the self-adjoint operator   $ D_\Theta$  satisfies the semi-Weyl commutation relations
\begin{equation}\label{ccrsym1}
U_t^* D_\Theta U_t= D_\Theta+tI\quad \text{on }\quad \Dom ( D_\Theta), \quad t\in \bbR, \quad |\Theta|=1.
\end{equation}
Here $U_t=e^{-it \cQ}$ is the unitary  group generated by the self-adjoint
 operator $\cQ$  of multiplication by independent variable   on the graph $\bbY$.

However,  if the  graph $\bbY$ is in Cases (ii) and (iii), then we only have  the commutation relations with respect to a discrete subgroup $$\bbZ\ni n\mapsto U_{n\frac{2\pi}{\ell}}$$
of the group $U_t$. That is,
\begin{equation}\label{per}
U_{n\frac{2\pi}{\ell}}^* D_\Theta U_{n\frac{2\pi}{\ell}}= D_\Theta+n\frac{2\pi}{\ell}I\quad \text{on }\quad \Dom ( D_\Theta), \quad n\in \bbZ, \quad |\Theta|=1.
\end{equation}

This phenomenon has the following topological explanation:
the set of self-adjoint  extensions  $D_\Theta$, $|\Theta|=1$, is in one-to-one correspondence  with the  unit circle $\bbT$. The map
$$
D_\Theta\to D_{\Theta_t}=U_t^* D_\Theta U_t,
$$
determines the flow  $\Theta\mapsto \Theta_t$  on $\bbT$. Using boundary conditions \eqref{bcI}-\eqref{bcIII}
it is straightforward to see that
$$
\Theta_t=\Theta \begin{cases}
1, & \text{in  Case } (i)\\
e^{i\ell t},& \text{in Cases } (ii)\text{ and } (iii)\\
 \end{cases},
$$
so that in Cases (ii) and (iii)   $\Dom (D_\Theta)$ is not invariant with respect to the whole group $U_t$, $t\in \bbR$,
but only to its subgroup   $U_{n\frac{2\pi}{\ell}}$, $n\in \bbZ$.
In particular, the flow  $\Theta\mapsto \Theta_t$ has
 no fixed point, whenever the graph $\bbY$ is in Cases (ii) and (iii) (notice that the Euler characteristics of $\bbT$ is zero and hence Proposition \ref{Lef} is not applicable).
In this  regard, it is worth mentioning
the fall to the center  ``catastrophe" in Quantum  Mechanics \cite{LL,PP,PZ}.
 For a related discussion of the Efimov Effect in three-body systems  see \cite{Ef,FM,FMin}
 where the collapse in a three-body system with  point interactions has been discovered, also see \cite{MM} and references therein.

More generally, suppose that  $\cA(x)$ is  a real-valued piecewise continuous function on $\bbY$. Prescribing the magnetic potential $\cA(x)$ to all edges  of the graph, consider the {\it magnetic} differentiation
 operator $ D_\Theta+\cA(x)$.

 If the graph $\bbY$ is in Case (i), the local gauge transformation
 $$
 f(x)\longmapsto e^{i\phi(x)}f(x),  \quad f \in L^2(\bbY),
 $$
where  $\phi(x)$ is  a solution to the differential equation
$$
\phi'(x)=\cA(x), \quad x\in \bbY,
$$
$$
\phi(0)=0,
$$
eliminates the magnetic potential and one shows that
 the self-adjoint operators
$ D_\Theta$ and $ D_\Theta+\cA(x)$ are unitarily equivalent, with the unitary equivalence
performed by the unitary operator
\begin{equation}
(Vf)(x)=e^{i\phi(x)}f(x),  \quad f \in L^2(\bbY).
\end{equation}
Clearly
  $\dom (D_\Theta)$ is $V$-invariant, that is,
 $$
V(\dom (D_\Theta)) =\dom (D_\Theta),
$$
and therefore
\begin{equation}\label{magham}
 D_\Theta=V^*( D_\Theta+\mathcal{A}(x))V.
\end{equation}

If the graph $\bbY$ is in Cases (ii) and (iii),
the gauge transformation still eliminates the magnetic potential but changes the boundary conditions.
That is,
$$
V(\dom (D_\Theta)) =\dom (D_{\Theta\cdot  e^{-i\Phi}}),
$$
where
$$
\Phi=\int_0^\ell \cA(s)ds,\quad \text{the flux of the magnetic field},
$$
and hence
\begin{equation}\label{flowrel}
 D_{\Theta \cdot e^{-i\Phi}}=V^*( D_\Theta+\mathcal{A}(x))V.
\end{equation}

Notice, that in the particular case of a constant potential
$$\cA(x)\equiv t,
$$
 one gets the commutation relations \eqref{ccrsym1} and \eqref{per} as a corollary of \eqref{magham} and \eqref{flowrel}, respectively.
\end{remark}

Having in mind the unitary equivalences \eqref{magham} and \eqref{flowrel}  we adopt the following definition.

\begin{definition} The self-adjoint differentiation operator $D_\Theta$ for $|\Theta|=1$ referred to in Theorem \ref{gengen} will be called the magnetic Hamiltonian.
\end{definition}

Notice that
in Cases (ii) and (iii),
the boundary conditions \eqref{bcII} and \eqref{bcIII} are not local vertex conditions.
Bearing in mind applications in quantum mechanics, in Cases (ii) and (iii) one can  identify the end points of the interval $[0, \ell]$ to
get a  one-cycle graph $\overline{\bbY}$.
As it has been explained  in Remark \ref{betti}, in Case (iii)    one can also assign two additional vertices to the external edges at $\pm \infty$ of the one-cycle graph $\overline{\bbY}$,
so that the one-cycle graphs in Case (ii) and (iii) have
 the Euler characteristics $\chi (\overline{\bbY}) $ zero with    the corresponding first Betti numbers equal  to  one.
 In this case, the graph   $\overline{\bbY}$ can be considered  to be  the  Aharonov-Bohm ring,
  the configuration space for the quantum system
with  the magnetic Hamiltonian $D_\Theta$. This system  describes  a (massless)  quantum particle moving on the edges of the graph
and   the argument of the parameter $\Theta$ that determines the magnetic Hamiltonian $D_\Theta$  can be interpreted to be the flux of the (zero) magnetic field through the cycle (see \eqref{flowrel}).
For a related  information about graphs with Euler characteristic zero in the context of the inverse scattering theory we refer to \cite{Kur}.

Our next goal is to obtain an explicit description of all quasi-selfadjoint extensions  of the symmetric differentiation operators $\dot D$ introduced in Section \ref{s4intro}.

\begin{theorem}\label{allset} The differentiation operators $D_\Theta$, $\Theta\in \bbC\cup\{\infty\}$ with     $|\Theta|\ne 1$ referred to in Theorem \ref{gengen}    with boundary conditions \eqref{bcI}-\eqref{bcIII} is in one to one correspondence with the set of all quasi-selfadjoint extensions of the symmetric operator $\dot D$.
  \begin{remark}
  If $\Theta=\infty$, the boundary conditions   \eqref{bcI}-\eqref{bcIII} in Cases (i)-(iii) should be understood as follows
 \begin{align}
 f_\infty (0-)&=0,\label{bcI'}
 \\
 f_\ell(\ell)&=0 ,\label{bcII'}
\\
\begin{pmatrix}
k&\sqrt{1-k^2}\\0&0
\end{pmatrix}
\begin{pmatrix}
f_\infty(0+)\\
f_\ell(0)
\end{pmatrix}&=\begin{pmatrix}
f_\infty(0-)\\
f_\ell(\ell)
\end{pmatrix}
 \label{bcIII'},
\end{align}
respectively.

Notice that the boundary condition \eqref{bcIII'}
can be justified as follows.
Rewrite \eqref{bcIII} as
$$
\begin{pmatrix}
k& \sqrt{1-k^2} \Theta\\
 \sqrt{1-k^2}&-k\Theta
\end{pmatrix}^{-1}
\begin{pmatrix}
f_\infty(0+)\\
f_\ell(0)
\end{pmatrix}=
 \begin{pmatrix}
f_\infty(0-)\\
f_\ell(\ell)
\end{pmatrix}
$$
and observe that
\begin{align*}
\begin{pmatrix}
k&\sqrt{1-k^2}\\0&0
\end{pmatrix}
&=\lim_{\Theta\to \infty}\begin{pmatrix}
k& \sqrt{1-k^2} \Theta\\
 \sqrt{1-k^2}&-k\Theta
\end{pmatrix}^{-1}\\
\\&= \lim_{\Theta\to \infty}\frac{1}{-\Theta}\begin{pmatrix}
-k\Theta &-\sqrt{1-k^2}\Theta\\
-\sqrt{1-k^2}&k
\end{pmatrix}
\end{align*}
to get \eqref{bcIII'} as a limiting case.

Notice that the boundary conditions \eqref{bcIII'} can also be rewritten as
$$
kf_\infty(0+)+\sqrt{1-k^2}f_\ell(0)=f_\infty(0-),
$$
$$
f_\ell(\ell)=0.
$$

 \end{remark}

\end{theorem}

\begin{proof} If $\bbY$ is in Case (i) or (ii), the corresponding result is well known (see, e.g., \cite{AkG}).

Suppose that the metric graph  $\bbY$ is in Case (iii).  We will describe the required one to one correspondence explicitly.

Denote by  $D^\varkappa$  ($\varkappa \in \bbC$, $|\varkappa|\ne 1$) a quasi-selfadjoint extension of $\dot D$ such that
\begin{equation}\label{fin}
\dom (D^\varkappa)=\dom (\dot D)+\text{span}\{g_+-\varkappa g_-\},
\end{equation}
 where the deficiency elements $g_\pm$ are given by \eqref{hyddef1} and \eqref{hyddef2}.

If
\begin{equation}\label{domtheta}
f =\begin{pmatrix}
f_\infty \\
f_\ell
\end{pmatrix} \in \Dom (D^\varkappa),
\end{equation}
then
 \begin{align*}
f_\infty (x)&=\alpha \sqrt{1-k^2}e^x\chi_-(x)+\alpha \varkappa \sqrt{1-k^2}e^{\ell-x}\chi_+(x)+h_\infty (x), \quad x \in \bbR,
\\
f_\ell(x)&=\alpha e^x-\alpha \varkappa  k e^{\ell-x}+h_\ell(x), \quad x\in [0, \ell),
\end{align*}
for some $\alpha \in \bbC$ and some
 $$h=
 \begin{pmatrix}
h_\infty \\
h_\ell
\end{pmatrix} \in \dom (\dot D).
$$
In particular,
\begin{align*}
f_\infty (0-)&=\alpha \sqrt{1-k^2} +f_\infty (0-),\\
f_\infty(0+)&=\alpha \varkappa   \sqrt{1-k^2} e^\ell +h_\infty(0+),
\end{align*}
and
\begin{align*}
f_\ell(0)&= \alpha (1-k\varkappa  e^\ell)+h_\ell(0),\\
f_\ell(\ell)&= \alpha( e^\ell -k\varkappa )+h_\ell(\ell).
\end{align*}

Since $h\in \dom( \dot D)$, the boundary conditions \eqref{dotdom0} hold and therefore
\begin{align*}
f_\infty(0-)&=\alpha  \sqrt{1-k^2} +h_\infty (0-),\\
f_\infty (0+)&=\alpha \varkappa    \sqrt{1-k^2} e^\ell +kh_\infty (0-),
\\ f_\ell(0)&= \alpha (1  -k \varkappa  e^\ell) +  \sqrt{1-k^2}h_\infty (0-),
\\ f_\ell(\ell)&= \alpha( e^\ell -k \varkappa ).
\end{align*}
Equivalently,
$$
\begin{pmatrix}
f_\infty(0+)\\
f_\ell(0)
\end{pmatrix}=
\begin{pmatrix}
\varkappa    \sqrt{1-k^2} e^\ell &k\\
1  -k \varkappa  e^\ell& \sqrt{1-k^2}
\end{pmatrix}
\begin{pmatrix}
\alpha\\
h_\infty (0-)
\end{pmatrix}
$$
and
$$
\begin{pmatrix}
f_\infty(0-)\\
f_\ell(\ell)
\end{pmatrix}=
\begin{pmatrix} \sqrt{1-k^2}
&1\\
 e^\ell  -k \varkappa  &0
\end{pmatrix}
\begin{pmatrix}
\alpha\\
h_\infty (0-)
\end{pmatrix}.
$$
If  $e^\ell-k\varkappa\ne0$, one obtains that
\begin{equation}\label{dombond}
\begin{pmatrix}
f_\infty(0+)\\
f_\ell(0)
\end{pmatrix}=S \begin{pmatrix}
f_\infty(0-)\\
f_\ell(\ell)
\end{pmatrix},
\end{equation}
 where
\begin{align*}
S&=\begin{pmatrix}
\varkappa    \sqrt{1-k^2} e^\ell &k\\
1  -k \varkappa  e^\ell& \sqrt{1-k^2}
\end{pmatrix} \begin{pmatrix} \sqrt{1-k^2}
&1\\
 e^\ell  -k \varkappa  &0
\end{pmatrix}^{-1}
\\&
=\begin{pmatrix}
k&\sqrt{1-k^2} \Theta\\
 \sqrt{1-k^2}&-k\Theta
\end{pmatrix}.
\end{align*}
Combining  \eqref{domtheta}, \eqref{dombond}  and \eqref{bcIII} shows that
$$
D^\varkappa =D_\Theta,
$$
where
\begin{equation}\label{autoc}
\Theta=-\frac{\varkappa  e^\ell-k}{k\varkappa-e^{\ell} }.
\end{equation}
If  $e^\ell -k\varkappa =0$,  then necessarily
$
f_\ell(\ell)=0
$
and
\begin{align*}
kf_\infty(0+)+\sqrt{1-k^2}f_\ell(0)&=k\left (\alpha \varkappa \sqrt{1-k^2}e^\ell+kh_\infty(0-)\right )
\\&+\sqrt{1-k^2}\left (\alpha (1  -k \varkappa  e^\ell) +  \sqrt{1-k^2}h_\infty (0-)\right )
\\&
=\alpha \sqrt{1-k^2} +h_\infty (0-)
\\&=f_\infty (0-),
\end{align*}
which shows that boundary conditions \eqref{bcIII'} holds ($\Theta=\infty$, formally).

It remains to consider the case of the quasi-selfadjoint extension $D^\infty$
defined on
\begin{equation}\label{finsing}
\dom (D^\infty)=\dom (\dot D)+\text{span}\{g_-\},
\end{equation}
which  corresponds to  the infinite value of the von Neumann  parameter $\varkappa $ ($\varkappa =\infty$).

If \eqref{finsing} holds ($\varkappa =\infty$), a similar computation shows that the corresponding quasi-selfadjoint extension corresponds to the boundary condition
\eqref{bcIII} with
$$
\Theta=-\frac{e^\ell}{k},
$$
which is well defined ($k\ne0$  by the  hypothesis).

Notice that  \eqref{autoc} gives the link between  the boundary condition parameter $\Theta$ and the von Neumann extension parameter
$\varkappa $ from \eqref{fin} and  thus establishes the required correspondence.

\end{proof}

\begin{remark}\label{remAA}
Observe that in Case (i)
the operator   $ D_\Theta$  satisfies the semi-Weyl commutation relations
\begin{equation}\label{reschebnia}
U_t^* D_\Theta U_t= D_\Theta+tI\quad \text{on }\quad \Dom ( D_\Theta), \quad t\in \bbR,
\end{equation}
for all $\Theta\in \bbC\cup\{\infty\}$,
where $U_t=e^{-it \cQ}$ is the unitary  group generated by the self-adjoint
 operator $\cQ$  of multiplication by independent variable   on the graph $\bbY$.

 If the  graph $\bbY$ is in Cases (ii) and (iii), the commutation relations  \eqref{reschebnia} hold only if $\Theta=0$ or $\Theta=\infty$.
 Otherwise, we only have  the commutation relations with respect to the discrete subgroup $\bbZ\ni n\mapsto U_{n\frac{2\pi}{\ell}}$, cf. Remark \ref{remsamd},
\begin{equation}\label{pergunt}
U_{n\frac{2\pi}{\ell}}^* D_\Theta U_{n\frac{2\pi}{\ell}}= D_\Theta+n\frac{2\pi}{\ell}I\quad \text{on }\quad \Dom ( D_\Theta), \quad n\in \bbZ, \quad |\Theta|=1.
\end{equation}
\end{remark}

It is interesting to notice
that   the metric graph
$\bbY=(-\infty, 0)\sqcup (0,\infty)\sqcup (0,\ell) $
in Case (iii)
 serves as the configuration space
for a  minimal {\it self-adjoint dilation}  of almost all  (with the only one exception)  maximal dissipative differentiation operators on the finite interval
$(0,\ell)$.  Actually, the corresponding self-adjoint dilations coincide with the set of magnetic Hamiltonians $D_\Theta$, $|\Theta|=1$, in Case (iii).

\begin{theorem}[{cf. \cite{PavDil}}]\label{dilthm} The self-adjoint operator $D_\Theta$, $|\Theta|=1$, on the metric graph  $\bbY=(-\infty, 0)\sqcup (0,\infty)\sqcup (0,\ell) $ in Case $(iii)$
with the boundary conditions \eqref{bcIII} $(0<k<1)$
is a $($minimal$)$ self-adjoint dilation of the maximal dissipative differentiation operator $ \widehat d_\Theta$ on its subgraph
     $\bbK=(0,\ell)$ determined by the
 boundary condition
\begin{equation}\label{domneinv}
\Dom (\widehat d_\Theta)=\{f_\ell \in W_2^1((0, \ell))\, |\, f_\ell(0)=-k\Theta f_\ell(\ell)\}.
\end{equation}

That is,
\begin{equation}\label{dil}
P(D_\Theta-zI)^{-1}|_K=(\widehat d_\Theta-zI)^{-1}, \quad z\in \bbC_-,
\end{equation}
where  $P$ is  the orthogonal projection from $L^2(\bbY)$ onto the subspace  $K=L^2(\bbK)=L^2((0,\ell))$.
In particular,
$$
e^{it\widehat d_\Theta}=Pe^{itD_\Theta}|_K, \quad t\ge 0.
$$

\end{theorem}

\begin{proof}
Let $g=(g_\infty, g_\ell)^T\in L^2(\bbY)$ and
$$
f=(D_\Theta-zI)^{-1}g, \quad z\in \bbC_-.
$$
Since $f=(f_\infty, f_\ell)^T\in \dom (D_\Theta)$, the boundary conditions
\begin{equation}\label{zzzz}
\begin{pmatrix}
f_\infty(0+)\\
f_\ell(0)
\end{pmatrix}=\begin{pmatrix}
k& \sqrt{1-k^2} \Theta\\
 \sqrt{1-k^2}&-k\Theta
\end{pmatrix}
 \begin{pmatrix}
f_\infty(0-)\\
f_\ell(\ell)
\end{pmatrix}
\end{equation}
hold.
We have
$$
i \frac{d}{dx}f_\infty(x) - zf_\infty(x) =g_\infty(x), \quad x\in (-\infty, 0)\cup (0, \infty)\subset \bbY,
$$
and
\begin{equation}\label{compr}
i \frac{d}{dx}f_\ell(x) - zf_\ell (x) =g_\ell(x), \quad x\in (0, \ell)\subset \bbY.
\end{equation}
If $g\in K$, then $g_\infty=0$ and hence
$$
i \frac{d}{dx}f_\infty(x) - zf_\infty(x) =0,  \quad x\in (-\infty, 0)\cup (0, \infty)\subset \bbY.
$$
Since  $z\in \bbC_-$,  the function $f_\infty$ has to vanish on the negative real-axis. In particular,  $f_\infty(0-)=0$. From
\eqref{zzzz} it follows that
$$
f_\ell(0)=-k\Theta
f_\ell(\ell).
$$
Therefore the boundary condition \eqref{domneinv} holds and  hence $f\in \Dom(\widehat d_\Theta)$. Combined with \eqref{compr} this means that
$$
f_\ell=(\widehat d_\Theta-zI)^{-1}g_\ell, \quad z\in \bbC_-,
$$
which  proves \eqref{dil} and eventually   shows that
$D_\Theta$ is a self-adjoint dilation of $ \widehat d_\Theta$.

\end{proof}

\begin{remark} Theorem \ref{dilthm} does not say anything about the dilation of the
(exceptional)  maximal dissipative differentiation operator $ \widehat d$
defined on
$$
\Dom (\widehat d)=\{f_\ell \in W_2^1((0, \ell))\, |\, f_\ell(0)=0\}
$$
(it is explicitly assumed that $k\ne 0$ in the boundary condition \eqref{domneinv}).

In fact, the corresponding
self-adjoint  dilation coincides with the self-adjoint realization of
$i\frac{d}{dx}$ on  the metric graph
$$\bbY=(-\infty, 0)\sqcup (0,\ell) \sqcup (\ell,\infty),$$
which can be identified with the real axis. Therefore,  in the exceptional case the configuration space of the dilation can be identified with the graph $\bbY$ in Case (i).

Indeed, to treat the exceptional case, assume that $g\in L^2(\bbR)
$ is supported by the finite interval $[0, \ell]$.
Then the element
$$
f=(D-zI)^{-1}g, \quad z\in \bbC_,
$$
is supported by the positive semi-axis and its continuous representative satisfies the boundary condition
$
f(0)=0.
$
In particular,
$$
i\frac{d}{dx}f(x)-zf(x)=g(x), \quad x\in [0,\ell],\quad f(0)=0,
$$
and therefore the compressed resolvent of $D$ in the lower half-plane coincides with the resolvent of $\widehat d$ proving that
$D$ dilates the dissipative operator $\widehat d$.

\end{remark}

\section{The Liv\v{s}ic function  $s_{(\dot D, D_\Theta)}(z)$}\label{s6}

The main goal of this  and the forthcoming section is to  describe those  unitary invariants  of  the prime symmetric operator $\dot D$
that characterize the operator up to unitary equivalence. Here $\dot D$ is the symmetric differentiation operator on  the metric graph $\bbY$  in one of the Cases (i)-(iii)  with boundary conditions
\eqref{dotos}, \eqref{dotint} and \eqref{dotdom0}, respectively.

To do so, we need to fix  a (reference) self-adjoint extension of the operator $\dot D$.
 We choose as such an extension
the   self-adjoint  realization $D=D_\Theta|_{\Theta=1}$ of the differentiation operator referred to in Theorem \ref{gengen}.
Recall  that  the domain of the self-adjoint operator $D=D_1$ is characterized by the following boundary conditions
  \begin{align}
f_\infty (0+)&=- f_\infty (0-),\label{bcI*}
 \\
 f_\ell(0)&=- f_\ell (\ell),\label{bcII*}
\\
\begin{pmatrix}
f_\infty(0+)\\
f_\ell(0)
\end{pmatrix}&=\begin{pmatrix}
k& \sqrt{1-k^2} \\
 \sqrt{1-k^2}&-k
\end{pmatrix}
 \begin{pmatrix}
f_\infty(0-)\\
f_\ell(\ell)
\end{pmatrix}, \label{bcIII*}
\end{align}
whenever the graph $\bbY$ is in Cases (i)--(iii), respectively.

  We start with the following important  observation.



\begin{lemma}\label{razgon} Let    $g_\pm$ be  the  deficiency elements of the symmetric operator $\dot D$ referred to in Lemma \ref{defeff}.
Then
\begin{equation}\label{vazhno}
f=g_+-g_-\in \Dom(D).
\end{equation}

\end{lemma}

\begin{proof}
In Case (i) we have the representation
 $$
 f(x)=\sqrt{2}(e^x\chi_{(-\infty,0)}(x)-e^{-x}\chi_{(0,\infty)}(x))
 $$
 so that
$
f(0-)=\sqrt{2}=-f(0+)
$
and therefore  $f\in \dom (D)$.

In Case (ii),
$$
f(x)=\frac{\sqrt{2}}{\sqrt{e^{2\ell}-1}}(e^{ x}-e^{\ell-x}), \quad x\in [0,\ell],
$$
which implies
$$
f(0)=\frac{\sqrt{2}}{\sqrt{e^{2\ell}-1}}(1-e^{\ell})=-\frac{\sqrt{2}}{\sqrt{e^{2\ell}-1}}(e^{\ell}-1)=-f(\ell)
$$
thus showing that  $f\in\dom (D)$ as well.

Finally, in Case (iii), from \eqref{hyddef1} and \eqref{hyddef2} it follows that the element $f$ admits the representation
$$f=\frac{\sqrt{2}}{\sqrt{e^{2\ell}-k^2}}(f_\infty, f_\ell)^T,$$
where
$$f_\infty(x)=\sqrt{1-k^2}e^{x}\chi_{(-\infty,0)}(x)+\sqrt{1-k^2}e^{\ell-x}\chi_{(0,\infty)}(x), \quad x\in \bbR,
$$
and
$$
f_\ell(x)=e^x-ke^{\ell-x}, \quad x\in [0,\ell].
$$
We have
$$
 f_\infty(0-)=\sqrt{1-k^2},\quad
 f_\infty(0+)=\sqrt{1-k^2}e^{\ell},
$$
and
$$
  f_\ell(0)=1- ke^{\ell},\quad
f_\ell(\ell)=e^{\ell}- k.
$$
Here $k$,
 $0<k<1,$ is the parameter from the boundary conditions  \eqref{dotdom0} and \eqref{bcIII*} describing the domains $\dom(\dot D)$ and $\dom(D)$
in Case (iii), respectively.

As a consequence, the incoming
$F^{\text{in}}=\begin{pmatrix}
f_\infty(0-)\\
f_\ell(\ell)
\end{pmatrix}$
and outgoing
$F^{\text{in}}=\begin{pmatrix}
f_\infty(0+)\\
f_\ell(0)
\end{pmatrix}
$
boundary data are related as
\begin{align*}
\begin{pmatrix}
f_\infty(0+)\\
f_\ell(0)
\end{pmatrix}&=
\begin{pmatrix}
\sqrt{1-k^2}e^{\ell}\\
1- ke^{\ell}
\end{pmatrix}
=
\begin{pmatrix}
k& \sqrt{1-k^2} \\
 \sqrt{1-k^2}&-k
\end{pmatrix}
\begin{pmatrix}
\sqrt{1-k^2}\\
e^{\ell}- k
\end{pmatrix}
\\&=S(1)
\begin{pmatrix}
f_\infty(0-)\\
f_\ell(\ell)
\end{pmatrix},
\end{align*}
where the bond scattering matrix  $S(1)$  is given by \eqref{bcIII} for $\Theta=1$, which shows   that
$$
f\in \Dom(D(1))=\Dom (D).
$$

\end{proof}

 Based on Lemma  \ref{razgon}, now we are ready to evaluate the Liv\v{s}ic function  associated with the pair $(\dot D , D)$, which is one of the unitary  invariants that characterizes the pair
  $(\dot D , D)$ up to unitary equivalence.

  \begin{lemma}\label{pomni} The Liv\v{s}ic function  associated with the pair $(\dot D , D)$  admits the representation
 \begin{equation}\label{srepr}
s_{(\dot D, D)}(z)=
 \begin{cases}
 0, & \text{in Case  $(i)$}
\\ \frac{e^{iz\ell}-e^{-\ell}}{e^{-\ell} e^{iz\ell}-1},& \text{in Case $ (ii)$}
\\k \frac{ e^{iz\ell}-e^{-\ell}}{k^2e^{-\ell} e^{iz\ell}-1},&\text{in Case $ (iii)$}
 \end{cases}.
\end{equation}
Here $k$, $0<k<1$, is the parameter from the boundary conditions  \eqref{dotdom0} and \eqref{bcIII*} describing the domains $\dom(\dot D)$ and $\dom(D)$
in Case $(iii)$.
\end{lemma}

\begin{proof} Denote by   $g_\pm$ the  deficiency elements of the symmetric operator $\dot D$ referred to in Lemma \ref{defeff}.
By Lemma \ref{razgon},
\begin{equation}\label{vzh}
g_+-g_-\in \Dom(D)
\end{equation}
in all Cases (i)-(ii).
As long as
\eqref{vzh} is established,
 in accordance with definition \eqref{charsum},
 the Liv\v{s}ic function associated with the pair  $ (\dot D, D)$
can be evaluated as
 $$
s_{(\dot D, D)}(z)=\frac{z-i}{z+i}\cdot \frac{(g_z, g_-)}{(g_z,g_+)}, \quad z\in \bbC_+.
$$
Here the deficiency elements $g_z \in \Ker ((\dot D)^*-zI)$, $z\in \bbC\setminus \bbR$
are given by  \eqref{defsub1}, \eqref{defsub2}, \eqref{defsub01} and \eqref{defsub02}
in Cases (i)-(iii), respectively.

In Case (i), one observes that
$
g_z\perp g_-, \quad z\in \bbC_+,
$
and hence
$$
 s_{(\dot D, D)}(z)=0, \quad   z\in \bbC_+.
 $$

In Case (ii), one computes
\begin{align*}
s_{(\dot D, D)}(z)&=\frac{z-i}{z+i}\cdot \frac{(g_z, g_-)}{(g_z,g_+)}=e^\ell
\frac{z-i}{z+i}\cdot \frac{\int_0^\ell e^{(-iz-1)x}dx}{\int_0^\ell e^{(-iz+1)x}dx}
=
e^{\ell} \frac{e^{(-iz-1)\ell}-1}{e^{(-iz+1)\ell}-1}
\\ &=\frac{e^{iz\ell}-e^{-\ell}}{e^{-\ell} e^{iz\ell}-1}.
\end{align*}

Finally, in Case (iii), we have
\begin{align*}
s_{(\dot D, D)}(z)&=\frac{z-i}{z+i}\cdot \frac{(g_z, g_-)}{(g_z,g_+)}
\\&
=\frac{z-i}{z+i}\cdot
\frac{k\int_0^\ell e^{(-iz-1)x}dx}
{(1-k^2) \int_{-\infty}^0 e^{(-iz+1)x}dx+\int_0^\ell e^{(-iz+1)x}dx}e^{\ell}
\\&
=
\frac{k( e^{(-iz-1)\ell}-1)}
{(1-k^2) + e^{(-iz+1)\ell}-1}e^{\ell}
\\&
=
\frac{k ( e^{-iz\ell}-e^{\ell})}
{e^\ell e^{-iz\ell}-k^2 }=k
\frac{  e^{iz\ell}-e^{-\ell} }
{k^2   e^{-\ell}  e^{iz\ell}-1}
\\&=k
\frac{  e^{iz\ell}-e^{-\ell} }
{k^2   e^{-\ell}  e^{iz\ell}-1}.
\end{align*}
 Combing these results proves  \eqref{srepr}.

\end{proof}

\begin{remark}
  The representation \eqref{srepr}  in Cases (i) and (ii) is known (see, e.g., \cite{AkG}).
\end{remark}

The following corollary
 provides a complete characterization of prime symmetric operators with deficiency indices $(1,1) $ satisfying the commutation relations \eqref{1/2Wdot}
(see Problem (I) a)   in the Introduction).

\begin{corollary}\label{muhly1} {\it Let $\dot A $ be a symmetric operator referred to in Hypothesis \ref{muhly}. Suppose that $\dot A$ is a prime operator.
Then $\dot A$ is unitarily equivalent to one of the differentiation operators $\dot D=i\frac{d}{dx}$  on a metric graph $\bbY$ in  one of the Cases $(i)$-$(iii)$ introduced in Section 4  $($see eqs.  \eqref{dotos}, \eqref{dotint} and
\eqref{dotdom0}$)$.
}
\end{corollary}

\begin{proof} By Corollary \ref{ssss}, $\dot A$ admits a self-adjoint extensions such that the Liv\v{s}ic function associated with the pair $(\dot A, A)$
coincides with the one  referred to in Lemma \ref{pomni}. Since $\dot A$ is a prime operator,  by the Uniqueness Theorem  \ref{unitar} in Appendix C,
the operator  $\dot A$ is unitarily equivalent to the symmetric differentiation operator  on the  metric graph $\bbY$ in  one of the cases Cases (i)-(iii).
 \end{proof}

More generally, the Liv\v{s}ic function associated with the  pair $(\dot D, D_\Theta)$, $|\Theta|=1$,
where $D_\Theta$ is   the  self-adjoint  realization of the differentiation operator referred to in Theorem \ref{gengen},
 differs  from the function
 $s_{(\dot D, D)}(z)$  evaluated above in Lemma \ref{pomni} by  a constant unimodular factor.
For the sake of completeness, we  present the following result.

\begin{theorem}\label{ind11} Suppose that $|\Theta|=1$.
The Liv\v{s}ic function  $s_{(\dot D, D_\Theta)}(z)$ associated with the pair $(\dot D, D_\Theta)$ admits the representation
  \begin{equation}\label{srepr*}
s_{(\dot D, D_\Theta)}(z)=e^{-2i\alpha}
 \begin{cases}
 0, & \text{in Case}  (i)
\\ \frac{e^{iz\ell}-e^{-\ell}}{e^{-\ell} e^{iz\ell}-1},& \text{in Case } (ii)
\\k \frac{ e^{iz\ell}-e^{-\ell}}{k^2e^{-\ell} e^{iz\ell}-1},&\text{in Case } (iii)
 \end{cases}.
\end{equation}
Here  $k$, $0<k<1$, is the parameter from the boundary conditions  \eqref{dotdom0} and \eqref{bcIII} describing the domains $\dom(\dot D)$ and $\dom(D_\Theta)$ in Case $(iii)$, respectively.

Here  $\alpha$ and the boundary condition parameter $\Theta$ are related as follows
  \begin{equation}\label{prostoeqq}
e^{2i\alpha}=
\begin{cases}
 \Theta, & \text{ in Case  }  (i)
\\ \frac{\Theta+e^{-\ell}}{e^{-\ell}\Theta+1},&  \text{ in Case }  (ii)
\\ \frac{\Theta+e^{-(\ell+\ell')}}{e^{-(\ell+\ell')}\Theta+1},& \text{ in Case }  (iii)
 \end{cases},
\quad \text{
with}\quad
\ell'=\ln \frac1k.
\end{equation}

In particular,
\begin{equation}\label{xyxy}
s_{(\dot D, D_\Theta)}(z)=e^{-2i\alpha}s_{(\dot D, D)}(z),
\end{equation}
where
$$
D=D_\Theta|_{\Theta=1}.
$$
\end{theorem}

\begin{proof}  From Theorem \ref{gengen} it follows that \begin{equation}\label{figa}
F=g_+-e^{2i\alpha}g_-\in \Dom(D_\Theta),
\end{equation}
where   $g_\pm$ are the deficiency elements referred to in Lemma \ref{defeff} and $e^{2i\alpha}$ is given by \eqref{prostoeqq}.
  Now   \eqref{srepr*} follows from   \eqref{srepr} by  Lemma \ref{ind0} in Appendix E.
\end{proof}

We conclude this section with several remarks of analytic character.

\begin{remark}\label{s2s3}  (i). One observes that   the Liv\v{s}ic function
 $$
s^{II}(z;\ell)= \frac{e^{iz\ell}-e^{-\ell}}{e^{-\ell} e^{iz\ell}-1}, \quad z\in \bbC_+,
$$ given by \eqref{srepr*} in Case (ii)  admits the representations
$$
s^{II}(z; \ell)=\frac{e^{i\ell z} -e^{-\ell } }{  e^{-\ell } e^{i\ell z}-1}=\frac{\sin (z-i)\frac\ell2}{\sin (z+i)\frac\ell2 }=\frac{z-i}{z+i}\cdot \frac{\prod\limits_{n\in \bbZ} \left (1-\left (\frac{z-i}{2\pi n}\ell\right)^2\right )}{\prod\limits_{n\in \bbZ} \left (1-\left (\frac{z+i}{2\pi n}\ell\right)^2\right )}.
$$
Therefore,
 $ s^{II}(z; \ell) $ is
a    pure Blaschke product  with zeroes $z_n$ located on the lattice
$$z_n=i+\frac{2\pi}{\ell} n, \quad n\in \bbZ.$$

(ii).  A direct computation shows that   the Liv\v{s}ic function
$$
s^{III}(z;k, \ell)=k \frac{ e^{iz\ell}-e^{-\ell}}{k^2e^{-\ell} e^{iz\ell}-1}, \quad z\in \bbC_+, \quad 0<k<1,
$$
(cf.  \eqref{srepr*} in Case (iii))
can be obtained
 by an analytic continuation  of
 $
s^{II}(z;\ell)$  with an appropriate identification of the parameters. That is,
\begin{equation}\label{ancont}
s^{III}(z;e^{-\ell'}, \ell)=s^{II}\left (\frac{\ell z+i\ell'}{\ell+\ell'};\ell+\ell'\right ), \quad z\in \bbC_+.
\end{equation}

(iii). In  the inner-outer factorization of the Liv\v{s}ic function in Case (iii)
$$
s^{III}(z;k, \ell)=k \frac{ e^{iz\ell}-e^{-\ell}}{k^2e^{-\ell} e^{iz\ell}-1}=   s^{III}_{\text{in}}(z)\cdot s^{III}_{\text{out}}(z)
$$
the inner factor $ s^{III}_{\text{in}}(z) $  coincides with the Liv\v{s}ic function in Case (ii), i.e.,
$$
 s^{III}_{\text{in}} =s^{II}(z;\ell)= \frac{e^{iz\ell}-e^{-\ell}}{e^{-\ell} e^{iz\ell}-1}.
$$

Indeed,  $$
s^{III}(z;e^{-\ell'}, \ell)
=\frac{\sin ( z-i)\frac{\ell}{2}}{\sin \left ((z+i)\frac\ell2 +i\ell'\right )}=s^{II}(z;\ell)\cdot \frac{\sin (z+i)\frac\ell2}{\sin \left ((z+i)\frac\ell2 +i\ell' \right )}.
$$
(In particular, the functions $ s^{III}(z;e^{-\ell'}, \ell)$ and $s^{II}(z;\ell)$ have the same set of zeros).

To  complete the proof of the claim  it remains to show that
the function
\begin{equation}\label{sinsin}
t(z)= \frac{\sin (z+i)\frac\ell2}{\sin \left ((z+i)\frac\ell2 +i\ell' \right )}
\end{equation}
is an outer function.

First, one observes that   $t(z)$ is a  contractive function in the upper half-plane.
Next, let
$$
t(z)=t_{\text{in}}(z) t_{\text{out}}(z)
$$
be its inner-outer factorization. Since $t(z)$ does not vanish in the upper half-plane, the  inner factor of $t(z)$ is necessarily a singular inner function.
Since $t(z)$ admits an analytic continuation into a strip in the lower half-plane, the singular measure in the exponential representation
 of $t_\text{in}(z)$ does not charge bounded sets and therefore $$
t_{\text{in}}(z)=e^{iLz}
$$
for some $L\ge 0$,  ``mass" at infinity.  In particular,
$$\lim_{y\to \infty} t_{\text{in}}(iy)=0$$ unless $L=0$.
However,  from \eqref{sinsin} it follows that
$$
\lim_{y\to \infty} t(iy)=e^{-\ell'},
$$
which implies that $L=0$. Therefore  $t_{\text{in}}(z)=1$  and hence $t(z)$ is an outer function.

(iv).  In fact,  for the outer factor $ s^{III}_{\text{out}}(z)=t(z)$  one gets the representation
  \begin{equation}\label{trepr}
  s^{III}_{\text{out}}(z)=\sqrt{\frac{\sinh \ell}{\sinh(\ell+2\ell')}}\exp \left (
\frac i{2\pi} \int_\bbR\left  (\frac{1}{\lambda +z}+\frac{\lambda}{1+\lambda^2}\right )\rho(\lambda)  d\lambda
\right ),
  \end{equation}
  where the density is given by
  $$
  \rho (\lambda)=
 \log\frac{P_{e^{-\ell-2\ell'}}(\lambda \ell)}{P_{e^{-\ell}}(\lambda \ell)}.
  $$Here
$$
 P_r(\theta)=\frac{1-r^2}{1+r^2-2r\cos \theta}
 $$  is the  Poisson kernel.

Indeed, since
 $t(z)$ is an outer function in the upper half-plane,  we have the representation \cite{Koosis}
\begin{equation}\label{vnesh}t(z)=
\exp \left (
\frac i\pi \int_\bbR\left  (\frac{1}{\lambda +z}+\frac{\lambda}{1+\lambda^2}\right ) \log |t(\lambda)|d\lambda
\right ).
\end{equation}
Using \eqref{sinsin} one computes that
\begin{align*}
\log |t (\lambda)|&=\log \left | e^{-\ell'} \frac{ e^{iz\ell}-e^{-\ell}}{e^{-2\ell'}e^{-\ell} e^{iz\ell}-1} \right |
\\&
=-\ell'+\frac12\log\frac{(\cos \ell\lambda-e^{-\ell})^2+\sin^2 \ell\lambda}{(e^{-\ell-2\ell'}\cos \ell\lambda-1)^2+ e^{-2\ell-4\ell'}\sin^2 \ell\lambda}
\\&	
=-\ell'+\frac12\log\frac{1-2\cos \ell\lambda e^{-\ell}+e^{-2\ell}}{1-2e^{-\ell-2\ell'}\cos \ell\lambda+ e^{-2\ell-4\ell'}}
\\&
=\frac12 \log \left (e^{-2\ell'} \cdot \frac{1-e^{-2\ell}}{1-e^{-2\ell-4\ell'}}\cdot\frac{P_{e^{-\ell-2\ell'}}(\lambda \ell)}{P_{e^{-\ell}}(\lambda \ell)}
\right )
\\&=
\frac12 \log \frac{\sinh \ell}{\sinh(\ell+2\ell')}+\rho (\lambda),
\end{align*}
and since
$$
\frac i\pi \int_\bbR\left  (\frac{1}{\lambda +z}+\frac{\lambda}{1+\lambda^2}\right ) d\lambda=1, \quad z\in \bbC_+,
$$
representation \eqref{vnesh} simplifies to \eqref{trepr}.

\end{remark}

\section{The Weyl-Titchmarsh function $M_{(\dot D, D_\Theta)} (z)$}\label{s7}

Along with the  Liv\v{s}ic function,
 the Weyl-Titchmarsh function associated with a pair consisting of a prime symmetric operator and its self-adjoint extension  characterizes the pair  up to  mutual unitary  equivalence.
So that our next goal is to evaluate the Weyl-Titchmarsh function associated with the symmetric differentiation
$\dot D$ on the metric graph $\bbY$ and its self-adjoint reference extension.

Suppose that $\bbY$ is the metric graph $
 \bbY$ in one of the Cases (i)--(iii).
As it follows from Lemma \ref{pomni},
the Weyl-Titchmarsh function  $M_{(\dot D, D)} $ associated with the pair $(\dot D, D) $ has the form
 \begin{equation}\label{Mrepr}
M_{(\dot D, D)}(z)=
 \begin{cases}
 i, & \text{in Case  (i)}
\\ \coth\frac\ell2\tan \frac\ell2 z,& \text{in Case  (ii)}
\\\coth\frac{\ell+\ell'}2\tan \left (\frac\ell2 z+i\frac{\ell'}2\right ),&\text{in Case  (iii)}
 \end{cases},
\end{equation}
where
$$
\ell'=\ln \frac1k\quad (0<k<1)
$$
and   $k$,   $0<k<1$, is the parameter from the boundary conditions  \eqref{dotdom0} and \eqref{bcIII*} describing the domains $\dom(\dot D)$ and $\dom(D)$
in Case (iii).

Indeed,   in Case (iii)   one observes that
\begin{align*}
M_{(\dot D, D)}(z)&=
\frac1i\frac{s_{(\dot D, D)}(z)+1}{s_{(\dot D, D)}(z)-1}
=\frac1i\frac{\frac{k e^{iz\ell}-k e^{-\ell} }
{k^2    e^{-\ell} e^{iz\ell}-1}+1}{\frac{k  e^{iz\ell}-k e^{-\ell} }
{k^2  e^{-\ell} e^{iz\ell}-1}-1}
\\&
=
\frac1i\cdot
\frac{1+k e^{-\ell}}{1-k e^{-\ell}}\cdot \frac{k e^{i\ell z}-1}{ke^{i\ell z}+1}
=\coth \frac{\ell+\ell'}{2}\cdot \tan\left (\frac{\ell}2 z+i\frac{\ell'}{2}\right ).
\end{align*}
 Case (ii) then  follows by setting $\ell'=0$, equivalently $k=1$, and the corresponding representation in  Case (i) is obvious (formally take   the limit as $\ell\to \infty $  in Case (ii)).

More generally, we have the following result.
\begin{theorem}  Let
$D_{\Theta}$, $|\Theta|=1$,
 be the one-parameter family of  self-adjoint  reference  operators referred to in Theorem \ref{gengen}.
Then the  Weyl-Titchmarsh  function  $M_{(\dot D, D_\Theta)}$ associated with the pair $(\dot D, D_\Theta)$ admits the representation
\begin{equation}\label{forM}
M_{(\dot D, D_\Theta)} (z)=\begin{cases}
i, & \text{ in Case  } (i)
\\
A(\Phi)\tan \left ( \frac{\ell}{2}z-\frac{\Phi}{2}\right )+\frac{\sin \Phi}{\sinh \ell},& \text{ in Case }  (ii)
\\
A(\Phi)\tan \left ( \frac{\ell}{2}z+\frac{\ell'}2 i-\frac{\Phi}{2}\right )
+\frac{\sin \Phi}{\sinh (\ell+\ell')},&\text{ in Case }   (iii)
\end{cases},
\end{equation}
where
$$
\Phi=\arg \Theta, \quad \ell'=\log \frac1k \quad (0<k<1)
$$
and   $k$,  $0<k<1$, is the parameter from the boundary conditions  \eqref{dotdom0} and \eqref{bcIII} describing the domains $\dom(\dot D)$ and $\dom(D_\Theta)$
in Case $(iii)$.

Here  in Case $(iii)$ the  amplitude $A(\Phi)$ is given by the convex combination
  \begin{equation}\label{ampli}
  A(\Phi)=  \cos^2 \frac{ \Phi}2\cdot  \coth \frac{\ell+\ell'}2 + \sin^2 \frac{\Phi}2\cdot \tanh \frac{\ell+\ell'}2 ,
  \end{equation}
and in Case $(ii)$  $A(\Phi)$ is given by the same expression with $\ell'=0$.
\end{theorem}
\begin{proof}

 In Case (i) there is nothing to prove, since   $s_{(\dot D, D_\Theta)}(z)=0$ and hence $M_{(\dot D, D_\Theta)}(z)=i$ for all $z\in \bbC_+$.

 In Case (ii), by Theorem \ref{ind11} (see eq. \eqref{xyxy}), we have
 \begin{equation}\label{xyxyxy}
s_{(\dot D, D_\Theta)}(z)=e^{-2i\alpha}s_{(\dot D, D)}(z),
\end{equation}
where
$$
D=D_\Theta|_{\Theta=1}
$$ and $\alpha $ and $\Theta $ are related as in \eqref{prostoeqq}.

 From  \eqref{xyxyxy} and Lemma \ref{ind0} in Appendix E
 it follows
\begin{equation}\label{9797}
M_{(\dot D, D_\Theta)}(z)=\frac{M_{(\dot D, D)}(z)-\tan \alpha}{1+\tan \alpha\cdot  M_{(\dot D, D)}(z)}.
\end{equation}

By \eqref{Mrepr},
\begin{equation}\label{9798}
M_{(\dot D, D)}(z)=\coth \frac \ell2 \tan  \frac\ell2 z=m \tan \zeta,
\end{equation}
where\begin{equation}\label{inview} m=\coth \frac \ell2
\quad \text{and}\quad
\zeta=\frac\ell2 z,
\end{equation}
and therefore
$$
M_{(\dot D, D_\Theta)}(z)=\frac{m\tan \zeta -\tan \alpha}{1+m\tan \alpha\tan \zeta }.
$$

Using the identity
$$
\frac{m\tan \zeta -\tan \alpha}{1+m\tan \alpha\tan \zeta }=\frac{1+\tan^2\alpha}{1+m^2\tan^2\alpha}m\frac{\tan \zeta -m\tan \alpha}{1+m\tan \alpha\tan \zeta }
+\frac{m^2\tan \alpha -\tan \alpha}{1+m^2\tan^2\alpha},
$$
we obtain the following representation
\begin{equation}\label{soki}
M_{(\dot D, D_\Theta)}(z)=\frac{1+\tan^2\alpha}{1+m^2\tan^2\alpha}M_{(\dot D, D)}(\zeta-t)+\frac{m^2\tan \alpha -\tan \alpha}{1+m^2\tan^2\alpha},
\end{equation}
where
$$
\tan t=m \tan \alpha.
$$
Therefore, \eqref{soki} can be rewritten as
$$
M_{(\dot D, D_\Theta)}(z)=\frac{1+\frac{1}{m^2}\tan^2t }{1+\tan^2t }M_{(\dot D, D)}(\zeta-t)+\frac{m-\frac1m}{1+\tan^2t}\tan t.
$$
In view of  \eqref{9798} and \eqref{inview}, we have
\begin{equation}\label{alongw}
M_{(\dot D, D_\Theta)}(z)=\left (\coth \frac\ell2 \cos^2t+\tanh \frac\ell2\sin^2 t\right )
\tan\left ( \frac{\ell}{2}z-t\right )+\frac{1}{\sinh \ell}\sin 2t.
\end{equation}

From \eqref{prostoeqq} it follows that
\begin{align*}
\tan \alpha &=\frac1i\frac{ \frac{\Theta-e^{-\ell}}{e^{-\ell}\Theta-1}-1}{ \frac{\Theta-e^{-\ell}}{e^{-\ell}\Theta-1} +1}
=\tanh \frac\ell2 \cdot\tan \left ( \frac12\arg \Theta\right ),
\end{align*}
so that
$$
\tan t=m
\tan \alpha =\coth  \frac\ell2\cdot  \tan \alpha  =\tan  \left ( \frac12 \arg \Theta\right ).
$$
In particular,
\begin{equation}\label{ttdd}
t= \frac12 \arg \Theta=\frac12\Phi,
\end{equation}
which along with \eqref{alongw} shows that
\begin{align*}
M_{(\dot D, D_\Theta)}(z)=A(\Phi)
\tan \left ( \frac{\ell}{2}z-\frac{\Phi}{2}\right )
+\frac{\sin \Phi}{\sinh \ell},
\end{align*}
proving \eqref{forM}  with $A(\Phi)$ given by \eqref{ampli} in Case (ii).

To prove \eqref{forM} in Case (iii), in  a similar way one gets
  (cf. \eqref{9797}) $$
M_{(\dot D, D)}(z)=\coth\frac{\ell+\ell'}2\tan \left (\frac\ell2 z+i\frac{\ell'}2\right )
$$
and establishes that \eqref{soki} holds where now
$$
\zeta=\frac\ell2 z+i\frac{\ell'}2,\quad  m=\coth\frac{\ell+\ell'}2
$$
and
$$
\tan t=m \tan \alpha.
$$
Observing that
$$\tan \alpha =\frac1i\frac{ \frac{\Theta+e^{-\ell-\ell'}}{e^{-\ell-\ell'}\Theta+1}-1}{ \frac{\Theta+e^{-\ell-\ell'}}{e^{-\ell-\ell'}\Theta+1}+1}
=\tanh \frac{\ell+\ell'}2 \tan  \left (\frac12\arg \Theta\right ),
$$
one justifies \eqref{ttdd} in Case (iii) as well.
Literally repeating the reasoning above  one obtains the representation  \eqref{forM} in Case (iii).

\end{proof}

Our last result in this section  shows that the spectral measure  of the reference operator  $D_\Theta$,
 $|\Theta|=1$, is rather sensitive to the magnitude of the ``flux'' $\Phi=\arg \Theta$.

\begin{corollary}\label{quasimain*}  {\it The  Weyl-Titchmarsh function  $M_\Theta(z)=M_{(\dot D, D_\Theta)}(z) $ associated with the pair $(\dot D, D_\Theta)$
admits the representation
$$ M_\Theta(z)=\int_\bbR \left (\frac{1}{\lambda-z}-\frac{\lambda}{1+\lambda^2}\right )d\mu_\Theta(\lambda),\quad  |\Theta|=1,
$$
where
 $\mu_\Theta(d\lambda)$, the spectral measure, is
 \begin{itemize}
 \item[(i) ] the absolutely continuous measure with a constant density
 \begin{equation}\label{mucon*}
 \mu_\Theta(d \lambda)=\frac1\pi d\lambda \quad \text{in Case $(i)$};
\end{equation}
\item[(ii)]
 the discrete pure point measure
 \begin{equation}\label{muintt*}
 \mu_\Theta(d\lambda)=  \frac{2}{\ell}A(\arg \Theta)
 \sum_{k\in \bbZ} \delta_{\frac{(2 k+1)\pi+\Phi}{\ell}}(d\lambda)  \quad \text{in Case $(ii)$},
\end{equation}
with  $\delta_x(d\lambda)$  the  Dirac mass at $x$ and
$$
A(\Phi)=  \cos^2 \frac{ \Phi}2\cdot  \coth \frac{\ell}2 + \sin^2 \frac{\Phi}2\cdot \tanh \frac{\ell}2 ;
$$
\item[(iii)]
the absolutely continuous measure  with  a   periodic density
 \begin{equation}\label{tuda}
 \mu_\Theta(d\lambda)=\frac1\pi A(\arg \Theta) P_{e^{-\ell'}}\left (\ell \lambda-\pi -\arg \Theta\right )d \lambda.
\end{equation}
Here
$$
 P_r(\varphi)=\frac{1-r^2}{1+r^2-2r\cos \varphi}
 $$
 is the Poisson kernel,
 $$
 \ell'=\log \frac1k,
 $$
 with   $k$,  $0<k<1$,  the parameter from the boundary conditions  \eqref{dotdom0}, \eqref{bcIII},
 and
  $$
A(\Phi)=  \cos^2 \frac{ \Phi}2\cdot  \coth \frac{\ell+\ell'}2 + \sin^2 \frac{\Phi}2\cdot \tanh \frac{\ell+\ell'}2 .
$$
\end{itemize}
}
 \end{corollary}

\begin{proof}

Indeed, since
$$
i=\frac1\pi \int_\bbR \left (\frac{1}{\lambda-z}-\frac{\lambda}{1+\lambda^2}\right )d\lambda, \quad z\in \bbC_+,
$$
 \eqref{mucon*} follows from  the  equality $M_\Theta=i$ in Case (i).

To check \eqref{muintt*}, we use the representation
\begin{align}
\tan(z)&=-\sum _{{k=0}}^{{\infty }}\left({\frac {1}{z-(k+{\frac {1}{2}})\pi }}+{\frac {1}{z+(k+{\frac {1}{2}})\pi }}\right)\label{tantan}
\\&=\int_\bbR \left (\frac{1}{\lambda-z}-\frac{\lambda}{1+\lambda^2}\right )d\nu(\lambda),\nonumber
 \end{align}
 where $\nu(d\lambda)$ is a discrete point measure
 \begin{equation}\label{muintan}
 \nu(d\lambda)=  \sum_{k\in \bbZ} \delta_{( k+\frac12)\pi }(d\lambda),
\end{equation}
which proves \eqref{muintt*} in the particular case of $\ell=2$, and $\Theta=1$, and then the   general case follows by  making a simple change of variables.

Using the explicit representation \eqref{forM}  for the Weyl-Titchmarsh function $M_\Theta(z)$  in Case (iii) one observes that
$M_\Theta(z)$ is bounded in the upper half-plane. Therefore, the representing  measure $\mu_\Theta(d\lambda)$ in Case (iii) is absolutely continuous with
the density given by
\begin{align}
 \mu_\Theta(d\lambda)
 &=\frac1\pi \Im \, M_{\Theta}(\lambda+i0)d\lambda=\frac1\pi A(\arg \Theta) \cdot \Im
\tan \left ( \frac{\ell \lambda+\ell' i-\arg\Theta}{2}\right ) d\lambda\nonumber
\\&=\frac1\pi A(\arg \Theta) P_{e^{-\ell'}}\left (\ell \lambda-\pi -\arg \Theta\right )d\lambda,\label{mugid}
\end{align}
which proves \eqref{tuda}.

Here we have used
the  representation for  the imaginary part of the tangent function of a complex argument via the  the Poisson kernel,
\begin{align*}
\Im \tan \left (\lambda+i\tau\right)&=\Im \frac1i\cdot\frac{e^{i(\lambda+i\tau)}-e^{-i(\lambda+i\tau)}}{e^{i(\lambda+i\tau)}+e^{-i(\lambda+i\tau)}}
\\
&=
-\Re \frac{e^{i\lambda}e^{-\tau}-e^{-i\lambda}e^{\tau}}{e^{i\lambda}e^{-\tau}+e^{-i\lambda}e^{\tau}} \cdot
\frac{e^{-i\lambda}e^{-\tau}+e^{i\lambda}e^{\tau}}{e^{-i\lambda}e^{-\tau}+e^{i\lambda}e^{\tau}}
\\ &=\frac{e^{2\tau}-e^{-2\tau}}{e^{2\tau}+2\cos  2\lambda +e^{-2\tau}}
=P_{e^{-2\tau}}(2\lambda-\pi),
\end{align*}
$$
\quad \lambda\in \bbR, \quad \tau > 0.
$$

\end{proof}

\begin{remark}
Observing that  $$
\frac1{2\pi}P_r(\lambda)d\lambda \to \sum_{n}\delta_n (d\lambda) \quad \text{as} \quad r\uparrow 1,$$
we get the following spectral
hierarchy of the representing  measures  given by \eqref{tuda}, \eqref{muintt*} and \eqref{mucon*}:
\begin{align*}
\frac1\pi A(\arg \Theta) \cdot P_{e^{-\ell'}}&\left (\ell \lambda-\pi -\arg \Theta\right )d\lambda
\\&
\downarrow \quad (\ell'\to 0)
\\
 \frac{2}{\ell}A(\arg \Theta)\cdot
 &\sum_{k\in \bbZ} \delta_{\frac{(2\pi k+1)\pi+\Phi}{\ell}} (d\lambda)
\\&
\downarrow  \quad (\ell\to \infty)
\\
 & \frac1\pi \,d\lambda,
\end{align*}
with the limits  as $\ell'\to 0$ and $\ell\to \infty$ taken in the sense of the  weak convergence of the measures. Here we used the inequality (see \eqref{ampli})
$$
\tanh\frac{\ell+\ell'}{2}\le A(\Phi)\le \coth\frac{\ell+\ell'}{2}
$$
on the last step that ensures that the amplitude $A(\Phi)$ approaches  $1$ as  $\ell \to \infty$.

\end{remark}




\section{The model dissipative operators}\label{secmodop}

Given a metric graph $\bbY$  in one of the Cases (i)--(iii) and a real parameter $k$, $0\le k<1$, we construct a family of
{\it   model} maximal dissipative differentiation operators
on $\bbY$. In what follows we will refer to the parameter $k$ as  {\it  the quantum gate coefficient  on the graph} $\bbY$.

If the metric graph $\bbY$ is in Case (i),
$$
\bbY=(-\infty, 0)\sqcup (0, \infty),
$$
denote  by  $$\widehat  D=\widehat  D_I(k)=i\frac{d}{dx},\quad   0\le k<1,$$the   maximal dissipative differentiation operator
with  the boundary condition at the origin
\begin{equation}\label{domdomi}
f_\infty (0+)=kf_\infty (0-),\quad 0\le k<1.
 \end{equation}

Notice that the case $k=0$ is exceptional in the sense that the point spectrum of the dissipative operator $ \widehat  D_I(0)$ fills in the entire (open) upper half-plane.

The corresponding strongly continuous semi-group generated by $\widehat  D_I(k)$ describes the motion  from left to right  of a (quantum)
particle which is  emitted outside (see Fig. 1) with probability $1-k^2$
through the quantum gate at the origin and  keeps moving  along the axes with probability $k^2$.

\vskip 1cm

\begin{pspicture}(12,4)%

\psline(.7, 2)(11.5,2)

\pscurve(1,2.2)(1.8,2.4)(2.5,3.5)(3.2,2.4)(4,2.2)

\psline[linewidth=1.3pt]{->}(3,2.8)(4.2,2.8)

\pscurve(8,2.1)(8.8,2.2)(9.5,2.75)(10.2,2.2)(11,2.1)

\psline[linewidth=1.3pt]{->}(10,2.4)(11.2,2.4)

\psline[linewidth=.5pt]{->}(6.5,2)(7.1,3)

\psline[linewidth=.5pt]{->}(6.5,2)(7.1,1)

\psline[linewidth=.5pt]{->}(6.5,2)(7.3,1.7)

\psline[linewidth=.5pt]{->}(6.5,2)(7.3,2.3)

\rput(6.5,3){$k$}

\psdot*[dotscale=2](6.5,2)

\psarc[linestyle=dashed,linewidth=.5pt]{<-}(6.5,2){.7}{20}{160}

\rput(6.5,0.7){{\bf Fig. 1} {\sc Dynamics on  the  metric graph $\bbY$  in Case} (i)}
\rput(6.5,0.2){{\sc with the quantum gate coefficient $k$ }}

\end{pspicture}
\vskip 1cm

If the metric graph $\bbY$ is in Case (ii) and
$$
\bbY=(0, \ell)
$$
for some $\ell>0$,
denote  by  \begin{equation}\label{domdomii**}\widehat  D=\widehat  D_{II}(k,\ell)=i\frac{d}{dx},
\end{equation}
 the   maximal dissipative differentiation operator
 determined by  the boundary condition
\begin{equation}\label{domdomii*}
 f_\ell(0)=kf_\ell(\ell),\quad 0\le k<1.
 \end{equation}

Notice  that  the case $k=0$ is also exceptional. That is, the dissipative differentiation operator $ \widehat  D_{II}(0,\ell)$
corresponding to the boundary condition
\begin{equation}\label{domdomii}
 f_\ell(0)=0
 \end{equation}
 has no spectrum.

The (nilpotent) semi-group generated by   $ \widehat  D_{II}(0,\ell)$ describes the motion of a particle which is emitted with probability one at  the right end-point of the   finite interval $[0,\ell]$ (see Fig. 2).

\begin{pspicture}(12,5)%

\psline(4.11, 3) (8.8,3)

\pscurve(5,3.2)(5.8,3.4)(6.5,4.5)(7.2,3.4)(8,3.2)

\psdot*[dotscale=2](8.8,3)

\psline[linewidth=1.3pt]{->}(7,3.9)(8.2,3.9)

\psarc(4, 3){.12}{0}{360}

\psline[linewidth=.5pt]{->}(9,3)(10,3)

\psline[linewidth=.5pt]{->}(9,3)(9.6,4)

\psline[linewidth=.5pt]{->}(9,3)(9.6,2)

\psline[linewidth=.5pt]{->}(9,3)(9.7,2.7)

\psline[linewidth=.5pt]{->}(9,3)(9.7,3.3)


\rput(6.5,1.5){{\bf Fig. 2} {\sc Dynamics on  the  metric graph $\bbY$  in Case} (ii)}
\rput(6.5,1){{\sc with the quantum gate coefficient $k=0$ }}

\end{pspicture}

Finally, if  the graph $\bbY $ is in Case (iii),
$$\bbY=(-\infty, 0)\sqcup (0, \infty) \sqcup(0, \ell),
$$
denote  by  $\widehat  D= \widehat  D_{III}(k, \ell)=i\frac{d}{dx}$ the   maximal dissipative differentiation operator on $
\bbY$ with   the boundary conditions
\begin{equation}\label{domdomiii}
\begin{pmatrix}
f_\infty(0+)\\
f_\ell(0)
\end{pmatrix}=\begin{pmatrix}
k& 0\\
 \sqrt{1-k^2}&0
\end{pmatrix}
 \begin{pmatrix}
f_\infty(0-)\\
f_\ell(\ell)
\end{pmatrix}, \quad 0< k<1.
\end{equation}

The dynamics associated with the strongly continuous  semi-group  generated by  $ \widehat  D_{III}(k, \ell)$
describes the motion
a  wave-packet
  moving from  left to  right (see Fig. 3).
\vskip 3cm



\begin{pspicture}(12,4)

\pscurve(1,4.2)(1.8,4.4)(2.5,5.5)(3.2,4.4)(4,4.2)

\psline[linewidth=1.3pt]{->}(9.9,5.4)(11.1,5.4)

 \pscurve(8,5.1)(8.8,5.2)(9.5,5.75)(10.2,5.2)(11,5.1)

\psline[linewidth=1.3pt]{->}(8.9,3.3)(9.9,3.3)

 \pscurve(7,3.05)(7.8,3.1)(8.5,3.35)(9.2,3.1)(10,3.05)

\psline[linewidth=1.3pt]{->}(2.9,5)(4.1,5)


\psline(1,4)(6,4)

\psline(6.5,5)(12,5)

\psline(6.5,3)(10.5,3)

\psdot*[dotscale=2](6,4)
\psline(6,4)(6.5,5)
\psline(6,4)(6.5,3)




\psdot*[dotscale=2](10.5,3)


\psarc[linestyle=dashed,linewidth=.5pt]{<-}(6,4){1.5}{50}{170}

\psarc[linestyle=dashed,linewidth=.5pt]{->}(6,4){1.5}{190}{300}

\rput(5.9,5.9){$k$}

\rput(5.9,2){$\sqrt{1-k^2}$}

\psline[linewidth=.5pt]{->}(10.7,3)(12,3)

\psline[linewidth=.5pt]{->}(10.7,3)(11.6,4)

\psline[linewidth=.5pt]{->}(10.7,3)(11.6,2)

\psline[linewidth=.5pt]{->}(10.7,3)(11.7,2.7)

\psline[linewidth=.5pt]{->}(10.7,3)(11.7,3.3)

\rput(6.5,.5){{\bf Fig. 3} {\sc Dynamics on  the  metric graph $\bbY$  in Case} (iii)}
\rput(6.5,0){{\sc with the quantum gate coefficient $k$ }}

\end{pspicture}
\vskip 1cm

 If the packet is  initially supported by the negative semi-axis, after  the interaction with the scatterer located at the center of the graph, the particle
continues its rightward motion along the real axis with its initial shape amplified by the factor $k$
while  a copy  of the  wave-packet   amplified by $\sqrt{1-k^2}$ turns right onto an appendix of length $\ell$
attached to the obstacle. When the  wave-packet
approaches the right end of the interval $[0,\ell]$ the wave is terminated.

From the boundary conditions \eqref{domdomiii} it follows that  the quantum Kirchhoff rule
(at the junction)
\begin{equation}\label{qkr}|f_\infty(0-)|^2= |f_\infty(0+)|^2 +|f_\ell(0)|^2
\end{equation}holds.

Taking into account  wave-particle duality, one can also say that
the corresponding particle with probability $k^2$ keeps moving  along the real axis and with probability $1-k^2$ enters the appendix.
Then, assuming that the initial profile of the wave-packet
was supported by the interval $[-L,0$), the particle   is emitted   with probability one  after time $t=\ell+L$ has elapsed.

Notice that a wave-packet   initially supported to the right of the obstacle moves to the right without changing its shape regardless whether
the wave is  supported by the semi-axis $(-\infty, 0]$ or by the finite interval $[0, \ell]$.
To complete the description of the dynamics in the general case, one applies the superposition principle.

 \begin{remark}\label{remdisd}
Notice that the  boundary conditions \eqref{domdomi}, \eqref{domdomii}  and \eqref{domdomiii}  (but not \eqref{domdomii*} with $k\ne0$) are the local vertex conditions, which means that different vertices do not interact.
In particular, the   domains of the corresponding dissipative  operators $\widehat D$
  are invariant with respect to the group of local  gauge transformations.
As a corollary,  the dissipative operators $\widehat D=\widehat D_{I, II, III}  $  satisfy the commutation relations
 \begin{equation}\label{ccrdisss}
U_t^*\widehat  DU_t=\widehat  D+tI\quad \text{on }\quad \Dom (\widehat  D), \quad t\in \bbR,
\end{equation}
where $U_t=e^{-it \cQ}$ is the unitary group generated by
 the operator $\cQ$ of multiplication by independent variable
on the graph $\bbY$.

This can be justified  immediately  but it also follows from a more general considerations below (cf. Remark \ref{remdotd}).

Let   $\cA(x)$ denote a real-valued piecewise continuous function on $\bbY$.  We remark that  the operators
$\widehat  D$ and $\widehat  D+\cA(x)$ are unitarily equivalent.
Indeed,  let
$\phi(x)$ be  any solution to the differential equation
$$
\phi'(x)=\cA(x),
$$
on the edges of the graph and continuous at the origin $\{0\}\in\bbY$.
Denote by $V$ the unitary local gauge transformation
\begin{equation}\label{gaugege3}
(Vf)(x)=e^{i\phi(x)}f(x),  \quad f \in L^2(\bbY).
\end{equation}
Then, taking into account the boundary conditions \eqref{domdomi},
\eqref{domdomii}  and
\eqref{domdomiii}
one concludes that the domain of $\widehat  D$ is $V$-invariant
$$
V(\dom (\widehat D)) =\dom (\widehat  D),
$$
and  moreover,
$$
\widehat D=V^*(\widehat D+\mathcal{A}(x))V.
$$
In particular,    \eqref{ccrdisss} holds.
\end{remark}

\begin{remark}\label{arlin}
Notice that the model  dissipative differentiation operators $\widehat D$  extend
 the  symmetric differentiation operator $\dot D$ on the graph
 $\bbY$   with   the boundary conditions
\eqref{dotos}, \eqref{dotint} and \eqref{dotdom0}, respectively.  Moreover, the  symmetric operator $\dot D$
is uniquely determined by $\widehat D$
and
$$
\dot D=\widehat D|_\cD\quad \text{where}\quad  \cD=\dom(\widehat D)\cap \dom((\widehat D)^*).$$

If $\widehat D$ is in Case (i), then the corresponding
 symmetric operator $\dot D$ admits a quasi-selfadjoint extension  the point spectrum of which fills in the entire upper half-plane.
Notice that  this property characterizes the operator up to  unitary equivalence. That is, any prime closed  symmetric operator   with deficiency indices $(1,1)$ that admits  a
quasi-selfadjoint  extension
with point spectrum filling in the whole upper half-plane   is unitarily equivalent to the operator $\dot D$ in Case (i)  \cite[Ch. IX, Sec. 114]{AkG}.
Apparently, any point from $\bbC_+$ is an eigenvalue for the  extension  $\widehat  D_I(0)$.

Moreover, the symmetric operator $\dot D$
has a relatively poor family of unitarily inequivalent  (dissipative) quasi-selfadjoint extensions.
The reason is that  any two (dissipative) extensions with the same
absolute value of the von Neumann parameter $k$, $(0\le k \le 1)$,
 are unitarily equivalent to the operator  $\widehat  D_I(k)$.
Recall that the absolute value of the von Neumann parameter of the dissipative operators in question is well defined (see Remark \ref{mnogopros}).
 This property also characterizes $\dot D$ up to unitary equivalence:
any prime closed  symmetric operator with deficiency indices $(1,1)$ that admits two distinct unitarily equivalent
quasi-selfadjoint  extensions is unitarily equivalent to the operator $\dot D$ \cite[Theorem 2]{ArDTs}.

The dissipative operator $\widehat D=\widehat D_{II}(0,\ell)$  in Case (ii) given by \eqref{domdomii**} and \eqref{domdomii*}
is the only dissipative extension of   $\dot D_{II}(\ell)$ whose resolvent set coincides with the whole complex plane $\bbC$.  Moreover,
 any dissipative
quasi-selfadjoint  extension  of a  prime closed  symmetric operator   with deficiency indices $(1,1)$
 without spectrum
 is unitarily equivalent to the symmetric differentiation operator $\dot D_{II}(\ell)$ on a finite interval of length $\ell$  \cite[Theorem 14]{L46}.

We remark  that in contrast to  Case (i), in  Cases (ii) and (iii)  the  maximal dissipative operator $\widehat D$ with the boundary  conditions
\eqref{domdomi} and  \eqref{domdomii}, respectively, is  the only one  maximal dissipative  extension  of the  symmetric  differentiation operator $\dot D$ with  a  gauge invariant domain.
Indeed, the boundary conditions \eqref{bcII} and \eqref{bcIII} are gauge invariant if  either   $\Theta=0$ or  $\Theta=\infty$. If $\Theta=\infty$, the corresponding quasi-selfadjoint extension is not dissipative, which proves the claim.
\end{remark}

The following structure theorem  shows that the differentiation operator $\widehat D$,
 $$\widehat D=\widehat D_{III}(k, \ell),$$
in  Case (iii) can be obtained  as  the result of an operator coupling (spectral synthesis) of  the more ``elementary"  dissipative differentiation operators
$\widehat D=\widehat D_{I}(k)$ in Case (i) and
$\widehat D=\widehat D_{II}(0,\ell)$
in Case (ii). For the concept of an operator coupling we refer to Appendix  \ref{littem}.

\begin{theorem}\label{opcup} The differentiation operator $\widehat D_{III}(k, \ell)$  on the metric graph $$\bbY_{III}=(-\infty, 0]\sqcup [0, \infty) \sqcup[0, \ell]$$
with the quantum gate coefficient $k$  is an operator coupling
of the  differentiation operator $\widehat D_{I}(k)$  on the edge $$\bbY_{I}=(-\infty, 0]\sqcup [0, \infty)$$
with the same quantum gate coefficient $k$
and the operator  $\widehat D_{II}(0,\ell)$  on  the remaining  edge $$\bbY_{II}=[0,\ell]$$ with the quantum gate coefficient $0$, respectively.
 That is,
 $$\bbY_{III}=
 \bbY_{I}\sqcup\bbY_{II} $$
 and
\begin{equation}\label{ccoouu}
\widehat D_{III}(k, \ell)= \widehat D_{I}(k)\uplus \widehat D_{II}(0,\ell).
\end{equation}
\end{theorem}
\begin{proof}

To see that,  set  $\cH_1=L^2((-\infty, \infty))$ and $\cH_2=L^2((0,\ell))$. One observes that
$$
  \widehat D_{III}(k, \ell)|_{ \dom ( \widehat D_{III}(k,\ell))\cap  \cH_1}= \widehat D_{I}(k)
$$
and
hence  the requirement (i) in the definition  \ref{defcoup} (see Appendix G) of a coupling of two dissipative operators is met.

Next,
$
\dom ( \widehat D_{I}(0,k))\oplus \dom ( (\widehat D_{II}(\ell))^*)
$
consists of the three-component functions $f=(f_-, f_+, f_\ell)^T$, with
$$f_-\oplus f_+\oplus f_\ell\in
 W_2^1((-\infty, 0))\oplus W_2^1((0,\infty))\oplus W_2^1((0, \ell))
$$
such that
\begin{equation}\label{bcond2}
 f_+(0+)=
k f_-(0-)
\quad \text{ and } \quad
 f_\ell(\ell-)=0
\end{equation}
and hence
$$
\dom (\dot  D)\subset \dom ( \widehat D_{I}(k))\oplus \dom ( ( \widehat D_{II}(0,\ell))^*).
$$
In particular,
$$
\dot  D\subset  \widehat D_{I}(k)\oplus(\widehat D_{II}(0,\ell))^*
$$
 and hence the requirement (ii) in the definition of an operator coupling is met as well. Therefore,
\eqref{ccoouu} holds.
\end{proof}

For the further references it is convenient to adopt the following Hypothesis.

\begin{hypothesis}\label{tripleops}
Assume that the metric graph $\bbY$ is in one of the Cases (i)--(iii) and $\widehat D$ is the model dissipative operator with the quantum gate coefficient $k$ on
$\bbY$ given by the boundary conditions \eqref{domdomi}, \eqref{domdomii*} and \eqref{domdomiii}, respectively.
 Let $\dot D$ be the restriction of $\widehat D$ on $\cD=\dom(\widehat D)\cap \dom((\widehat D)^*)$.
 Assume that $D_\Theta$ is the self-adjoint reference extension of $\dot D$   referred to in Theorem \ref{gengen} (see \eqref{bcI}, \eqref{bcII}, and \eqref{bcIII}).
If the graph $\bbY$ is in Case (ii)  assume that $k=0$ and if $\bbY$ is in Case (iii) we require that $k\ne0$.

\end{hypothesis}

\begin{definition}\label{modeltrip} Under 
Hypothesis \ref{tripleops} suppose  that  $k=0$ whenever
 the graph $\bbY$ is in Case (ii) and  that $k\ne0$  whenever $\bbY$ is in Case (iii).
Under this assumption
we
call  the triple of differentiation operators $(\dot D, \widehat D, D_\Theta)$
\it the model triple on $\bbY$ with the quantum gate coefficient $k$.
\end{definition}

More explicitly,
each of the differentiation  operators from the triple  $(\dot D, \widehat D, D_\Theta)$   in the Hilbert $L^2(\bbY)$  is given by the differential expression
 $$
 \tau=i\frac{d}{dx}
 $$
(on the edges of the graph $\bbY$) initially defined  on the Sobolev space $W_2^1(\bbY)$,
$$
W_2^1(\bbY)=\begin{cases}
W_2^1((-\infty, 0))\oplus W_2^1((0, \infty))&\\
W_2^1((0,\ell))&\\
(W_2^1((-\infty, 0))\oplus W_2^1((0, \infty))\oplus W_2^1((0,\ell))&
\end{cases}.
$$

In Case (i), the metric graph  has the form
$ \bbY=(-\infty, 0)\sqcup(0,\infty)$,
and
\begin{align*}
\Dom(\dot D)&=\{f_\infty\in W_2^1(\bbY)\, |  f_\infty(0+)=f_\infty(0-)=0\},
\\
\Dom(\widehat D)&=\{f_\infty\in W_2^1(\bbY)\, | \, f_\infty(0+)=kf_\infty(0-)\},
\\
\Dom(D_\Theta)&=\{f_\infty\in W_2^1(\bbY)\, | \, f_\infty(0+)=-\Theta f_\infty(0-)\},
\end{align*}
$$
 0\le k<1.
$$

In Case (ii),
$ \bbY=(0,\ell)$,
\begin{align*}
\Dom(\dot D)&=\{f_\ell \in W_2^1(\bbY)\, | f_\ell(0)=  f_\ell(\ell)=0\},
\\
\Dom(\widehat  D)&=\{f_\ell\in  W_2^1(\bbY)\, | f_\ell(0)=0\},
\\
\Dom(D_\Theta)&=\{f_\ell\in  W_2^1(\bbY)\, | f_\ell(0)=-\Theta f_\ell(\ell)\}.
\end{align*}

  In Case (iii),
$\bbY =(-\infty, 0)\sqcup(0,\infty)\sqcup(0,\ell)$, and
\begin{align*}
\Dom(\dot D)&=\left \{ f_\infty\oplus f_\ell\in  W_2^1(\bbY)\,|\,  \begin{cases}
f_\infty(0+)&=k f_\infty(0-)
\\
f_\ell(0+)&=\sqrt{1-k^2} f_\infty(0-)\\
 f_\ell(\ell)&=0
\end{cases}
\right \},
\\
\Dom(\widehat  D)&=\left \{ f_\infty\oplus f_\ell\in  W_2^1(\bbY)\, |\,  \begin{cases}
f_\infty(0+)&=k f_\infty(0-)
\\
f_\ell(0+)&=\sqrt{1-k^2} f_\infty(0-)\\
\end{cases}
\right \},
\\
\Dom(  D_\Theta )&=\left \{ f_\infty\oplus f_\ell\in  W_2^1(\bbY)\,|\,
 \begin{cases}
f_\infty(0+)&=k f_\infty(0-) + \sqrt{1-k^2}\Theta  f_\ell(\ell)
\\
f_\ell(0+)&=\sqrt{1-k^2} f_\infty(0-)-k\Theta f_\ell(\ell)\\
\end{cases}
\right \},
\end{align*}
$$
0<k<1.
$$




\section{The characteristic function $S_{(\dot D, \widehat D, D_\Theta)}(z)$ }\label{s9}

Now we are ready  to evaluate the characteristic function $S_{(\dot D, \widehat D, D_\Theta)}(z)$ of the triple  $(\dot D, \widehat D, D_\Theta)$  of differentiation operators on a metric graph $\bbY$ in  Cases (i)-(iii).




First, we evaluate the characteristic function of the triple
$(\dot D, \widehat D, D)$
for a particular choice of the reference self-adjoint operator
$D$  referred to in Theorem \ref{gengen},
in Cases (i), (ii) and (iii) with $\Theta=1$, respectively.
Recall that the operator $D$ is determined by the boundary conditions
$$
f_\infty(0+)=- f_\infty(0-),
$$
$$ f_\ell(0)=-f_\ell(\ell),
$$
$$
 \begin{cases}
f_\infty(0+)&=k f_\infty(0-) + \sqrt{1-k^2}  f_\ell(\ell)
\\
f_\ell(0+)&=\sqrt{1-k^2} f_\infty(0-)-k f_\ell(\ell)\\
\end{cases}
$$
in Cases (i)--(iii), respectively.

\begin{lemma}\label{key}
Let   $(\dot D, \widehat D, D)$ be the model triple of the differentiation operators on the metric graph $\bbY$ in one of the Cases $(i)-(iii)$  as above.
Then  the characteristic function of the  triple  $(\dot D, \widehat D, D)$  admits the representation
\begin{equation}\label{dlaS}
S_{(\dot D, \widehat D, D)}(z)= \begin{cases}
 k, & \text{ in Case }  (i)
\\e^{i\ell z}, &  \text{ in Case }  (ii)
\\ke^{i\ell z}, &\text{ in Case }  (iii)
 \end{cases},
 \quad z\in \bbC_+,
\end{equation}
where  $0\le k<1$  $($in Case $(i)$$)$ and $0<k<1$ $($in Case $(iii)$$)$.
\end{lemma}

\begin{proof}

To check \eqref{dlaS} we proceed as follows.

Let $g_\pm$ be  the deficiency elements given by  \eqref{deftip1} in Case (i), \eqref{deftip2} in Case (ii) and
\eqref{hyddef1}, \eqref{hyddef2} in Case (iii).

We claim that  \begin{equation}\label{SDD}
f=g_+-S(i)g_-\in \Dom (\widehat D),
\end{equation}
where we have used the shorthand notation
$$ S(z)=S_{(\dot D, \widehat D, D)}(z).$$

It suffices to check that $f$  satisfies the boundary conditions \eqref{domdomi}, \eqref{domdomii} and \eqref{domdomiii} in Cases (i)-(iii), respectively, and therefore
$$f=g_+-S(i)g_-\in \Dom (\widehat D).$$

 Indeed, in Case (i),
 \begin{equation}\label{for1}
 S(i)=k,
 \end{equation} and hence
 $$
 f_\infty(x)=\sqrt{2}(e^x\chi_{(-\infty,0)}(x)+ke^{-x}\chi_{(0,\infty)}(x)).
 $$
Clearly,
$$
f_\infty(0-)=kf_\infty(0+),
$$
and therefore  $f_\infty \in  \dom (\widehat D)$.

In Case (ii),
 \begin{equation}\label{for2}
S(i)=e^{-\ell},
\end{equation}
 and therefore the element
$$
f_\ell(x)=\frac{\sqrt{2}}{\sqrt{e^{2\ell}-1}}(e^{ x}-e^{-\ell}e^{\ell-x}), \quad x\in [0,\ell],
$$
satisfies the Dirichlet boundary condition
$
f_\ell(0)=0
$
which proves that $f_\ell\in\dom (\widehat D)$.

Finally, in Case (iii) we have that
 \begin{equation}\label{param}S(i)=ke^{-\ell}.\end{equation}
 Therefore, the element
 $
 f=g_+-S(i)g_-=g_+-ke^{-\ell}g_-
 $
admits the representation (see \eqref{hyddef1} and \eqref{hyddef2})
$$f=\frac{\sqrt{2}}{\sqrt{e^{2\ell}-k^2}}(f_\infty, f_\ell)^T,$$
where
$$
f_\infty(x)=\sqrt{1-k^2}e^{x}\chi_{(-\infty, 0)}(x)+ke^{-\ell}\sqrt{1-k^2}e^{\ell-x}\chi_{(0, \infty)}(x),  \quad x\in \bbR,
$$
and
$$
f_\ell(x)=e^x-(ke^{-\ell})ke^{\ell-x}, \quad x\in [0, \ell].
$$
We have \begin{align*}
 f_\infty(0-)&=\sqrt{1-k^2},
 \\ f_\infty(0+)&=k\sqrt{1-k^2},\\
 f_\ell(0)&=1- k^2,\\
\end{align*}
which shows that
 the boundary conditions \eqref{domdomiii} hold. Therefore  $
f\in \Dom(\widehat D).
$

Since  $$g_+-g_-\in \Dom (D)$$ and
$$g_+-S(i)g_-\in \Dom (\widehat D),$$
one computes
  $$
S_{(\dot D,\widehat D, D)}(z)=\frac{s_{(\dot D, D)}(z)-S(i)}
{\overline{S(i)}s_{(\dot D, D)}(z)-1}=\frac{s_{(\dot D, D)}(z)-ke^{-\ell}}
{ke^{-\ell}s_{(\dot D, D)}(z)-1}.
  $$
It remains to remark that since  the Liv\v{s}ic function  $s_{(\dot D, D)}(z)$   is given by \eqref{srepr},  one gets    \eqref{dlaS}  by  a direct computation.

\end{proof}

\begin{remark}

Notice that    Lemma \ref{key},  in particular, states  that the characteristic function of the  triple  $(\dot D, \widehat D, D)$  in Case (iii)
is the product of the characteristic functions  in Cases (i) and (ii), respectively. That is,
 \begin{equation}\label{prodrule}
S_{(\dot D,\widehat D, D)}(z)=k\cdot e^{i\ell z}.
 \end{equation}
In view of Theorem \ref{opcup}, the rule \eqref{prodrule} can also be obtained
 as  a corollary  of the Multiplication Theorem \ref{opcoup} in Appendix G.

Also, comparing \eqref{for1}, \eqref{for2}  and \eqref{param}, one observes  that the von Neumann parameter $ke^{-\ell}$ of the coupling  $\widehat D_{III}(k, \ell)$ associated with the bases \eqref{hyddef1} and  \eqref{hyddef2}
        is the product of the von Neumann parameters $k$ and $e^{-\ell}$
     of the dissipative operators $\widehat D_{I}(k)$ and $\widehat D_{II}(\ell)$
       with respect to the bases \eqref{deftip1} and \eqref{deftip2}, respectively (see Remark \ref{vonbas} for the terminology).
        This observation
      illustrates   the multiplicativity property for  the absolute values of the  von Neumann extension parameters  under coupling
(see   \cite[Theorem  5.4]{MT-MAT} or Theorem \ref{opcoup} in Appendix G).
Recall that the concept of  absolute value of the von Neumann parameter  is well defined by Remark \ref{mnogopros}.
\end{remark}

The  more general  result  below  can be understood as the solution of the following inverse problem: find a triple with a prescribed characteristic function referred to in Theorem  \ref{lem:diss}  (cf. \cite[Theorem 20]{JM}).

\begin{theorem}\label{ind12}
The characteristic function $S_{(\dot D, \widehat D, D_\Theta)}(z)$ of the  model triple  $(\dot D, \widehat D, D_\Theta)$, $|\Theta|=1$, admits the representation
\begin{equation}\label{dlaSS}
S_{(\dot D, \widehat D, D_\Theta)}(z)= e^{-2i\alpha}\begin{cases}
 k, & \text{ in Case   $(i)$}
\\e^{i\ell z}, &  \text{ in Case   $(ii)$}
\\ke^{i\ell z}, &\text{ in Case   $(iii)$}
 \end{cases},
 \quad z\in \bbC_+,
 \end{equation}
where  $0\le k<1$  $($in Case $(i)$$)$ and $0<k<1$ $($in Case $(iii)$$)$.

Here  $\alpha$ and the boundary condition parameter $\Theta$ are related as follows
  \begin{equation}\label{prostoeqqq}
e^{2i\alpha}=
\begin{cases}
 \Theta, & \text{ in Case   $(i)$}
\\ \frac{\Theta+e^{-\ell}}{e^{-\ell}\Theta+1},&  \text{ in Case   $(ii)$}
\\ \frac{\Theta+e^{-(\ell+\ell')}}{e^{-(\ell+\ell')}\Theta+1},& \text{ in Case   $(iii)$, \quad
with}
 \quad
\ell'=\ln \frac1k.
 \end{cases}.
\end{equation}

In particular,
\begin{equation}\label{xyxyxy2}
S_{(\dot D, \widehat D, D_\Theta)}(z)=e^{-2i\alpha}S_{(\dot D,\widehat D,  D)}(z).
\end{equation}
\end{theorem}

\begin{proof} In view of Lemma \ref{ind0} in Appendix $E$, the assertion of the theorem  is a direct consequence of Theorem \ref{ind11}.
\end{proof}

\begin{remark}\label{singr}
Observe that  in Case (ii)
the characteristic function of the triple  $(\dot D, \widehat D, D_\Theta)$ is  a singular inner function with
``mass at infinity."
\end{remark}

\begin{corollary}\label{gogol} {\it The model triples  $(\dot D, \widehat D, D_\Theta)$, $|\Theta|=1$, and  $(\dot D, \widehat D, D_{\Theta'})$, $|\Theta'|=1$,
$\Theta'\ne \Theta$, are not mutually unitarily equivalent unless the graph $\bbY$ is in Case $(i) $ and the point spectrum of the dissipative differentiation operator $\widehat D$ fills in the
whole upper half-plane $\bbC_+$.
 In the latter case the triples in question are mutually unitarily equivalent to one another for any $\Theta$ and $\Theta'$ $(|\Theta|=|\Theta'|=1)$.
}
\end{corollary}

\begin{proof}
Combining \eqref{dlaSS}, \eqref{prostoeqqq} and \eqref{xyxyxy2} shows that the characteristic functions   $S_{(\dot D,  \widehat D, D_{\Theta})}(z)$
and  $S_{(\dot D,  \widehat D, D_{\Theta'})}(z)$ are different for $\Theta\ne \Theta'$ unless
$$
 S_{(\dot D,  \widehat D, D_{\Theta})}(z)=0 \quad \text{for all} \quad z\in \bbC_+
$$
for some (and therefore  for all) $|\Theta|=1$. In the latter case $\widehat D$ has no regular points in the upper half-plane and the
 triples  $(\dot D, \widehat D, D_\Theta)$, $|\Theta|=1$, and  $(\dot D, \widehat D, D_{\Theta'})$, $|\Theta'|=1$, are mutually unitarily equivalent by the uniqueness Theorem \ref{unitar}  in Appendix C.
 \end{proof}

\begin{remark}\label{remgend} Let  $\cA$ be  a real-valued piecewise continuous function on the metric graph  $\bbY$ in  Cases (i)-(iii). Combining Remarks \ref{remdotd}, \ref{remsamd} and \ref{remdisd}
imply that   if the graph $\bbY$ is in Case (i), then
 the triple   $(\dot D+\mathcal{A}(x), \widehat  D+\mathcal{A}(x), D_\Theta+\mathcal{A}(x))$ is  mutually unitarily equivalent
 to the triple $ (\dot D, \widehat D, D_{\Theta })$. If the graph is in Cases (ii) or (iii), then  $(\dot D+\mathcal{A}(x), \widehat  D+\mathcal{A}(x), D_\Theta+\mathcal{A}(x))$
is  mutually unitarily equivalent to the triple
  $ (\dot D, \widehat D, D_{\Theta e^{-i\Phi}})$,
where
 $$
\Phi=\int_0^\ell\mathcal{A}(x)dx .$$
Moreover, the corresponding  unitary equivalence is given by a gauge transformation.
\end{remark}

The knowledge of the characteristic function of the triple   $(\dot D, \widehat D, D_\Theta)$, which is its complete unitary invariant, enables us to obtain  the converse to structure Theorem \ref{lem:diss}.

  \begin{theorem} \label{thmconv}
 Let $(\dot A, \widehat A, A)$  be a triple of operators in a  Hilbert space $\cH$
where $\dot A$ is a prime symmetric  operator with deficiency indices $(1,1)$, $\widehat  A$ is its  dissipative quasi-selfadjoint extension  and $A$ is a reference self-adjoint extension of $\dot A$

Suppose that  the characteristic function $S(z)=S_{(\dot A, \widehat A, A)}(z)$  of the triple   $(\dot A, \widehat A, A)$ admits the representation
\begin{equation}\label{savtoconv}
S(z)=k e^{i\ell z}, \quad z\in \bbC_+,
\end{equation}
for some  $|k|\le1$ and $\ell \ge 0$.  We also assume that  if $\ell=0$, then necessarily $|k|<1$ and if $|k|=1$, then $\ell>0.$

Then there exists a unitary group $U_t$ such that the domains $ \Dom ( \dot  A)$ and $ \Dom ( \widehat A)$ are $U_t$-invariant and
\begin{equation}\label{orderconv0}
U_t^* \dot  AU_t=\dot  A+tI\quad \text{on }\quad \Dom ( \dot  A),
\end{equation}
\begin{equation}\label{orderconv}
U_t^* \widehat AU_t=\widehat A+tI\quad \text{on }\quad \Dom ( \widehat A).
\end{equation}

\end{theorem}

\begin{proof} By the uniqueness Theorem \ref{unitar} in Appendix C, the triple   $(\dot A, \widehat A, A)$ is mutually unitarily equivalent to the triple $(\dot D, \widehat D, D_\Theta)$
referred to in Theorem \ref{ind12} for some choice of the extension parameter $\Theta$.
By \eqref{ccrsym} and \eqref{ccrdisss},
$$
e^{it \cQ}\dot  De^{-it \cQ}=\dot   D+tI\quad \text{on }\quad \Dom (\dot  D), \quad t\in \bbR,
$$
and
$$
e^{it \cQ}\widehat  De^{-it \cQ}=\widehat  D+tI\quad \text{on }\quad \Dom (\widehat  D), \quad t\in \bbR,
$$
where  $\cQ$ the self-adjoint  operator of multiplication by independent variable  on the graph $\bbY$. Pulling back the group  $e^{-it \cQ}$ in $L^2(\bbY)$ to the original Hilbert space
$\cH$ proves the assertion.
\end{proof}

It is worth mentioning that the choice of the orientation of the graph $\bbY$ was {\it ad hoc} from the very beginning. To complete the exposition,
along with the graph $\bbY$,
consider   the metric graph  $\bbY^*$ obtained from   $\bbY$ by reversing the orientation.  The corresponding differentiation operator $-(\widehat D)^*$ on $\bbY^*$
extends the (symmetric)  differentiation operator
$$- \dot D=-i\frac{d}{dx},$$
and its domain is determined
by the following boundary conditions

in Case (i):
\begin{align}
 f_\infty(0-)=kf_\infty(0+);    \label{dotdom*1}
\end{align}

in Case (ii):
\begin{align}
f_\ell(\ell)=0;  \label{dotdom*2}
\end{align}

 in Case (iii):

 \begin{align}
\begin{cases}
f_\infty (0-)&=k f_\infty (0+)+\sqrt{1-k^2} f_\ell(0)
\\
f_\ell(\ell)&=0
\end{cases}.   \label{dotdom*}
\end{align}

Notice  that    the graph $\bbY^*$ in Case (iii)  (as opposed to the graph $\bbY$)
has only two incoming and only one outgoing bonds which is reflected in a slightly different way of posing boundary conditions (cf. \eqref{domdomiii} and \eqref{dotdom*}).
Meanwhile, both  $-(\dot D)^*$ and $-(\widehat D)^*$ solve the commutation relations \eqref{ccrdisss}.

On  the  algebraic level, it can   be seen   by observing that
 the relations
$$
U_t^* \widehat AU_t=\widehat A+tI\quad \text{on }\quad \Dom ( \widehat A).
$$
and
$$
V_t^* (-\widehat A)^*V_t=(- \widehat A)^*+tI\quad \text{on }\quad \Dom ( \widehat A),
$$
 with $V_t=U_{-t}$, imply one another.

In fact, reversing  the orientation of the graph $\bbY$ does not lead to the new solutions as far as the classification  up to unitary equivalence  is concerned.

\begin{lemma}\label{clon}  The triples $(\dot D, \widehat D, D_\Theta)$ and  $(-\dot D, -(\widehat D)^*, -D_{\overline{\Theta}})$ are mutually unitarily equivalent.
\end{lemma}

\begin{proof}  From  Theorem \ref{ind12} it follows that
$$
S_{  (\dot D, \widehat D, D_{\Theta})}(z)=e^{-2i\alpha}ke^{i\ell z},
$$
where $\alpha$ is given by \eqref{prostoeqq}.
Applying Lemma \ref{comin} in Appendix F,  we have that
$$
S_{  (-\dot D, -(\widehat D)^*,- D_{\overline{\Theta}})}(z)=\overline{S_{  (\dot D, \widehat D, D_{\overline{\Theta}})}(-\overline{z})}=e^{-2i\alpha}ke^{i\ell z}=S_{ (\dot D, \widehat D, D_\Theta)}(z),
$$ which ensures a mutual unitary equivalence of the triples  in question by the uniqueness Theorem \ref{unitar} in Appendix C (the symmetric operator $\dot D$ is prime by Lemma \ref{primeD}, so is  the operator $-\dot D$).

\end{proof}

\section{The transmission coefficient and the characteristic function }\label{transmission}

Recall that if the  metric graph $\bbY$ is in Cases (ii) or (iii),  the differentiation operators $D_\Theta$, $\Theta\in \bbC\cup \{\infty\}$, referred to in Theorem \ref{allset}
satisfy the commutation relations
$$
U_t^* D_\Theta U_t= D_\Theta+tI, \quad t\in \frac{2\pi}{\ell}\bbZ,
$$
  with respect to a discrete subgroup of one-parameter strongly continuous group of unitary operators $U_t$. On fact, one can choose
  $$U_t=e^{-it\cQ}, \quad t\in \bbR, $$ where $\cQ$ is the multiplication operator by independent variable on the graph $\bbY$.
 However, in the exceptional cases  $\Theta=\infty$, the semi-Weyl relations
 $$
 U_t^* D_\Theta U_t= D_\Theta+tI, \quad t\in \bbR,
\quad
(\Theta=0\quad \text{or}\quad \Theta=\infty),
$$
 hold ($D_\infty=-D_0^*$).

Suppose that the metric graph $\bbY$  is in Case (ii), that is,
$$
\bbY=(0,\ell).
$$

Our fist goal is to evaluate the characteristic function of  a dissipative  triple   on the graph $\bbY$.

To be more specific,
let
  $(\dot d, \widehat d_\Theta, d)$ be the triple, where
 $\dot d$ is the symmetric differentiation on
$$
\Dom (\dot  d)=\{f_\ell \in W_2^1((0, \ell))\, |\, f_\ell(0)= f_\ell(\ell)=0\},
$$
$$\widehat d_\Theta=\widehat D_{II}(-k\Theta, \ell)\quad (0<k<1)$$ is the  maximal dissipative  differentiation operator  on
$$
\Dom (\widehat d_\Theta)=\{f_\ell \in W_2^1((0, \ell))\, |\, f_\ell(0)=-k\Theta f_\ell(\ell)\},
$$
and  $d$ is  the self-adjoint differentiation  operator   defined on
$$
\dom (d)=\{f_\ell \in W_2^1((0, \ell))\, |\, f_\ell(0)=- f_\ell(\ell)\}.
$$

\begin{lemma}
The characteristic function  $ S_{(\dot d, \widehat d_\Theta, d)}(z)$ of the triple
$(\dot d, \widehat d_\Theta, d )$ has the form
\begin{equation}\label{charchar}
 S_{(\dot d, \widehat d_\Theta, d)}(z)=
 \frac{\Theta+e^{-\ell }k}{k e^{-\ell} \Theta+1}\cdot B(z), \quad z\in \bbC_+,
\end{equation}
where
$$
B(z)=\frac{e^{i\ell z}+k\Theta }{\Theta+e^{i\ell z}k}=\frac{\cos\frac{\ell z-\arg \Theta -i\ell'}{2}}{\cos\frac{\ell z-\arg \Theta +i\ell'}{2}}, \quad z\in \bbC_+,
$$
 is  a pure Blaschke product with simple
zeros $z_n$ given by
\begin{equation}\label{gamovpol}
z_n=i\frac{\ell'}{\ell}+(2\pi n+1)\frac{\pi}{\ell}+\frac{\arg\Theta}{\ell} , \quad\text{with}\quad \ell'=\log \frac1k,
\end{equation}
$$
0<k<1.
$$

In particular, the characteristic function $S_{(\dot d, \widehat d_\Theta, d)}(z)$ of the triple  $(\dot d, \widehat d_\Theta, d)$ is a periodic function
$$
S_{(\dot d, \widehat d_\Theta, d)}\left  (z+\frac{2\pi}{\ell}\right)=S_{(\dot d, \widehat d_\Theta, d)}(z), \quad z\in \bbC_+,
$$
with the minimal  period $ T=\frac{2\pi}{\ell}$.

\end{lemma}

\begin{proof}

Let  $g_\pm$ be the deficiency elements of the symmetric operator  $\dot d$ given by \eqref{deftip2}.

Since $\widehat  d_\Theta$ is a quasi-selfadjoint extension of $\dot d$, its domain can be represented as
$$
\Dom ( \widehat  d_\Theta)=\Dom (\dot d) \dot+\text{span}\{g_+-\varkappa g_-\},
$$
for some $|\varkappa|<1$.
In particular the function  $g_+-\varkappa g_-$ satisfies the boundary condition
$$
(g_+-\varkappa g_-)(0)=-k\Theta (g_+-\varkappa g_-)(\ell),
$$
which allows to relate  the von Neumann  extension parameter $\varkappa$ and the coefficient $ k$  as
 $$
1-\varkappa e^{\ell}=-k\Theta  (e^\ell - \varkappa).
$$
Therefore,
$$\varkappa=\frac{k \Theta+e^{-\ell}}{k\Theta e^{-\ell}+1}.
$$

Since
$
g_+-g_-\in \dom (d),
$
one computes  (see \eqref{obr}) that
\begin{equation}\label{slif}
S(z)=S_{(\dot d,\widehat d_\Theta, d)}(z)=\frac{s(z)-\varkappa}{\overline{\varkappa}s(z)-1},
\end{equation}
where $s(z)=s_{(\dot d, d)}(z)$ is the Liv\v{s}ic function associated with the pair  $(\dot d, d)$
given by Lemma \ref{pomni} as
$$
s(z)= s_{(\dot d, d)}(z)=\frac{e^{i\ell z}-e^{-\ell}}{e^{-\ell}e^{i\ell z}-1}.
$$

We claim   that the inner functions $B(z)$ and $S(z)$
 have the same set of (simple) roots.

 Indeed, if $S(z_0)=0$, then
 \begin{equation}\label{vanish}
s(z_0)=\frac{e^{i\ell z_0}-e^{-\ell}}{e^{-\ell}e^{i\ell z_0}-1}=\varkappa.
\end{equation}
Therefore,
$$
e^{i\ell z_0}=\frac{e^{-\ell}- \varkappa}{1-\varkappa e^{-\ell}}=\frac{1-e^{\ell} \varkappa}{e^{\ell}-\varkappa}=-k\Theta,
$$
which implies that $B(z_0)=0$, and vice versa.

 Since both $B(z)$ and $S(z)$ are Blaschke products, we get
\begin{equation}\label{inview1}
B(z)=\frac{B(i)}{S(i)}S(z).
\end{equation}
To complete the proof it remains to observe that
\begin{equation}\label{inview2}
\frac{B(i)}{S(i)}=\frac{B(i)}{\varkappa}=\frac{e^{-\ell }+k\Theta }{\Theta+e^{-\ell }k}\cdot  \frac{k\Theta e^{-\ell}+1}{k \Theta+e^{-\ell}}
=\frac{k\Theta e^{-\ell}+1}{\Theta+e^{-\ell }k}.
\end{equation}

\end{proof}

Next, recall  that    by Theorem \ref{dilthm} the self-adjoint magnetic Hamiltonian  $D_\Theta$, $|\Theta|=1$,
on the metric graph  $$\bbY=(-\infty, 0)\sqcup (0,\infty)\sqcup (0,\ell) $$
in Case (iii)    referred to in  Theorem \ref{dilthm}  dilates the   maximal dissipative differentiation operator $\widehat d_\Theta$.

  Define the  transmission coefficient $t(\lambda )$ in the scattering problem on the graph $\overline{\bbY}$  (obtained from  $\bbY $  by identifying the right vertex of the edge $[0, \ell]$ with the vertex at the origin)
as  the amplitude
of the   generalized eigenfunction of the Hamiltonian $D_\Theta$, the solution  $f$ to the equation
\begin{equation}\label{scatprob}
i \frac{d}{dx}f=\lambda f\quad  \text{on }\quad \bbY=(-\infty, 0)\sqcup (0,\infty)\sqcup (0,\ell)
\end{equation}   that coincides with   $e^{-i\lambda x}$ on the incoming edge $(-\infty, 0)$ of the graph $\bbY$,
 equals  $t(\lambda) e^{-i\lambda x}$ on the outgoing edge $(0,\infty)$,
  and $f=(f_\infty, f_\ell)$  satisfies the boundary conditions  \eqref{bcIII},
\begin{align}
\begin{pmatrix}
f_\infty(0+)\\
f_\ell(0)
\end{pmatrix}&=\begin{pmatrix}
k& \sqrt{1-k^2} \Theta\\
 \sqrt{1-k^2}&-k\Theta
\end{pmatrix}
 \begin{pmatrix}
f_\infty(0-)\\
f_\ell(\ell)
\end{pmatrix}. \label{bcIIII}
\end{align}

The analytic counterpart  of the dilation Theorem \eqref{dilthm} is  as follows.

 \begin{theorem} The  transmission coefficient $t(\lambda )$  in the scattering problem \eqref{scatprob}, \eqref{bcIIII} has  the form
\begin{equation}\label{vyrt}
t(\lambda )=\frac{\Theta+e^{i\ell \lambda }k}{e^{i\ell \lambda }+k\Theta }, \quad \lambda \in \bbR.
\end{equation}

\end{theorem}
\begin{proof}
Let $f$ be the solution to the scattering problem \eqref{scatprob}.

We have
 $$
f_\infty (\lambda , x) =e^{-i \lambda x}\chi_{(-\infty, 0)}(x)+t(\lambda )e^{-i \lambda x}\chi_{(0, \infty)}(x), \quad x\in \bbR,
$$
and
$$
f_\ell (\lambda ,x)=a(\lambda)e^{-i\lambda x}, \quad x\in (0, \ell).
$$
From \eqref{bcIII} it follows that
$$
\begin{pmatrix}
t(\lambda )\\a(\lambda)
\end{pmatrix}=\begin{pmatrix}
k& \sqrt{1-k^2} \Theta\\
 \sqrt{1-k^2}&-k\Theta
\end{pmatrix} \begin{pmatrix}
1\\a(\lambda)e^{-i\lambda \ell}
\end{pmatrix}.
$$
Solving for $a(\lambda)$ we get that
$$
a(\lambda )=\frac{\sqrt{1-k^2}}{1+k \Theta e^{-i\lambda\ell}}
$$
and hence
\begin{align}
t(\lambda)&=k+ \sqrt{1-k^2} \Theta e^{-i \lambda\ell}a(k)=k + \sqrt{1-k^2} \Theta e^{-i\lambda\ell}\frac{\sqrt{1-k^2}}{1+k\Theta e^{-i\lambda \ell}}
\nonumber\\
&=\frac{\Theta+e^{i\ell \lambda}k}{e^{i\ell \lambda}+k\Theta }, \quad \lambda\in \bbR. \nonumber
\end{align}
\end{proof}

\begin{remark}
We observe  that  if one sets $\Theta=1$ in \eqref{vyrt}, then
$$
t(\lambda)=\frac{1}{s_{\ell'}\left (\frac\ell{\ell'}\lambda\right )},
$$
where $\ell'=\log \frac1k$ and
$$
s_{\ell'}(z)=\frac{e^{iz\ell'}-e^{-\ell'}}{e^{-\ell'} e^{iz\ell'}-1}
$$
is the Liv\v{s}ic function associated with the pair $(\dot D, D)$ on  the metric graph
$$\bbY=(0,\ell')$$
in Case (ii)   referred to in Lemma \ref{pomni}.
\end{remark}

\begin{corollary}\label{transs}{\it
Let $t(\lambda)$ be the transmission coefficient in the scattering problem \eqref{scatprob}, \eqref{bcIIII}.
}

{\it
 Then,
 \begin{equation}\label{rama}
t(\lambda)=  \frac{\Theta+e^{-\ell }k}{k e^{-\ell} \Theta+1}\cdot S^{-1}(\lambda+i0), \quad \lambda\in \bbR,
\end{equation}
where  $S(z)$ is  the characteristic function of the triple $(\dot d, \widehat d_\Theta, d)$ in Case $(ii)$.
}

{\it
In particular,  the poles of the analytic (meromorphic)  continuation of the   transmission coefficient $t(\lambda )$ to the upper half-plane
 coincide with  the eigenvalues of the dissipative operator $\widehat d_\Theta$.
}

\end{corollary}
\begin{proof}
Representation \eqref{rama} follows from \eqref{charchar} and \eqref{vyrt}. Since $S(z)$ is analytic in $\bbC_+$, the transmission coefficient $t(\lambda)$
can be meromorphically continued on the whole complex plane. The (simple) poles of this continuation are located at the zeroes of the characteristic function $S(z)$
which are the eigenvalues of the dissipative operator $\widehat d_\Theta$.

\end{proof}

\begin{remark}  The exceptional case of  the triple $(\dot d, \widehat  d, d)$,
where    $\widehat d$  is  the differentiation operator   on
\begin{equation}\label{dominavar}
\Dom (\widehat d)=\{f_\ell \in W_2^1((0, \ell))\, |\, f_\ell(0)=0\}
\end{equation}
deserves a special discussion.
Notice that $\widehat d$  is the only one  dissipative quasi-selfadjoint   extension of  the symmetric differentiation $\dot d$ which is not in the family \eqref{domneinv}
with $k\ne 0$ (cf.
 Frostman's observation of general character:
if $S$ is an inner function,  then $\frac{S-\varkappa}{1-\overline{\varkappa} S}$ is  a pure Blaschke product   for almost all
$\varkappa \in \bbD$).

By Lemma \ref{key}, the characteristic function of the triple
$$
S_{(\dot d, \widehat  d, d)}=e^{i\ell z}, \quad z\in \bbC_+,
$$
 is a singular inner function (see Remark \ref{singr}).

 On the other hand, the transmission coefficient of the self-adjoint  dilation $D$ of the dissipative operator  $\widehat d$  on the one-cycle graph
 can be evaluated by solving  the equation $$
i \frac{d}{dx}f=\lambda f
$$
 on the metric graph $\bbY$ with boundary conditions  \eqref{bcIII} with $k=0$, and $\Theta=1$,
 $$
\begin{pmatrix}
f_\infty(0+)\\
f_\ell(0)
\end{pmatrix}=\begin{pmatrix}
0& 1 \\
 1&0
\end{pmatrix}
 \begin{pmatrix}
f_\infty(0-)\\
f_\ell(\ell)
\end{pmatrix}$$
  to get an analog of \eqref{rama}.  In this case,
  $$
  t(\lambda)=e^{-i\ell z}=S^{-1}_{(\dot d, \widehat  d, d)}(\lambda),  \quad \lambda\in \bbR.
  $$

\end{remark}

\section{Uniqueness results}\label{sec12}

So far we were interested in  characterizing solutions to  the commutation relations \eqref{1/2Wdot}
under the assumption   that  the unitary group $U_t$ is  given.
Our  next  goal is to show that  if  the symmetric operator $\dot A$  (from Hypothesis \ref{muhly}) is a prime symmetric operator, then
the commutation relations  \eqref{1/2Wdot}  uniquely  determine the group $U_t$ up to a character $t\to e^{it\mu}$ (with the only one exception).

We start with a preliminary observation  that  a prime symmetric operator with   deficiency indices $(1,1)$
  has a rather poor set of symmetries.

\begin{lemma}\label{nudada}
 Suppose that
$
\dot A$
is a symmetric operator with deficiency indices $(1,1)$ and $\cU$  is a unitary operator. Assume that  the  operator $\cU$  commutes with
$\dot A$ in the sense that
\begin{equation}\label{ravny}
\cU(\dom(\dot A))=\dom(\dot A)
\end{equation}
and
$$
\dot A\cU f=\cU \dot A f\quad \text{for all} \quad f\in \dom(\dot A).
$$
Then   the subspaces  $$
 \cH_\pm=\overline{\Span_{z\in \bbC_\pm}\Ker ((\dot A)^*-zI)}
$$ are invariant for $\cU$. Moreover,   the corresponding restrictions of $\cU$ onto those subspaces are multiples of the identity. That is,
\begin{equation}\label{upm}
\cU|_{\cH_\pm}=\Theta_\pm I_{\cH_\pm}\quad \text{for some } \quad |\Theta_\pm|=1.
\end{equation}

\end{lemma}

\begin{proof} Suppose that $f\in \dom ((\dot A)^*)$. Then
$$
(\dot Ag,\cU f)=(\cU^*\dot A g,f)= (\cU^*\dot A \cU\cU^*g,  f)\quad  \text{for all} \quad g\in \dom (\dot A).
$$
From  \eqref{ravny} it follows that   $ \cU^*g\in \Dom(\dot A)$, $g\in \dom(\dot A)$.
Therefore,
$$
(\cU^*\dot A \cU\cU^*g,  f)=(\dot A \cU^*g,  f)=(g, \cU (\dot A )^*f).
$$
That is,
$$
(\dot Ag,\cU f)=(g, \cU (\dot A )^*f)\quad  \text{for all} \quad g\in \dom (\dot A),
$$
which means that
$$
\cU (\dom ((\dot A)^*))\subset   \dom ((\dot A)^*)
$$
and
\begin{equation}\label{arttt}
(\dot A)^*\cU=\cU (\dot A)^*
\quad
\text{on}\quad  \dom ((\dot A)^*).
\end{equation}

Since \eqref{arttt} holds,
 the deficiency subspace $$N_z=\Ker ((\dot A)^*-zI), \quad z\in \bbC_+,$$
is an eigensubspace of  $\cU$.
Therefore,  the subspace
$$
\cH_+=\overline{\Span_{z\in \bbC_+}\Ker ((\dot A)^*-zI)}
$$
is invariant for $\cU$.

Next we claim   that  the deficiency  subspaces $N_z$ and  $N_\zeta$
 are not orthogonal to each other  for $z,\zeta \in\bbC_+$.

Indeed, let $A$ be a self-adjoint extension of $\dot A$. Suppose that
 $g_+\in N_i$ with  $g_+\ne 0$.
 Then the element
$g_z=(A-iI)(A-zI)^{-1}g_+$ generates the subspace $N_z$,
$\Im(z)\ne 0$.
Therefore,
\begin{align*}
(g_{z},g_\zeta)&=\left ((A-iI)(A-zI)^{-1}g_+,(A-iI)(A-\zeta I)^{-1}g_+\right )
\\&
=\int_\bbR\frac{\lambda^2+1}{(\lambda-z)(\lambda-\overline{\zeta})}d\mu(\lambda)
\\
&
=\frac{1}{z-\overline{\zeta}}\int_\bbR\left (\frac{1}{\lambda-z}-
\frac{1}{\lambda-\overline{\zeta}}\right )(\lambda^2+1)d\mu(\lambda)
,\end{align*}
where $\mu(d\lambda)$ is the spectral measure of the element $g_+$ associated with
the self-adjoint operator $A$. That is,
$$
d\mu(\lambda)=(dE(\lambda)g_+, g_+),
$$
where $E(\lambda)$ is the resolution of identity for the self-adjoint operator $A$,
$$
A=\int_\bbR\lambda dE(\lambda).
$$

 Clearly,
$$
\Im\left (\frac{1}{\lambda-z}-
\frac{1}{\lambda-\overline{\zeta}}\right )>0,\quad \text{whenever} \quad z,\zeta \in \bbC_+,
$$
and therefore
$$
(g_{z},g_\zeta)\ne 0, \quad z,\zeta\in \bbC_+,
$$
which proves the claim.

Finally, since   $\Ker ((\dot A)^*-zI)$ and  $\Ker ((\dot A)^*-\zeta I)$ for $z,\zeta \in \bbC_+$ are not orthogonal to each other,  the restrictions of $\cU$ onto
these subspaces have the same eigenvalues, proving that
the restriction of $\cU$ onto $\cH_+$
is a multiple of the identity.

The same reasoning
 shows that   the restriction of $\cU$ onto its invariant subspace
$$
\cH_-=\overline{\Span_{z\in \bbC_-}\Ker ((\dot A)^*-zI)}
$$
is also a (possibly different) multiple of the identity as well.

The proof is complete.

\end{proof}

\begin{lemma}\label{uniuni}  Suppose that
$
\dot A$
is a symmetric operator with deficiency indices $(1,1)$.
Assume that $U_t$ and $V_t$ are strongly continuous unitary groups such that the commutation relations
$$
U_t^*\dot AU_t=V_t^*\dot AV_t=\dot A+tI\quad \text{on }\quad \Dom ( \dot  A)
$$
hold.

Then the subspaces  $$
 \cH_\pm=\overline{\Span_{z\in \bbC_\pm}\Ker ((\dot A)^*-zI)}
$$ reduce  the groups $U_t $ and $V_t$ and
$$
V_t|_{\cH_\pm}=e^{it\mu_\pm}U_t|_{\cH_\pm}\quad \text{for some} \quad \mu_\pm\in\bbR.
$$

\end{lemma}
\begin{proof} In a similar way as in the proof of Lemma  \ref{nudada}, one observes that
the commutation relations
$$
U_t^*(\dot A)^*U_t=V_t^*(\dot A)^*V_t=(\dot A)^*+tI\quad \text{on }\quad \Dom ( (\dot  A)^*)
$$
for the adjoint operator $(\dot A)^*$ hold.

Since  obviously
$$
U_t(\Ker ((\dot A)^*-zI))=V_t(\Ker ((\dot A)^*-zI))=(\Ker ((\dot A)^*-(z-t)I)), \quad t\in \bbR,
$$
the subspaces $\cH_\pm$ are invariant for both  $U_t$ and $V_t$ for all $t$, and therefore $\cH_\pm$ reduce the groups $U_t$ and $V_t$.
Since   the unitary operator  $\cU_t=U_t^*V_t$  commutes with $\dot A$ on $\dom(\dot A)$,
by Lemma \ref{nudada}  one gets that
\begin{equation}\label{adad}
U_t^*V_t|_{\cH_\pm}=\cU_t|_{\cH_\pm}=e^{i\alpha_\pm (t)} I_{\cH_\pm}, \quad t\in \bbR,
\end{equation}
for some continuous real-valued functions $\alpha_\pm(t)$ (the continuity of the function $\alpha(t)$ follows from the hypothesis that the groups $U_t$ and $V_t$ are strongly continuous).
That is,
$$
V_t=e^{i\alpha_\pm(t)}U_t \quad \text{on}\quad \cH_\pm.
$$
Since $U_t$ and $V_t$ are one-parameter groups, it follows that the functional equation
$$
\alpha_\pm(t+s)=\alpha_\pm (t)+\alpha_\pm (s)
$$
holds
and hence, due to the continuity of $\alpha_\pm$, we conclude  that
$$
\alpha_\pm (t)=\mu_\pm  t,
$$
for some $\mu_\pm \in \bbR$, which combined  with \eqref{adad} completes the proof.
\end{proof}

\begin{theorem}\label{maininv}
Suppose that   $\widehat A$ is a maximal  dissipative extension of a prime symmetric operator  $\dot A$  with deficiency indices $(1,1)$.

Assume that $U_t$ and $V_t$ are strongly continuous unitary groups such that the commutation relations
\begin{equation}\label{zilot}
U_t^*\widehat  AU_t=V_t^*\widehat  AV_t=\widehat A+tI\quad \text{on }\quad \Dom ( \widehat   A)
\end{equation}
hold.

If $\widehat A$ is self-adjoint assume, in addition, that
\begin{equation}\label{zilotdot}
U_t^*\dot AU_t=V_t^*\dot AV_t=\dot A+tI\quad \text{on }\quad \Dom ( \dot  A).
\end{equation}

If $\widehat  A$ has a regular point in the upper half-plane,
 in particular, if $\widehat A$ is self-adjoint,
 then
 \begin{equation}\label{u=v}
V_t=e^{it\mu}U_t\quad \text{ for some  $\mu\in \bbR$.}
\end{equation}

Moreover, in the exceptional case when  the spectrum of $\widehat A$ fills in  the whole upper half-plane,  the subspaces    $$
 \cH_\pm=\overline{\Span_{z\in \bbC_\pm}\Ker ((\dot A)^*-zI)}
$$ are orthogonal, reduce  the groups $U_t $ and $V_t$ and
$$
V_t|_{\cH_\pm}=e^{it\mu_\pm}U_t|_{\cH_\pm}\quad \text{for some  $\mu_\pm\in \bbR.$}
$$

\end{theorem}

\begin{proof}
Let $z$ be  a regular point of $\widehat A$ in the upper half-plane. As it has been explained in the proof of Theorem \ref{lem:diss}, if $\widehat A$ is not self-adjoint,
then  \eqref{zilotdot} holds automatically. Therefore,
Lemma \ref{uniuni} is applicable and hence
 $$
 \cU_tg=e^{i\mu_-t} g \quad \text{for some}\quad \mu_-\in \bbR,
 $$
 where $\cU_t=U_t^*V_t$ and $0\ne g\in
 \Ker ((\dot A)^*+iI)\subset \cH_-$.

 Set
 $$
 f=(\widehat A+iI)(\widehat A-zI)^{-1}g.
 $$
Since   $f\in \Ker ((\dot A)^*-zI)\subset \cH_+$, by Lemma \ref{uniuni},
$$
\cU_tf=e^{i\mu_+t} f \quad \text{for some}\quad \mu_+\in \bbR.
$$
 On the other hand,
since  the unitary operator $\cU_t$ commutes with $\widehat A$, we have
 $$
 \cU_t f=(\widehat A+iI)(\widehat A-zI)^{-1}\cU_tg=(\widehat A+iI)(\widehat A-zI)^{-1} e^{i\mu_-t} g=e^{i\mu_-t} f.
$$
Therefore, $\mu_+=\mu_-$.

Finally, taking into account that  $\dot A$ is a prime operator, it follows  that the subspaces $\cH_\pm$ span the whole Hilbert space $\cH$, and the claim follows.

  In view of Lemma   \ref{uniuni}, to prove the last assertion it remains to show that $\cH_\pm$ are mutually orthogonal  whenever
 the spectrum of  $\widehat  A$ fills in the upper half-plane.

Since \eqref{zilot} holds and the dissipative operator $\widehat A$ has no regular points in the upper half-plane,  one can apply Theorem \ref{lem:diss}
to conclude that for any self-adjoint extension of $\dot A$  the characteristic function of the triple  $(\dot A, \widehat A, A)$ is identically zero.
By Lemma \ref{key}, the characteristic function of the triple  $(\dot D, \widehat D_I(0),D)$ on the metric graph $\bbY$ in Case (i) with quantum gate coefficient $k=0$
also vanishes identically in the upper-half-plane. The operators  $\dot A$ and $\dot D$ are prime symmetric operators, therefore
 $\dot A$ is unitarily equivalent to $\dot D$ on the metric graph $\bbY$ in Case (i), where $\widehat D$ is in the exceptional case, that is,
its point spectrum  fills in $\bbC_+$. In this case the subspaces $\cH_\pm(\dot D)=\overline{\Span_{z\in \bbC_\pm}\Ker ((\dot D)^*-zI)}=L^2(\bbR_\pm)$ for the operator $\dot D$ are  orthogonal,
so are the subspaces
$$\cH_\pm(\dot A)=\overline{\Span_{z\in \bbC_\pm}\Ker ((\dot A)^*-zI)}
$$ for the symmetric operator $\dot A$.

\end{proof}

\section{Dissipative solutions to the CCR}\label{sec11}

Now we are prepared to  get  a complete classification  (up to  unitary equivalence) of the  simplest non-self-adjoint  maximal dissipative solutions $\widehat A$ to the commutation relations \eqref{mmm}.

More generally, we have the following result.

\begin{theorem}\label{premain}
{\it Assume Hypothesis \ref{muhly}. Suppose, in addition, that $\dot A$ is a prime operator and $A$ its self-adjoint extension. Suppose that $\widehat A$ is a maximal  dissipative  extension of $\dot A$ such that
$$
U_t^* \widehat AU_t=\widehat A+tI\quad \text{on }\quad \Dom ( \widehat A).
$$

\begin{itemize}
\item[(i)]
If $\widehat A$ is self-adjoint, then there exists a unique $\Theta$, $|\Theta|=1$, such that the triple $(\dot A, \widehat A, A)$ is mutually unitarily equivalent to the  triple $(\dot D,  D, D_\Theta)$ on the metric graph $\bbY$ in Case $(i)$
 and therefore $\Theta$ is
 a unitary invariant of   $(\dot A, \widehat A, A)$.

\item[(ii)]
If $\widehat A$ is not self-adjoint, then the triple $(\dot A, \widehat A, A)$ is mutually unitarily equivalent to the model triple $(\dot D, \widehat D, D_\Theta)$ in one of the Cases $(i)-(iii)$ for some $|\Theta|=1$.
 If, in addition,  $\widehat A$ has at least one regular point in the upper half-plane, then  the parameter $\Theta$ is uniquely determined by the  triple  $(\dot A, \widehat A, A)$
 and therefore $\Theta$ is
 a unitary invariant of   $(\dot A, \widehat A, A)$ in this case. That is, if some triples $(\dot A, \widehat A, A)$ and $(\dot A', \widehat A', A')$ are mutually unitarily equivalent, then the corresponding parameters $\Theta$ and $\Theta'$ coincide.
\end{itemize}
}
\end{theorem}

\begin{proof}

(i). If $\widehat A$ is self-adjoint, we argue as follows.  By Theorem \ref{sacr}, the Weyl-Titchmarsh function of $(\dot A, \widehat A)$ coincides with $i$ in the upper half-plane.
Therefore,  since $\dot A$ is a prime operator,
the pair  $(\dot A, \widehat A)$     is mutually unitarily equivalent to the pair $(\dot D,  D)$ on the metric graph $\bbY$ in Case (i).
Since $A$ is a self-adjoint extension of $\dot A$,  the triple $(\dot A, \widehat A, A)$ is mutually unitarily equivalent to the  triple $(\dot D,  D, D_\Theta)$ for some $|\Theta|=1$, which proves the existence of such a $\Theta$.

 To establish  the uniqueness,  suppose that the triples $(\dot D,  D, D_{\Theta_1})$  and  $(\dot D,  D, D_{\Theta_2})$ in Case (i) (recall that $D$ is self-adjoint here)  are  mutually unitarily equivalent.

 In particular, there exists a unitary operator $U$ such that
$$
U(\dom(\dot D )=\dom(\dot D )), \quad   U(\dom( D ))=\dom( D ),
$$
$$
U\dot D f=\dot D U f \quad \text{ for all}\quad f\in \dom(\dot D),
$$
$$
U D f=D U f \quad \text{ for all}\quad f\in \dom( D)
$$
and
$$
U^* D_{\Theta_1}U=D_{\Theta_2}.
$$

By Lemma \ref{uniuni}, the subspaces $L^2(\bbR_\pm)$ are eigensubspaces for the unitary operator $U$, and since
$ U(\dom( D ))=\dom( D )$, the corresponding eigenvalues of $U$ coincide. Therefore, $U$ is necessarily  a (unimodular) multiple of the identity
and hence
$$
D_{\Theta_2}=U^*D_{\Theta_1}U=D_{\Theta_1}$$
so that
$$\Theta_1=\Theta_2.
$$

 (ii). Suppose that $\widehat A$ is not self-adjoint.
By Theorem \ref{lem:diss}, the characteristic function of the triple admits the representation \eqref{savto}.
Combining   Theorem \ref{lem:diss} and Lemma \ref{key} one concludes that
there exists a (possibly different) self-adjoint extension $A'$ of $\dot A$ such that   the triples $(\dot A, \widehat A, A')$ and $(\dot D, \widehat D, D)$
are mutually unitarily equivalent.

In particular,  $\cU^{-1}\dot A\cU=\dot D$ and $\cU^{-1}A'\cU=D$ for some unitary operator. Hence   $\cU^{-1} A\cU$ is a self-adjoint extension of $\dot D$
and hence $\cU^{-1}A\cU=D_\Theta$ for some $ \Theta$. It is the unitary transformation $\cU$ that establishes the required mutual unitary equivalence of the triples.

To prove the last assertion, one observes that since $\widehat A$  has at least one regular point in the upper half-plane,  the characteristic function of the triple is not
 identically zero.
In particular, if  $A'$ is another reference operator, then the triples $(\dot A, \widehat A, A)$ and  $(\dot A, \widehat A, A')$ are mutually unitarily equivalent if and only if  $A'$ and $A$ coincide by Lemma \ref{ind0} in Appendix E.

Therefore, taking into account that  the triple $(\dot A, \widehat A, A)$ is mutually unitarily equivalent to  $(\dot D,  \widehat D, D_\Theta)$ for some $|\Theta|=1$, one concludes that
in this case  the unimodular parameter  $\Theta$ is uniquely determined by the triple $(\dot A, \widehat A, A)$.

\end{proof}

\begin{remark} If $\widehat A$ has no regular points in the upper half-plane,
then   $(\dot A, \widehat A, A')$ is mutually unitarily equivalent to $ (\dot D,  \widehat D, D_{\Theta})$ for some (and therefore for all) $\Theta$, $|\Theta |=1$,
where $ \widehat D$
is in  the exceptional case (Case (i) with $k=0$) (see Corollary \ref{gogol}). Therefore, in this exceptional case, the parameter $\Theta$ is not determined uniquely by the triple  $(\dot A, \widehat A, A')$.
\end{remark}

\begin{remark}\label{muhlyresolution}
On account of the  remarks that we made in Section \ref{secmodop}, structure Theorem  \ref{opcup} combined with Theorem \ref{premain}
provides the following intrinsic characterization of all symmetric operators satisfying Hypothesis \ref{muhly}
thus giving  a {\it complete solution of the J$\clock$rgensen-Muhly problem} for symmetric operators in the case of deficiency indices $(1,1)$ (see Problem (I) b) in the Introduction):

 either

\begin{itemize}
\item[i)]  $\dot A$ admits a (dissipative) quasi-selfadjoint extension with the point spectrum filling in $\bbC_+$, equivalently,   $\dot A$ admits
a pair of distinct quasi-selfadjoint extensions that are unitarily equivalent,
\item[]or,
\item[ii)]  $\dot A$ admits a quasi-selfadjoint extension with no spectrum,
\item[]or, finally,
\item[iii)]  $\dot A$ is the symmetric part  of an operator coupling
of a dissipative extension of the symmetric operator  $\dot A$ without  point spectrum  in case i)   and the dissipative extension of  $\dot A$  with no  spectrum in case ii)
(see Remark  \ref{metmet} for the definition of the symmetric part of a dissipative operator in connection with Hypothesis \ref{muhly}).
\end{itemize}

\end{remark}

\section{Main results}\label{s13}

In this section  we provide the  complete classification of the simplest solutions   to  the restricted Weyl commutation relations
 $$
V_sU_t=e^{ist}U_tV_s, \quad t\in \bbR, \quad s\ge 0,
$$
for
 a strongly continuous  group of unitary operators   $U_t$ and    a strongly continuous semi-group of contractions $V_s$  in a separable Hilbert space $\cH$.

\begin{hypothesis}\label{refref}
 Let $(-\infty, \infty) \ni t \to U_t =e^{iBt}$
 be a strongly continuous  group of unitary  operators and  $[0, \infty)\ni s\to V_s=e^{i\widehat A s}$ a semi-group of  contractions in a separable Hilbert space $\cH$.
Suppose that   the restricted Weyl commutation relations   \begin{equation}\label{resthyp}
V_sU_t=e^{ist}U_tV_s, \quad t\in \bbR, \quad s\ge 0,
\end{equation}
hold.
\end{hypothesis}

We remark that in the light of Corollary \ref{muhly1} one could have  started from the  much weaker  Hypothesis \ref{muhly}.

The following two results characterize the  simplest solutions  to  the restricted Weyl commutation relations.
We start with the case where  $V_s$ is a semi-group of isometries.

\begin{theorem}\label{thm01} Assume Hypothesis \ref{refref}.
Suppose, in addition, that the generator
$\widehat A $ of the semi-group $V_s$ is a prime symmetric operator with deficiency indices $(0,1)$.

Then there exists a unique metric  graph $\bbY=(\mu,\infty)$, $\mu \in \bbR$,
such that
 the pair $(\widehat A, B)$ is mutually unitarily equivalent to the  pair $(\widehat \cP, \cQ)$,
where $\widehat \cP=i\frac{d}{dx}$ is the differentiation operator in $L^2(\bbY)=L^2((\mu, \infty))$
on
$$
\Dom (\widehat \cP)=\{f\in W_2^1((\mu, \infty))\, |\, f(\mu)=0\}
$$
and $\cQ$ is the  operator of multiplication by  independent variable in   $L^2(\bbY)$.
\end{theorem}

\begin{proof}
Since $\widehat A $ is a generator of a semi-group,  $\widehat A $ is  a closed operator.
By the Stone-von Neumann uniqueness result (see, e.g., \cite[Theorem 2, Ch. VIII, Sec. 104]{AkG}) there exists a isometric map $\cU$ from $\cH$ onto
$L^2((0, \infty))$ such that
$
\cU \widehat A \cU^{-1}
$  coincides with the differentiation operator in   $L^2((0, \infty))$ with the Dirichlet boundary condition at the origin.

Lemma \ref{uniuni} and Lemma \ref{popopo} in Appendix B
show  that there exists a $\mu$ such that  the pair $(\widehat A, B)$ is mutually unitarily equivalent to the  pair $(\cP_0, \cQ_0+\mu I)$,
  where $$\cP_0=i\frac{d}{dx}$$ is the differentiation operator in $L^2((0, \infty))$
on
$$
\Dom (\cP_0)=\{f\in W_2^1((0, \infty))\, |\, f(0)=0\},
$$
and $\cQ_0$ is the operator of multiplication by independent variable in   $L^2((0,\infty))$.  Clearly the pairs  $(\cP_0, \cQ_0+\mu I)$  and
$(\widehat \cP, \cQ)$ in the Hilbert spaces
 in $L^2((0,\infty)) $ and
in $L^2((\mu, \infty))$, respectively,  are mutually unitarily equivalent.
By  construction, the spectrum of the generator $B$ coincides with the semi-axis $[\mu, \infty)$.
Therefore,  $\mu$ is a unitary invariant  which is  uniquely determined by the  pair $(\widehat A, B).$ In particular, the graph $\bbY=(\mu, \infty)$ is also uniquely determined by the pair $(\widehat A, B).$
\end{proof}

In a completely analogous way one proves the following result.

\begin{theorem}\label{thm10}   Assume Hypothesis \ref{refref}.
Suppose, in addition, that the generator
$\widehat A $ of the semi-group $V_s$ is a maximal dissipative extension of a prime symmetric operator with deficiency indices $(1,0)$.

Then there exists a unique metric  graph $\bbY=(-\infty, \nu)$, $\nu \in \bbR$,
such that  the pair $(\widehat A, B)$ is mutually unitarily equivalent to the  pair $(\widehat \cP, \cQ)$
where $\widehat \cP =i\frac{d}{dx}$ is the differentiation operator in $L^2(\bbY)=L^2((-\infty, \nu))$
on
$$
\Dom (\widehat \cP)=W_2^1((- \infty, \nu))
$$
and $\cQ$ is the   operator of multiplication by  independent variable in   $L^2(\bbY)$.
\end{theorem}
\begin{remark} We remark that if $\widehat A$ is a prime symmetric operator with deficiency indices $(0,1)$ or $(1,0)$ in a Hilbert space $\cH$, then $\cH$ is necessarily separable.
\end{remark}
\begin{remark}
Theorems \ref{thm01} and \ref{thm10} are dual to each other: if $(\widehat A, B)$ satisfies the hypotheses of Theorem \ref{thm01}, then
$(-\widehat A^*, -B)$  satisfies the hypotheses of Theorem \ref{thm10}, however  the pairs $(\widehat A, B)$ and  $(-\widehat A^*, -B)$ are not mutually unitarily equivalent.
In this case the self-adjoint operator $B$ is semi-bounded from below  but $-B$ is semi-bounded from above. Moreover, the   generator $\widehat A$ has no point spectrum while the point spectrum of $(-\widehat A)^*$ fills in the entire  open upper half-plane.
\end{remark}

Next we treat the case where the dissipative generator of the semi-group  $V_s$  is a quasi-selfadjoint extension of a prime symmetric operator
with deficiency indices $(1,1)$.

To formulate the corresponding uniqueness result, we need some preparations.
\begin{definition}\label{canonical}
Let  $\bbY$ be a metric graph in
the following  cases
 \begin{equation}\label{vaivai}
\bbY=
\begin{cases}
 (-\infty, \nu) \sqcup(\mu, \infty) & \text{ Case I$^*$}  \quad (\nu \ne \mu)
\\
(-\infty, \mu) \sqcup(\mu, \infty)  & \text{ Case I}
\\
 (\mu, \nu)& \text{ Case II}  \quad (\mu < \nu)
 \\
 (-\infty, \mu) \sqcup(\mu, \infty)\sqcup (\mu,\nu) & \text{ Case III}  \quad (\mu < \nu).
\end{cases}.
\end{equation}

 Given  a real number $k$, $0\le k<1$, define
  \begin{itemize}
\item[]  the {\it position} operator $\cQ$ as  the  operator  of multiplication  by  independent variable on the edges of the graph $\bbY$
\item [] \,\,\,\,and

\item []
  the  {\it momentum}
 operator  $\widehat \cP$ as the differentiation operator $i\frac{d}{dx}$  on the edges
of the graph $\bbY$,
\begin{equation}
( \widehat \cP f)(x)= i \frac{d}{dx}f(x) \quad \text{a. e.} \,\,x\in
e  \text{ on every edge $e$ of } \bbY
\end{equation}
on
 $$\Dom(\widehat \cP)\subsetneqq \bigoplus\limits_{e\subset \bbY}W_2^1(e)\subsetneqq L^2(\bbY),$$
the space of locally absolutely continuous functions  on the edges with   the following  vertex   boundary conditions \end{itemize}

$$
\Dom(\widehat \cP)=\begin{cases}
\left \{f_-\oplus f_+\in W_2^1((-\infty, \nu))\oplus  { W}_2^1(( \mu, \infty)) \,\, |\,\,  f_+(\mu+)=0\right \},  & \text{in Case I$^*$}
\\
\left \{f_\infty\in W_2^1((-\infty, \mu))\oplus W_2^1(( \mu, \infty)) \,\, |\,\, f_\infty (\mu+)=k f_\infty(\mu-)\right \}, &\text{in Case I}
\\
 \left \{f_\ell\in W_2^1(( \mu, \nu))\,\, |\,\, f_\ell(\mu+)=0\right \},& \text{in Case II}
 \end{cases}.
$$
Here, in Case I we require that $0<k<1$ and in Case II we assume that $\nu=\mu+\ell>\mu$.

In Case III,
$
\Dom(\widehat \cP )$ consists  of the two-component vector-functions $f=(f_\infty, f_\ell)^T$,
$$f_\infty \oplus f_\ell\in
 \left (W_2^1((-\infty,\mu))\oplus W_2^1((\mu, \infty))\right )\oplus W_2^1((\mu , \nu))
$$
that satisfy  the ``boundary conditions"
\begin{equation}\label{bcond*}
 f_\infty(\mu+)=
k f_\infty(\mu-)
\quad \text{ and } \quad
 f_\ell(\mu+)=\sqrt{1-k^2} f_\infty(\mu-) \quad (\text{Case III}).
\end{equation}

By definition the pair $(\widehat \cP, \cQ)$ is said to be   the dissipative  {\it canonical pair with the quantum gate coefficient} $k$, $0\le k<1$ on the metric graph $\bbY$. In Case III
we always assume that $k>0$ and
 formally set $k=0$ whenever the graph $\bbY$ is in Case I$^*$  or in Case II.
 We also call the triple $(\dot \cP, \widehat \cP, \cQ)$, where  $$
 \dot \cP=\widehat \cP|_{\dom ( \widehat \cP)\cap \dom (\widehat \cP^*)},
$$
{\it the canonical dissipative triple on $\bbY$ with the quantum gate coefficient $k$}.

\end{definition}

\begin{remark}
In  Case I$^*$ the metric graph is ``disconnected'' whenever $\nu <\mu$, while if  $\mu<\nu$,  one may  think that the edges of the graph eventually  ``overlap'' over  the finite interval $[\mu, \nu]$.
Also, in Case III the boundary conditions \eqref{bcond*} at the junction point $\mu$ of the graph $\bbY$, the center of the graph, yield  the quantum Kirchhoff rule
$$
|f_\infty(\mu+)|^2+ |f_\ell(\mu+)|^2= |f_\infty(\mu-)|^2.
$$

\end{remark}

\begin{remark}\label{leble}
It is easy to see  that the spectrum of the position operator $\cQ$ is given by
\begin{equation}\label{specq}
\spec (\cQ)= \begin{cases}
(-\infty, \nu]\cup [\mu, \infty),&\text{   in Case I$^*$}\\
(-\infty, \infty),&\text{   in Case I }\\
[\mu, \nu],&\text{  in Case II}\\
(-\infty, \infty),&\text{  in Case III }
\end{cases}.
\end{equation}

From \eqref{specq} it follows that if $\bbY$  is in Case I$^*$ with  $\nu >\mu$ or  in
Case III, then
 the spectrum of the position operator $\cQ$ has multiplicity 2 on the finite interval $[\mu, \nu]$.  In Case I and II
the position operator $Q$ has simple Lebesgue  spectrum filling in the whole real axis $(-\infty, \infty)$  and the finite interval $[\mu, \nu]$ respectively.

We also notice that the spectrum of the dissipative momentum operator $\widehat \cP$ is
\begin{equation}\label{specp}
\spec (\widehat \cP)= \begin{cases}
\bbC_+\cup (-\infty, \infty),&\text{   in Case I$^*$ and Case I with $k=0$}\\
(-\infty, \infty),&\text{    in Case I with $k>0$ }\\
\emptyset, &\text{   in Case II}\\
(-\infty, \infty),&\text{   in Case III }
\end{cases}.
\end{equation}
Notice that  Case I$^*$ and  Case I with $k=0$ are exceptional in the sense that any point in the (open) upper half-plane is an eigenvalue of $\widehat \cP$.
\end{remark}

If a metric graph $\bbY$ is in Case I, we also introduce the concept of {\it the  Weyl canonical triple}
 on $\bbY$.

\begin{definition}\label{canwey} Let $\bbY$ be a metric graph in Case I, that is,
$$
\bbY=(-\infty, \mu)\sqcup (\mu, \infty)\quad \text{for some}\quad \mu\in \bbR.
$$
Let $\cP=i\frac{d}{dx}$ be the self-adjoint differentiation operator  on
\begin{equation}\label{sspp}
\Dom(\cP)=\{f\in W_2^1(\bbY)\, |\, f(\mu-0)=f(\mu+0\},
\end{equation}
$\dot \cP$ its symmetric restriction on
$$
\Dom(\dot \cP)=\{f\in W_2^1(\bbY)\, |\, f(\mu-0)=f(\mu+0)=0\},
$$
and $\cQ$ the position operator on $\bbY$.
We call $(\dot \cP, \cP, \cQ)$ {\it the Weyl canonical  triple on $\bbY$} (centered at $\mu$).
\end{definition}

\begin{remark}
If  $(\widehat \cP  , \cQ)=(\widehat \cP(k) , \cQ)$ is   the dissipative  {\it canonical pair} on the  metric graph  $
\bbY=(-\infty, \mu]\cup [\mu, \infty)
$  {\it with the quantum gate coefficient} $k$,
then
$$
\slim_{k\to1}\widehat \cP(k)=\cP,
$$
where    $\cP$ is the self-adjoint differentiation  operator defined  on \eqref{sspp} and the limit is taken in the strong resolvent sense.
Therefore,  the Weyl canonical  triple  on the metric graph $\bbY$
 can be considered  the limiting case of the dissipative triple  $(\dot \cP,\widehat  \cP, \cQ)=(\dot \cP(k),\widehat  \cP(k), \cQ)$ with the quantum gate coefficient $k$
as $k\to 1$.
\end{remark}

Our first auxiliary result is that the pair $(\bbY, k)$ is a  unitary invariant of a dissipative canonical pair.
\begin{lemma}\label{unickmu} Suppose that  the  canonical pairs  $(\widehat \cP(k), \cQ)$ and $(\widehat \cP'(k'),\cQ')$ with the quantum gate coefficients $ k$ and $k'$ on metric graphs $\bbY$ and $\bbY'$, respectively,  are mutually unitarily equivalent. Then
$$
\bbY=\bbY'
\quad \text{and}\quad
 k=k'.
$$
 \end{lemma}

\begin{proof}
As it has been explained in Remark \ref{leble}, there are two options: either  the position operator  $\cQ$  has simple spectrum or
$\cQ$ has spectrum of multiplicity 2  filling in  a finite interval.

Assume, first,   that the position operator  $\cQ$ has spectrum of multiplicity 2 supported by  a finite interval $[\mu , \nu]$, $\nu>\mu$. So does $\cQ'$.
Therefore the graphs $\bbY$ and $\bbY'$ have the same vertices but may possibly be in different cases, in Case I$^*$ or in Case III only.

     Suppose that $\bbY$ is in Case I$^*$and therefore $k=0$.
Then   the point spectrum of the dissipative momentum operator $\widehat \cP=\widehat \cP(0)$
fills in  the whole upper half-plane $\bbC_+$, so does  the dissipative momentum operator $\widehat \cP'$
since  $\widehat \cP$ and $\widehat \cP'$ are unitarily equivalent. Therefore, $\bbY'$ is in  Case I$^*$   with $k'=0$ as well.
Analogously, if $\bbY$ is in Case III, then   $\bbC_+$ belongs to the resolvent set of  $\widehat \cP(k)$. Again, since  $\widehat \cP(k)$ and $\widehat \cP'(k')$ are unitarily equivalent,
$\bbC_+$ belongs to the resolvent set of  $\widehat \cP'(k')$
 and then necessarily  $\bbY'$  is in Case III. Thus   $\bbY$ and $\bbY'$ have the same vertices and are in the same  cases.  Therefore,  $\bbY=\bbY'$.

It remains to treat the case where  the multiplication operator $\cQ$ has  simple spectrum.  There are two options: either both $\bbY$ and $\bbY'$ are in Case II, or both of them are in Case I.

If they are in  Case II,   the knowledge of the spectrum of
 $\cQ$ ($\cQ')$   uniquely determines the location of the  vertices of  the graph $\bbY$ ($\bbY'$) and   the graph(s)  itself.

If both  $\bbY$ and $\bbY'$ are in Case I, we proceed as follows.
 Since  the pairs $(\widehat \cP(k), \cQ) $ and  $(\widehat \cP'(k'), \cQ')$ are mutually  unitarily equivalent and  $\cQ=\cQ'$,
there exists a unitary operator $U$ commuting with the multiplication operator $\cQ$
such that
$$
\widehat \cP'(k')=U^*\widehat \cP(k) U.
$$
Since $\cQ$ has simple spectrum and the unitary operator  $U$ commutes with $\cQ$, the  operator  $U$ is the multiplication operator by a unimodular function $u$. We have
\begin{equation}\label{domeqq}
\dom(\widehat \cP(k))=U(\dom (\widehat \cP'(k'))).
\end{equation}

Suppose that  the vertices $\mu$ and $\mu'$ of the graphs $\bbY$ and $\bbY'$ are different, that is, $\mu\ne \mu'$.
From \eqref{domeqq} it follows that the function $u(x)f(x)$ is a continuous function in a neighborhood of the point $\mu'$ for all $f\in \dom (\widehat \cP'(k'))$, so is
the function $|f(x)|=|u(x)f(x)|$, which is incompatible with the boundary condition $f(\mu'+)=k'f(\mu'-)$ for all $f\in \dom (\widehat \cP'(k'))$, since $k'<1$.
Therefore, the vertices of the graphs $\bbY $ and $ \bbY'$ coincide,  $\mu=\mu'$,  and hence  $\bbY=\bbY'$.

To prove  that $k=k'$,    notice that  if the metric graph $\bbY$ and therefore $\bbY'$ is in Cases I$^*$ or   II, then $k=k'=0$ by definition.

  Suppose  that  $\bbY=\bbY'$   is in Cases I,
$$
\bbY=(-\infty, \mu) \sqcup(\mu, \infty)=\bbY' .
$$
In this case, the absolute values of the von Neumann parameters of $\widehat \cP(k)$ and $\widehat \cP'(k')$ (more precisely, of the corresponding triples)
coincide with $k$ and $k'$, respectively. By the hypothesis,  $\widehat \cP(k)$ and $\widehat \cP'(k')$  are unitarily equivalent. Therefore,
$$k=k',$$ since the absolute value of the von Neumann parameter is a unitary invariant of a dissipative operator by Remark \ref{mnogopros}.

Next, assume that  $\bbY=\bbY'$   is in Case III,
$$
\bbY=
 (-\infty, \mu) \sqcup(\mu, \infty)\sqcup (\mu,\nu)=\bbY'   \quad (\mu < \nu).
$$
By Theorem \ref{ind12}, the absolute values of the von Neumann parameters of $\widehat \cP(k)$ and $\widehat \cP'(k')$
are $ke^{-\ell}$ and  $k'e^{-\ell}$,  respectively (see the relation \eqref{sac}),
where $$\ell=\nu-\mu.$$ Therefore, $k=k'$, since $\widehat \cP(k)$ and $\widehat \cP'(k')$  are unitarily equivalent by the hypothesis.

\end{proof}

Now we are ready to present the central result of the first part of the  book.

\begin{theorem}\label{vNvN} Assume Hypothesis \ref{refref}.
Suppose, in addition, that the generator
$\widehat A $ of the semi-group $V_s$  is not self-adjoint and that  the restriction
$$
\dot A=\widehat A|_{\dom (\widehat A)\cap \dom ((\widehat A)^*)}
$$
is a prime symmetric operator with deficiency indices $(1,1)$.

Then there exists a unique metric   graph $\bbY$  in one of the Cases I$^*$, I--III
and a unique   $ k\in[0,1) $   such that
 the pair  $ (\widehat A, B)$  $($triple
$ (\dot A, \widehat A, B))$ is mutually unitarily equivalent to the canonical  dissipative
 pair  $(\widehat \cP(k),\cQ)$ $($triple  $(\dot \cP,\widehat \cP(k),\cQ))$ on $\bbY$, respectively.
\end{theorem}

\begin{proof}  By the hypothesis the restricted Weyl relations \eqref{resthyp} hold. Therefore (see \cite{EM,Si})
\begin{equation}\label{comrel}
U_t^* \widehat AU_t=\widehat A+tI\quad \text{on }\quad \Dom ( \widehat A), \quad t\in \bbR.
\end{equation}

As in the proof of Theorem \ref{lem:diss} one shows that symmetric operator $\dot A$ solves the commutation relations
\begin{equation}\label{nonnadot}
U_t^*\dot AU_t=\dot A+tI\quad \text{on}\quad \Dom (\dot A).
\end{equation}
Therefore, the operator $\dot A$ satisfies Hypothesis \ref{muhly}. In this situation  one
 can apply Theorem \ref{premain} (ii) to see that  there is  a metric graph $\bbY_0$ in one of the Cases (i)-(iii)
with the quantum  gate coefficient $k$ such that
the dissipative operator
$\widehat  A$ is unitarily equivalent to  the one of the following model dissipative differentiation operators:
\begin{itemize}
\item[(i)] $\widehat D=\widehat D_I(k)=i\frac{d}{dx}$ on $\bbY_0=(-\infty, 0) \sqcup(0, \infty) $
with the boundary condition
 \begin{equation}\label{domdomi8}
f_\infty (0+)=kf_\infty (0-),\quad 0\le k<1;
 \end{equation}

\item[] or
$$
$$
\item[(ii)]  $\widehat D=\widehat D_{II}(0,\ell)=i\frac{d}{dx}$ on $\bbY_0=(0,\ell)$ with the boundary condition
\begin{equation}\label{domdomii8}
 f_\ell(0)=0;
 \end{equation}

\item[] or
$$
$$
\item[(iii)] $\widehat D=\widehat D_{III}(k,\ell)=i\frac{d}{dx}$  on $\bbY_0=(-\infty, 0) \sqcup(0, \infty)\sqcup (0,\ell) $
with the boundary conditions
\begin{equation}\label{domdomiii8}
\begin{pmatrix}
f_\infty(0+)\\
f_\ell(0)
\end{pmatrix}=\begin{pmatrix}
k& 0\\
 \sqrt{1-k^2}&0
\end{pmatrix}
 \begin{pmatrix}
f_\infty(0-)\\
f_\ell(\ell)
\end{pmatrix}, \quad 0< k<1.
\end{equation}
\end{itemize}
  That is, there exists a unitary map $\cW$ from $\cH$ onto  $L^2(\bbY_0)$ such that
  \begin{equation}\label{qqww}
  \cW\widehat A \cW^{-1}=\widehat D.
  \end{equation}
In particular, from \eqref{comrel} it follows that
  \begin{equation}\label{comrel*}
W_t^* \widehat DW_t=\widehat D+tI\quad \text{on }\quad \Dom ( \widehat D), \quad t\in \bbR,
\end{equation}
where $W_t$ is the unitary group on $L^2(\bbY)$ given by
$$
W_t=\cW U_t \cW^{-1}, \quad t\in \bbR.
$$
On the other hand, $$
e^{it\cQ_0}\widehat D e^{-it\cQ_0}=\widehat D+tI\quad \text{on } \quad \Dom ( \widehat D), \quad t\in \bbR,
 $$
 where $\cQ_0$ is the  operator of multiplication by independent variable on the graph $\bbY_0$.

 Applying  Theorem   \ref{maininv} to the dissipative operator $\widehat D$, one obtains
  \begin{equation}\label{qqwwqq}
W_t=\cW e^{iBt} \cW^{-1}=e^{-i\mu t}e^{i\cQ_0 t} \quad \text{for some}\quad \mu \in \bbR,
\end{equation}
whenever  $\widehat D$, and therefore $\widehat A$,  has a  regular point in the upper half-plane.

In this case,  combining \eqref{qqww} and  \eqref{qqwwqq} one concludes that
 the  pair
$ ( \widehat A, B)$ is mutually unitarily equivalent to the pair
  $(\widehat D,\cQ-\mu I )$ on the graph $\bbY_0$   (with the quantum gate  coefficient $k$).
   The pair
    $(\widehat D,\cQ-\mu I )$ is  in turn
mutually unitarily equivalent to the canonical dissipative pair $(\widehat \cP(k),\cQ)$ with the quantum  gate coefficient $k$ on the metric graph $\bbY$  centered at $\mu$.
Notice that $\bbY$
 can be obtained from   the graph  $\bbY_0$ by a shift.

If the dissipative operator $\widehat A$ has no regular points in the upper half-plane, Theorem   \ref{premain} asserts that
 $\widehat A$ is unitarily equivalent   the model differentiation operator $\widehat D=\widehat D_I(0)$ on the graph $\bbY_0=(-\infty, 0) \sqcup(0, \infty) $
 in Case (i) with quantum  gate coefficient $k=0$.

The same reasoning as above shows that  in this exceptional case,
 the  pair
$ ( \widehat A, B)$ is mutually unitarily equivalent to the pair
  $(\widehat D_I(0),\cQ-\mu_-\cR_--\mu_+\cR_+ )$  on the graph $\bbY_0=(-\infty, 0) \sqcup(0, \infty)$ for some  $\mu_\pm \in \bbR$.
 Here   $\cR_-$ and $\cR_+$  are  the orthogonal projections
in $$L^2(\bbY_0)=L^2((-\infty, 0)) \oplus L^2((0,\infty))$$
 onto the subspace  $L^2((-\infty, 0)) $ and  $L^2((0,\infty))$,
 respectively.

 If $\mu_+=\mu_-=\mu$,  the pair  $(\widehat D_I(0),\cQ-\mu_-\cR_--\mu_+\cR_+ )=(\widehat D_I(0),\cQ-\mu I )$  is mutually unitarily equivalent the canonical dissipative pair $(\widehat \cP(0), \cQ)$
on the graph $\bbY=(-\infty, \mu) \sqcup(\mu, \infty) $ in Case I with quantum gate coefficient $k=0$.

If  $\mu_+\ne \mu_-$,   then  the pair  $(\widehat D_I(0),\cQ-\mu_-\cR_--\mu_+\cR_+ )$ on the metric graph
$
\bbY_0=(-\infty, 0) \sqcup(0, \infty)$ (in Case(i)) is mutually unitarily equivalent  to the canonical dissipative  pair $(\widehat \cP, \cQ)$
on the graph $\bbY=(-\infty, \mu_-) \sqcup(\mu_+, \infty) $ in Case  I$^*$.

 The uniqueness part of the statement  is an immediate consequence of  Lemma \ref{unickmu}.

\end{proof}

\begin{remark} \label{sutkin}
 If in addition to the hypotheses of Theorem \ref{vNvN} one assumes that  $A$ is a self-adjoint (reference)  extension of $\dot A$, then we immediately get that there exists a self-adjoint  extension  $\cP$ of $\dot \cP$
 such that
 the quadruple
$(\dot A, \widehat A, A, B)$ is mutually unitarily equivalent to the  quadruple $(\dot \cP, \widehat \cP(k) ,\cP, \cQ) $ on  the metric graph $\bbY$ in Cases I$^*$, I--III with  the quantum gate coefficient $k\ne 0$ for some $k\in [0, 1)$. The extension $\cP $ is determined by the quadruple
$(\dot A, \widehat A, A, B)$ uniquely unless the graph $\bbY$  is in Case $I^*$ or in Case  $I$ with
 the quantum gate coefficient $k= 0$.
\end{remark}

With a  minor modification,  the result of Theorem \ref{vNvN} extends  to the case where the generator $\widehat A$ is self-adjoint.

\begin{theorem}\label{vNvNad} Assume Hypothesis \ref{refref}.
Supose, in addition, that  $\widehat A=A$  is self-adjoint and that   $\dot A $ is  a prime symmetric restriction  of $\widehat A $ with deficiency indices $(1,1)$ such that
\begin{equation}\label{addnucomnu}
U_t^*\dot AU_t=\dot A+tI\quad \text{on}\quad \Dom (\dot A).
\end{equation}

Then there exists a unique metric  graph $\bbY=(-\infty, \mu)\sqcup(\mu,\infty)$ $($in Case $I)$ centered at  $\mu\in \bbR$
such that the  triple
$ (\dot A,A, B)$ is mutually unitarily equivalent to the Weyl canonical triple $(\dot \cP, \cP, \cQ)$ on $\bbY$ $($see Definition \ref{canwey}$)$.
\end{theorem}

\begin{proof} From the hypothesis it follows that in fact (unrestricted)   Weyl commutation relations   $$
V_sU_t=e^{ist}U_tV_s, \quad t\in \bbR, \quad s\in \bbR,
$$
hold
and therefore
\begin{equation}\label{nucomnu}
U_t^* AU_t=A+tI\quad \text{on }\quad \Dom ( A), \quad t\in \bbR.
\end{equation}

By Theorem \ref{sacr},
 the Weyl-Titchmarsh function associated with the pair $(\dot A , A)$
  has the form
$$
M(z)=i, \quad z\in \bbC_+.
$$
So does  the Weyl-Titchmarsh function associated with the pair $(\dot \cP, \cP)$  on the graph $\bbY_0=(-\infty, 0)\sqcup(0,\infty)$ in Case (i).

Since $\dot A$ and    $\dot \cP$ are prime operators, the pair
 $(\dot A ,  A)$ is   mutually unitarily equivalent to the pair $(\dot \cP, \cP)$ on the graph $\bbY_0$.
  That is,
there exists
a unitary operator $\cW:\cH\to \L^2(\bbR)$ such that
$$
\cW A \cW^{-1}=\cP \quad
 \text{ and }
\quad
\cW\dot A\cW^{-1}=\dot  \cP.
$$
From \eqref{addnucomnu} and \eqref{nucomnu} it follows that
 $$
W_t^* \cP W_t=\cP +tI\quad \text{on }\quad \Dom ( A)
$$
and
 $$
W_t^* \dot \cP W_t=\dot \cP +tI\quad \text{on}\quad \Dom (\dot A),
$$
where
$$
W_t=\cW U_t^* \cW^{-1}.
$$

By the definition of the Weyl canonical triple  $(\dot\cP,\cP,\cQ)$,
the commutation relations
$$
e^{it\cQ}\cP e^{-it\cQ}=\cP+tI \quad \text{ and }
\quad e^{it\cQ}\dot \cP e^{-it\cQ}=\dot \cP+tI
$$
hold. Now one can apply Theorem  \ref{maininv} to see that   there exists a   $\mu \in \bbR$ such
that
$$
W_t=\cW U_t\cW^{-1}=e^{-i\mu t}e^{it\cQ}= e^{it(\cQ-\mu I)}.
$$

Therefore,  the triple $(\dot A, A, B)$ is mutually unitarily equivalent to the triple
$(\dot \cP, \cP, \cQ-\mu I)$ on the metric graph $\bbY_0$.

In turn,  the triple  $(\dot \cP, \cP, \cQ-\mu I)$   is  mutually unitarily equivalent to the  Weyl canonical triple
$(\dot \cP_\mu, \cP, \cQ)$ on the metric graph $\bbY_\mu=(-\infty, \mu)\sqcup(\mu,\infty)$ in Case I.
(Recall that $
\dom(\dot \cP_\mu)=\{f\in W_2^1(\bbR)\,|\, f(\mu)=0\}
$.)

To complete the proof of the theorem  it remains to show that
the Weyl triples  $(\dot \cP_\mu, \cP, Q)$   and $(\dot \cP_{\mu'}, \cP, Q)$ on the graphs  $\bbY_\mu$ and $\bbY_{\mu'}$,  respectively, are not mutually unitarily equivalent unless $\mu=\mu'$. Indeed, assume that they are. Denote by $\cU$ the unitary operator that establishes the mentioned mutual unitary equivalence.
Since $\cU$ commutes with $\cP$ and $Q$, the operator  $\cU$
is a (unimodular) multiple of the identity. Therefore, $\cU^*\dot \cP_\mu U=\dot \cP_{\mu'}$ implies $\dot \cP_\mu =\dot \cP_{\mu'}$ and hence $\mu=\mu'$, a contradiction.

\end{proof}

\begin{remark}\label{independeny}
Comparing the assumptions of Theorems \ref{vNvN} and    \ref{vNvNad} it is clearly seen that the main difference is that in Theorem   \ref{vNvNad}
one has to require the commutation relation for the symmetric operator $\dot A$, while in the case of Theorem \ref{vNvN} the corresponding relations hold automatically.
As we have already mentioned in the Introduction, the existence of a symmetric operator $\dot A$ with the required properties
in the hypothesis of Theorem \ref{vNvNad}
follows from the Stone-von Neumann uniqueness result.

Indeed, let  $(\cP,\cQ)$ be the canonical pair with $\cP=i\frac{d}{dx}$ and $\cQ$  the operator of multiplication  by independent variable in $L^2(\bbR)$.

Suppose that
 $\cW:\cH\to \L^2(\bbR)$ is a unitary operator such that
\begin{equation}\label{uneq2}
\cW A \cW^{-1}=\cP \quad
 \text{ and }
\quad
\cW B\cW^{-1}=  \cQ.
\end{equation}

Then  the symmetric restriction $\dot \cP_\mu$ of $\cP$ on
$$
\dom(\dot \cP_\mu)=\{f\in W_2^1(\bbR)\,|\, f(\mu)=0\}
$$
has deficiency indices $(1,1)$  and  satisfies the commutation relations
$$
e^{it\cQ}\dot \cP_\mu e^{-it\cQ}=\dot \cP_\mu+tI.
$$It remains to choose
\begin{equation}\label{perlo}
\dot A= \cW^{-1} \dot P_\mu \cW
\end{equation}
and the existence of a restriction   with the required properties follows.

From the uniqueness part of Theorem \ref{vNvNad} it also follows that
if   a closed symmetric restriction $\dot \cP$ of  $\cP$  (with deficiency indices $(1,1)$)
satisfies commutation relations
$$
e^{it\cQ}\dot \cP e^{-it\cQ}=\dot \cP+tI,
$$ then $\dot \cP=\dot \cP_\mu$ for some $\mu\in \bbR$ (cf. \cite{Jorg1}).
 In this sense as far as the unitary equivalence \eqref{uneq2} is established (based on the Stone-von Neumann uniqueness result),  the choice of $\dot A$ via  \eqref{perlo}
 in the hypothesis of Theorem \ref{vNvNad} is canonical.

Notice that as long as the existence of  such a restriction is established/required
the reasoning  above  can be considered
an independent  proof of the Stone-von Neumann uniqueness result.

\end{remark}

The following corollary can be considered an {\it important extension of the Stone-von Neumann uniqueness theorem}.

 \begin{corollary}\label{bell}
{\it Suppose that strongly continuous  groups of unitary  operators
$ V_t =e^{i A  t}$ and  $U_t =e^{iBt}$ in the Hilbert space $\cH$
solve    the  Weyl commutation relations   $$
V_sU_t=e^{ist}U_tV_s, \quad s,  t\in \bbR.
$$
Assume  that the  self-adjoint operator $A$ has simple spectrum. Without loss of generality
suppose that $\dot A$ is a closed symmetric restriction of $A$ with deficiency indices $(1,1)$
such that
\begin{equation}\label{nucomnudot}
U_t^* \dot AU_t=\dot A+tI\quad \text{on }\quad \Dom ( \dot A), \quad t\in \bbR.
\end{equation}

If $A'$ is any other  self-adjoint extension of $\dot A$,
then  the  Weyl commutation relations
\begin{equation}\label{nui0}
V_s'U_t=e^{ist}U_tV_s', \quad\text{with }\quad V_s'=e^{isA'},\quad  s,  t\in \bbR,
\end{equation}
hold.}
 \end{corollary}

\begin{proof}
By Theorem \ref{vNvNad},
there exists a unique metric  graph $\bbY=(-\infty, \mu)\sqcup(\mu,\infty)$ such that the  triple
$ (\dot A,A, B)$ is mutually unitarily equivalent to the Weyl canonical triple $(\dot \cP, \cP, \cQ)$ on $\bbY$,  so that
$$
\dot A=\cU \dot \cP\cU^{-1},
\quad
A= \cU \cP\cU^{-1}
\quad \text{and} \quad
B=\cU \cQ\cU^{-1}
$$ for some unitary map $\cU$ from $L^2(\bbY)$ onto $\cH$.  Let $A'$ be a self-adjoint extension of $\dot A$. Therefore, $\cP'=\cU^{-1}A'\cU$ is a self-adjoint
extension of  $\dot \cP$ on
$$
\Dom(\cP')=\{ f\in W_2^1((-\infty ,\mu))\oplus W_2^1((\mu, \infty))\, |\, f(\mu_-)=\Theta f(\mu+)\}
$$
for some $|\Theta|=1$.

We have
\begin{equation}\label{nui1}
e^{it\cQ}  \cP'e^{-it\cQ} =\cP'+tI\quad \text{on }\quad \Dom ( \cP'), \quad t\in \bbR,
\end{equation}
and therefore
\begin{equation}\label{nui2}
U_t^*  A'U_t=\dot A'+tI\quad \text{on }\quad \Dom (A'), \quad t\in \bbR,
\end{equation}
which in turn implies \eqref{nui0}.
\end{proof}

\begin{remark}
In the situation in question one can state more, cf. Remark \ref{sutkin}.
For instance,  there exists a unique $\mu \in \bbR$
and a unique $\Phi\in [0, 2\pi)$ such that the quadruples $(\dot A, A, A', B)$ and  $(\dot \cP, \cP, \cP', \cQ)$
are  mutually unitarily equivalent.

Here  $\dot \cP, \cP$ and $\cP'$ are differentiation operators in $L^2(\bbR)$
defined on
\begin{align*}
\dom (\dot \cP)&=\{f\in W_2^1(\bbR)\,|\, f(\mu)=0\},\\
\dom (\cP)&=W_2^1(\bbR),\\
\dom (\cP')&=\{f\in W_2^1((-\infty, 0))\oplus W_2^1 ((0, \infty))\, |\, f(\mu+)=e^{i\Phi}f(\mu-)\},
\end{align*}
respectively, and $\cQ$ is the operator  of multiplication  by  independent variable in  $L^2(\bbR)$.

More generally, if $\widehat  A$ is a maximal dissipative extension of $\dot A$, then the restricted Weyl commutation relations
$$
\widehat V_sU_t=e^{ist}U_t\widehat V_s, \quad\text{with }\quad  \widehat V_s=e^{is\widehat A},\quad    t\in \bbR,\quad  s\ge 0,
$$
hold.

In this case,  there exists a unique point
 $(\mu, \Phi, k)\in   \bbR\times  [0, 2\pi)\times [0,1)$ such that
 the quadruples $(\dot A, \widehat A, A', B)$ and  $(\dot \cP, \widehat \cP, \cP', \cQ)$ are  mutually unitarily equivalent.
Here the operators $ \dot \cP, \cP' $ and $  \cQ$ are as above and
$\widehat \cP$ is differentiation operators in $L^2(\bbR)$
 on
 $$
 \dom (\widehat \cP)=\{f\in W_2^1((-\infty, \mu))\oplus W_2^1( (\mu, \infty))\, |\, f(\mu+)=kf(\mu-)\}.
 $$
\end{remark}

To be complete, we provide the description
of the  dynamics associated with the strongly continuous  semi-group $V_s=e^{i\widehat \cP s}$ generated by the dissipative momentum operator $\widehat \cP=\widehat \cP(k)$ with the quantum gate coefficient $k\in [0, 1)$ on a metric graph $\bbY$ in  Cases I$^*$, I--III (see  Section \ref{secmodop} for a more informal description of the dynamics).



\begin{theorem}\label{cases}
Let $\widehat \cP=\widehat \cP(k)$ be the canonical dissipative momentum operator  with the quantum gate coefficient $0\le k<1$ on  a metric graph  $\bbY$ in one of the Cases I*, I-III.
Then the strongly continuous   semi-group $ \widehat V_s$ of contractions generated by $\widehat \cP(k)$ in the Hilbert space $L^2(\bbY)$ admits the following explicit description.

In Case I*, the semigroup $ \widehat V_s$ acts as the right shift on the semi-axis $[\mu, \infty)$
and as the truncated right shift on  $(-\infty, \nu]$
$$
(\widehat V_sF)_+(x)=\chi_{[s+\mu, \infty)}(x)f_+(x-s),
\quad x\in[\mu,\infty),
$$and
$$
(\widehat V_sF)_-(x)=f_-(x-s),\quad x\in (-\infty,\nu],
$$
where
$$F=(f_-, f_+)^T\in L^2(\bbY) =L^2((-\infty, \nu))\oplus\L^2((\mu, \infty)). $$

In Case I,  we have that
$$
(\widehat V_sF )(x)=(\chi_{(-\infty, \mu)}(x)+k\chi_{[ \mu, \infty)}(x))f_\infty(x-s),\quad x\in \bbR,
$$
$$F=f_\infty\in L^2(\bbY)=L^2((-\infty, \mu))\oplus\L^2((\mu, \infty)).$$

In Case II,   the semi-group $V_s$
is  a nilpotent shift with index $\ell=\nu-\mu >0$ ($V_s=0$ for $ s\ge \ell
$)
$$(\widehat V_sF)_\ell(x)=\chi_{[\mu+s, \nu]}(x)f_\ell(x), \quad x\in [\mu, \nu],
$$
$$
F=f_\ell\in L^2(\bbY)=L^2((\mu,\nu)).
$$

In Case III,
the action of the semi-group $\widehat V_s$  is given by
$$
(\widehat V_sF)_\infty(x)=(\chi_{(-\infty, \mu)}(x)+k\chi_{[ \mu, \infty)}(x))f_\infty(x-s),\quad x\in \bbR,
$$
$$
(\widehat V_sF)_\ell(x)=\chi_{[\mu+s, \nu]}(x)f_\ell(x-s)+\sqrt{1-k^2}f_\infty(x-s), \quad x\in [\mu, \nu],
$$
where
$$F=(f_\infty, f_\ell)^T\in L^2(\bbY)=L^2((-\infty, \mu))\oplus L^2((\mu, \infty))\oplus L^2((\mu, \nu))$$ and  $\ell=\nu-\mu$.
\end{theorem}

\begin{proof}
If the metric graph $\bbY$ is in Cases I* or II, there is nothing to prove.

 We provide a complete proof  when $\bbY$ is in Case I.

 Without loss we may assume that $\mu=0$.
From the definition of the semi-group $V_s$ it follows that
\begin{equation}\label{cfi}
\lim_{s\downarrow 0} s^{-1}( \widehat V_s-I)f=-f'=i\widehat \cP f, \quad f\in C_0^\infty(\bbR\setminus \{0\}).
\end{equation}

Introduce the functions
$$
 g(x)=\begin{cases} k e^{-x},  & x\ge 0\\
 e^x,  & x<0
\end{cases}\quad \text{ and }
 \quad h(x)=\begin{cases}
k e^{-x},  & x\ge 0\\
- e^x,  & x<0
\end{cases}.
$$

Clearly, $g\in \Dom (i\widehat \cP)$ and
$
i\widehat \cP g=h.
$
It is sufficient to show that
\begin{equation}\label{cifii}
\lim_{s\downarrow 0} s^{-1}( \widehat  V_s-I)g=h=i\widehat \cP g.
\end{equation}
Indeed, \eqref{cfi} and \eqref{cifii} mean that the generator of $V_s$ restricted on
the dense linear set $\cD=C_0^\infty(\bbR\setminus \{0\})+\text{span}\{g\}$
coincides with the operator $\widehat \cP|_\cD$. Since the generator
is a closed operator and   the closure of $\widehat \cP|_\cD$ coincides with $\widehat \cP$
one proves that $\widehat \cP$ is the generator of the semi-group $ \widehat  V_s$.

Thus, it remains to prove \eqref{cifii}.

We proceed as follows.
Note that
$$
(\widehat  V_sg)(x)=
\begin{cases}
e^{x-s}, & x<0\\
k e^{x-s}, &0\le x<s\\
e^{s-x}, & x-s\ge0
\end{cases},
\quad \text{ a. e. } x\in \bbR,
$$
and hence
\begin{align*}
 \|s^{-1}&(\widehat  V_s g-g)-h \|^2=
\int_{-\infty}^0
|s^{-1}( e^{x-s}-e^{x})+e^{x}|^2 dx
\\ &
+\int_{0}^s|s^{-1}(k  e^{x-s} -k e^{x})- k e^{-x}|^2dx
\\ &
+\int_{s}^{\infty }|s^{-1}( k e^{s-x} -k e^{-x})-k e^{-x}|^2dx
\\ &
=\left (\frac{e^{-s}-1+s}{s}\right )^2\int_{-\infty}^0
e^{2x} dx
+s^{-2}k^2  \int_0^{s}e^{2x}\left  |
(e^{-s} -1- se^{-2x}\right |^2dx
\\ &
+k \left  (\frac{ e^{s} - 1-s}{s}
\right )^2\int_{s}^\infty e^{-2x}dx\to 0 \quad \text{ as } s\to 0,
\end{align*}
proving \eqref{cifii}.

The proof  is complete.

\end{proof}

\begin{remark}
Taking into account that the generator of the group $V_s$ in Case III is an operator coupling of the ones in Cases I and II, cf. Theorem  \ref{opcup}, the restriction of the dissipative dynamics in case III to its invariant subspace
$L^2((\mu, \nu))$ (with $\nu=\mu+\ell>\mu)$ gives rise to the dissipative dynamics in Case II, while the  compression  $P_{\cH} V_s|_{\cH}$ of the dissipative dynamics $V_s$  onto its  coinvariant subspace $\cH=L^2((-\infty, \mu))\oplus L^2((\mu, \infty))$ leads to the dissipative dynamics in Case I.

One can also compress the dynamics to the channel   $L^2((-\infty,\nu))\oplus L^2 ((\mu,\infty))$ to obtain the semi-group $\widehat V_s$ in Case I$^*$, provided that $\nu <\mu$.
\end{remark}

\section{Unitary dynamics on  the full graph}\label{s14}

The  metric graph $\bbY$    given by  \eqref{vaivai}
can naturally be considered a subgraph of the full metric graph  $\bbX=\bbY_\mu\sqcup\bbY_\mu$   composed of
 two identical copies of  the metric graph $ \bbY_\mu=(-\infty, \mu) \sqcup(\mu, \infty)$.

In turn, the dynamics on   the metric graph $\bbY$ can be dilated to the  unitary groups $\widehat V_s$ and $\widehat U_t$
generated by the canonical variables on $\bbX$.

To be more precise, we proceed as follows.

In the Hilbert space $L^2(\bbX)$ introduce
the unitary  group of shifts   $V_t $ that  can be recognized as the evolution  operator that maps the initial values of
the solution of
the following  first order hyperbolic system
\begin{equation}\label{system}
\partial_t\begin{pmatrix}u\\v\end{pmatrix}+\partial_x\begin{pmatrix}u\\v\end{pmatrix}=0
\end{equation}
with the boundary condition at the vertex $x=\mu$
\begin{equation}\label{systembc}
\begin{pmatrix}u\\v\end{pmatrix}(\mu+)=
\begin{pmatrix}k &-\sqrt{1-k^2}\\\sqrt{1-k^2}&k\end{pmatrix}
\begin{pmatrix}u\\v\end{pmatrix}(\mu-)
\end{equation}
into their value at $t$.  Here $k$ is a parameter such that $0\le k\le 1$.

The action
of  the  group $V_t$ can easily  be described   explicitly.

If $t>0$, one gets that
\begin{equation}\label{unitgruppa}
\left ( V_t\begin{pmatrix}
u\\
v\end{pmatrix}\right )(x)=
\begin{pmatrix}
u(x-t)\\
v(x-t)\end{pmatrix}
\end{equation}
whenever $ x<\mu$ or $ t\le x+\mu$, and
\begin{equation}\label{unitgruppa1}
\left ( V_t\begin{pmatrix}
u\\
v\end{pmatrix}\right )(x)=
\begin{pmatrix}
k \, u(x-t)-\sqrt{1-k^2} \,v(x-t)\\
\sqrt{1-k^2} \,u(x-t)+k \,u(x-t)\end{pmatrix}
\end{equation}
whenever $  \mu\le x <t$.

If $t<0$, one clearly has that
$$
V_t=(V_{-t})^*, \quad t<0.
$$

The self-adjoint generator $ \cP$  of the group
is the following  self-adjoint realization of the differentiation operator on the full graph $\bbX$ on
 $
\Dom(\cP )$ consisting of the two-component functions $f=(f_{\uparrow},f_{\downarrow})^T$,
$$(f_\uparrow,f_\downarrow)\in\left [
 W_2^1((-\infty,\mu))\oplus W_2^1((\mu, \infty))\right ]\oplus \left [W_2^1((-\infty,\mu))\oplus W_2^1((\mu, \infty))\right ]
$$
that satisfy the  boundary conditions
\begin{equation}\label{bcdilation}
\begin{pmatrix}
 f_\uparrow(\mu+)\\
 f_\downarrow(\mu+)
\end{pmatrix}=
\begin{pmatrix}
k&- \sqrt{1-k^2}\\
\sqrt{1-k^2}&k
\end{pmatrix}
\begin{pmatrix}
 f_\uparrow(\mu-)\\
  f_\downarrow(\mu-)
\end{pmatrix}.
\end{equation}
We will call the  generator  $\cP$ the  {\it self-adjoint momentum operator with the quantum gate coefficient} $0\le k \le 1$ on the full graph $\bbX$.

The unitary dynamics $ V_t$ can be illustrated on the  example of wave propagation along the  transmission line (see \eqref{systembc} and also  Fig. 4 below) and can be described
informally as follows.

The wave  packet  initially located on  $(-\infty, \mu)$ in the upper channel is transmitted to the upper channel of the semi-axis $(\mu, \infty)$  with the quantum gate coefficient $k $.
The  ``rest" of the packet gets amplified by   $\sqrt{1-k^2}$ and then  is transmitted to the lower channel of the semi-axis $(\mu, \infty)$.
The wave packet  on $(-\infty, \mu)$ in the lower channel
gets amplified by the factor  $-\sqrt{1-k^2}$ and  by the quantum gate coefficient $k$
and then is transmitted to the upper
and lower
channel  of the  semi-axis $(\mu, \infty)$, respectively.

$$
$$

\begin{pspicture}(12,6)%

\psline[linewidth=1.5pt]{->}(4.2,5.4)(5.2,5.4)

\psline[linewidth=1.5pt]{->}(4.2,2)(5.2,2)

\psline(0,3)(5.5,3)

\psline(0,5)(5.5,5)

\psline(6.5,5)(11,5)

\psline(6.5,3)(11,3)

\psdot*[dotscale=2](6,4)

\psline(5.5, 5)(6,4)

\psline(6,4)(6.5,5)

\psline(6,4)(6.5,3)

\psline(5.5, 3)(6,4)


\psline{->}(6.8,4)(6.2,4)

\rput(7,4){$\mu$}

\psarc[fillcolor=gray](2,4.8){1}{11}{169}

\psarc[fillcolor=gray](7.5,4.8){.9}{14}{167}

\psarc[fillcolor=gray](7.5,2.3){1.1}{39}{140}

\psarc*[fillcolor=gray](4.5,3){.5}{180}{360}

\psarc*[fillcolor=gray](10,4.2){.95}{57}{123}

\psarc*[fillcolor=gray](10,3.5){.7}{225}{315}

\rput(6.5,.5){{\bf Fig. 4} {\sc Unitary dynamics on  the  full metric graph $\bbX$ } }

\end{pspicture}

Based on  the explicit description of the semi-group of contractions  $ \widehat V_s$ provided by Lemma \ref{cases},  we arrive to the main result of this section
that shows that
the dissipative dynamics $ \widehat  V_s, $ $s\ge0$,  on the metric graph  $\bbY$ in any of the Cases I*, I-III can be dilated to a unitary one  in the Hilbert space  $L^2( \bbX)$
where $\bbX$ is the full graph. Equivalently,
any solution to the restricted Weyl commutation relations
$$
U_t \widehat  V_s=e^{ist} \widehat  V_sU_t, \quad s\ge 0, \,\,\,t\in \bbR,
$$
such that the generator of the semi-group   $ \widehat  V_s$ belongs to the class $\fD(\cH)$
(see Appendix G for the definition of the class $\fD(\cH)$)
 is unitarily equivalent
 to a  compression of the canonical  solution to the Weyl commutation relations
$$
 U_t  V_s=e^{ist}V_s  U_t, \quad s,t \in \bbR,
$$
in $L^2(\bbR, \bbC^2)$ onto an appropriate   coinvariant  subspace $K\subset L^2(\bbR, \bbC^2)$ that reduces the multiplication group $U_t $.

The precise statement is as follows (cf. \cite[Theorem 15]{JM}).

\begin{theorem}\label{cbcb}
Let  $\bbY\subset \bbX$  be a  metric graph  in one of the Cases I*, I-III given by \eqref{vaivai}.
Suppose that   $(\widehat \cP(k), \cQ(\bbY))$   is  the   canonical dissipative  pair with  the quantum gate coefficient $k$  on the metric  graph $\bbY$
 and $(\cP(k), \cQ(\bbX))$ is  the   canonical pair on the full graph $\bbX$.

 Then
 \begin{equation}\label{sila}
 e^{i\widehat \cP(k)s}=P_{L^2(\bbY)}e^{i\cP(k) s}|_{L^2(\bbY)},\quad  s\ge 0,
\end{equation}
and
 \begin{equation}\label{znanie}
e^{i\cQ(\bbY) t}=P_{L^2(\bbY)}e^{i\cQ (\bbX)t}|_{L^2(\bbY)},\quad  t\in \bbR.
\end{equation}
Here $P_{L^2(\bbY)}$  stands for  the orthogonal projection from the space  $ L^2(\bbX)$ onto the subspace  $ L^2(\ \bbY)\subset L^2(\bbX)$.
\end{theorem}
\begin{proof} It is convenient to  identify
the  Hilbert space  $L^2(\bbX)$ as the  von Neumann integral
$$L^2(\bbX)=
\begin{pmatrix}
L^2((-\infty, \mu))& L^2((\mu, \infty))\\
L^2((-\infty, \mu))& L^2((\mu, \infty))\\
\end{pmatrix}.
$$

Since any  metric graph $\bbY$   in  one  of the Cases I*, I-III can naturally be considered  a subgraph of the full graph $\bbX$,
 the Hilbert space $L^2( \bbY) $ can be identified with
 $$
L^2( \bbY) \approx
\begin{cases}
\begin{cases}
\begin{pmatrix}
L^2((-\infty, \nu)), & L^2((\mu, \infty))\\
0& 0\\
\end{pmatrix}, &(\nu <\mu )\\
\begin{pmatrix}
L^2((-\infty, \mu))& L^2((\mu, \infty ))\\
0& L^2((\mu, \nu))\\
\end{pmatrix}, &(\nu \ge \mu)\end{cases},& \text{in Case I*}\\
\begin{pmatrix}
L^2((-\infty, \mu))& L^2((\mu, \infty))\\
0&0
\end{pmatrix},
& \text{in Case I}\\
\begin{pmatrix}
0& 0\\
0&L^2((\mu , \nu))
\end{pmatrix} ,& \text{in Case II}\\
\begin{pmatrix}
L^2((-\infty, \mu))& L^2((\mu, \infty))\\
0&L^2((\mu , \nu))
\end{pmatrix}, &\text{in Case III}
\end{cases}.
$$
Clearly, the subspace $L^2(\bbX)\ominus L^2(\bbY)$ splits into the direct sum of incoming $\cD_-$ and outgoing subspaces $\cD_+$ for  the group $\widehat V_t$,
$$
L^2(\bbX)\ominus L^2(\bbY)=\cD_-\oplus \cD_+.
$$
For instance, in Case III,
$$
\cD_-=\begin{pmatrix}
0& 0\\
L^2((-\infty, \mu))&0
\end{pmatrix}  \quad \text{and}\quad \cD_+=\begin{pmatrix}
0& 0\\
0&L^2(( \nu,\infty)).
\end{pmatrix}
$$
Therefore the restriction of the unitary group $ V_t=e^{i\cP(k) t}$  onto its coinvariant subspaces $K=L^2(\bbY)$ is a strongly continuous semi-group of contractions.

Comparing \eqref{unitgruppa}, \eqref{unitgruppa1} with the explicit description  for the semi-groups $ \widehat  V_s$ action provided by Lemma \ref{cases} in each of the cases,
one gets that
 the unitary evolution $V_s$, $s\in \bbR$,  in $L^2 (\bbX)$  compressed to the (coinvariant) subspace
$L^2 ( \bbY)$ gives rise to
the  corresponding semi-groups of contractions
$$
\widehat  V_s=P_{L^2(\bbY)} V_s|_{L^2(\bbY)},\quad  s\ge 0,
$$
which proves \eqref{sila}. To obtain \eqref{znanie} it remains to observe that the subspace $L^2(\bbY)$ reduces the multiplication group
$U_t=e^{-i\cQ(\bbX) t}$.
\end{proof}


\newpage
\part{\Large Applications}

\section{Continuous monitoring of the quantum systems }

The aim  of this section is to present  a general discussion of possible outcomes  in  the  frequent  quantum  measurement theory.

Let  $\cH$ be the Hilbert space used in the description of a quantum system with the Hamiltonian $H$, a self-adjoint operator in $\cH$.
Recall that the time evolution of an initial  state $\phi$  of the system, a unit vector  in the Hilbert space $\cH$, is described by a one-parameter group of unitary operators
$U_t=e^{-it/\hbar H}$. Notice that if   $\phi \in \Dom (H)$, the vector $\psi(t)= e^{-it/\hbar H}\phi$ satisfies the time-dependent Schr\"odinger equation
$$
i\hbar \frac{\partial \psi }{\partial t}=H\psi.
$$

Let    $p(t)$  denote  the  survival probability
\begin{equation}\label{sprob}
p(t)=|(e^{-it/\hbar H}\phi,\phi)|^2=|(e^{it/\hbar H}\phi,\phi)|^2.
\end{equation}

One of the central problems of  frequent quantum  measurement theory is to study  the time behavior of  $[p(t/n)]^m$ for large $m$ and $ n$, and,
in particular, the computation of the limit
 $$p_c(t)=\lim_{n\to \infty}[p(t/n)]^n.$$

We will call  $p_c(t)$ the {\it survival probability under continuous monitoring} of the system and
 focus on the following  three possible scenarios:
\begin{itemize}
\item[$\alpha)$] $p_c(t)=1$, the quantum Zeno effect;
\item[$\beta)$] $p_c(t)=0$,  the quantum Anti-Zeno effect;
\item[$\gamma)$]$p_c(t)=e^{-\tau |t|}$ for some $\tau >0$, the Exponential Decay.
\end{itemize}

\subsection{Quantum Zeno effect}

\begin{hypothesis}\label{raspadpad}
Suppose that $H$ is a self-adjoint operator in the Hilbert space $\cH$ expressed by
$$
H=\int_\bbR\lambda dE_H(\lambda),
$$
where $\EE_H(\lambda)$ is  a resolution of the identity.
Let  $\phi$  be a unit vector (state) in $\cH$ and
 $\nu_\phi(d\lambda)$ denote  the spectral measure of the state $\phi$,
$$
\nu_\phi(d\lambda)=(\EE_H(d\lambda)\phi,\phi).$$
Assume that   $N(\lambda)$ is the corresponding right-continuous  distribution function
\begin{equation}\label{comdistr}
N(\lambda)=\nu_\phi((-\infty, \lambda]),\quad  \lambda\in \bbR.
\end{equation}
\end{hypothesis}

\begin{definition}We say that $\phi$  is a Zeno state under continuous monitoring of the quantum  unitary evolution $\phi \to e^{itH}\phi$  if
$$
\lim_{n\to \infty}|(e^{it/nH}\phi, \phi)|^{2n}=1 \quad \text{for all} \quad t\ge  0.
$$
\end{definition}

Recall the following necessary and sufficient conditions for the Quantum Zeno effect  to occur.
\begin{proposition}[{\cite{Atm}}]\label{zenoprop}
{\it Assume Hypothesis \ref{raspadpad}. Then
the  state $\phi$ is a Zeno state if and only if the light tails requirement
 \begin{equation}\label{lt}
\lim_{\lambda\to \infty}\lambda(1-N(\lambda)+N(-\lambda))=0
\end{equation}
holds.}

{\it In particular,  if  $\phi\in \Dom(|H|^{1/2})$, then $\phi $ is a Zeno state.}
\end{proposition}

\begin{remark}\label{fogel}
The discovery of the phenomenon  that the evolution of  quantum system can be  eventually frozen  under  continuous monitoring  is due to L. Khalfin \cite{Kh} while
the term  the quantum Zeno effect was coined by B. Misra and E.~C.~G.~Sudarshan \cite{Misra}.
The necessary and sufficient condition for the occurrence of the quantum Zeno effect  is due to H. Atmanspacher, W. Ehm and T. Gneiting  \cite{Atm}
 where the authors
 explored  the well known fact (see \cite[Theorem 1, p. 232]{F}) that  the light tails requirement \eqref{lt} is necessary and sufficient for   the  weak law of large numbers  to hold.
 We also refer to  \cite{Schmidt} for an excellent introduction to the subject. 
The quantum Zeno dynamics of a relativistic system is discussed in \cite{MenBel}. 
As for an experimental confirmation of the effect see \cite{Zeno}.

\end{remark}
\begin{remark}
Notice that the membership  $  \phi\in\Dom (|H|^{1/2})$ means that the spectral (probability) measure
 $(E_H(d\lambda)\phi, \phi)$ has  the first moment and hence $\phi $ is a Zeno state. This can also  be seen directly  (cf. \cite{Kh}) as follows.

Consider a sequence  $\xi_1, \xi_2, \dots$ of independent copies of a random variable $\xi$
 with the common distribution function $N(\lambda)$ given by \eqref{comdistr}.
By the strong law of large numbers,
 $$
 \lim_{n\to \infty} \frac{\xi_1+\xi_2+\dots +\xi_n}{n}=a \quad \text{almost surely},
 $$
 where
 $$
 a=\EE\xi=(\sign (H) |H|^{1/2}\phi, |H|^{1/2}\phi)\in \bbR
 $$
is    the mathematical expectation of the  random variable $\xi$.
Recall that $\phi \in \Dom (|H|^{1/2})$ and therefore the right hand side  $ (\sign (H) |H|^{1/2}\phi, |H|^{1/2}\phi) $ is well defined.
 Since the random variables $\xi_1, \xi_2, \dots$ are independent and equidistributed,
 we have  \begin{align*}
  (e^{it/nH}\phi, \phi)^n&= (\EE e^{it/n\xi})^n=\underbrace{\EE e^{it/n\xi_1}\cdot \EE e^{it/n\xi_2}\dots \cdot \EE e^{it/n\xi_n}}_{n \text{ times}}
  \\
 &=\EE e^{it \frac{\xi_1+\xi_2+\dots +\xi_n}{n} }\to e^{it  a }\quad \text{as} \quad n\to \infty,
 \end{align*}
 which  shows that
 \begin{equation}\label{ZZZZ}
 \lim_{n\to \infty} |(e^{it/nH}\phi, \phi)|^{2n}=1.
 \end{equation}
 That is, $\phi $ is a Zeno state.

 \end{remark}

\subsection{Anti-Zeno effect}
Frequent observations can also accelerate the decay  process and the corresponding phenomenon  is known as the quantum  anti-Zeno effect.

\begin{definition} We say that $\phi$, $\|\phi\|=1$, is an anti-Zeno state under continuous monitoring of the quantum  unitary evolution $\phi \to e^{itH}\phi$  if
\begin{equation}\label{antizeno}\lim_{n\to \infty}|(e^{it/nH}\phi, \phi)|^{2n}=0 \quad \text{for all} \quad t> 0.
\end{equation}
\end{definition}

The following lemma provides a simple sufficient condition for  a state to be an anti-Zeno state.
We state the corresponding result   using the
language of the theory of limit distributions of sums of independent random variables  (we refer to Appendix H for the  terminology and a brief exposition of the theory).

\begin{lemma}[cf. \cite{Ex}]\label{AZS} Assume Hypothesis \ref{raspadpad}. Suppose that the distribution $N(\lambda)$ of the   state $\phi$ belongs to the domain of attraction of an $\alpha$-stable law with  $0<\alpha<1$.
Then $\phi$ is  an anti-Zeno state.
\end{lemma}

\begin{proof}  The characteristic function of the distribution $N(\lambda)$ coincides with $ (e^{itH}\phi, \phi)$ and therefore admits
the representation (by Remark \ref{strem} in Appendix H)
$$
(e^{itH}\phi, \phi)=\exp \left (
-\sigma |t|^\alpha\tilde h(t)(1-i\beta \frac{t}{|t|} \omega(t,\alpha) )\right ),
$$ where $\tilde h(t)$ is slowly varying as $t\to 0$.
In particular,
$$
|(e^{it/nH}\phi, \phi)|^{2n}=\exp (
-2cn^{1-\alpha} \tilde h(t/n) |t|^\alpha(1+o(1)) \quad \text{as }\quad n\to \infty.
$$
Since (see, e.g., \cite[Appendix 1]{IL})
$$
\lim_{n\to \infty}n^{1-\alpha} \tilde h(t/n)=+\infty,
$$
one concludes that
$$
\lim_{n\to \infty}|(e^{it/nH}\phi, \phi)|^{2n}=0, \quad t\ne 0.
$$

\end{proof}
\begin{remark}
 Necessary and sufficient conditions for the quantum anti-Zeno effect to occur
can be found  in \cite[Theorem 2]{Atm}.
\end{remark}

\begin{remark} It is worth mentioning that  the situation is quite different  if the distribution function  belongs to the domain of attraction of an $\alpha$-stable law with  $1<\alpha\le 2$.
In this case
 the probability measure
$\nu_\phi(d\lambda)$ has the first moment and therefore the state $\phi$ is a Zeno state.
\end{remark}

\subsection{The exponential decay}

Next we turn to the borderline case of $\alpha$-stable distributions with $\alpha =1$ which play an exceptional role
in explanation of   the exponential decay  phenomenon    under continuous monitoring.

\begin{definition}
We say that  $\phi$ is   a {\it resonant state}  under continuous monitoring of the quantum  unitary evolution $\phi \to e^{itH}\phi$  if
\begin{equation}\label{expsc}
\lim_{n\to \infty}
|(e^{it/nH}\phi, \phi)|^{2n}=e^{-\tau |t|}, \quad\text{for some }\quad \tau>0  \quad \text{  and  all} \quad t\ge 0.
\end{equation}
\end{definition}

\begin{remark} Notice that one can also consider  exponentially decaying (resonant) states    by requiring that
the survival probability $p(t)=|(e^{itH}\phi, \phi)|^{2}$ tends to zero exponentially fast as $|t|$ approaches infinity. In this case, however, the spectrum of the Hamiltonian
$H$ has to fill in the whole real axis which excludes from the consideration the quantum systems with semi-bounded Hamiltonians.

It can be easily seen as follows.
The requirement that  the survival probability  $p(t)$  falls off exponentially  implies that the survival probability amplitude $ (e^{itH}\phi, \phi)$ is the  Fourier transform of an absolute continuous measure with the density that is analytic in a strip
containing the real axis. In particular,  the Radon-Nykodim derivative $\frac{d}{d\lambda}(E_H(\lambda) \phi,\phi)$ of the spectral measure of the element $\phi $ is positive almost everywhere which shows that $\spec(H)=\bbR$. If the Hamiltonian $H$ has a gap in its spectrum, then  there are no exponentially decaying states whatsoever unless the quantum system is under continuous monitoring. The geometric reason behind   is that the unitary group $e^{itH}$ does not have orthogonal incoming and outgoing subspaces as it follows from the Hegerfeldt Theorem \cite{Hegerfeldt}.
\end{remark}

 A sufficient condition for the exponential decay  \eqref{expsc} is provided by the following corollary of the   Gnedenko-Kolmorogov limit theorem (see Theorem  \ref{stthm} in Appendix H).

\begin{theorem}\label{levy}  Assume Hypothesis \ref{raspadpad}.
Suppose, in addition, that
\begin{equation}\label{zapozdalo}
\lim_{\lambda\to \infty}\lambda(1-N(\lambda))=\frac{1+\beta}{\pi}\sigma
\quad \text{
and }\quad
\lim_{\lambda\to \infty}\lambda N(-\lambda)=\frac{1-\beta}{\pi}\sigma
\end{equation}
for some
 $\sigma > 0$ and  $\beta\in [-1,1]$.  Then $\phi $ is a resonant state
and
\begin{equation}\label{eds}
\lim_{n\to \infty}|(e^{it/n H}\phi, \phi)|^{2n}=e^{-2\sigma |t|}, \quad t\in \bbR.
\end{equation}

\end{theorem}
\begin{proof}
By Theorem \ref{stthm} and Remark \ref{strem} in Appendix H, the distribution $N(\lambda)$ belongs to the domain of normal  attraction of the  $1$-stable law. In particular,  there are constants $A_n$ such  that
$$
\lim_{n\to \infty}(e^{it/n H}\phi, \phi)^ne^{iA_n}=\exp \left (-\sigma |t| \left (1+i\beta \frac2\pi \frac{t}{|t|} \log |t|\right )\right ),
$$
from which  \eqref{eds} follows.
\end{proof}

\begin{remark} If $\sigma>0$ and therefore  $\phi$ is a resonant state,
the probability measure $\nu_\phi(d\lambda)$ does not have the first moment and hence
$$
\phi\notin\Dom( {|H|}^{1/2}).
$$
In this case
 the  ``total energy" of the quantum system  in the state $\phi$ is infinite.
 Introducing   the ``ultra-violet" cut-off Hamiltonian
$$
H_E=\int_{|\lambda|\le E}\lambda dE_H(\lambda),
$$
one observes that
the ``truncated" energy $(|H_E|\phi, \phi)$ of the state $\phi$ is log-divergent  as the truncation parameter $E $ approaches infinity. That is,
\begin{equation}\label{rashod}
(|H_E|\phi, \phi)=\frac{2}\pi \sigma \log E+o(\log E)\quad \text{as } \quad E\to \infty.
\end{equation}
In particular, the parameter $\sigma$ determines the rate of convergence of the  mean-value cut-off energy in the logarithmic scale  as
$$
\sigma=\frac\pi2\lim_{E\to \infty}\frac{(|H_E|\phi, \phi)}{\log E}.
$$
Indeed,
one gets that
\begin{equation}\label{tosha}
(H_E\phi, \phi)=\int_{[-E, E]}\lambda dN(\lambda)=\int_{[-E, 0)}\lambda dN(\lambda)+\int_{(0, E]}\lambda dN(\lambda).
\end{equation}
Integrating by parts (see, e.g.,  \cite[Theorem 3.36]{Folland}) one obtains
\begin{align*}
\int_{(0, E]}\lambda dN(\lambda)&=\int_{(0, E]}\lambda d(N(\lambda)-1)\nonumber
\\&=E(N(E)-1)+\int_0^E(1-N(\lambda))d\lambda
\\&=-\frac{1+\beta}{\pi} \sigma+\left (\frac{1+\beta}{\pi} \sigma \log E+o(\log E)\right )\quad \text{as}\quad E\to \infty,
\end{align*} where we have used   \eqref{zapozdalo} on the last step.

Therefore,
\begin{equation}\label{chardiv}
\int_{(0, E]}\lambda dN(\lambda)=\frac{1+\beta}{\pi} \sigma \log E+o(\log E)\quad \text{as} \quad  E\to \infty.
\end{equation}

In a similar way one shows that
\begin{equation}\label{tosha1}
\int_{[-E, 0)}\lambda dN(\lambda)=-\frac{1-\beta}{\pi }\sigma \log E +o(\log E)\quad \text{as} \quad  E\to \infty,
\end{equation}
which together with \eqref{chardiv} implies \eqref{rashod}.

 Combing \eqref{tosha}, \eqref{chardiv} and \eqref{tosha1}, one
also justifies  the logarithmic divergence of the averaged truncated  energy  of the state ($\beta \ne 0$), that is,
\begin{equation}\label{aven}
(H_E\phi, \phi)=\frac{2\beta}\pi \sigma \log E+o(\log E)\quad \text{as } \quad E\to \infty.
\end{equation}
\end{remark}

We conclude this subsection by an abstract result  that takes place  for  a special class of hyperbolic (quantum) systems the Hamiltonian of which is a self-adjoint dilation of a dissipative operator (cf. \eqref{dil} for the definition of a self-adjoint dilation).

\begin{lemma}\label{dissss}
Suppose that a self-adjoint operator $H$ in a Hilbert space $\cH$ dilates   a maximal dissipative operator  $\widehat A$ acting in a subspace  $\cK\subset \cH$.
 Assume that  $\phi \in \Dom (\widehat A)\subset \cK$
is such that $\|\phi\|=1$.

Then
 \begin{equation}\label{raspad100}
\lim_{n\to \infty}|(e^{it/n H } \phi, \phi)|^{2n}=\lim_{n\to \infty}|(e^{it/n \widehat A } \phi, \phi)|^{2n}=e^{-\tau t }, \quad t\ge 0,
\end{equation}
where   $$
\tau=2 \Im (\widehat A \phi, \phi).
$$

\end{lemma}

\begin{proof} Since $\phi \in \Dom (\widehat A)$,
we have
$$
(e^{i\varepsilon\widehat A}\phi, \phi)=1+i\varepsilon   (\widehat A \phi, \phi)+o(\varepsilon)\quad \text{as }\quad \varepsilon\downarrow 0.
$$
Therefore,
$$
|(e^{i\varepsilon\widehat A}\phi, \phi)|=1-\varepsilon  \Im (\widehat A \phi, \phi)+o(\varepsilon)\quad \text{as} \quad \epsilon \downarrow 0,
$$
and hence
$$
\lim_{n\to \infty}|(e^{it/n \widehat A } \phi, \phi)|^{2n}=e^{- 2 \Im (\widehat A \phi, \phi)t }, \quad t\ge 0.
$$

 By the hypothesis, the Hamiltonian $H$ dilates $\widehat A$ and hence
$$
\lim_{n\to \infty}|(e^{it/n H } \phi, \phi)|^{2n}=\lim_{n\to \infty}|(e^{it/n \widehat A } \phi, \phi)|^{2n}=e^{-\tau t }, \quad t\ge 0.
$$

\end{proof}

\subsection{Frequent measurements and the time-energy uncertainty principle}\label{sub15.4}

The study of the limit behavior of the survival probability $[p_c(t/n)]^m$ as $m,n \to \infty$ and $m\ne n$ in appropriate time-scales  is also of definite interest.
For instance,  if $\frac{m}{n}\to \infty$, we deal with the case of
  {\it prolonged  } frequent quantum measurements.

For instance,  prolonged frequent   measurements with  $m(n)=n^2$ can  eventually ``unfreeze" states that were {\it a priory} Zeno states.

Indeed, assume
  that $\phi\in \Dom(H)$ and therefore $\phi$ is a  Zeno state.
  We have
  $$
  (e^{itH}\phi, \phi)=1+i(H\phi,\phi)t-\frac 12(H\phi,H\phi) t^2+o(t^2), \quad \text{as }\quad t\to0.
  $$
  In particular,
  $$
  | (e^{itH}\phi, \phi)|^2=1-t^2(\Delta H)^2+o(t^2),
  $$
  where
  $$
  \Delta H=\left (\|H\phi\|^2-(H\phi, \phi)^2\right )^{1/2}<\infty.
  $$
Therefore,
\begin{equation}\label{cltzeno}
   \lim_{n\to \infty} | (e^{it/nH}\phi, \phi)|^{2n^2}=e^{-(\Delta H)^2t^2}
  \end{equation}
and the state $\phi$  exhibits an exponential decay in a non-linear  time scale.
In other words,  if the quantum system is observed  $n^2$ times with  the ``frequency'' $\omega=n/t$ ($t$ is fixed and  $n$ is large)
the results of  the {\it  prolonged frequent measurements}  of the survival probability
unfreezes the  Zeno state $\phi\in \dom (H)$.

  More precisely, the state $\phi$ is a resonant state in the time-scale
$$\mathfrak{t}(t)=\sqrt{t}
$$
in the sense that
\begin{equation}\label{CLT}
\lim_{n\to \infty}[p(\mathfrak{t}(t/n))]^n=\exp\left (- (\Delta H)^2 |t|\right ).
\end{equation}

In this situation one can give an   estimate for the survival  probability     from above
via the angle $\theta(t)=\arccos  |(e^{itH}\phi, \phi)|$ between the  states $\phi $ and $e^{itH}\phi$,  which is an important  geometric characteristics  of the trajectory
$\psi(t)= e^{itH}\phi$ in the Hilbert space.

 To do so, we
use  the Mandelstam-Tamm time-energy uncertainty relation \cite{MTamm} (also see \cite{PFR})
  \begin{equation}\label{PFF}
 \theta(t)\le t\Delta H \quad\left ( 0\le t\Delta H\le\frac{\pi}{2 }\right ),
  \end{equation}
one gets an   {\it aposteriory}  estimate
$$
\lim_{n\to\infty}|(e^{it/nH}\phi, \phi)|^{2n^2}\le e^{-\theta^2(t)}.
$$
via the angle $\theta(t)$.

\section{The Quantum Zeno versus Anti-Zeno effect alternative}\label{s16}

In this section  we focus our attention  on  continuous monitoring of massive one-dimensional particles on a semi-axis.
We assume that
the Hamiltonian for a particle with one degree of freedom is the  one-dimensional Schr\"odinger operator $$H=- \frac{\hbar^2}{2m}\frac{d^2}{dx^2}$$
 in the Hilbert space $L^2((0,\infty))$.
  In the system of units where $\hbar =1$ and mass $m=1/2$ the Hamiltonian is given by the differential expression
$$H=- \frac{d^2}{dx^2}
$$
with appropriate boundary conditions at the origin.
It turns out that the results of frequent measurements for such quantum systems   depend on  the specific choice
of the boundary conditions at the origin and they differ qualitatively. For instance,
in the case of the Dirichlet Schr\"odinger operator, any smooth initial state with $\phi(0)\ne 0$ is an Anti-Zeno  state under the continuous monitoring. In contrast to this,
if the quantum evolution is governed by any other self-adjoint realization of the second order differentiation operator, then  all smooth initial states are Zeno states.

The proper understanding of this phenomenon requires   a more thorough   analysis of the decay properties of quantum systems, which we will proceed below.

We start with a definition of a resonant state under continuous monitoring in a non-linear time-scale.

\begin{definition}
Let $\mathfrak{t}(t)$  be an increasing continuous function  of $t$ such that
$$
\mathfrak{t}(0)=0.
$$
We say that the state $\phi$  is a resonant state  under continuous monitoring of the quantum  unitary evolution $\phi \to e^{itH}\phi$
in the time-scale  $\mathfrak{t}(t)$ if
$$
\lim_{n\to \infty}[p(\mathfrak{t}(t/n))]^n=e^{-\sigma |t|}.
$$
where
$$
p(t)=|(e^{itH}\phi, \phi)|^2
$$
is the survival probability.
\end{definition}

\begin{remark}
If $\phi\in \dom(H)$,  then $\phi$ is a Zeno state by Proposition \ref{zenoprop} in the standard (linear) time-scale
$$
\mathfrak{t}(t)=t,
$$
but $\phi$ is simultaneously a resonant state in the non-linear time-scale
$$
\mathfrak{t}(t)=\sqrt{t}
$$
as it follows from \eqref{cltzeno} (see Subsection \ref{sub15.4} where the concept of  a prolonged frequent measurement is discussed).
\end{remark}

First, we treat  the case of the   Schr\"odinger operator on the positive semi-axis with the Dirichlet boundary condition at the origin.

Before formulating the corresponding result recall (see \cite{F}) that  if $\cN$ denotes the normal distribution
$$
\cN(\lambda)=\frac{1}{\sqrt{2\pi}}\int_{\-\infty}^\lambda e^{-\frac12 y^2}dy,
$$
then
\begin{equation}\label{levydis}
F_\sigma(\lambda)=2\left [1- \cN\left (\sqrt{\frac{\sigma}{\lambda}}\right )\right],\quad \lambda>0,
\end{equation}
defines one-sided stable (L\'evy) distribution with index of stability  $\frac12$
the characteristic function $f(t)$ of which  is given by
 \begin{equation}\label{f12}
 f(t)=
\exp\left (-\sigma |t|^{1/2}\left (1-i\, \frac{t}{|t|}\right)\right ).
\end{equation}
It is remarkable that along with the Gaussian and Cauchy  distributions, the probability density function of the L\'evy distribution  is known in closed form
\cite{F}
$$
\rho(\lambda)=\left (\frac{\sigma}{2\pi} \right )^{1/2}\frac{1}{\lambda^{3/2}}e^{-\frac{\sigma}{2\lambda}}, \quad \lambda>0.
$$

Our first result shows that for  a typical initial state $\phi$ ($\phi \in W_2^2((0, \infty))$ such that $\phi(0)\ne 0$)
the spectral measure $(E_H(d\lambda)\phi,\phi)$ of the state $\phi$ has $r$-moments for all $r<1/2$ but not for $r=1/2$.
Here $H$ denotes the Schr\"odinger operator with the Dirichlet boundary condition at the origin. In particular this means that such states are anti-Zeno states under continuous
monitoring of the quantum evolution $\phi\mapsto e^{itH}\phi$. However, in the time scale $\mathfrak{t}(t)=t^2$ such states do  exhibit exponential decay.  Equivalently,
short frequent measurements yield
$$
\lim_{n\to\infty}\left |(e^{it/n H}\phi,\phi)\right |^{2\sqrt{n}}=e^{-\sigma |t|^{1/2}}
$$
for some $\sigma >0$.

\begin{theorem}\label{1/2thm}   Let $\cH=L^2((0, \infty))$ and
\begin{equation}\label{schr1}
H=- \frac{d^2}{dx^2}
\end{equation}
be  the Schr\"odinger operator with the Dirichlet boundary condition at the origin
$$
\Dom ( H)=\{f\in W_2^2((0, \infty))\,|\, f(0)=0\}.
$$

  Suppose that $\phi\in W_2^2((0, \infty))$  is such that $\phi(0)\ne 0$ and  $\|\phi\|=1$.

Then the distribution function $N(\lambda)$  of
the spectral measure
$$\nu_\phi(d\lambda)=(E_H(d\lambda)\phi, \phi)$$ of the element $\phi$ belongs to the domain of normal attraction  of
the one-sided $\frac12$-stable L\'evy distribution $F_\sigma $  \eqref{levydis} the characteristic function of which is given by \eqref{f12} with
$$
\sigma=\sqrt{\frac2\pi}|\phi(0)|^2.
$$

In particular, the state $\phi$ is a resonant state in the time-scale
$$\mathfrak{t}(t)=t^2.
$$ That is,
\begin{equation}\label{nuono}
\lim_{n\to \infty}[p(\mathfrak{t}(t/n))]^n=\exp\left (-2 \sigma |t|\right ),
\end{equation}
where
$$
p(t)=|(e^{itH}\phi, \phi)|^2
$$
is the survival probability.

\end{theorem}
\begin{proof}Let $\dot H$ be the restriction of $H$ on
$$
\Dom(\dot H)=\{f\in \Dom (H)\, |\, f(0)=f'(0)=0\}.
$$

It is known that  the Weyl-Titchmarsh function $M(z)$ associated with the pair $(\dot H, H)$  admits the representation   \cite{GMT}
$$
M(z)=i\sqrt{2z}+1, \quad z\in \bbC_+.
$$
By the Stieltjes inversion formula,  we have that
$$
M(z)=\int_0^\infty\left (\frac{1}{\lambda-z}-\frac{\lambda}{\lambda^2+1}\right )d\mu(\lambda),
$$
where $\mu(d\lambda)$ is an absolutely continuous measure supported by      the positive semi-axis with the density
\begin{equation}\label{forsss}
\frac{d\mu(\lambda)}{d\lambda}=\frac1\pi  \Im (M(\lambda+i0))=\frac{\sqrt{2}}{\pi}\sqrt{\lambda}, \quad \lambda>0.
\end{equation}

Suppose that $g_\pm$,  $\|g_\pm \|=1$,
 are deficiency elements   $g_\pm \in \Ker ((\dot H)^*\mp i I)$ such that
\begin{equation}\label{proverka}
g_+-g_-\in \Dom(H).
\end{equation}

In fact, the  deficiency  elements $g_\pm $ of the symmetric operator $\dot H$
 can be chosen as  (see, e.g., \cite{GMT})
\begin{equation}\label{nudada1}
g_+(x)=2^{1/4}e^{i\frac{\sqrt{2}}{2}x}e^{-\frac{\sqrt{2}}{2}x}
\quad \text{and}\quad
g_-(x)=\overline{g_+(x)}, \quad x\ge 0.
\end{equation}
In this case, $g_+(0)-g_-(0)=0$ which shows that \eqref{proverka} holds.

 One can apply  Theorem \ref{unitar} in Appendix C to conclude that  there is a unitary map $U$  from $L^2((0, \infty))$ onto $L^2((0, \infty); d\mu)$ such that
$UHU^{-1}$ is the operator of multiplication  by  independent variable in $L^2((0, \infty); d\mu)$.

Since \eqref{proverka} holds, from Remark \ref{snoska} in Appendix C it follows that
\begin{equation}\label{vania}
(Ug_\pm)(\lambda)=\frac{\Theta}{\lambda \mp i},\quad \lambda >0,
\end{equation}
for some $|\Theta|=1$.

By the hypothesis,
$
\phi \in W_2^2((0, \infty))= \Dom((\dot H)^*)
$.
Therefore, in accordance  with von Neumann's formula
the element $\phi$ admits the representation
\begin{equation}\label{phirep}
\phi=\alpha g_++\beta g_-+h
\end{equation}
for some uniquely determined $\alpha, \beta \in \bbC$ and
$
h\in \Dom (\dot H)\subset \dom (H).
$

We claim
 that the distribution function $N(\lambda)$   of
the spectral measure $$ \nu_\phi(d\lambda)=(E_H(d\lambda)\phi, \phi)$$ of the element $\phi$
 admits the asymptotic representation
 \begin{equation}\label{polust}
1-N(\lambda)=\frac{2\sqrt{2}}{\pi}\frac{|\alpha+\beta|^2}{\sqrt{\lambda}}+o(\lambda^{-5/4})  \quad \text{as } \lambda \to \infty.
\end{equation}

Indeed, from \eqref{vania}  and  \eqref{phirep}
 it follows that
\begin{equation}\label{abrep}
(U\phi)(\lambda)=\frac{a}{\lambda-i}+\frac{b}{\lambda+i}+(Uh)(\lambda),
\end{equation}
where $a=\Theta\alpha  $ and $b=\Theta\beta$. In particular,
\begin{equation}\label{abba}
 |a+b|=|\alpha+\beta|.
 \end{equation}
Working out the computations in the model representation provided by Theorem \ref{unitar},
in Appendix C, we obtain for   the  distribution function  $N(\lambda)$  the representation
\begin{align}
1-N(\lambda)&=\int_{\lambda}^\infty \left |  \frac a{s-i}+ \frac b{s+i} +(Uh)(s)\right |^2d\mu(s)
=|a +b |^2\int_\lambda^\infty
\frac{d\mu(s) }{s^2+1}\label{odmin}
\\&+2\Re \, a \overline b\int_\lambda^\infty  \left (
\frac{1}{(s-i)^2}-\frac{1}{s^2+1}\right )d\mu(s)  \nonumber
\\&+2 \Re \int_\lambda^\infty \left ( \frac a{s-i}+ \frac b{s+i}\right )\overline{(Uh)(s)}d\mu(s)  \nonumber
\\&+\int_\lambda^\infty\left | (Uh)(s)\right |^2d\mu(\lambda), \quad \lambda\ge 0.\nonumber
\end{align}

From \eqref{forsss} it follows
\begin{equation}\label{ddd1}
\int_\lambda^\infty
\frac{d\mu(s) }{s^2+1}=\frac{\sqrt{2}}{\pi} \int_\lambda^\infty\frac{\sqrt{s}}{s^2+1}d s=\frac{2\sqrt{2}}{\pi} \frac{1}{\sqrt{\lambda}}+O(\lambda^{-3/2})
\quad\text{as }\lambda \to \infty.
\end{equation}
Therefore, for the first  term of the right hand side of \eqref{odmin} we have the asymptotic representation
$$
|a +b |^2\int_\lambda^\infty
\frac{d\mu(s) }{s^2+1}=|\alpha+\beta|^2\frac{2\sqrt{2}}{\pi} \frac{1}{\sqrt{\lambda}}+O(\lambda^{-3/2}).
$$
Here we have used \eqref{abba}.

The remaining  three terms in \eqref{odmin} can be estimated as follows
\begin{align}
\int_{\lambda}^\infty \left |\frac{1}{(s -i)^2}-\frac{1}{s^2+1}\right |d\mu(s)
&\le \frac{3}\lambda \int_{\lambda}^\infty  \frac{d\mu(s) }{s^2+1}
=O(\lambda^{-3/2}),
 \label{ddd2}
\end{align}
\begin{align}
\left |\int_{\lambda}^\infty \frac{\overline{(Uh)(s)}}{s\pm i}d\mu(s)\right |&\le \frac1\lambda \sqrt{\int_{\lambda}^\infty  \frac{d\mu(s) }{s^2+1}}\cdot\sqrt{\int_\lambda^\infty (1+s^2)|(Uh)(s)|^2 d\mu(s)}\label{ddd3}
\\&
=o(\lambda^{-5/4}), \nonumber
\end{align}
and
\begin{equation}\label{ddd4}
\int_{\lambda}^\infty\left | (Uh)(s)\right |^2d\mu(s)\le\frac1{\lambda^2}\int_\lambda^\infty (1+s^2)|(Uh)(s)|^2 d\mu(s)=o(\lambda^{-2}),
\end{equation}
$$\text{as}\quad \lambda\to \infty.
$$
Here, in \eqref{ddd3} and \eqref{ddd4} we have used that $h\in \Dom (\dot H)\subset  \Dom (H)$, so that
   $$Uh\in L^2(\bbR; (1+\lambda^2)d\mu(\lambda)).$$

Combining \eqref{ddd1} and the asymptotic estimates \eqref{ddd2}-\eqref{ddd4},  from \eqref{odmin}, we get
\begin{equation}\label{polust0}
1-N(\lambda)=\frac{2\sqrt{2}}{\pi}\frac{|\alpha+\beta|^2}{\sqrt{\lambda}}+o(\lambda^{-5/4}) \quad\text{as}\quad  \lambda  \to \infty.
\end{equation}

Next, we evaluate   $|\alpha+\beta|^2$  via   the boundary data  $|\phi(0)|^2$.

One observes  (see \eqref{phirep}) that the boundary condition
$$
\phi(0)=\alpha g_+(0)+\beta g_-(0)+h(0)
$$
holds. Now,
since $h\in \Dom (\dot H)$,  we have $h(0)=0$, so that
$$
\phi(0)=\alpha g_+(0)+\beta g_-(0)=2^{1/4}(\alpha+\beta)
$$
as it follows from \eqref{nudada1}.
Hence,
$$
|\alpha+\beta|^2=\frac{|\phi(0)|^2}{\sqrt{2}}.
$$

Now, taking into account
that  $\phi(0)\ne0$, from \eqref{polust0} we get the asymptotic representation   $$
N(\lambda)=1-\frac{2}{\pi}\frac{|\phi(0)|^2}{\lambda^{1/2}}(1+o(1))
\quad \text{as}\quad \lambda\to \infty.
$$
Moreover, since $H$ is a non-negative operator,   we obviously have
$$
N(\lambda)=0, \quad \lambda<0.
$$

 By Theorem  \ref{stthm}  in Appendix H,   the distribution $N(\lambda)$ belongs to the domain  of normal attraction of
 the one-sided stable L\'evy distribution $F_\sigma $
with the characteristic function
\begin{equation}\label{lawlaw}f(t)=
\exp\left (-\sigma |t|^{1/2}\left (1-i\, \frac{t}{|t|}\right )\right ),
\end{equation}
 where
 $$
\sigma=\sqrt{\frac{2}{\pi}}|\phi(0)|^2.
$$

 Indeed, the distribution $N(\lambda)$ satisfies the conditions \eqref{as+} and \eqref{as-}  of Theorem \ref{stthm} in Appendix H with
 $\alpha =\frac12
 $,
 $$c_1= \frac{2}{\pi}|\phi(0)|^2 \quad \text{and }\quad c_2=0.
 $$
 Therefore, $N(\lambda) $ belongs to the domain  of normal attraction of   a stable law with the characteristic function
$$f(t)=
\exp\left (-\sigma |t|^{1/2}\left (1-i\, \beta \frac{t}{|t| }\omega\left (t, \frac12\right)\right )\right ).
$$
Here
\begin{align*}
\sigma&= (c_1+c_2)d\left ( \frac12\right )
=\frac{2}{\pi}|\phi(0)|^2\cdot d\left ( \frac12\right ),
\\
\beta&=\frac{c_1-c_2}{c_1+c_2}=1,
\\
\omega\left (t, \frac12\right)&=\tan \left (\frac{\pi}{4} \right )=1,
\end{align*}
and
$$
d\left (\frac12\right )=\Gamma (1/2)\cos \frac\pi4=\frac12\sqrt{2\pi}.
$$
The main assertion of the theorem is now proven.

To complete the proof of the theorem it remains to apply the $1/2$-stable limit theorem
to see that
\begin{equation}\label{nuonomb}
\lim_{n\to \infty}\left |\left (\exp \left (i n^{-2}t  H\right )\phi, \phi\right )\right |^{2n}=\exp\left (-2 \sigma |t|^{1/2}\right ),
\end{equation}
which justifies \eqref{nuono} by a  change of variables.

\end{proof}

The situation is quite different for any other self-adjoint realization of the  free  Schr\"odinger operator $H'$ on the semi-axis. In this case,
the spectral measure $(E_H(d\lambda)\phi,\phi)$ of a typical state $\phi \in W_2^2((0,\infty))$ has $r$-moments for all $r<3/2$. As a consequence,  such states are Zeno states under continuous
monitoring of the quantum evolution $\phi\mapsto e^{itH'}\phi$. However, $\phi$ becomes a resonant state in the time scale $\mathfrak{t}(t)=t^{2/3}$.
 Equivalently,
prolonged  measurements ``unfreeze" the quantum  system  and
$$
\lim_{n\to\infty}\left |(e^{it/n H}\phi,\phi)\right |^{2n\sqrt{n}}=e^{-\sigma' |t|^{3/2}}
$$
for some $\sigma' >0$.

\begin{theorem}\label{3/2thm}
Let $\cH=L^2((0, \infty))$, $\gamma \in \bbR$ and
\begin{equation}\label{schr2}
H'=- \frac{d^2}{dx^2}
\end{equation}
be  the Schr\"odinger operator with the  mixed boundary condition at the origin
$$
\Dom ( H')=\{f\in W_2^2((0, \infty))\,|\,f'(0)+\gamma f(0)=0\}.
$$

Suppose that $\phi\in W_2^2((0, \infty))$, $\|\phi\|=1$,  and assume, in addition,   that $$\phi '(0)+\gamma \phi(0)\ne 0.$$

Then
the  distribution function $N(\lambda)$ of
the spectral measure
$$ \nu_\phi(d\lambda)=(E_{H'}(d\lambda)\phi, \phi)$$ of the element $\phi$ belongs to the domain of normal attraction of the $3/2$-stable law
with the characteristic function
 $$
f(t)=\exp\left (-\sigma' |t|^{3/2}(1+i\, \sign(t))\right ),
$$
where
$$
 \sigma'=\frac23 \sqrt{\frac2\pi}|\phi'(0)+\gamma \phi (0)|^2.
$$

In particular, the state $\phi$ is a resonant state in the time-scale
$$\mathfrak{t}(t)=t^{2/3}.
$$ That is,
\begin{equation}\label{nuonohom}
\lim_{n\to \infty}[p(\mathfrak{t}(t/n))]^n=\exp\left (-2 \sigma' |t|\right ),
\end{equation}
where
$$
p(t)=|(e^{itH}\phi, \phi)|^2
$$
is the survival probability.

\end{theorem}

\begin{proof} Let  $\dot H$  be the symmetric restriction of the operator $H'$ on
$$
\Dom(\dot H)=\{f\in \Dom (H')\, |\, f(0)=f'(0)=0\}.
$$

Denote by  $g_\pm$, $\|g_\pm\|=1$,  the deficiency elements of the symmetric operator $\dot H$
\begin{equation}\label{gpmrep}
g_+(x)=2^{1/4}e^{i\frac{\sqrt{2}}{2}x}e^{-\frac{\sqrt{2}}{2}x}\quad\quad \text{and }\quad
g_-(x)=\overline{ g_+(x)}\Theta,
\end{equation}
where $\Theta$ is chosen in such a way to ensure  that
\begin{equation}\label{minusy}
g_+-g_-\in \Dom (H').
\end{equation}

The parameter $\Theta$ can be determined as follows. From \eqref{minusy} it follows that   $g_+(x)-g_-(x)$  should  satisfy the boundary condition
$$
(g_+'(0)-g_-'(0))+\gamma (g_+(0)-g_-(0))=0
$$
and hence
\begin{equation}\label{eqfort}
(\zeta -\bar \zeta \Theta)+\gamma  (1-\Theta)=0,
\end{equation}
where
\begin{equation}\label{delo}
\zeta=
\frac{\sqrt{2}}{2}(1-i).
\end{equation}
Solving \eqref{eqfort} for $\Theta$  yields
$$
\Theta=\frac{\zeta +\gamma}{\bar\zeta +\gamma}.
$$

Since \eqref{minusy} holds,   one can apply Theorem \ref{unitar} in Appendix C, and
the same reasoning as the one in the proof of Theorem \ref{1/2thm} shows that the leading term of the asymptotics of the distribution function $N(\lambda)$ is given by
\begin{equation}\label{justi}
1-N(\lambda)= |\alpha +\beta|^2\int_\lambda^\infty  \frac{d\mu(s)}{s^2+1}+o(\lambda^{-7/4}) \quad \text{as}\quad \lambda \to \infty.
\end{equation}
Here
$\mu(d\lambda) $ is the measure associated with the Herglotz-Nevanlinna decomposition for the Weyl-Titchmarsh  function   $M(z)$ associated  with the pair $(\dot H, H')$
$$
M(z)=\int_{\text{spec}(H')}\left (\frac{1}{\lambda-z}-\frac{\lambda}{\lambda^2+1}\right )d\mu(\lambda)
$$
  and  $\alpha$ and $\beta $ are determined by   the von Neumann decomposition
\begin{equation}\label{vnrepr}
\phi= \alpha g_++\beta g_-+h, \quad h\in \Dom (\dot H).
\end{equation}

To justify the asymptotic representation \eqref{justi} we argue as follows.

First recall, that it is known that
the Weyl-Titchmarsh function associated with the pair $(\dot H,H')$ has    the form
\begin{equation}\label{mspec}
M(z)=\frac{\cos \alpha +\sin \alpha (i\sqrt{2z}+1)}{\sin  \alpha -\cos\alpha (i\sqrt{2z}+1)}, \quad z\in \bbC_+.
\end{equation}
Here the boundary condition parameter $\gamma$ and (the von Neumann extension) parameter $\alpha$ are related as   \cite{GMT}
\begin{equation}\label{tutsi}
\gamma=2^{-1/2} (1- \tan \alpha), \quad \alpha\ne \frac\pi2.
\end{equation}
From \eqref{mspec} it follows that the restriction of the measure $\mu(d\lambda)$ on the positive semi-axis is an absolutely continuous measure with the density given by
$$
\frac{d\mu(\lambda)}{d\lambda}=\frac1\pi \Im (M(\lambda+i0))d\lambda, \quad \lambda>0.
$$
Explicit computations show that
\begin{align*}
\frac1\pi \Im (M(\lambda+i0))d\lambda&
=\frac1\pi \Im \frac{\cos \alpha +\sin \alpha (i\sqrt{2\lambda}+1)}{\sin  \alpha -\cos\alpha (i\sqrt{2\lambda}+1)}
\\&
=\frac1\pi \Im  \frac{(\cos \alpha +\sin \alpha (i\sqrt{2\lambda}+1))(\sin  \alpha -\cos\alpha (-i\sqrt{2\lambda}+1))}{(\sin  \alpha -\cos\alpha)^2+2\lambda \cos^2\alpha }
\\&
=\frac1\pi \Im  \frac{(\cos \alpha   +\sin \alpha+ i \sqrt{2\lambda}\sin \alpha )(\sin  \alpha -\cos\alpha +i\sqrt{2\lambda}\cos\alpha )}{(\sin  \alpha -\cos\alpha)^2+2\lambda \cos^2\alpha }
\\&
=\frac1\pi  \frac{\sqrt{2\lambda}}{(\sin  \alpha -\cos\alpha)^2+2\lambda \cos^2\alpha }, \quad \lambda>0.
\end{align*}
Therefore,
\begin{equation}\label{meraa}
d\mu(\lambda)=\frac1\pi  \frac{\sqrt{2\lambda}}{(\sin  \alpha -\cos\alpha)^2+2\lambda \cos^2\alpha }d\lambda, \quad \lambda>0.
\end{equation}

To justify \eqref{justi}, in particular, to see that the  error term is of the order of $o(\lambda^{-7/4}) $ as $\lambda \to \infty$,  we argue  exactly as  in the proof of Theorem \ref{1/2thm}.  To do so,  we need to estimate  the following three integrals (we use the notation from the proof of Theorem \ref{1/2thm})
\begin{align*}
I&=\int_{\lambda}^\infty \left |\frac{1}{(s -i)^2}-\frac{1}{s^2+1}\right |d\mu(s),
\\
II&=\left |\int_{\lambda}^\infty \frac{\overline{(Uh)(s)}}{s\pm i}d\mu(s)\right |,
\end{align*}
and
$$
III=\int_{\lambda}^\infty\left | (Uh)(s)\right |^2d\mu(s).
$$

We have (as $\lambda\to  \infty)$
\begin{equation}\label{ddd22}
I\le \frac{3}\lambda \int_{\lambda}^\infty  \frac{d\mu(s) }{s^2+1}=O(\lambda^{-5/2}).
 \end{equation}
 Since  $h\in \Dom (\dot H)\subset  \Dom (H)$, and therefore
   $Uh\in L^2(\bbR; (1+\lambda^2)d\mu(\lambda)))$, we also have the asymptotic estimates
\begin{align}
II&\le \frac1\lambda \sqrt{\int_{\lambda}^\infty  \frac{d\mu(s) }{s^2+1}}\cdot\sqrt{\int_\lambda^\infty (1+s^2)|(Uh)(s)|^2 d\mu(s)}
=o(\lambda^{-7/4})
\label{ddd32}
\end{align}
and
\begin{equation}\label{ddd42}
III\le \frac1{\lambda^2}\int_\lambda^\infty (1+s^2)|(Uh)(s)|^2 d\mu(s)=o(\lambda^{-2})
\end{equation}
$$\text{as}\quad \lambda\to \infty.
$$
Therefore,
$$I+II+III=o(\lambda^{-7/4}) \quad \text{as} \quad  \lambda\to  \infty,
$$
which completes the justification of the representation  \eqref{justi}.

Next, combining \eqref{justi} and \eqref{meraa} we obtain
 \begin{align}
1-N(\lambda)&=  |\alpha +\beta|^2\int_\lambda^\infty \frac1\pi  \frac{\sqrt{2s}}{(\sin  \alpha -\cos\alpha)^2+2s\cos^2\alpha }\frac{ds}{s^2+1}+o(\lambda^{-7/4})
\nonumber
\\
&=|\alpha +\beta |^2\frac{\sqrt{2}}{3\pi\cos^2\alpha}\lambda^{-3/2} (1+o(1)) +o(\lambda^{-7/4})
 \nonumber
\\
&
=|\alpha +\beta |^2\frac{\sqrt{2}}{3\pi} ((\sqrt{2} \gamma-1)^2+1)\lambda^{-3/2}(1+o(1))+o(\lambda^{-7/4})
\quad \text{as}\quad \lambda\to \infty.\label{pochti}
\end{align}
Here we have  used the relation
$$
\frac{1}{\cos^2\alpha}=((\sqrt{2} \gamma-1)^2+1)
$$
that easily follows from \eqref{tutsi}. Recall that $\alpha\ne \frac\pi2$ and therefore $\cos \alpha\ne 0$.

Our next claim is that
\begin{equation}\label{pochtichti}
|\alpha +\beta |^2=\frac{1}{\sqrt{2}}\frac{1}{(\gamma -\frac{\sqrt{2}}{2})^2+\frac12}|\phi'(0)+\gamma\phi(0)|^2.
\end{equation}

From \eqref{vnrepr} it follows that
\begin{equation}\label{abel1}
\phi(0)=\alpha g_+(0)+\beta g_-(0)=\alpha g_+(0)+\beta \Theta \overline{g_+(0)}=(\alpha +\beta \Theta)2^{1/4}
\end{equation}
and
\begin{equation}\label{abel2}
\phi'(0)=\alpha g_+'(0)+\beta g_-'(0)=(\alpha\zeta  +\beta\overline{\zeta } \Theta)2^{1/4},
\end{equation}
where $\zeta$ is given by \eqref{delo}.

Rewriting \eqref{abel1} and \eqref{abel2} as
$$
\begin{pmatrix}1&\Theta\\
\zeta&\bar \zeta \Theta
\end{pmatrix}
\begin{pmatrix}\alpha\\\beta
\end{pmatrix}=2^{-1/4}
\begin{pmatrix}\phi(0)\\
\phi'(0)
\end{pmatrix},
$$
and solving this system of algebraic equations
one obtains
$$
\begin{pmatrix}\alpha\\\beta
\end{pmatrix}= \frac{1}{\Theta (\bar \zeta-\zeta)}
\begin{pmatrix}\bar \zeta \Theta&-\Theta\\
-\zeta&1
\end{pmatrix}
2^{-1/4}
\begin{pmatrix}\phi(0)\\
\phi'(0)
\end{pmatrix}.
$$
Therefore,
\begin{align*}
\alpha +\beta & =\frac{1}{2^{1/4}\Theta (\bar \zeta-\zeta)}[(\bar \zeta \Theta-\zeta)\phi(0)+(1-\Theta)\phi'(0)]
\\&=\frac{1-\Theta}{2^{1/4}\Theta (\bar \zeta-\zeta)}\left [\frac{\bar \zeta \Theta-\zeta}{1-\Theta}\phi(0)+\phi'(0)\right ].
\end{align*}
From \eqref{eqfort} it follows that
$$
\gamma=\frac{\bar \zeta \Theta-\zeta}{1-\Theta},
$$
so that
$$
\alpha +\beta=\frac{1-\Theta}{2^{1/4}\Theta (\bar \zeta-\zeta)}[\gamma\phi(0)+\phi(0)].
$$
One also observes that
$$
\frac{1-\Theta}{\Theta(\bar \zeta-\zeta)}=\frac{1}{ \zeta+\gamma},$$
which yields
\begin{align*}
|\alpha +\beta |^2&=\frac{1}{\sqrt{2}}\frac{1}{|\zeta +\gamma |^2}|\phi'(0)+\gamma\phi(0)|^2
\\&=\frac{1}{\sqrt{2}}\frac{1}{(\gamma -\frac{\sqrt{2}}{2})^2+\frac12}|\phi'(0)+\gamma\phi(0)|^2,
\end{align*}
and the claim  \eqref{pochtichti} follows.

Combining \eqref{pochti} and \eqref{pochtichti}  and taking into account that  $\phi'(0)+\gamma\phi(0)\ne 0$ (by the hypothesis),
we
finally obtain the asymptotic representation
\begin{align}
1-N(\lambda)&=|\alpha +\beta |^2\frac{\sqrt{2}}{3\pi} ((\sqrt{2} \gamma-1)^2+1)\lambda^{-3/2}(1+o(1))\label{polustH}\\
&=\frac{1}{\sqrt{2}}
\cdot  \frac{\sqrt{2}}{3\pi} \frac{(\sqrt{2} \gamma-1)^2+1}{(\gamma -\frac{\sqrt{2}}{2})^2+\frac12}|\phi'(0)+\gamma\phi(0)|^2 \lambda^{-3/2}(1+o(1))
\nonumber \\
&=\frac{2}{3\pi} |\phi'(0)+\gamma\phi(0)|^2 \lambda^{-3/2} (1+o(1))\quad \text{as } \lambda \to \infty. \nonumber
\end{align}

Notice that for $\gamma <0$  the operator  $H'$ has a simple eigenvalue $\lambda_0=-\gamma^2$ and therefore $N(\lambda)=0$ whenever $\lambda <-\gamma^2$,   and   $N(\lambda)=0$
for all  $\lambda<0$ if $\gamma \ge 0$.
Therefore,
\begin{equation}\label{naosi}
\lim_{\lambda\to -\infty}|\lambda|^{3/2}N(\lambda)=0.
\end{equation}

 By Theorem  \ref{stthm} in Appendix H,   the distribution $N(\lambda)$ belongs to the domain  of normal attraction of
the one-sided $\frac32$-stable  distribution
with the characteristic function
\begin{equation}\label{lawlaw23}
f(t)=
\exp\left (-\sigma' |t|^{3/2}\left (1+i\, \frac{t}{|t|}\right )\right ),
\end{equation}
 where
 $$
\sigma'= \frac23 \sqrt{\frac2\pi}|\phi'(0)+\gamma \phi (0)|^2.
$$

 Indeed, the distribution $N(\lambda)$ satisfies the conditions \eqref{as+} and \eqref{as-}  of Theorem \ref{stthm} in Appendix H with
 $\alpha =\frac32
 $,
 $$c_1= \frac{2}{3\pi}|\phi'(0)+\gamma \phi (0)|^2 \quad \text{and }\quad c_2=0.
 $$
 By Theorem \ref{stthm}, $N(\lambda) $ belongs to the domain  of normal attraction of   a stable law with the characteristic function
$$f(t)=
\exp\left (-\sigma' |t|^{3/2}\left (1-i\, \beta \frac{t}{|t| }\omega\left (t, \frac32\right)\right )\right ),
$$
where the parameters $\sigma'$ and $\beta$ are given by
\begin{align*}
\sigma'&= (c_1+c_2)d\left ( \frac32\right )
=\frac{2}{\pi}|\phi'(0)+\gamma \phi (0)|^2\cdot d\left ( \frac32\right ),
\\
\beta&=\frac{c_1-c_2}{c_1+c_2}=1,
\end{align*}
with
$$
d\left (\frac32\right )=\Gamma (-1/2)\cos \frac{3\pi}4=\frac12\sqrt{2\pi},
$$
and
$$
\omega\left (t, \frac32\right)=\tan \left (\frac{3\pi}{4} \right )=-1.
$$
Now \eqref{lawlaw23} follows.

To complete the proof of the theorem it remains to apply the $3/2$-stable limit theorem
to see that
\begin{equation}\label{nuonomb'}
\lim_{n\to \infty}\left |\left (\exp \left (i n^{-2/3}t  H\right )\phi, \phi\right )\right |^{2n}=\exp\left (-2 \sigma' |t|^{3/2}\right ),
\end{equation}
which justifies \eqref{nuonohom} by a  change of variables.
\end{proof}

\begin{remark}
The right hand side of \eqref{nuonomb'} is the characteristic functions of the Holtsmark distribution \cite{HM}. The Holtsmark distribution is a special case of a symmetric stable distribution with the index of stability   $\alpha=3/2$ and skewness parameter $\beta=0$  (see Appendix \ref{gnedkol}, eqs. \eqref{stlaw}, \eqref{data}
 with $\alpha=3/2$ and  $\beta=\gamma=0$).

\end{remark}

{\bf Scholium.} The $1/2$- and $3/2$-central limit theorems, Theorems  \ref{1/2thm} and  \ref{3/2thm}, respectively, show  that the results of continuous monitoring of the quantum evolution of a smooth state $\phi$ are rather sensitive to the choice of a self-adjoint realization of the Hamiltonian, the Schr\"odinger operator \eqref{schr1} and\eqref{schr2}, respectively.

For instance,  for the   Schr\"odinger operator   $H$   with the Dirichlet boundary condition at the origin  we have
$$
\lim_{n\to \infty}|(e^{it/nH}\phi, \phi)|^{2\sqrt{n}}=e^{-2\sigma |t|^{1/2}},
$$
where
$$
\sigma=\sqrt{\frac2\pi}|\phi(0)|^2.
$$
Therefore,  if the probability density      $|\phi(0)|^2$  to find a quantum particle at the origin  does not vanish, then  the state $\phi$ is an anti-Zeno state. That is,
$$
\lim_{n\to \infty}|(e^{it/nH}\phi, \phi)|^{2n}=0. $$

In the meanwhile,  for the  Schr\"odinger operator $H'$
with the  mixed boundary condition
$$
f'(0)+\gamma f(0)=0,
$$
one obtains that
$$
\lim_{n\to \infty}|(e^{it/nH'}\phi, \phi)|^{2n\sqrt{n}}=e^{-2\sigma' |t|^{3/2}} ,
$$
where
$$
 \sigma'=\frac23 \sqrt{\frac2\pi}|\phi'(0)+\gamma \phi (0)|^2.
$$

Hence, $$
\lim_{n\to \infty}|(e^{it/nH'}\phi, \phi)|^{2n}=1 .
$$
In other words, any smooth  state $\phi$  is a Zeno  state under the continuous monitoring of the evolution $\phi\mapsto e ^{it/nH'}\phi$ where $H'$ is any self-adjoint realizations of the  second differentiation operator different form the Friedrichs extension $H$ of $\dot H$.

We summarize the observations above in a more formal way.

\begin{corollary}\label{frid}
{\it Suppose that $
\phi\in W_2^2((0, \infty))
$, $\|\phi\|=1$.

\begin{itemize}
\item[(i)]  Let $H$  be  the Schr\"odinger operator with the Dirichlet boundary condition at the origin.
Then  $\phi$ is a Zeno state under the continuous monitoring of the unitary evolution
$\phi\mapsto e^{itH}\phi$ if and only if $\phi(0)=0$. Otherwise, $\phi$ is an anti-Zeno state.

\item[(ii)]
If $H'$ is any other self-adjoint realization of the differential expression
$$\tau=-\frac{d^2}{dx^2}
$$ different from its Friedrichs extension, then $\phi $ is a Zeno state under the continuous monitoring of the unitary evolution
$\phi\mapsto e^{itH'}\phi$.
\end{itemize}
}
\end{corollary}

\begin{proof} (i). If $\phi(0)=0$, then $\phi\in \Dom (H)$ and therefore $\phi$ is a Zeno state  under the continuous monitoring of the unitary evolution
$\phi\mapsto e^{itH}\phi$. If $\phi(0)\ne 0$, by Theorem \ref{1/2thm} the distribution function of the spectral measure of the element $\phi$ belongs to the domain of attraction
of a $1/2$-stable law, and therefore $\phi $ is an anti-Zeno state by Lemma \ref{AZS}.

(ii).
Notice that  for any self-adjoint extension $H'$ different from the Friedrichs extension $H$
the domain of the quadratic form of $H'$ coincides with the Sobolev class $W_2^1((0,\infty))$. Since  $\phi \in W_2^2((0,\infty))\subset  W_2^1((0,\infty))$, we have that
  the distribution $N(\lambda)$ of the spectral measure $(\EE_{H'}(d\lambda)\phi, \phi)$
of the state $\phi$ has the first moment and hence $\phi$ is necessarily a Zeno state.

\end{proof}
\begin{remark} (i). In the case of the Schr\"odinger operator $H$ with the Dirichlet boundary condition at the origin, one can slightly  relax the smoothness requirement on the state $\phi$ that $\phi\in W_2^2((0, \infty))$:
If $\phi\in W_2^1((0, \infty))$ only and  $\phi(0)=0$, then the  state $\phi$ belongs to the domain of the quadratic form of the  Schr\"odinger operator $H$.  In this case,  $\phi$
 is also a Zeno state under the continuous monitoring of the unitary evolution
$\phi\mapsto e^{itH}\phi$  by Proposition \ref{zenoprop}.

 (ii). From Theorem \ref{3/2thm} it follows that  the spectral measure
 $$\nu_\phi (d\lambda)=(\EE_{H'}(d\lambda)\phi, \phi)$$
of the state $\phi$ has  moments of order $r$ for all  $r<\frac32$. In particular,  the state $\phi$ belongs to  the domain of the quadratic form of $H'$ and therefore,
$\phi$ is a Zeno state under the continuous monitoring of the unitary evolution
$\phi\mapsto e^{itH'}\phi$  by Proposition \ref{zenoprop}.

(iii). As far as the domain issues are concerned, we have the following inclusions
$$
\Dom (\dot H)\subset \Dom ( H)\subset\Dom ((\dot H)^*)=W_2^2((0, \infty))$$
and
$$
 \Dom ((\dot H)^*)\cap \Dom ( H^{1/2})= \Dom ( H)=\{f\in W_2^2((0, \infty))\,|\,  f(0+)=0\}
$$
for  the Friedrichs extension $H$  of $\dot H$, $H\ge 0$.
For any self-adjoint extension $H'$ different from the Friedrichs extension $H$ we have
$$
\Dom (\dot H)\subset \Dom ( H')\subset\Dom ((\dot H)^*) \subset \Dom ( |H'|^{1/2})=W_2^1((0, \infty))
$$
and therefore
$$
 \Dom ((\dot H)^*)\cap \Dom ( |H'|^{1/2})=\Dom ((\dot H)^*)=W_2^2((0, \infty)).
$$

Notice that  for the Friedrichs extension $H$ we have
$$
 \Dom ( H^{1/2})= \{f \in W_2^1((0, \infty))\,|\, f(0+)=0\}\ne  \Dom ( |H'|^{1/2})=W_2^1((0, \infty)),
$$
which  explains the peculiar ``phase transition" in the relative geometry of domains (the Sobolev spaces) when replacing the Friedrichs extension $H$  with  any other self-adjoint extension $H'$.
\end{remark}

\section{The Quantum Zeno Effect versus Exponential Decay   alternative  }\label{s17}

Throughout this section we assume that  $\bbY$  is a  metric graph  in one of the  Cases (i)-(iii) (see the classification in the beginning of Section 4).
Denote by
$(\dot D, \widehat D, D)$  the triple of  differentiation operators on $\bbY$  as introduced in Section \ref{s9}.

Recall that  in Case (i), the metric graph  has the form
$ \bbY=(-\infty, 0)\sqcup(0,\infty)$, in Case (ii),
$ \bbY=(0,\ell)$, in  Case (iii),
$\bbY =(-\infty, 0)\sqcup(0,\infty)\sqcup(0,\ell)$.
Also recall that the reference self-adjoint operator   $D$ is   the differentiation operator  on  the graph $\bbY$
 defined on
\begin{align*}
\Dom(D)&=\{f_\infty\in W_2^1(\bbY)\, | \, f_\infty(0+)=- f_\infty(0-)\},
\\
\Dom(D)&=\{f_\ell\in  W_2^1(\bbY)\, |\, f_\ell(0)=-f_\ell(\ell)\},
\end{align*}
$$
\Dom(  D)=\left \{ f_\infty\oplus f_\ell\in  W_2^1(\bbY)\,
\Bigg|
\,
 \begin{cases}
f_\infty(0+)&=k f_\infty(0-) + \sqrt{1-k^2}  f_\ell(\ell)
\\
f_\ell(0+)&=\sqrt{1-k^2} f_\infty(0-)-k f_\ell(\ell)\\
\end{cases}
\right \},
$$
in Cases (i)--(iii), respectively.
Here
$
0<k<1$ is the parameter from the boundary condition \eqref{dotdom0} (the quantum gate coefficient)  that determines the symmetric operator $\dot D$ in Case (iii).

More generally, see Theorem \ref{decr} below, we will also  deal with the triples $(\dot D, \widehat D, D_\Theta)$ where $D_\Theta$, $|\Theta|=1$ is the self-adjoint operator referred to in Theorem \ref{gengen}.

Our main concern is to
study small-time asymptotic behavior of the quantum survival probability
$$
p(t)=|(e^{itH}\phi,\phi)|^2 \quad \text{as } \quad t\to 0,
$$
where $H=D$, or, more generally,   $H=D_\Theta$, the magnetic Hamiltonian. We also assume that the   state $\phi$  belongs to the test space
$$
\cL=\Dom ((\dot D)^*).
$$
We obviously have the inclusion
$$
\cL\subsetneqq \bigoplus_{e\subset \bbY}W_2^1(e),
$$
where  the sum is taken over all edges $e$ of the graph $\bbY$.

The main goal of this section
is to show that
 the  survival probability under  continuous monitoring   of the quantum evolution
 $$\phi \mapsto e^{itH}\phi, \quad \phi \in \cL, $$
 on the metric graph
either  experiences an exponential decay  or, alternatively,  the quantum Zeno effect takes place. This
   justifies  the complementarity of  the Exponential Decay  and  the Quantum Zeno Effect  scenarios for hyperbolic systems first indicated in  \cite{KMPY}.

We start  our analysis with the observation that the normalized  deficiency elements   of the symmetric operator $\dot D$
are resonant states under continuous monitoring of the unitary evolution $\phi \mapsto e^{itH}\phi$ where
$
H=D.
$

\begin{lemma}\label{equidi}
Suppose that a metric graph $\bbY$ is in  one of the Cases $(i)-(iii)$ and $\dot D$ is the symmetric differentiation operator on
$\bbY$ with boundary conditions \eqref{dotos}, \eqref{dotint} and \eqref{dotdom0}, respectively.
 Let $
g_\pm \in \Ker ((\dot D)^*\mp iI)$,  $ \|g_\pm\|=1$,
be   normalized  deficiency elements $g_\pm$ of the symmetric operator $\dot D$.
 Then   $g_\pm$
 are equidistributed, that is, $g_\pm$  have  the same spectral measure \begin{equation}\label{defnu}
\nu(d\lambda)=(\EE_{H}(d\lambda)g_+, g_+)=(\EE_{H}(d \lambda)g_-, g_-),\end{equation}
where  the Hamiltonian $H$ is given by the differentiation operator $D$.

Moreover,
\begin{align}
\lim_{\lambda\to +\infty} \lambda\nu\left ((\lambda, \infty)\right) &=\lim_{\lambda\to +\infty} \lambda \nu \left ((-\infty, -\lambda)\right ) \label{kkkk}
\\&=\frac1\pi
\begin{cases}
1, & \text{in Case $(i)$}\\
\coth\frac\ell 2,& \text{in Case $(ii)$}\\
 \coth\frac{\ell+\ell'}{2},& \text{in Case $(iii)$}
 \end{cases}.\nonumber
\end{align}
Here, in Case $(iii)$,
$$\ell'=\log \frac1k
$$
and  $0<k<1$ is the quantum gate coefficient from  the boundary condition \eqref{dotdom0} that determines the symmetric operator $\dot D$ in Case $(iii)$.

In particular,  the deficiency elements $g_\pm$ are resonant states with respect to the continuous monitoring of the  unitary dynamics $g_\pm \mapsto e^{itH}g_\pm$.

In this case,
\begin{equation}\label{raspad}
\lim_{n\to \infty}|(e^{it/n H}g_\pm,g_\pm)|^{2n}=e^{-\tau |t|},
\end{equation}
where the decay constant $\tau$ is given by
$$
\tau =2\begin{cases}
1, & \text{in Case $(i)$}\\
\coth\frac\ell 2,& \text{in Case $(ii)$}\\
 \coth\frac{\ell+\ell'}{2},& \text{in Case $(iii)$}
 \end{cases}.
$$

 \end{lemma}
 \begin{proof}
 Let $M(z)$ be the Weyl-Titchmarsh function associated with the pair $(\dot D, D)$.
 By Corollary   \ref{quasimain*},
\begin{equation}\label{wtwt}M(z)=\int_\bbR \left (\frac{1}{\lambda-z}-\frac{\lambda}{1+\lambda^2}\right )d\mu(\lambda),
\end{equation}
where
 the measure $\mu(d\lambda)$ is given by
\begin{equation}\label{mera}\mu(d\lambda)=\frac1\pi
 \begin{cases}
  d\lambda, &\text{ in Case (i)}\\
\\ \frac{2\pi }{\ell}\coth\frac\ell2
 \sum_{k\in \bbZ} \delta_{\frac{(2 k+1)\pi}{\ell}} (d\lambda),& \text{ in Case (ii)}
\\  \coth\frac{\ell+\ell'}{2} P_{e^{-\ell'}}\left (\ell \lambda-\pi \right )d \lambda ,& \text{ in Case (iii)}
\end{cases}.
\end{equation}
Here, in Case (iii),
$$
 P_r(\varphi)=\frac{1-r^2}{1+r^2-2r\cos \varphi}
 $$
 denotes  the Poisson kernel.

Recall that by Lemma \ref{primeD} the operator  $\dot D$ is a prime symmetric operator.
  Therefore,  Theorem \ref{unitar} in Appendix C  ensures the existence of a unitary map $\cU$  from $L^2(\bbY)$ onto
the Hilbert space  $L^2(\bbR, d\mu)$, where  $\mu(d\lambda)$  is given by  \eqref{mera}, with the following properties:
\begin{itemize}
\item[(i)] $\cU D\cU^{-1}$ coincides with the  operator of multiplication  by  independent variable
and
\item[(ii)] the deficiency elements $g_\pm$ get mapped to simple fractions
$$
(\cU g_\pm)(\lambda)=\frac{\Theta_\pm}{\lambda\mp i}\quad \text{for some}\quad |\Theta_\pm|=1.
$$
\end{itemize}

In particular,   for any Borel set $ \delta\subset \bbR$ we have
$$
(\EE_{H}(\delta)g_+, g_+)=(\EE_{H}(\delta)g_-, g_-)=\int_\delta\frac{d\mu(s)}{s^2+1},
$$
which shows that $g_\pm$ are equidistributed and  hence the spectral measure $\nu(d\lambda)$ in \eqref{defnu} is well defined.

It follows that
$$
 \lambda\nu((\lambda, \infty))=\lambda\int_\lambda^\infty\frac{d\mu(s)}{s^2+1}.
$$

In Case (i),  in view of \eqref{mera}  we have the following asymptotic representation
$$
\lambda\nu((\lambda, \infty))=\lambda\int_\lambda^\infty\frac{1}{\pi}\frac{ds}{s^2+1}=\frac1\pi\left (1+o(1)\right)\quad \text{as}\quad  \lambda\to +\infty,
$$
which proves that the first limit in  \eqref{kkkk} exists and coincides with the right hand side of  \eqref{kkkk}.

In Case (ii), by  \eqref{mera},
\begin{align*}
\lambda\nu((\lambda, \infty))&=\lambda \frac{2}{\ell}\coth\frac\ell2\,\,  \sum_{\frac{(2 k+1)\pi}{\ell}\ge \lambda}  \frac{1}{\left (\frac{(2 k+1)\pi}{\ell}\right )^2+1}\\
&= \frac{2}{\ell}\coth\frac\ell2\,\, \left (\frac{\ell}{2\pi}\right )^{2}\lambda\int_{ \frac{\lambda \ell}{2\pi}}^\infty  \frac{dk}{k^2} \cdot\left (1+o(1)\right)
\\&=  \frac{2}{\ell}\coth\frac\ell2\cdot  \frac{\ell}{2\pi} \left (1+o(1)\right)
\\&={ \frac1\pi \coth\frac\ell2} \left (1+o(1)\right) \quad \text{as}\quad  \lambda\to +\infty,
\end{align*}
 proving the first equality \eqref{kkkk} in Case (ii).

Finally, in Case (iii), using \eqref{mera} we have
\begin{align*}
\lambda\nu((\lambda, \infty))&=  \frac1\pi \coth\frac{\ell+\ell'}{2}\lambda \int_\lambda^\infty P_{e^{-\ell'}}\left (\ell s-\pi \right )\frac{d s}{s^2+1}
\\&={\frac1\pi \coth\frac{\ell+\ell'}{2}}\left (1+o(1)\right) \quad \text{as}\quad  \lambda\to +\infty,
\end{align*}
 which shows    that the first limit in  \eqref{kkkk} exists and coincides with the right hand side of  \eqref{kkkk}.
Here we used that the Poisson kernel admits  the representation
$$
P_r(s)=\frac{1-r^2}{1+r^2-2r\cos s}
=1+G_r(s),
$$
where $G_r(s)$ is a bounded  $2\pi$-periodic function with zero mean over the period such that
$$\lim_{\lambda\to +\infty}\lambda \int_\lambda^\infty G_{e^{-\ell'}}(\ell s-\pi)\frac{d s}{s^2+1}=0.
$$
Notice that the  equality above  can be justified by  integration by parts.

In a completely similar way  one shows that in all  Cases (i)-(iii) the second
limit in  \eqref{kkkk} exists and coincides with the right hand side of  \eqref{kkkk}.

 \end{proof}

 \begin{remark}   In Case (i),  one can apply  the residue theorem to see that the survival probability  amplitude $(e^{it D}g_\pm , g_\pm)$ itself is  exponentially decaying as
  $$
 (e^{it D}g_\pm , g_\pm)=\frac1\pi\int_{-\infty}^\infty e^{i\lambda t}\frac{d\lambda}{\lambda^2+1}=e^{- |t|}
 ,$$
 which in particular implies    \eqref{raspad}.  In this case the result of continuous monitoring of the corresponding quantum system on the time interval  $[0,t]$ and  a ``one time  observation''  at the moment of time $t$ are identical. That is,
 $$
 (e^{it/n D}g_\pm , g_\pm)^n= (e^{it D}g_\pm , g_\pm)
 $$
 and therefore
  $$
 |(e^{it/n D}g_\pm , g_\pm)|^{2n}= |(e^{it D}g_\pm , g_\pm)|^2 \quad \text{for all }\quad t.
 $$
  In this exceptional $($resonant$)$ case the continuous monitoring can neither stop  nor modify the evolution.
 \end{remark}

To understand better fine decay properties of a particular state from the test space $\cL=\Dom (\dot D^*)$ we need a comprehensive information about the boundary functionals
associated with  the von Neumann  decomposition of the test space
 \begin{equation}\label{VNF}
\cL= \Dom ((\dot D)^*)=\Ker((\dot D)^*-iI)\dot +\Ker((\dot D)^*+iI)\dot +\Dom (\dot D).
 \end{equation}

\begin{lemma}\label{bloody}  Suppose that a metric graph $\bbY$ is in  one of the Cases $(i)$-$(iii)$ and $\dot D$ is the symmetric differentiation operator on
$\bbY$ with boundary conditions \eqref{dotos}, \eqref{dotint} and \eqref{dotdom0}, respectively.  Denote by   $g_\pm$ the  deficiency elements of the symmetric operator $\dot D$ referred to in Lemma \ref{defeff}.

Assume   that $\phi\in\cL=\Dom((\dot D)^*)$ and  let
\begin{equation}\label{bfunc}
\phi =\alpha g_++\beta g_-+f, \quad  \phi\in \Dom((\dot D)^*),
\end{equation}
be  the decomposition associated with  von Neumann's formula \eqref{VNF},
where  $\alpha, \beta\in C$ and $f\in \Dom (\dot D)$.

Then
\begin{equation}\label{buch}
\alpha +\beta=\frac{1}{\sqrt{2}}
\begin{cases}
\phi_\infty(0-)+\phi_\infty(0+), &\text{in Case $(i)$}\\
\sqrt{\tanh\frac{\ell}{2}}\left (\phi_\ell(0)+\phi_\ell(\ell)\right ),&\text{in Case $(ii)$}
\\
\sqrt{\tanh\frac{\ell+\ell'}{2}}
\left  ( \phi_\ell(\ell)-\frac{ \phi_\infty(0+)-k\phi_\infty(0-)}{\sqrt{1-k^2}}\right ), &\text{in Case $(iii)$}
\end{cases},
\end{equation}
where, in Case $(iii)$,
$$
\ell'=\log \frac1k
$$
and  $0<k<1$ is the quantum gate coefficient from  the boundary condition \eqref{dotdom0} that determines the symmetric operator $\dot D$ in Case $(iii)$.

\end{lemma}

\begin{proof} In Case (i), we have
$$
\phi_\infty(x)=\alpha \sqrt{2}e^{x}\chi_{(-\infty, 0)}(x)+\beta \sqrt{2}e^{-x}\chi_{( 0, \infty)}(x)+f(x)
$$
for some $f\in \Dom (\dot D)$.

Since $f(0-)=f(0+)=0$, we have
$$
\phi_\infty(0-)=\alpha \sqrt{2}\quad \text{and}\quad \phi_\infty(0+) =\beta\sqrt{2}.
$$
Therefore
$$
\alpha+\beta=\frac{\phi_\infty(0+)+\phi_\infty(0+)}{\sqrt{2}},
$$
proving \eqref{buch} in that case.

In Case (ii),
$$
\phi_\ell(x)=\alpha\sqrt{\frac{2}{e^{2\ell}-1}}e^x+\beta \sqrt{\frac{2}{e^{2\ell}-1}}e^{\ell-x}+f(x)
$$
and therefore
$$\alpha+e^{\ell}\beta= \sqrt{\frac{e^{2\ell}-1}{2}}\phi_\ell(0)
$$
and
$$
e^\ell \alpha+\beta= \sqrt{\frac{e^{2\ell}-1}{2}}\phi_\ell(\ell).
$$
Hence
$$
\alpha+\beta= \sqrt{\frac{e^{\ell}-1}{e^{\ell}+1}}\frac{\phi_\ell(0)+\phi_\ell(\ell)}{\sqrt{2}}=\sqrt{\tanh \frac\ell2}\cdot\frac{\phi_\ell(0)+\phi_\ell(\ell)}{\sqrt{2}},
$$
proving \eqref{buch} in Case (ii).

In Case (iii), the elements $\phi$ and $f$ from the von Neumann  decomposition \eqref{bfunc} are the two-component vector functions
$$
\phi =\begin{pmatrix}
\phi_\infty \\
\phi_\ell
\end{pmatrix}
\quad \text{and}\quad
f=
 \begin{pmatrix}
f_\infty \\
f_\ell
\end{pmatrix}.
$$
From \eqref{bfunc} it follows that
$$
\phi_\infty (x)=\alpha\xi \sqrt{1-k^2}e^x\chi_{(-\infty,0)}(x)-\beta  \xi  \sqrt{1-k^2}e^{\ell-x}\chi_{(0,\infty)}(x)+f_\infty (x), \quad x \in \bbR,
$$
and
$$
\phi_\ell(x)=\alpha\xi e^x+\beta \xi k e^{\ell-x}+f_\ell(x), \quad x\in [0, \ell),
$$
where the norming constant $\xi$ is given by
$$
\xi =\sqrt{\frac{2}{e^{2\ell}-k^2}}.
$$

In particular,
\begin{align*}
\phi_\infty (0-)&=\alpha\xi \sqrt{1-k^2} +f_\infty (0-),
\\
\phi_\infty(0+)&=-\beta\xi  \sqrt{1-k^2} e^\ell +f_\infty(0+),
\\
\phi_\ell(\ell)&= \alpha \xi e^\ell +k\beta \xi+f_\ell(\ell).
\end{align*}

Since $f\in \dom( \dot D)$, the boundary conditions
\begin{align*}
f_\infty(0+)&=kf_\infty (0-),
\\
f_\ell(0)&=\sqrt{1-k^2}f_\infty(0-),
\\
f_\ell(\ell)&=0
\end{align*}
hold  and hence
\begin{align*}
\phi_\infty(0-)&=\alpha \xi \sqrt{1-k^2} +\gamma,\\
\phi_\infty (0+)&=-\beta  \xi  \sqrt{1-k^2} e^\ell +k\gamma,
\\\phi_\ell(\ell)&= \alpha \xi e^\ell +k\beta \xi ,
\end{align*}
where we use the shorthand notation
$$
\gamma=f_\infty (0-).
$$

Combing the obtained equations we arrive at  the following system of equations
$$
\begin{pmatrix}\sqrt{1-k^2}&0 &1\\
0&-e^{\ell}\sqrt{1-k^2}&k\\
e^\ell &k&0
\end{pmatrix}
\begin{pmatrix}
\alpha \xi  \\\beta \xi  \\\gamma
\end{pmatrix}=
\begin{pmatrix}
x\\y\\z
\end{pmatrix},
$$
where
$$
\begin{pmatrix}
x\\y\\z
\end{pmatrix}=
\begin{pmatrix}
\phi_\infty (0-)\\\phi_\infty (0+)\\
\phi_\ell(\ell)
\end{pmatrix}.
$$

Taking into account that the inverse matrix of  the system
is of the form

$$
\frac{1}{(e^{2\ell}-k^2)\sqrt{1-k^2}}
\begin{pmatrix}-k^2&k&e^\ell\sqrt{1-k^2}\\ke^\ell&-e^\ell&-k\sqrt{1-k^2}\\
e^{2\ell}\sqrt{1-k^2}&-k\sqrt{1-k^2}&-e^\ell(1-k^2)
\end{pmatrix}
$$
one easily obtains that
$$
\alpha \xi  =\frac{-k^2x+ky+e^\ell\sqrt{1-k^2}z}{(e^{2\ell}-k^2)\sqrt{1-k^2}}
\quad \text{and}\quad
\beta \xi  =\frac{ke^\ell x -e^\ell y -k\sqrt{1-k^2}z}{(e^{2\ell}-k^2)\sqrt{1-k^2}}.
$$
Therefore,
\begin{align*}
\alpha+\beta&=\xi^{-1}\frac{(ke^\ell-k^2)x+(k-e^\ell)y+(e^\ell-k) \sqrt{1-k^2}z}{(e^{2\ell}-k^2)\sqrt{1-k^2}}
\\&=\xi^{-1}\frac{(ke^\ell-k^2)\phi_\infty(0-)+(k-e^\ell)\phi_\infty(0+)+(e^\ell-k) \sqrt{1-k^2}\phi_\ell(\ell)}{(e^{2\ell}-k^2)\sqrt{1-k^2}}
\\&=\xi^{-1}\frac{1}
{(e^{\ell}+k)}\left  ( \phi_\ell(\ell)-\frac{ \phi_\infty(0+)-k\phi_\infty(0-)}{\sqrt{1-k^2}}\right )
\\&=\frac{1}{\sqrt{2}}\sqrt{\tanh \frac{\ell+\ell'}{2}}\left  ( \phi_\ell(\ell)-\frac{ \phi_\infty(0+)-k\phi_\infty(0-)}{\sqrt{1-k^2}}\right )
,
\end{align*}
which completes the proof of \eqref{buch} in Case (iii).

\end{proof}

The main result of this section is  the following

\begin{theorem}[{\sc Exponential Decay-Quantum Zeno Effect alternative}]\label{decr}

$\quad $ Suppose that a metric graph $\bbY $ is in one of the Cases $(i)$-$(iii)$. Let $\dot D$ be the symmetric differentiation operator given by
\eqref{dotos}, \eqref{dotint} and \eqref{dotdom0}, respectively.
Assume, in addition, that $\phi\in \Dom ( (\dot D)^*)$, $\|\phi\|=1$.

Let $H=D_\Theta$, $|\Theta|=1$,  be the $($magnetic$)$ Hamiltonian  referred to in Theorem \ref{gengen}.

Then
\begin{equation}\label{dvadva}
\lim_{n\to \infty}|(e^{it/n H}\phi, \phi)|^{2n}=e^{- \tau (\Theta) |t|}, \quad t\in \bbR,
\end{equation}
 where the decay constant $\tau(\Theta)$ is given by
\begin{equation}\label{tautheta}
\tau(\Theta)=
\begin{cases}
|\Theta\phi_\infty(0-)+\phi_\infty(0+)|^2, & \text{in Case $(i)$}\\
|\Theta \phi_\ell(\ell)+\phi_\ell(0)|^2, & \text{in Case $(ii)$}\\
\left | \Theta \phi_\ell(\ell) - \frac{\phi_\infty(0+)-
k \phi_\infty(0-)}{\sqrt{1-k^2}}\right |^2, & \text{in Case $(iii)$}
\end{cases}.
\end{equation}
Here $0<k<1$ is the quantum gate coefficient from  the boundary condition \eqref{dotdom0} that determines the symmetric operator $\dot D$ in
Case $(iii)$.

In particular,  the state  $\phi\in \cL= \Dom ( (\dot D)^*)$ is a resonant state  under continuous monitoring of the quantum  unitary evolution $\phi \to e^{itH}\phi$ if and only if
$$\phi\notin \dom(H)=\Dom (D_\Theta).$$

Otherwise, the state $\phi$ is a Zeno state.

\end{theorem}

\begin{proof} {\it Part 1.} First, we prove the assertion in the particular case of $\Theta=1$,
where the Hamiltonian $H$  is given by the differentiation operator $D$, i.e.,
 $$H=D=D_\Theta|_{\Theta=1}.$$

Since $\phi \in \Dom((\dot D)^*)$,
by the von Neumann formula, the element $\phi$ admits a unique decomposition
\begin{equation}\label{modneu}
\phi =\alpha g_++\beta g_-+f,
\end{equation}
where   $\alpha, \beta \in \bbC$, $
f\in \dom (\dot D)
$.
Here we choose
the deficiency elements $g_\pm$ to be  given by \eqref{deftip1}, \eqref{deftip2}, and finally by  \eqref{hyddef1} and \eqref{hyddef2}
whenever the graph $\bbY$ is in Cases (i),  (ii),  and (iii), respectively.

 Without loss of generality, we may assume that
the operator $H=D$ is already realized in its model representation in the Hilbert space $L^2(\bbR; d\mu)$ as the  operator of multiplication by independent variable
with the measure $\mu(d\lambda)$  determined by \eqref{mera}.

Indeed, the Weyl-Titchmarsh function $M_{(\dot D,D)}(z)$ associated with the pair $(\dot D,D)$  and  given by \eqref{Mrepr} admits the representation \eqref{wtwt}
with the measure $\mu(d\lambda)$ from \eqref{mera}. By Lemma \ref{primeD}, the symmetric differentiation operator $\dot D$ is prime and therefore the Hamiltonian $H=D$ is unitarily equivalent
to  its model representation in the Hilbert space $L^2(\bbR; d\mu)$ by Theorem \ref{unitar} in Appendix C.
By Lemma \ref{razgon},
$$g_+-g_-\in\Dom(H)= \Dom(D).
$$
Therefore, one can also assume that the decomposition \eqref{modneu}
takes place in the  model Hilbert space $L^2(\bbR; d\mu)$, where
the deficiency elements $g_\pm$ are given by the  partial fractions
(see Remark \ref{snoska} in Appendix C)
$$
g_\pm = \frac{1}{\lambda\mp i},
\quad \lambda \in\bbR \,\,\,\mu-\text{a.e.},
$$ and
\begin{equation}\label{fgut}f\in L^2(\bbR;(1+\lambda^2) d\mu(\lambda)).
\end{equation}

The spectral measure  $(\EE_H(d\lambda)\phi,\phi)$ of the element $\phi $ can be evaluated as follows
$$
(\EE_H(\delta)\phi,\phi)=\int_\delta \left | \alpha \frac1{\lambda-i}+\beta \frac1{\lambda+i} +f(\lambda)\right |^2d\mu(\lambda),$$
with $\delta\subset \bbR$  a Borel set.

Therefore,
\begin{align}
(\EE_H(\delta)\phi,\phi)&=|\alpha +\beta |^2\int_\delta
\frac{d\mu(\lambda) }{\lambda^2+1}\label{tailss}
\\&+2\Re \, \alpha \overline \beta \int_\delta  \left (
\frac{1}{(\lambda -i)^2}-\frac{1}{\lambda^2+1}\right )d\mu(\lambda) \nonumber
\\&+2 \Re \int_\delta \left ( \alpha \frac1{\lambda-i}+\beta \frac1{\lambda+i}\right )\overline{f(\lambda)}d\mu(\lambda)  \nonumber
\\&+\int_\delta \left | f(\lambda)\right |^2d\mu(\lambda).\nonumber
\end{align}

It turns out that the first term in  \eqref{tailss} determines the leading term of the  asymptotics
 in the heavy-tailed distribution of the
spectral measure $(\EE_H(d\lambda)\phi,\phi)$ whenever
$$
\alpha+\beta\ne 0.
$$

Indeed,  we have
$$
I=:\int_{|s|>\lambda} \left |\frac{1}{(s -i)^2}-\frac{1}{s^2+1}\right |d\mu(s)\le \frac{2}\lambda \int_{|s|>\lambda}  \frac{d\mu(s) }{s^2+1}.
$$
By Lemma \ref{equidi},
the following limit exists,
$$
\lim_{\lambda \to +\infty} \lambda \int_{|s|>\lambda}  \frac{d\mu(s) }{s^2+1}<\infty
$$
and therefore$$I=O(\lambda^{-2}) \quad \text{as }\quad \lambda \to+\infty.
$$Next,
\begin{align*}
II&=:\left |\int_{|s|>\lambda} \frac{\overline{f(s)}}{s\pm i}d\mu(s)\right |
\\&\le \frac1\lambda \sqrt{\int_{|s|>\lambda}  \frac{d\mu(s) }{s^2+1}}\cdot \sqrt{\int_{|s|>\lambda}^\infty (1+s^2)|f(s)|^2 d\mu(s)}\\
&=o(\lambda^{-3/2}) \quad \text{as }\quad \lambda \to+\infty,
\end{align*}
where we have used \eqref{fgut}.

Finally,
\begin{align*}
III&=:\int_{|s|>\lambda}\left | f(s)\right |^2d\mu(s)
\\&\le \frac1{\lambda^2}\int_{|s|>\lambda} (1+s^2)|f(s)|^2 d\mu(s)= o(\lambda^{-2})\quad\text{as}\quad \lambda \to\infty.
\end{align*}
Therefore,
$$I+II+III=o(\lambda^{-3/2}) \quad \text{as }\quad \lambda \to+\infty
$$
and from \eqref{tailss}
we obtain
\begin{align*}
\lambda(\EE_H((\lambda, \infty))\phi,\phi)&= |\alpha+\beta|^2 \lambda \int_\lambda^\infty \frac{d\mu(s)}{s^2+1}+o(\lambda^{-1/2})
\quad \text{as } \quad \lambda \to \infty.
\end{align*}
In a similar way one proves that
$$\lambda(\EE_H((-\infty, -\lambda))\phi,\phi)= |\alpha+\beta|^2 \lambda \int_{-\infty}^{-\lambda} \frac{d\mu(s)}{s^2+1}+o(\lambda^{-1/2})
\quad \text{as } \quad \lambda \to \infty.
$$
By Lemma \ref{equidi},
\begin{align*}
\lim_{\lambda\to+ \infty}\lambda \int_\lambda^\infty \frac{d\mu(s)}{s^2+1}&=\lim_{\lambda\to +\infty}\lambda \int_{-\infty}^{-\lambda} \frac{d\mu(s)}{s^2+1}
&=\begin{cases}
1, & \text{in Case (i)}\\
\coth\frac\ell 2,& \text{in Case (ii)}\\
 \coth\frac{\ell+\ell'}{2},& \text{in Case (iii)}
 \end{cases}.
\end{align*}

Therefore,
\begin{align*}
 \lim_{\lambda\to + \infty}\lambda(\EE_H((\lambda, \infty))\phi,\phi)&=\lim_{\lambda\to +\infty}\lambda(\EE_H((-\infty, -\lambda))\phi,\phi)
 \\&=\frac1\pi |\alpha+\beta|^2
 \begin{cases}
1, & \text{in Case (i)}\\
\coth\frac\ell 2,& \text{in Case (ii)}\\
 \coth\frac{\ell+\ell'}{2},& \text{in Case (iii)}
 \end{cases}.
 \end{align*}

 On the other hand, from  Lemma \ref{bloody} it follows that
 \begin{equation}\label{buch1}
|\alpha +\beta|^2=\frac{1}{2}
\begin{cases}
\left |\phi_\infty(0-)+\phi_\infty(0+)\right |^2, &\text{in Case (i)}\\
\tanh\frac{\ell}{2}\left |\phi_\ell(0)+\phi_\ell(\ell)\right |^2,&\text{in Case (ii)}
\\
\tanh\frac{\ell+\ell'}{2}
\left  | \phi_\ell(\ell)-\frac{ \phi_\infty(0+)-k\psi_\infty(0-)}{\sqrt{1-k^2}}\right |^2, &\text{in Case (iii)}
\end{cases}.
\end{equation}
Here, in Case (iii),
$$\ell'=\log \frac1k.
$$

Hence
\begin{align*}
 \lim_{\lambda\to + \infty}\lambda(\EE_H((\lambda, \infty))\phi,\phi)&=\lim_{\lambda\to + \infty}\lambda(\EE_H((-\infty, -\lambda))\phi,\phi)
 \\&=\frac1{2\pi}
 \begin{cases}
\left |\phi_\infty(0-)+\phi_\infty(0+)\right |^2, & \text{in Case (i)}\\
\left |\phi_\ell(0)+\phi_\ell(\ell)\right |^2,& \text{in Case (ii)}\\
 \left  | \phi_\ell(\ell)-\frac{ \phi_\infty(0+)-k\phi_\infty(0-)}{\sqrt{1-k^2}}\right |^2,& \text{in Case (iii)}
 \end{cases}.
 \end{align*}
To complete the proof of \eqref{tautheta}  in case $\Theta=1$ it remains to apply
 Theorem \ref{levy}.

{\it Part 2}.
Now we can treat  the general case of  an arbitrary $\Theta$, $|\Theta|=1$.

 Let $H=D_\Theta$ be the magnetic Hamiltonian in Case (i).
Denote by $U_\Theta$ the unitary operator in $L^2(\bbY)=L^2(\bbR)$ defined as
$$
(U_\Theta f)(x)=\overline{\Theta}\chi_{(-\infty, 0)}(x)f(x)+ \chi_{(0,\infty)}(x)f(x).
$$
One verifies that
$$
D_\Theta=U_\Theta D U^*_\Theta
$$
and therefore
$$\lim_{n\to \infty}|(e^{it/n H}\phi, \phi)|^{2n}=\lim_{n\to \infty}|(e^{it/n D_\Theta}\phi, \phi)|^{2n}=
\lim_{n\to \infty}|(e^{it/n D}U_\Theta^*\phi,U_\Theta^* \phi)|^{2n}.
$$

By Part 1 of the proof,
\begin{align*}
\lim_{n\to \infty}|(e^{it/n D}U_\Theta^*\phi,U_\Theta^* \phi)|^{2n}&=
\exp \left (-|(U_\Theta^*\phi_\infty)(0+)+(U_\Theta^*\phi_\infty)(0-)|^2|t|\right )
\\&=\exp \left (-|\phi_\infty(0+)+\Theta \phi_\infty(0-)|^2|t|\right ).
\end{align*}
Therefore,
$$
\tau(\Theta)=|\phi_\infty(0+)+\Theta \phi_\infty(0-)|^2,
$$
which proves \eqref{tautheta} in Case (i).

In Cases (ii) and (iii),
we have the commutation relation (cf. \eqref{ccrsym1})
$$ D_\Theta=U_\Theta\left ( D+\frac{\arg \Theta}{\ell}I\right )U_\Theta^*,
$$
where   $U_\Theta$  the unitary multiplication operator   in $L^2(\bbY)$ given by
$$
(U_\Theta\phi(x))=e^{-i\frac{\arg \Theta}{\ell} x}\phi(x), \quad x\in \bbY.
$$
Therefore,
\begin{align*}
\lim_{n\to \infty}|(e^{it/n H}\phi, \phi)|^{2n}&=
\lim_{n\to \infty} | e^{it\frac{\arg \Theta}{n\ell}}(e^{it/n D}U^*\phi,U^* \phi)|^{2n}
\\&=\lim_{n\to \infty}|(e^{it/n D}U^*\phi,U^* \phi)|^{2n}.
\end{align*}
Now the claim \eqref{tautheta} follows by applying the result of Part 1 of the proof  to the state $U^*\phi$.

To complete the proof it remains to show that
under the hypothesis that $\phi\in \Dom ((\dot D)^*)$ the decay constant $\tau(\Theta)$ vanishes if and only if $\phi\in \Dom (D_\Theta)$.
Indeed, if $\phi \in \Dom (D_\Theta)\subset \Dom (|D_\Theta|^{1/2}) $, then $\phi$ is a Zeno state (by Proposition \ref{zenoprop}) and hence $\tau(\Theta)=0$. One can also see right away that   the boundary conditions \eqref{bcI}, \eqref{bcII} and \eqref{bcIII}
imply $\tau(\Theta)=0$.

The converse (under the hypothesis that $\phi\in \Dom (\dot D)^*)$) is also true. It is obvious in Cases (i) and (ii). In Case (iii) the equality $\tau(\Theta)=0$ implies
 \begin{equation}\label{bc111}
 \Theta \phi_\ell(\ell) = \frac{\phi_\infty(0+)-
k \phi_\infty(0-)}{\sqrt{1-k^2}}
 \end{equation}
 and since $\phi \in \Dom ((\dot D)^*)$, from \eqref{domsop} it also follows that
 \begin{equation}\label{bc222}
 \phi_\infty(0-)-k\, \phi_\infty(0+)
-\sqrt{1-k^2}\, \phi_\ell(0)=0.
  \end{equation}
 Multiplying \eqref{bc111} by $k$ and using \eqref{bc222} we obtain
   \begin{align*}
k\Theta \phi_\ell(\ell)&= \frac{k \phi_\infty(0+)-
k^2 \phi_\infty(0-)}{\sqrt{1-k^2}}
\\&=\frac{ \phi_\infty(0-)
-\sqrt{1-k^2}\, \phi_\ell(0+)-
k^2 \phi_\infty(0-)}{\sqrt{1-k^2}}
\\&
=\sqrt{1-k^2}\phi_\infty(0-)-\phi_\ell(0),
\end{align*}
 which shows that
  \begin{equation}\label{bc1111}
\phi_\ell(0)= \sqrt{1-k^2}\phi_\infty(0-) -k\Theta \phi_\ell(\ell).
\end{equation}
From \eqref{bc111} and \eqref{bc1111} we get
\begin{align*}
\begin{pmatrix}
\phi_\infty(0+)\\
\phi_\ell(0)
\end{pmatrix}&=\begin{pmatrix}
k& \sqrt{1-k^2} \Theta\\
 \sqrt{1-k^2}&-k\Theta
\end{pmatrix}
 \begin{pmatrix}
\phi_\infty(0-)\\
\phi_\ell(\ell)
\end{pmatrix}
\end{align*}
and hence \eqref{bcIII} holds proving that $\phi \in \Dom (D_\Theta)$.
\end{proof}

Let $\widehat D$ be   the maximal dissipative differential operator
  defined by \eqref{domdomi}  (with $k=0$) whenever the graph $\bbY$ is in Case (i) and
  by  \eqref{domdomii} and \eqref{domdomiii} whenever the graph $\bbY$ is in Cases  (ii) and (iii), respectively.
  Assume,  in addition, that the initial state $\phi$ is such that
  $ \phi\in \Dom (\widehat D)\cup \Dom (\widehat D^*)$.

The following lemma   shows that  under these assumptions  the decay rate of the state $\phi$ under continuous monitoring of the unitary evolution
\begin{equation}\label{elkaa}\phi \mapsto e^{itH }\phi
\end{equation}
 is determined by the state only  and  is, in fact,   independent of  the  self-adjoint realization   $H=D_\Theta$ of the symmetric differentiation operator $\dot D$.

 \begin{lemma}\label{palka}Assume the hypothesis of Theorem \ref{decr}.
 Let $\widehat D$ be   the maximal dissipative differential operator
  defined by \eqref{domdomi}  $($with $k=0$$)$ whenever the graph $\bbY$ is in Case $(i)$ and
  by  \eqref{domdomii} and \eqref{domdomiii} whenever the graph $\bbY$ is in Cases  $(ii)$ and $(iii)$, respectively.

 Then the decay constant
$ \tau=\tau(\Theta)$ given by \eqref{tautheta} does not depend on  $\Theta$ if and only if
$$
\phi\in \Dom (\widehat D)\cup \Dom ((\widehat D)^*).
$$

In this case,
\begin{equation}\label{tautheta1}
\tau=
\begin{cases}
|\phi_\infty(0-)|^2, & \text{in Case }(i)\\
|\phi_\ell(\ell)|^2, & \text{in Case }(ii)\\
\left | \phi_\ell(\ell) \right |^2, & \text{in Case }(iii)
\end{cases},\quad \text{whenever } \quad \phi\in\Dom(\widehat D),
\end{equation}
and
\begin{equation}\label{tautheta2}
\tau=
\begin{cases}
|\phi_\infty(0+)|^2, & \text{in Case }(i)\\
|\phi_\ell(0)|^2, & \text{in Case }(ii)\\
  \frac{|\phi_\infty(0+)-
k \phi_\infty(0-)|^2}{1-k^2}, & \text{in Case  }(iii)
\end{cases}, \quad \text{whenever } \quad \phi\in\Dom((\widehat D)^*).
\end{equation}
Here, in Case $(iii)$, $k$ is the quantum gate coefficient.

 \end{lemma}
\begin{proof}
It is easily seen from  \eqref{tautheta} that  the decay constant  $\tau(\Theta)$  does not depend on  $\Theta$ if and only if
 either
\begin{equation}\label{nez1}
\begin{cases}
\phi_\infty(0+)=0, &  \text{in Case (i)}
\\
  \phi_\ell(0)=0, &\text{in Case (ii)}
\\
\phi_\infty(0+)=
k \phi_\infty(0-), & \text{in Case (iii)}
\end{cases},
\end{equation}
or
\begin{equation}\label{nez2}
\begin{cases}
 \phi_\infty(0-)= 0 , &  \text{in Case (i)}
\\
\phi_\ell(\ell)=0   ,  &\text{in Case (ii)}
\\
\phi_\ell(\ell)=0,& \text{in Case (iii)},
\end{cases}
\end{equation} or both.

Recall that  the boundary conditions   \eqref{domdomi}  (with $k=0$),  \eqref{domdomii}, and  \eqref{domdomiii}
for the dissipative differentiation operator $\widehat D$
yield
\begin{align}
\phi_\infty(0+)&=0 \quad \text{(in Case (i))}\label{chtoxa}
\\
\phi_\ell(0)&=0 \quad \text{(in Case (ii))}\nonumber
\end{align}
and
\begin{equation}\label{chtoto*}
\begin{cases}
\phi_\infty(0+)=k \phi_\infty(0-)
\\
\phi_\ell(0)=\sqrt{1-k^2}\phi_\infty(0-)
\end{cases} \quad \text{(in Case (iii))}.
\end{equation}
Notice that in Case (iii) the first condition in \eqref{chtoto*} implies the second one whenever $\phi\in \Dom ((\dot D)^*) $. Indeed, the membership
 $\phi\in \Dom ((\dot D)^*) $ means that
$$
\phi_\infty(0-)-k\, \phi_\infty(0+)
-\sqrt{1-k^2}\, \phi _\ell(0)=0
$$
by Lemma \ref{domsopr} and the claim follows by a simple computation.

Now, it is straightforward to see that under the hypothesis that  $\phi\in\cL= \Dom ((\dot D)^*) $ the boundary conditions \eqref{nez1}   hold if and only if $\phi\in \Dom (\widehat D)$.
 In this case \eqref{tautheta1}  follows from \eqref{tautheta} in Theorem \ref{decr}.

Next, the boundary conditions for the adjoint operator $(\widehat D)^*$
 \eqref{dotdom*1} (with $k=0$), \eqref{dotdom*2}, and  \eqref{dotdom*}
 can be rewritten as
 \begin{align}
\phi_\infty(0-)&=0 \quad \text{(in Case (i))}\label{chtoxaxa}
\\
\phi_\ell(\ell)&=0 \quad \text{(in Case (ii))}\nonumber
\end{align}
and
\begin{equation}\label{chtototo}
\begin{cases}
\phi_\infty (0-)&=k \phi_\infty (0+)+\sqrt{1-k^2} \phi_\ell(0)
\\
\phi_\ell(\ell)&=0
\end{cases} \quad \text{(in Case (iii))}.
\end{equation}
   By Lemma  \ref{domsopr}, the first condition in \eqref{chtototo} simply means that $\phi\in \Dom ((\dot D)^*) $. Therefore, the boundary conditions \eqref{nez2} hold if and only if
and $\phi\in \Dom (\widehat D^*)$,
In this case \eqref{tautheta2}  follows from \eqref{tautheta} in Theorem \ref{decr}.

\end{proof}

If  the initial state
 $\phi$ is taken from
the somewhat narrower test space
$$
\cM=\Dom (\widehat D)\cup\Dom ((\widehat D))^*\subset \Dom ((\dot D)^*)=\cL,
$$
then Lemma \ref{palka} states  that  continuous monitoring of
the unitary evolution \eqref{elkaa}  is universal  (in the sense that the corresponding decay rate $
\tau$ referred to in Lemma \ref{palka} is independent of the choice
of the magnetic Hamiltonian $H$).

In this case, i.e. if $\phi \in \cM$,  the universal exponent  $\tau$ can also be recognized as the decay rate  associated with continuous monitoring of the unitary evolution
\begin{equation}\label{dub}\widehat \phi \mapsto e^{it\bbH }\widehat \phi
\end{equation}
in an extended Hilbert $\fH$  containing $L^2(\bbY)$ as a proper subspace. Here  $\fH=L^2(\bbX)$, where $\bbX$ is the full metric graph:
$\bbX=\bbY\sqcup
\bbY$ if the metric graph $\bbY$ in Cases (i) and (iii), and $\bbX$ can be identified with $\bbR$ if $\bbY=(0,\ell)$ is in Case (ii),
 the Hamiltonian $\bbH$ is a self-adjoint dilation  of the dissipative differentiation operator $\widehat D$,
and the new state $\widehat \phi \in L^2(\bbX)$ of the extended quantum system
 is  identified with  the initial state  $\phi$ being naturally  imbedded to the space
 $\fH=L^2(\bbX)$.

The precise statement is as follows.

\begin{corollary}{\it
 Let   $\bbH$ be
 a self-adjoint dilation  in the Hilbert space $\fH=L^2(\bbX)$  of the dissipative differentiation operator $\widehat D$.
 Assume that
 $$
\phi\in \cM=\Dom (\widehat D)\cup \Dom (\widehat D^*), \quad \|\phi\|=1.
$$
Denote by $\widehat \phi$ a state in $\fH$ such that
$$
\phi = P_ {L^2(\bbY)}\widehat \phi\quad \text{and}\quad (I-P_ {L^2(\bbY)}) \widehat \phi=0,
$$
where  $P_ {L^2(\bbY)}$ is the orthogonal projection  from
the Hilbert space $L^2(\bbX)$ onto its subspace $L^2(\bbY)$.

 Then
\begin{equation}\label{univer}
\lim_{n\to \infty}|(e^{it/n \bbH}{\widehat  \phi}, {\widehat \phi})|^{2n}
=e^{-\tau |t|},
\quad t\in \bbR, \end{equation}
where the decay constant
$\tau $
 is given by
\eqref{tautheta1} if  $\phi\in \Dom (\widehat D)$
  and by \eqref{tautheta2} if $\phi\in \Dom( (\widehat D)^*)
$, respectively.
}

\end{corollary}

 \begin{proof}

 Suppose first that $\phi \in \Dom (\widehat D)$.

 Denote by $\cD$ the differentiation operator $i\frac{d}{dx}$ on $$\Dom(\cD)= \bigoplus_{e\subset \bbY}W_2^1(e),$$
where  the sum is taken over all edges $e$ of the graph $\bbY$.

Integration by parts for $\phi \in \Dom (\cD)$ yields
\begin{equation}\label{lagrang}
\Im(\cD\phi, \phi)=\frac12
\begin{cases}
|\phi_\infty(0-)|^2-|\phi_\infty(0+)|^2 \\
|\phi_\ell(\ell)|^2-|\phi_\ell(0)|^2 \\
|\phi_\infty(0-)|^2-|\phi_\infty(0+)|^2+|\phi_\ell(\ell)|^2-|\phi_\ell(0)|^2
\end{cases}
\end{equation}
in Cases (i), (i) and (iii), respectively.
Since $\phi \in  \Dom (\widehat D)$, taking into account  the boundary conditions \eqref{chtoxa} and \eqref{chtoto*}, from \eqref{lagrang} we obtain
$$
\Im(\widehat D\phi, \phi)= \Im( \cD\phi, \phi)=\frac12
\begin{cases}
|\phi_\infty(0-)|^2,& \text{in Case (i)}\\
|\phi_\ell(\ell)|^2,& \text{in Case (ii)}\\
|\phi_\ell(\ell)|^2,& \text{in Case (iii)}
\end{cases}.$$
Therefore,
$$2\Im(\widehat D\phi, \phi)=
\tau,
$$
where $\tau $ is given by \eqref{tautheta1}.

By Lemma \ref{dissss},
$$
\lim_{n\to \infty}|(e^{it \widehat D}\phi, \phi)|^{2n}
=e^{-2 \Im (\widehat D\phi, \phi) t}, \quad t\ge 0, \quad \phi\in\Dom(\widehat D),
$$
and therefore
\begin{equation}\label{limd}
\lim_{n\to \infty}|(e^{it/n \widehat D}\phi, \phi)|^{2n}
=e^{-\tau t}, \quad t\ge 0, \quad \phi\in\Dom(\widehat D),
 \end{equation}
 where  $\tau$ given by \eqref{tautheta1}.

Since $\bbH$ dilates $\widehat D$, we have
\begin{equation}\label{10}
(e^{it/n \bbH}\widehat \phi, \widehat \phi)=(e^{it \widehat D}\phi, \phi), \quad t\ge 0,
\end{equation}
and therefore
\begin{equation}\label{11}
\lim_{n\to \infty}|(e^{it/n \bbH}\widehat \phi, \widehat \phi)|^{2n}=\lim_{n\to \infty}|(e^{it/n \widehat D}\phi, \phi)|^{2n}=e^{-\tau t}, \quad t\ge 0.
\end{equation}
To complete the proof of  \eqref{univer} for $\phi \in \Dom(\widehat D)$  it remains  to observe that the return probability $p(t)=|(e^{it/n \bbH}\widehat \phi, \widehat \phi)|^2$
is an even function in $t$.

Next, suppose that   $
 \phi\in \Dom ((\widehat D)^*).
$

As above, by Lemma \ref{dissss},
\begin{equation}\label{terra}
\lim_{n\to \infty}|(e^{-it/n( \widehat D)^*}\phi, \phi)|^{2n}=e^{2\Im ((\widehat D)^*\phi, \phi)t}, \quad t\ge 0.
\end{equation}

Since $\phi \in\dom((\widehat D)^*)$, one can use boundary conditions  \eqref{chtoxaxa},  \eqref{chtototo} and the equality \eqref{lagrang}
to obtain that
\begin{align*}
 \Im ((-\widehat D)^*\phi, \phi)&=
 \Im ((- \cD)\phi, \phi)
 \\&=\frac12
 \begin{cases}|\phi_\infty(0+)|^2,& \text{in Case (i)}\\
|\phi_\ell(0)|^2,& \text{in Case (ii)  }\\
|\phi_\infty(0-)|^2-|\phi_\infty(0+)|^2-|\phi_\ell(0)|^2
,& \text{in Case (iii)  }
\end{cases}.
\end{align*}
(In Case (iii), we took into account  the second condition in \eqref{chtototo} that  $\phi_\ell(\ell)=0$).

Now, using the first condition in \eqref{chtototo}, one computes
$$
|\phi_\infty(0-)|^2-|\phi_\infty(0+)|^2-|\phi_\ell(0)|^2
=\left |\frac{\phi_\infty(0+)-k\phi_\infty(0-)}{\sqrt{1-k^2}}\right |^2,
$$
which shows that
\begin{align*}
2 \Im ((\widehat D)^*\phi, \phi)=-\begin{cases}|\phi_\infty(0+)|^2,& \text{in Case (i)}\\
|\phi_\ell(0)|^2,& \text{in Case (ii)  }\\
\left |\frac{\phi_\infty(0+)-k\phi_\infty(0-)}{\sqrt{1-k^2}}\right |^2
,& \text{in Case (iii)  }
\end{cases}
=\tau,
\end{align*}
where $\tau$ is given by \eqref{tautheta2}.

Again, by Lemma \ref{dissss},
$$
\lim_{n\to \infty}|(e^{-it (\widehat D)^*}\phi, \phi)|^{2n}
=e^{2 \Im ((\widehat D)^*\phi, \phi) t}, \quad t\ge 0, \quad \phi\in\Dom((\widehat D)^*),
$$
 and therefore
  \begin{equation}\label{limd*}
 \lim_{n\to \infty}|(e^{-it/n (\widehat D)^*}\phi, \phi)|^{2n}=
e^{-\tau t}, \quad t\ge 0,  \quad \phi\in\Dom((\widehat D)^*),
\end{equation}
 where  $\tau$ given by \eqref{tautheta2}.

Since $\bbH$ dilates $\widehat D$, we obtain
\begin{align*}
 \lim_{n\to \infty}|(e^{-it/n (\widehat D)^*}\phi, \phi)|^{2n}&= \lim_{n\to \infty}|(\phi, (e^{-it/n (\widehat D)^*})^*\phi)|^{2n}
 = \lim_{n\to \infty}|( e^{it/n \widehat D}\phi,\phi)|^{2n}
 \\&=\lim_{n\to \infty}|(e^{it/n \bbH}\phi, \phi)|^{2n}=e^{-\tau t}, \quad t>0,
\end{align*}
which  proves \eqref{univer}  for $\phi\in \Dom((\widehat D)^*)$
 with $\tau$  given by \eqref{tautheta2} and then, by symmetry,  for all $t\in \bbR$.

\end{proof}

\begin{remark}\label{chtoto}
We remark  that
if $\phi \in \Dom (\widehat D)\cap  \Dom ((\widehat D)^*)=\dom (\dot D)$ and therefore $\phi\in \Dom (D_\Theta)$ for all $|\Theta|=1$, then  $\phi$ is a Zeno state under continuous monitoring of the  dynamics
$\phi \mapsto e^{itH}\phi$ for any self-adjoint realizations $H=D_\Theta $ of the differentiation operator.  Therefore, in this case the corresponding decay constant  $\tau=\tau(\Theta)=0$ is $\Theta$-independent for  an  obvious reason.

Also notice that  under the requirement that $
\phi\in\cM= \Dom (\widehat D)\cup \Dom ((\widehat D)^*)
$
the boundary data that determine the decay constant \eqref{tautheta1} and \eqref{tautheta2}
can also be evaluated as follows.

Assume, for instance, that the metric graph $\bbY$ is in Case (iii) with its  main vertex  at the origin ($\mu=0$).
Denote by   $\bbX$   the full metric  graph containing $\bbY$ as its subgraph and let $\bbH$ be
the self-adjoint dilation   in the extended Hilbert space $L^2(\bbX)$ of the dissipative operator
$\widehat D$  in $L^2( \bbY)$.

 Suppose that  a two-component vector-function $ \Psi=\begin{pmatrix}
\phi_{\uparrow}
\\
\phi_{\downarrow}
\end{pmatrix}\in L^2(\bbX)
$
is a continuation  of the function $\phi$ from the graph $\bbY$ onto the full graph $\bbX$,
\begin{equation}\label{exten}
\Psi(x)=\phi(x), \quad x\in \bbY,
\end{equation}
such that  $\Psi \in \Dom(\bbH) $.

Then  the decay constant  \eqref{tautheta1} can be evaluated via the second component of the vector-function  $\Psi$ as
$$
\tau=|\phi_{\downarrow}(\ell)|^2.$$

Indeed, since
$\Psi\in \Dom (\bbH)$, the two-component vector-function $\Psi(x)$
is continuous, so is  its second component  $\phi_{\downarrow}(x)$.
 In particular,
 $
 \Psi(\ell)=\phi(\ell),
 $
 so that
\begin{equation}\label{vl1}
\phi_\ell(\ell)=\phi_{\downarrow}(\ell)
\end{equation}
and hence
$$
\tau=|\phi_\ell(\ell)|^2=|\phi_{\downarrow}(\ell)|^2.$$

   Moreover,  for  the decay constant  \eqref{tautheta2}  we have a similar expression
      \begin{equation}\label{vtoraia}
\tau=|\phi_{\downarrow}(0-)|^2.
\end{equation}

Indeed, since  $\Psi \in \Dom (\bbH)$, by  \eqref{bcdilation} we get
$$
\begin{pmatrix}
\phi_{\uparrow}(0+)\\
\phi_{\downarrow}(0+)
\end{pmatrix}
=
\begin{pmatrix}
k&-\sqrt{1-k^2}\\
\sqrt{1-k^2}&k
\end{pmatrix}
\begin{pmatrix}
\phi_{\uparrow}(0-)\\
\phi_{\downarrow}(0-)
\end{pmatrix}.
$$
In particular,
$$
\phi_{\uparrow}(0+)=k \phi_{\uparrow}(0-)-\sqrt{1-k^2}\phi_{\downarrow}(0-).
$$
Hence
 \begin{equation}\label{vtor}
\phi_{\downarrow}(0-)=\frac{k \phi_{\uparrow}(0-)-\phi_{\uparrow}(0+)}{\sqrt{1-k^2}}=\frac{k \phi_{\infty}(0-)-\phi_{\infty}(0+)}{\sqrt{1-k^2}}.
\end{equation}
By  \eqref{tautheta2},
$$
\tau=\left |\frac{k \phi_{\infty}(0-)-\phi_{\infty}(0+)}{\sqrt{1-k^2}}\right |^2,
$$
which together  with \eqref{vtor}  proves \eqref{vtoraia}.
\end{remark}

\section{Preliminaries. Probabilities versus amplitudes}\label{s18}

To  discuss  applications of the continuous monitoring principle in connection with the exponential decay phenomenon in quantum mechanics, we   need to warm up with some preliminaries.

Recall that
``when we deal with probabilities under ordinary circumstances, there are the following ``rules of composition": 1)
if something can happen in alternative ways, we add the probabilities for each of the different ways: 2) if the event occurs as a succession of steps - or depends on a number of things happening |`concomitantly" (independently) - then we multiply the probabilities of each of the steps (or things)" \cite[Ch. 3]{Feyn1}.

   Apparently,  under certain circumstances the probability $P$  of an event that   can  be realized in   two  (at  first glance mutually exclusive) alternative ways  $A_1$ and $A_2$ is not necessarily equals the sum of  probabilities $P_1$ and $P_2$ of the  events $A_1$ and $A_2$, that is,
$$P\ne P_1+P_2 \quad \text{(in general)}.
$$
A more detailed analysis of the experimental data  shows that the concept of an alternative should be analyzed  more carefully  and one has to distinguish between  {\it exclusive} and {\it  interference alternatives}. The latter occurs if  there is  no (experimental)  evidence  available  to answer the question  of how  the final event  has been realized, via the occurrence of $A_1$ or $A_2$? In other words, ``when alternatives cannot possibly be resolved by any experiment, they always interfere" \cite[page 14]{Feyn2}.

If the alternative ways of a realization of the event are exclusive, one  has the usual  addition law of probabilities
\begin{equation}\label{problaw}
P_{\text{ex}}=P_1+P_2.
\end{equation}

 In the case of an interference alternative, the rules of composition  should  be applied to  the amplitudes of probability  instead. Recall that there are  complex numbers
 $\phi$, $\phi_1$, $\phi_2$ (the probability amplitudes),  obtained, for example,  by solving a kind of wave  equation,
such  that  $$
P_{\text{int}}=|\phi|^2 , \quad P_1=|\phi_1|^2 \quad \text{and }\quad P_2=|\phi_2|^2.
 $$
In particular, the  addition law for (probability) amplitudes
 \begin{equation}\label{amplaw1}
 \phi=\phi_1+\phi_2
 \end{equation}
yields
 \begin{equation}\label{amplaw2}
 P_{\text{int}}=|\phi|^2=|\phi_1+\phi_2|^2 \quad (\ne P_{\text{ex}} \text{ in general)}.
 \end{equation}
 In this context it should be stressed that
the  (experimental)  knowledge of the  probabilities $P_{1,2}$ (but not the amplitudes $\phi_{1,2}$)  only  gives the two sided-estimate  for the probability $P_{\text{int}}$  of the final  event
 $$
 P_1+P_2- 2\sqrt{P_1P_2}\le P_{\text{int}}\le  \min\{1,P_1+P_2+ 2\sqrt{P_1P_2}\}.
 $$

  All of this is well known and has been extensively discussed in detail in connection with the two slit experiment (see, e.g.,
 \cite{Feyn1, Feyn2,Hei}, also see \cite{Vaid1, Vaid2} for the concept of interaction-free measurements).

Our goal  is to provide a solid mathematical background for understanding the phenomenon on a simple one-dimensional example of a quantum system and develop a framework where the concepts of exclusive and interference alternatives can be rigorously discussed

\section{Massless particles on a ring}

Consider a quantum system the configuration space of which is  a ring $\mathbb{S}$ obtained by identifying  the end-points of  a finite interval $[0, \ell]$.
  The dynamics  of the system   is described by the strongly continuous group of unitary operators
$U(t)=e^{-it/\hbar H}$,
 where the Hamiltonian $H$ is given by the differentiation operator on the ring, or, equivalently,  by the differentiation operator on  the  finite interval $[0, \ell]$ with periodic boundary conditions. That is,
  \begin{equation}\label{givenby}
 H=ic\hbar \frac{d}{dx} \quad \text{on}\quad   \dom(H)=\{f\in W_2^1((0, \ell))\,|\, f(0)=f(\ell)\}.
 \end{equation}
 To motivate the  choice of the Hamiltonian we  use the energy-momentum relation
 $$
 E^2=(cP)^2+(mc^2)^2
 $$
 and assume that we are dealing with a massless particle $(m=0)$ and then choose a square root brunch of  $(cP)^2$ to define   the energy operator $H$  as (cf. \cite{Sus})
 $$
 H=-cP.
 $$
Here $P$ denotes the momentum of a particle moving  with no dispersion at  the speed of light on the ring $\mathbb{S}$
in the direction  from ``$x=0$ to $x=\ell$.''

 In order to save ourselves from inventing new words such as ``wavicles'', we have chosen to call these objects  ``$\lamed\gimel$-particles"  (cf. \cite[p. 85]{Feyn1}).
In our opinion,
 $\lamed\gimel$-particles may, for instance,  serve as a one-dimensional prototype of  low energy electrons  in the vicinity of an impurity
in a zero-gap semiconductor. Recall that such electrons  can formally
 be described by the two-dimensional Dirac-like Hamiltonian
 \[
  H=-ic\hbar \,{\boldsymbol \sigma} \cdot {\boldsymbol\nabla}+V\,,\quad \text{with }\quad c=\nu_F,
 \]
 where  $\nu_F$ is the Fermi velocity, $\boldsymbol\sigma =(\sigma_x,\sigma_y)$ are the $2\times 2$ conventional Pauli matrices
  and $V $ is a short range ``defect'' potential \cite{DM84}.
 In this simplified  model we will imagine these electrons as fake spin-zero electrons which, however, can carry the  charge $\mathfrak{e}$.

Suppose that $\phi\in L^2((0, \ell))$, $\|\phi\|=1$, is a  wave-function describing the initial state of the  quantum system with the Hamiltonian
$$H=ic\hbar \frac{d}{dx}$$
with periodic boundary conditions \eqref{givenby}.
 If  no observation is made whatsoever,
  the time evolution
 $$U(t)\phi=e^{-it/\hbar H}\phi$$ of the state $\phi$
 is given by the family of unitary transformations
   $$
( U(t)\phi)(x)=\widetilde \phi(x+ct), \quad x\in [0,\ell], \quad t\in \bbR.
 $$
Here   $\widetilde \phi $ denotes the periodic extension of the function $\phi(x)$ from the interval $[0, \ell]$  onto the full real axis.
In other words, the wave packet  $U(t)\phi$ is confined to  move
at the speed of light  $c$
without  dispersion on
the ring $\mathbb{S}$ of radius $\ell/2\pi$, obtained from the interval $[0,\ell]$  by identifying its end-points.

 In this case, the survival probability  $$
 p(t)=|(e^{-it/\hbar H}\phi ,\phi)|^2
 $$
  to the initial state $\phi $
 is a periodic function with the period  $T=\ell/c$.

 In the forthcoming sections we will learn that  under continuous monitoring  the quantum system on the ring $\mathbb{S}$ becomes an open quantum system,
the  particles can be emitted and the whole system can be considered as a kind of quantum antenna.

 \section{Continuous monitoring with interference}\label{quansz}

Throughout this section we
assume that   the initial state $\phi$ is a $ W_2^1((0, \ell))$-function   that is allowed to have  a discontinuity (jump) at the point of the observation $x=0\equiv\ell$, that is,
 $$
 \phi(0)\ne \phi(\ell), \text{in general}.$$
Notice that although    $\phi\in  W_2^1((0, \ell))$, the initial state $\phi$ is not  required to belong to the domain of the Hamiltonian $H$ in general. That is,
it is not assumed that $\phi \in W_2^1(\mathbb{S})$, with $\mathbb{S}$ the ring obtained by identifying the end-points of the interval $[0,\ell]$

 The decay properties of  states with a unique jump-point on the ring under continuous
 monitoring are  described by the following result.

 \begin{theorem}\label{quanmon} Suppose that $\bbY$ is the  metric graph  $\bbY=(0,\ell)$ in Case $(ii)$. In the Hilbert space $\cH=L^2(\bbY)$  denote by
  $H$ the differentiation operator
    \begin{equation}\label{operring}
 H=ic\hbar \frac{d}{dx} \quad \text{on}\quad   \dom(H)=\{f\in W_2^1((0, \ell))\,|\, f(0)=f(\ell)\}.
 \end{equation}

If $\phi\in W_2^1((0, \ell))$ and  $\|\phi \|=1$, then
 \begin{equation}\label{interfer}
\lim_{n\to\infty}|(e^{-it/(\hbar n)H }\phi, \phi)|^{2n}=e^{-\tau |t|},\quad t\in \bbR,
 \end{equation}
 where
  \begin{equation}\label{QED0}
 \tau=c|\Delta \phi|^2=c|\phi(\ell)-\phi(0)|^2.
  \end{equation}
\end{theorem}

 \begin{proof}   Let $D_\Theta$  be
 the differentiation operator $i\frac{d}{dx}$ on the finite interval $[0,\ell]$    defined on
 $$
 \dom(D_\Theta)=\{f\in W^1_2(0,\ell)\,|\, f(0)=- \Theta(\ell)\}.
 $$

Since $
 H=c\hbar D_\Theta,
 $ with $\Theta=-1$,
by   Theorem  \ref{decr} (see \eqref{tautheta} in Case (ii) with $\Theta=-1$)  we have
 \begin{align*}
\lim_{n\to \infty} |( e^{-it/(n\hbar) H}  \phi, \phi)|^{2n}&= \lim_{n\to \infty}|(e^{-i  ct/nD_{-1}}\phi, \phi)|^{2n}\\&=e^{-|-\phi(\ell)+\phi(0)|^2c|t|}
=e^{-c|\Delta \phi|^2|t|},
 \end{align*}
which proves   \eqref{interfer}.
 \end{proof}

More generally, just repeating  the proof presented above for an arbitrary $\Theta$, $|\Theta|=1$, we have the following
 \begin{corollary}\label{bohmcor}
{\it Let  $H^\Phi$, $ \Phi\in \bbR$, be the self-adjoint realization of the differential expression
$$
H^\Phi=ic\hbar \frac{d}{dx}
$$
on
$$
\dom (H^\Phi)=\{f\in W^1_2((0,\ell))\,|\, f(0)=e^{-i\Phi}f(\ell)\}.
$$

Then,
  \begin{equation}\label{QED} \lim_{n\to \infty} |(e^{-it/(\hbar n) H^\Phi} \phi,  \phi)|^{2n}=
  e^{-\tau_\Phi  |t|},
 \end{equation}
 where
 \begin{equation}\label{QED1}
 \tau_\Phi=c|\Delta_\Phi \phi|^2=c|e^{-i\Phi}\phi(\ell)-\phi(0)|^2 .
 \end{equation}

In particular,
if $$e^{-i\Phi}\phi(\ell)\ne\phi(0),$$
then  $ \phi$ is a resonant state
under continuous monitoring of the unitary  evolution $$\phi \mapsto e^{-it/\hbar H^\Phi} \phi$$
governed by the Hamiltonian $H^\Phi$.

}
 \end{corollary}

\begin{remark}\label{aabb} Notice that the magnetic Hamiltonian $H^\Phi$  is unitarily equivalent to the operator  $H+\mathfrak{e}\cA(x)$
with
$$
\Phi= \frac{\mathfrak{e}}{c\hbar}\int_0^\ell \cA(x)dx.
$$
Here  $\mathfrak{e}$ is  the ``charge" of the $\lamed\gimel$-particle, $  \cA(x)$ is the magnetic potential (we assume that $\cA(x)$ is a piecewise real-valued continuous function),
and   $\int_0^\ell \cA(x)dx $  is  the flux of the field through the  ring.

 Indeed,
denote by  $U$  the   unitary  multiplication operator
$$
(U\Psi)(x)=\exp \left [{i\frac{\mathfrak{e}}{c\hbar}\int_0^x\cA(s)ds}\right ]\cdot \Psi(x), \quad \Psi \in L^2((0, \ell)).
$$

 Then
 $$
U^*(H+\mathfrak{e}\cA(x))U=H^\Phi,
$$
 which  follows from  the equality $$ \left (ic\hbar \frac{d}{dx}+\mathfrak{e}\cA(x)\right ) [ \cE(x)\cdot \Psi(x)]=\cE(x)\cdot
 ic\hbar \frac{d}{dx}\Psi(x),
 $$
 where
 $$
 \cE(x)=\exp \left [{i\frac{\mathfrak{e}}{c\hbar}\int_0^x\cA(s)ds}\right ],
 $$
 and the observation that
 $$
 \Dom (H)=U(\Dom (H^\Phi)) =U\left (\{f\in W_2^1((0, \ell))\,|\, f(0)=e^{-i \Phi}f(\ell)\}\right ).
 $$

\end{remark}

 Theorem \ref{quanmon}  and Corollary \ref{bohmcor} clearly suggest  that continuous observation over a quantum system
  should rather be  treated in the framework  of open quantum systems theory.
  Below is  a suitable model for that.

  \begin{theorem}\label{model} Given $\Phi\in [0,2\pi)$,
  in the Hilbert space $
  \cH=L^2((0,\ell))\oplus \bbC$ introduce the maximal dissipative operator $\widehat H^\Phi$
 defined  on
  $$
  \Dom(\widehat H^\Phi)=\left \{\begin{pmatrix}
  f\\c
  \end{pmatrix}
  \,\Bigg| \, f\in W_2^1((0,\ell)), \, c=f(0) \right \}
  $$
 as
  $$
  \widehat H^\Phi\begin{pmatrix}
  f\\c
  \end{pmatrix}=i\begin{pmatrix}
\frac{d}{dx}f(x)\\
f(0)-e^{-i\Phi}f(\ell)
 \end{pmatrix}.
  $$

  If $\phi \in \cH$ is a state such that $\phi \in \Dom (\widehat H^\Phi)$,
  $$
  \int_0^\ell|\phi(x)|^2dx+|\phi(0)|^2=1,
  $$
  then
  \begin{equation}\label{QED2} \lim_{n\to \infty} |(e^{it/ n) \widehat H^\Phi} \phi,  \phi)|^{2n}=
  e^{-\tau_\Phi  |t|},
 \end{equation}
 where
 \begin{equation}\label{QED3}
 \tau_\Phi=|e^{-i\Phi}\phi(\ell)-\phi(0)|^2 .
 \end{equation}

  \end{theorem}
  \begin{proof}
  We have
  $$
  (\widehat H^\Phi f,f)=i\int_0^\ell f'(x)\overline{f(x)}dx+i(f(0)-
  e^{-i\Phi}f(\ell))\overline {f(0)},\quad f\in \Dom (\widehat H^\Phi).
  $$
  In particular,
 \begin{align*}
 \Im \left ( (\widehat H^\Phi f,f)\right )&=
 \frac12 |f(\ell)|^2-\frac12|f(0)|^2+|f(0)|^2-\Re\left (
  e^{-i\Phi}f(\ell)\overline {f(0)}\right)
  \\&=\frac12| e^{-i\Phi}f(\ell)-f(0)|^2\ge 0,
  \end{align*}
  which shows that $\widehat H^\Phi$ is a dissipative operator.
  To complete the proof it remains to check that  the lower half-plane belongs to the resolvent set of $\widehat H^\Phi$, so that $\widehat H^\Phi$ is a maximal dissipative operator, and then use the same    reasoning
  as the one in the proof of
   Lemma \ref{dissss}.

  \end{proof}
  \begin{remark}
The idea to add to the original Hilbert space $\cH$ a one-dimensional ``vacuum" subspace $\bbC$ is due to Schrader \cite{Schrader}, also see   \cite{Makthree}, \cite{Pavthree} and \cite{Shondin},
 where such extensions of a non-densely-defined symmetric operators found applications in modeling three-body systems with $\delta$-like interactions that are  free of the ``fall to the center" phenomenon.
 For the general extension theory for non-densely-defined operators and its applications we also refer to \cite{ABT,Krasn0,Krein0,Pav0}.

  \end{remark}

\subsection{Discussion}

The decay law  \eqref{interfer} shows that continuous monitoring eventually triggers an exponential decay of the  system. Meanwhile,  the  explicit expression  \eqref{QED0}
 for the decay constant $\tau$, $\tau=c|\Delta\psi|^2$,  suggests that we are dealing with  an interference alternative, which means that   we cannot  apply the laws of probabilities \eqref{problaw} and have to count on the composition  laws of amplitudes \eqref{amplaw1}.
 Indeed,
a particle arriving at the junction point, the point of observation,  has  two options: a) either to   stay on the track or b) be emitted. However, there is no way
to ``experimentally" confirm which option has been realized in reality. That is,
 we are not certain  about what  happened at the junction  point $x=0=\ell$, and  consequently,  the   description
 of  motion becomes an interference alternative.

 On the quantitative level, the reasoning presented above can be supported by the following considerations.

The incoming $\phi(\ell)$, outgoing $\phi(0)$ and  emission amplitudes  $\phi_\text{em}$ are to  satisfy the ``interference'' relation
  \begin{equation}\label{amplaw}
 \phi(\ell)=\phi_\text{em}+\phi(0).
 \end{equation}

 Since
 the quantity  $c|\phi_{\text{em}}|^2  t$ asymptotically describes
 the probability that the emitted particle can eventually be detected during the time interval $[0,t]$,
 \footnote{This can be justified as follows.
 Suppose we   put in  an ideal  detector   that  counts particles
 passing through the point $x_0$  in  the interval   $(0,\ell)$. If the initial state $\phi$ is a smooth function in a neighborhood of $x_0$, the probability $p_{x_0}(t)$ that the detector  will detect a particle during the time interval $t$ is asymptotically given by
 \begin{equation}\label{explog}
p_{x_0}(t)=c|\phi(x_0)|^2 t +o(t) \quad \text{as} \quad t\to 0.
\end{equation}

To justify the claim, recall that
 in accordance with the probabilistic interpretation of the wave-function, the probability to find the particle inside the interval $\delta \subset [0,\ell]$ is given by
 $$
 \text{Pr}\{ \text{``particle''}  \in \delta\}=\int_\delta |\phi(x)|^2dx.
 $$
In particular,
\begin{align*}
 \text{Pr}\{ \text{``particle''}  \in [x_0-\varepsilon,x_0]\}&=\int_{x_0-\varepsilon}^{x_0} |\phi(x)|^2dx
 \\&=|\phi(x_0)|^2 \varepsilon  +o(\varepsilon) \quad \text{as} \quad \varepsilon\to 0.
 \end{align*}

  If we repeat the experiment $N_\infty$ times,
  the quantity  $$
 \Delta N=N_\infty \int_{x_0-\varepsilon}^{x_0} |\phi(x)|^2dx
 $$
gives  the (average) number of outcomes when the (quantum) particle is accommodated by
the interval  $[x_0-\varepsilon,x_0]$.

One can change the point of view and assume that we are dealing with a beam of particles and that  initially  there were $N_\infty$   particles in the system.
Therefore,   $ \Delta N$ would have  the meaning
of the    averaged number   of particles
in the interval   $[x_0-\varepsilon,x_0]$.
For the quantum system in question wave-particle duality is exact and  hence we may assume
 that the particles are moving to the right
with speed of light $c$. Therefore,  in time
$$
t=\frac{\varepsilon}{c}
$$
all the  particles will leave the interval  $[x_0-ct,x_0]$ and will eventually be counted by  the detector.

Finally, the probability that the detector  will go off within the time interval $[0,t]$ is asymptotically given by
$$
 p_{x_0}(t)\sim\frac{\Delta N}{N_\infty}=\int_{x_0-ct}^{x_0} |\phi(x)|^2dx=
c|\phi(x_0)|^2 t +o(t)\quad \text{as} \quad t\to 0,
$$
which justifies the claim \eqref{explog}.
}
  the probability $P(t)$
of  staying on the ring
should fall off exponentially as
 \begin{equation}\label{eqqed}
 P(t)=e^{- c|\phi_{\text{em}}|^2  t}=e^{- c|\Delta \phi|^2  t}, \quad t\ge 0,
 \end{equation}
 where
  \begin{equation}\label{quantdecr}
 |\phi_{\text{em}}|^2= |\Delta \phi|^2=|\phi(\ell)-\phi(0)|^2,
 \end{equation}
 which gives a heuristic justification of decay law  \eqref{interfer} in Theorem \ref{quanmon}.

Notice that if  $\phi\in \Dom (H) $, then   the wave function is continuous at the junction point, that is,
 \begin{equation}\label{contj} \phi(\ell)=\phi(0).
 \end{equation} Therefore,   $|\phi_{\text{em}}|^2=0$ by \eqref{quantdecr} and hence
there is no emission of particles.
 In this case, the dynamics is frozen by the  continuous
  monitoring and  we face the quantum Zeno effect.

   However,  in the situation in question, in view of Corollary \ref{bohmcor} and Remark \ref{aabb}  one can  unfreeze the evolution (stopped by  continuous monitoring)
   by switching on the magnetic field through the ring.
 Indeed, since  the configuration space of the system (the ring $\mathbb{S}$) is not a simply connected set,
the effect of the magnetic potential will be to produce  the phase shift of the wave function  \cite{AhB} at the junction point
even if the magnetic field is absent in a neighborhood of the ring $\mathbb{S}$ (the Aharonov-Bohm effect)
  $$
  \phi(\ell)\mapsto \phi (\ell) e^{-i\Phi},
  $$
  where
  $$
  \Phi=\frac{\mathfrak{e}}{c\hbar}
   \int_0^\ell \cA(x)dx.
    $$
  Here  $\mathfrak{e}$ is  the ``charge" of the $\lamed\gimel$-particle and   $\int_0^\ell \cA(x)dx $  is  the flux of the field through the  ring.

    In this case,  the interference relation \eqref{amplaw} should be modified as  $$
 e^{-i\Phi} \phi(\ell)=\phi_\text{em}+\phi(0).
 $$
Therefore,  the decay properties of the state under continuous monitoring are  determined by the quantity
\begin{equation}\label{aharbom}
  |\Delta_\Phi \phi|^2=|e^{-i\Phi}\phi(\ell)-\phi(0)|^2,
 \end{equation}
which finally leads to the  decay law
$$
 P(t)=e^{- |\phi_{\text{em}}|^2 c t}=e^{- |\Delta_\Phi \phi|^2 c t}, \quad t\ge 0,
$$
 for the probability $P(t)$  to detect a particle remaining on the ring.

      In particular, under the assumption that  $|\phi(\ell)|=|\phi(0)|$, it follows from \eqref{aharbom} that
      the decrement
      $ |\Delta_\Phi \phi|^2
      $ experiences quite typical  Aharonov-Bohn oscillations which are periodic with respect to the flux of the field. That is,
     \begin{align}
       |\Delta_\Phi \phi|^2&= |e^{-i\Phi}e^{i \arg \phi(\ell)}-e^{i\arg\phi(0)}|^2\cdot |\phi(\ell)|^2\nonumber
       \\&=4 \sin^2\left (\frac{\Phi -\Delta \arg \phi}{2}\right )\cdot |\phi(\ell)|^2, \label{intell}
 \end{align}
where
      $$
     \Delta  \arg \phi=\arg \phi(\ell) -\arg \phi(0).
      $$

      The above discussion provides  a physically motivated example of a quantum system  with   the decay law
       \eqref{QED}, \eqref{QED1}, see  Corollary  \ref{bohmcor}.

 \section{Continuous monitoring with no interference}\label{classsz}

 Continuous monitoring of  open quantum systems leads to a completely  different understanding of decay processes.

 We will assume that the  time evolution of the system  is governed by a semi-group of contractive transformations generated by  a dissipative differentiation operator   such that  the initial state
 belongs to the domain of the  operator.

 \begin{theorem}\label{quanmondis}  Suppose that $\bbY$ is the  metric graph  $\bbY=(0,\ell)$ in Case $(ii)$.   Given   $|\varkappa|<1$, in the Hilbert space $\cH=L^2(\bbY)$  denote by  $\widehat H_\varkappa$
   the dissipative differentiation operator  $$
\widehat H_\varkappa=i c \hbar \frac{d}{dx}
 $$
 on
 \begin{equation}\label{bckappa}
 \Dom( \widehat H_\varkappa)=\left \{f\in W_2^1((0, \ell))\,|\, f(0)=\varkappa f(\ell)\right \}.
 \end{equation}

Assume that
 $\phi\in \Dom (\widehat H_\varkappa)
 $
and  $\|\phi \|=1$.

 Then
 \begin{equation}\label{clastau}
 \lim_{n\to\infty}|(e^{it/(\hbar n)\widehat H_\varkappa }\phi, \phi)|^{2n}=
e^{-\tau t},\quad t\ge 0,
\end{equation}
 where
\begin{equation}\label{clastau1}
 \tau=c\Delta |\phi|^2=c(|\phi(\ell)|^2-|\phi(0)|^2).
 \end{equation}

\end{theorem}
 \begin{proof} Integration by parts yields
$$
\int_0^\ell \phi'(x)\overline{\phi(x)}dx=|\phi(\ell)|^2 -|\phi(0)|^2-\int_0^\ell \phi(x)\overline{\phi'(x)}dx,
 $$
and therefore
 \begin{align*}
\Re  \int_0^\ell \phi'(x)\overline{\phi(x)}dx&=\frac{\int_0^\ell \phi'(x)\overline{\phi(x)}dx+\overline{\int_0^\ell \phi'(x)\overline{\phi(x)}dx}
}{2}
\\&=\frac{|\phi(\ell)|^2 -|\phi(0)|^2}{2}\ge 0  \quad (\text{for} \quad \phi\in \Dom(\widehat H_\varkappa).
\end{align*}

 Since  $\phi \in \Dom (\widehat H_\varkappa)$,
 as in the proof of Lemma \ref{dissss} one obtains
$$
\lim_{n\to \infty}|(e^{it/(n\hbar) \widehat H_\varkappa } \phi, \phi)|^{2n}=e^{- 2\hbar^{-1} \Im (\widehat H_\varkappa \phi, \phi)t }, \quad t\ge 0,
$$
with
 \begin{align*}
 2\hbar^{-1} \Im (\widehat H_\varkappa \phi, \phi)t&=2c\,\Im\left ( i\int_0^\ell \phi'(x)\overline{\phi(x)}dx\right )
 \\&= c(|\phi(\ell)|^2 -|\phi(0)|^2)=\tau.
 \end{align*}
 \end{proof}

 One can go back from the reduced  description of the open quantum system to the full one
 following the extended Hilbert space approach  presented below.

 Consider an open quantum system prepared in the state $\phi \in L^2(\mathbb{S})=L^2((0,\ell))$ the time evolution of which  generated by the dissipative operator
 $\widehat H_\varkappa$ with the boundary condition parameter $\varkappa$, $|\varkappa|<1$, referred to in Theorem \ref{quanmondis}.

 Suppose  that $\phi\in W_2^1((0, \ell))$ is  such that the radiation condition
  \begin{equation}\label{instr}
 |\phi(\ell)|> |\phi(0)|
 \end{equation}
holds. Assume, in addition, that $\phi \in\Dom ( \widehat H_\varkappa)$, that is,
 $$
   \varkappa=\frac{\phi(0)}{\phi(\ell)}.
$$

Along  with  the open  quantum system
in the state space  $L^2((0,\ell))$
 introduce
a new  quantum system in an extended Hilbert space $\fH$
 containing    $L^2((0,\ell))$ as a (proper) subspace. For the Hamiltonian  $\bbH$  of the new system in the extended Hilbert space
we choose a (minimal)  self-adjoint dilation of the dissipative operator $\widehat H_\varkappa$
and  the new state of the system  $\widehat \phi\in \fH $  is a clone of  $\phi$ considered as an element of the extended Hilbert space $\fH$.

The  results of continuous monitoring of  the quantum evolution  $\widehat \phi \mapsto  e^{-it/\hbar \bbH}\widehat \phi$    generated by the self-adjoint Hamiltonian
$\bbH$   in the Hilbert space $\fH$  can be described  as follows.

\begin{corollary}\label{smes0}
{\it
 Suppose that $\bbY$ is the  metric graph  $\bbY=(0,\ell)$ in Case $(ii)$.   Given   $|\varkappa|<1$, in the Hilbert space $\cH=L^2(\bbY)$  denote by  $\widehat H_\varkappa$
   the dissipative differentiation operator  $$
\widehat H_\varkappa=i c \hbar \frac{d}{dx}
 $$
 on
 \begin{equation}\label{bckappa0}
 \Dom( \widehat H_\varkappa)=\left \{f\in W_2^1((0, \ell))\,|\, f(0)=\varkappa f(\ell)\right \}.
 \end{equation}

Let $\bbH$  be a self-adjoint dilation of the dissipative operator $\widehat H_\varkappa$
in an extended Hilbert $\fH$ space containing the original Hilbert space $\cH$  as a $($proper$)$ subspace $\cH=L^2((0,\ell))\subset \fH$.
Suppose that
$$
\phi\in \cH, \quad \|\phi\|=1.
$$

Then
$$
 \lim_{n\to \infty} |(e^{-it/(\hbar n)\bbH } \phi,  \phi)|^{2n}= \lim_{n\to \infty} |(e^{i|t|/(\hbar n)\widehat H_\varkappa } \phi,  \phi)|^{2n}, \quad t\in \bbR,
 $$
 provided that at least one (and therefore both) of the limits exist.

 In particular, if
 $$
 \phi\in \Dom (\widehat H_\varkappa),
 $$
then
$$
 \lim_{n\to \infty} |(e^{-it/(\hbar n)\bbH } \phi,  \phi)|^{2n}=e^{-\tau  t},\quad t>0,
 $$
where the decay constant $\tau$ is given by
\begin{equation}\label{obscheetaudil0}
\tau=
c(|\phi(\ell )|^2- |\phi(0)|^2).
\end{equation}
}

\end{corollary}

 \subsection{Discussion}
 The decay law  \eqref{clastau},  \eqref{clastau1} suggests  that  the interference effects are definitely absent.
In order to get an adequate explanation for  the phenomenon, instead of applying the composition law of amplitudes \eqref{amplaw}  one has to use the calculus of probabilities
   \begin{equation}\label{problaw1}
 |\phi(\ell)|^2=|\phi_\text{em}|^2+|\phi(0)|^2.
 \end{equation}
Here is an argument supporting  \eqref{problaw1}:  a particle arriving at the junction point still has  two options: a) either to   stay on the track or b) be emitted.
However, the transition from the open quantum system with state space $\cH=L^2(\mathbb{S})$ referred to in Theorem \eqref{quanmondis} to the closed one in the extended Hilbert space $\fH$ containing $\cH$ as a proper subspace and discussed  in Corollary \ref{smes0} assumes that an additional scattering channel $\fH\ominus \cH$ is added and the emitted particles as well as the ones stayed on the track  can eventually  be counted.
In other words, we are dealing with an exclusive alternative.

 From the experimental viewpoint in this case, the arrangement of the corresponding Gedankenexperiment involves the installation of two additional detectors
 $D_{\ell}$ and $D_{0}$ that  count the particles that
 pass through the point $x=\ell$ and $x=0$. In other words,
we
 accept  the experimental condition  that the emitted particles can eventually  be counted (by combining the readings of the two detectors).

  Given \eqref{problaw1}, arguing as in Section 19
we
arrive  to  the exponential decay law
  \begin{equation}\label{clssika1}
 P(t)=e^{- |\phi_{\text{em}}|^2 c t}=e^{- c\Delta |\phi|^2  t}, \quad t\ge 0,
 \end{equation}
where $P(t)$  stands for the probability to detect  the particle on the ring at the moment of time $t$ and
 \begin{equation}\label{classdecr}
|\phi_{\text{em}}|^2= \Delta |\phi|^2=|\phi(\ell)|^2-|\phi(0)|^2.
 \end{equation}

Below we offer an heuristic explanation of the law \eqref{classdecr} based on purely classical interpretation of the nature of a $\lamed\gimel$-particle.

  When we watch the beam  of $\lamed\gimel$-particles
  by observing the readings of the two detectors,
  we indeed deal with an exclusive alternative. Denote by $N_\infty(t)$ the total amount of particles on the ring  at the moment of time $t$ and let $N_\ell$  be the amount of particles  that passed through the point $x=\ell$ and arrived at the check point $x=0$ during the time interval $[0,t]$. The arrived particles   ``have'' the alternative: either   to keep moving  on the ring or to be emitted.
  By checking the  readings of the detector $D_0$ we know  that     $N_0$
   out of  $N_\infty $ particles
   stayed traveling  along the ring. Next, taking into account the readings of the detector $D_\ell$, we conclude that
 the remaining  $\Delta N=N_\ell-N_0$ particles   have beed emitted  during the time interval $[0,t]$. (Here we implicitly assume that there is no other mechanism that causes the  particles
to radiate. Why this hypothesis is consistent with the way of the suggested reasoning will be explained later).

 Given the wave-function probabilistic interpretation above, it is easy to see that
 $$
 \frac{N_\ell}{N_\infty}= c|\phi(\ell-)|^2t   +o(t) \quad\text{ and }\quad  \frac{N_0}{N_\infty}= c|\phi(0+)|^2t +o(t)
 \quad
 \text{as}\quad  t\to 0.
 $$

Therefore,
$$
 \frac{ \Delta N }{N_\infty}  =\left ( |\phi(\ell)|^2-|\phi(0)|^2\right ) ct+o(t) .
 $$
 Repeating that monitoring    over and over, in the limit $ t\to 0$, we arrive at the differential equation
 $$
 \frac{dN}{dt}=-c\left (|\phi(\ell)|^2-|\phi(0)|^2\right )N,
 $$
 $$
 N(0)=N_\infty,
 $$
that  governs the  counting  process.

 Therefore,  under continuous monitoring with detectors that are going off,
 the total  number of particles $N(t)$   as a function of time falls off exponentially  as
 \begin{equation}\label{nuivot}
 N(t)=N_\infty e^{-  c\Delta |\phi|^2 t},
 \end{equation}
where  the  decrement  $ \Delta |\phi|^2>0$ is given by
\eqref{classdecr}
 (provided that   $|\phi|$ has a jump at the point $x=0$).

 Notice that if the state $\phi$ is a continuous function on the ring, then no emission is observed and
then the quantum Zeno effect  takes place.

Summarizing, we arrive at the conclusion that the computation of the emission probability for the quantum system $(H,\phi)$ referred to in Theorem \ref{quanmon}
requires the application of the composition law  of amplitudes \eqref{amplaw}  (the interference alternative scenario). Meanwhile, the decay rate  for the open quantum system  $(\widehat H_\varkappa, \phi)$ referred to in Theorem \ref{quanmondis} or for the quantum system $ (\bbH, \widehat \phi)$ in the extended Hilbert space $\fH$ (see Corollary \ref{classmoncor})
can be evaluated using the  rules of the calculus of probabilities \eqref{problaw1} (the exclusive alternative scenario).

\section{The Self-adjoint dilation}

The self-adjoint dilation $\bbH$ of the dissipative operator $\widehat H_\varkappa$ referred to in Corollary \ref{smes0} can be described explicitly as the differentiation operator
on the metric graph $\bbY$ with appropriate boundary conditions. However, the geometry of the metric
graph that determines the configuration space of the quantum system depends on whether  or not the parameter $\varkappa$ in the boundary  condition \eqref{bckappa} vanishes.

 If $\varkappa\ne 0$,
 as it follows from Theorem \ref{dilthm}, the extended Hilbert space can  be chosen to coincide with
 $\fH=L^2(\bbY)=L^2(\overline{\bbY})$
 where  $\bbY$ is the metric graph  $$\bbY=(-\infty,0)\sqcup(0,\ell)\sqcup(0,\infty)\quad \text{ in Case } (iii) $$ and $\overline{\bbY}$  denotes its one-cycle completion,
 the configuration space of the extended quantum system, while the  self-adjoint dilation
 $\bbH$, the Hamiltonian, coincides with  the differentiation operator
on the graph
    \begin{equation}\label{}
 \bbH=ic\hbar \frac{d}{dx}
  \end{equation}
defined on the domain of functions satisfying the boundary conditions
\begin{equation}\label{bbking}
\begin{pmatrix}
f_\infty(0+)\\
f_\ell(0)
\end{pmatrix}=
\begin{pmatrix}
|\varkappa|&- \sqrt{1-|\varkappa|^2} \frac{\varkappa}{|\varkappa|}\\
 \sqrt{1-|\varkappa|^2}&\varkappa
\end{pmatrix}
 \begin{pmatrix}
f_\infty(0-)\\
f_\ell(\ell)
\end{pmatrix}.
\end{equation}

 In the exceptional case  $\varkappa= 0$,  the boundary condition \eqref{bckappa}   that determines the dissipative operator $\widehat H_0$
 is local. Therefore, it is convenient to assume that the configuration space for the corresponding open quantum system (if $\varkappa =0$) is the finite interval $(0,\ell)$ rather than a ring.
Consequently, in this case  the extended Hilbert space may be chosen as
 $\fH=L^2(\bbY)$
 where  $\bbY$ is the metric graph
 $$
 \bbY=(-\infty,0)\sqcup(0,\ell)\sqcup(\ell,\infty)\quad \text{ in Case }(i).
 $$
The  self-adjoint dilation (the Hamiltonian)
 $\bbH$ of the  dissipative operator $\widehat H_0$ can be chosen  to be the differentiation operator
on the graph
    \begin{equation}\label{bich}
 \bbH=ic\hbar \frac{d}{dx}
  \end{equation}
with  the self-adjoint boundary condition
 \begin{equation}\label{bichbc}
f(\ell+)=\Theta f(\ell-), \quad |\Theta|=1,
  \end{equation}with an  arbitrary choice of the unimodular extension parameter $\Theta$.

Notice that in the limit $\varkappa\to 0$ along the ray $\varkappa=-|\varkappa|\Theta$  the boundary conditions
\eqref{bbking} split as
$$
f_\infty(0+)=\Theta f_\ell(\ell)
$$
and
\begin{equation}\label{posled}
f_\infty(0-)=f_\ell(0).
\end{equation}
In view of \eqref{posled}, the one-cycle graph $\overline{\bbY}$ ``unwinds" to a straight line
$$\bbY=(-\infty,0)\sqcup(0,\ell)\sqcup(\ell,\infty)$$ which can naturally  be identified with the real axis.
One can show that  the corresponding  self-adjoint dilations  of the dissipative operators $\widehat H_{-|\varkappa |\Theta}$
approach  (in the strong resolvent sense)
 the operator \eqref{bich}  with the boundary condition \eqref{bichbc}.

Notice that  emission  amplitude  $\phi_{\text{em}}$  from \eqref{classdecr}  can be evaluated
by directly solving   the Schr\"odinger equation
\begin{equation}\label{schonY}
i\hbar \frac{\partial}{\partial t}\Psi =\bbH \Psi
\end{equation}
on the one-cycle graph $ \overline{\bbY}$.

Indeed, representing the initial state  $\widehat \phi $ in the extended Hilbert space $\fH=L^2(\overline{\bbY})$ as the two-component vector function
$$
\widehat \phi=\begin{pmatrix}\widehat \phi_\infty\\
\widehat \phi_\ell
\end{pmatrix}
=\begin{pmatrix}0\\ \phi
\end{pmatrix},\quad \widehat \phi \in L^2(\overline{\bbY}),
$$
denote by $\Psi (t, x)$, $x\in \overline{\bbY}$  the solution of the Schr\"odinger equation
\eqref{schonY}
with the initial condition
$$
\Psi|_{t=0}=\widehat \phi=\begin{pmatrix}0\\ \phi
\end{pmatrix}.
$$
We claim that  emission  amplitude  $\phi_{\text{em}}$    can be evaluated  as the limiting value of the first component $\Psi_\infty(t, 0+)$
of the solution $\Psi$ as $t\to 0$.

Indeed, since
$$
\lim_{t\downarrow 0}\lim_{\varepsilon \downarrow 0}\Psi (t, \epsilon)=\begin{pmatrix}
f_\infty(0+)\\
f_\ell(0)
\end{pmatrix}=
\begin{pmatrix}
|\varkappa|&- \sqrt{1-|\varkappa|^2} \frac{\varkappa}{|\varkappa|}\\
 \sqrt{1-|\varkappa|^2}&\varkappa
\end{pmatrix}
 \begin{pmatrix}
 0\\
 \phi(\ell)
 \end{pmatrix},
$$
we have
$$
\phi_{\text{em}}=  \Psi_\infty (0+, 0+)=- \sqrt{1-|\varkappa|^2} \frac{\varkappa}{|\varkappa|}\phi(\ell).
$$
Therefore, the probability density of the event of emission of a particle is given by
$$
|\phi_{\text{em}}|^2= (1-|\varkappa|^2)|\phi(\ell)|^2=|\phi(\ell)|^2-|\phi(0)|^2=\Delta |\phi|^2,
$$
which agrees with  \eqref{clastau}, \eqref{clastau1} (cf. Remark \ref{chtoto}).  Here we have used that the boundary condition
$$
\phi(0)=\varkappa\phi(\ell)
$$
holds, that is,
$$
\phi\in \Dom (\widehat H_\varkappa).
$$

Notice, that in the exceptional case,  that is,   if the initial state is such that $\phi(0)=0$,
 the  dissipative operator $\widehat H_0$ associated with the corresponding open quantum system has the domain
$$
   \Dom( \widehat H_0)=\{f\in W_2^1((0,\ell))\, |\, f(0)=0\}.
   $$
In this case, the differentiation operator $\bbH$ on the whole real axis, $$
\bbH=ic\hbar \frac{d}{dx}
\quad \text{on}\quad
\dom (\bbH)=W_2^1((-\infty, \infty)),
$$
 dilates
$\widehat H_0$. Therefore, the configuration space of the extended quantum system is just the real axis $\bbY=\bbR$, not the one-cycle graph $\overline{\bbY}$.
The corresponding solution $\Psi(x,t)$ of  the Schr\"odinger equation \eqref{schonY}
with the initial data
$$
\Psi|_{t=0}=\phi
$$
(here we do not distinguish the function $\phi$ on $[0,\ell]$ and  its extension by zero on the  whole real axis $\bbR$)
is a one-component function given by
$$
\Psi(x,t)=\phi(x-ct).
$$
In this case the emission amplitude $ \phi_{\text{em}}$ can be evaluated as
$$
\phi_{\text{em}}=\Phi(x-ct)|_{x=\ell,
\, t=0}=\phi(\ell),
$$
so that again
$$
|\phi_{\text{em}}|^2= |\phi(\ell)|^2=|\phi(\ell)|^2-|\phi(0)|^2=\Delta |\phi|^2.
$$

It is also worth mentioning that the decrement  $ \Delta | \phi|^2$ given by \eqref{classdecr} and referred to in Theorem  \ref{quanmondis} is  gauge invariant
 while
  $ |\Delta \phi|^2$ defined  in \eqref{quantdecr} is not.

  Indeed, if
\begin{equation}\label{gauge}
(Vf)(x)=e^{i\lambda(x)}f(x), \quad x\in (0,\ell),
\end{equation}
 is  a (unitary) gauge transformation,
 where $\lambda(x) $ is a differentiable function on $[0, \ell]$,
 then
 $$
 \Delta | V\phi|^2=\Delta | \phi|^2
 $$
and
 $$
 |\Delta V \phi|^2 =|e^{i\Phi}\phi(\ell)-\phi(0)|^2 \ne   |\Delta \phi|^2\quad \text{(in general)}.
  $$
Here,
  $$
  \Phi=\lambda(\ell)-\lambda(0)=\int_0^\ell \frac{d}{dx}\lambda(x) dx.
  $$
 is a shift of the  relative phase of the wave function.

As we have already pointed out, if the wave function is continuous at the junction point, that is
$$
\phi(0)=\phi(\ell),
$$
then $$
\Delta | \phi|^2=
 |\Delta  \phi|^2=0.
 $$
 In this case  the quantum Zeno effect takes place regardless of whether  both  of the detectors $D_0$ and $D_\ell$ go off or only one of them does.
 Moreover,  if $V$ is a gauge transformation, we also have that
 $$
 \Delta | V\phi|^2=\Delta | \phi|^2=0
 $$that  shows that the Zeno effect  is stable with respect to the gauge  transformations when  both of the  detectors $D_0$ and $D_\ell$ go off.
The situation is quite different
 in the experiment when only one  of the detectors  goes off. A tiny  variation of the phase of the wave functions can easily
  transfer the system from the quantum Zeno mode to the exponential decay regime with
  the decay rate given by  the ``magnetic" decrement
$$
 |\Delta_\Phi \phi|^2=|e^{-i\Phi}\phi(\ell)-\phi(0)|^2.
$$
We remark that
 $$ (|\phi(\ell)|-|\phi(0)|)^2\le |\Delta_\Phi \phi|^2\le (|\phi(\ell)|+|\phi(0)|)^2.
 $$
Moreover,  by changing the ``flux" $\Phi$, the upper and as well as the lower  bound for the magnetic decrement
can easily  be attained.  In particular,
 $$ 0\le |\Delta_\Phi \phi|^2\le4 |\phi(0)|^2,
$$
whenever the continuity condition \eqref{contj} at the junction point holds.

\section{General open quantum systems on a ring}

Notice that while considering open quantum systems $(\widehat H_\varkappa, \phi)$ on a ring referred to in  Theorem \ref{quanmondis}, we assumed
  $$
 \phi\in \dom(\widehat H_\varkappa).
 $$
This requirement can be relaxed and we arrive at the following more general  result.

 \begin{theorem}\label{general}
Suppose that $\phi\in W_2^1((0, \ell))$ is a state, $ \|\phi\|=1$. Given   $|\varkappa|\le1$, in the Hilbert space $\cH=L^2(\bbY)$, $\bbY=(0,\ell)$,
denote by   $$ \widehat H_\varkappa=i c\hbar \frac{d}{dx} $$  the  differentiation operator with the  boundary condition
$$
f(0)=\varkappa f(\ell),
$$
that is,
 $$
 \Dom( \widehat H_\varkappa)=\left \{f\in W_2^1((0, \ell))\,|\, f(0)=\varkappa f(\ell)\right \}.
 $$

Then
\begin{equation}\label{decdec}
 \lim_{n\to \infty} |(e^{it/n\widehat H_\varkappa  } \phi,  \phi)|^{2n}=e^{-\tau  t},\quad t>0,
 \end{equation}
where
\begin{equation}\label{obscheetau}
\tau=
|\phi(0)-\varkappa \phi(\ell)|^2+(1-|\varkappa|^2)|\phi(\ell)|^2.
\end{equation}

 \end{theorem}

 \begin{proof} Without loss we will assume that we work in the system of units where $c=1$ and $\hbar =1$.

 It is sufficient to prove the asymtotic representation
  \begin{equation}\label{obschee}
\Re (e^{i t\widehat H_\varkappa }\phi, \phi)=1-\frac12  \tau t+o(t) \quad \text{as }\quad t\downarrow 0.
\end{equation}

  Denote by $W(t)=e^{it\widehat H_\varkappa}$ the contractive semi-group generated by the operator $\widehat H_\varkappa$.
 Notice that
 $$
 W(t+\ell)=\varkappa W(t), \quad t\ge 0 ,
 $$
 with
 $$
( W(t)\phi(x)=
\begin{cases}
\varkappa \phi(\ell+x-t),& 0<x<t
\\
\phi(x-t),&x<t<\ell
\end{cases}.
 $$
 Assume  that $\phi\in W_2^1((0,\ell))$ is the (absolutely)  continuous representative of the element $\phi$ (denoted by the same symbol). We have
 $$
 (W(t)\phi,\phi)=\int_0^t\varkappa  \phi(\ell+x-t)\overline{\phi(x)}dx +\int_t^\ell  \phi(x-t)\overline{\phi(x)}dx=I+J,
 $$
where
 $$
 I=\varkappa \int_0^t \phi(\ell+x-t)\overline{\phi(x)}dx
 $$
and
 $$J=\int_t^\ell  \phi(x-t)\overline{\phi(x)}dx.
 $$
Since $\phi$ is a continuous function, we get
 $$
 I=t \varkappa \phi(\ell)\overline{\phi(0)}+o(t)\quad \text{as }\quad t\downarrow 0.
 $$
 Moreover, the membership $\phi\in W_2^1((0,\ell))$ ensures the representation
   $$
 \phi(x-t)=\phi(x)-t\phi'(x)+\eta_t(x),\quad \text{for a.e. }\quad x\in [\ell-t, \ell]
 $$
 where
$\phi'$ denotes the generalized derivative of $\phi$ and
$$
\|\eta_t\|_{L^2((\ell-t, \ell))}=o(t)  \quad \text{as}\quad t\downarrow 0.
$$

Therefore,
 \begin{align*}
 J&=\int_t^\ell  \phi(x)\overline{\phi(x)}dx-t\int_t^\ell  \phi(x)\overline{\phi'(x)}dx+o(t)
 \\&=\int_0^\ell  \phi(x)\overline{\phi(x)}dx-\int_0^t  \phi(x)\overline{\phi(x)}dx-t\int_0^\ell  \phi(x)\overline{\phi'(x)}dx+o(t)
 \\&=\int_0^\ell  \phi(x)\overline{\phi(x)}dx-t|\phi(0)|^2-t\int_0^\ell  \phi(x)\overline{\phi'(x)}dx+o(t)
 \\&=1-t\left [ |\phi(0)|^2+\int_0^\ell  \phi(x)\overline{\phi'(x)}dx\right ]+o(t)\quad \text{as}\quad t\downarrow 0.
\end{align*}
Here we have used that $\phi$ is a state, that is,
$$
\|\phi\|^2=\int_0^\ell  \phi(x)\overline{\phi(x)}dx=1.
$$

Combining the asymptotic representations for $I$ and $J$ we get
$$
I+J=1-t\left [
|\phi(0)|^2+\int_0^\ell  \phi(x)\overline{\phi'(x)}dx -\varkappa \phi(\ell)\overline{\phi(0)}\right ]+o(t) \quad \text{as}\quad t\downarrow 0.
$$

Since
$$
\Re\left (  \int_0^\ell  \phi(x)\overline{\phi(x)}dx\right )=\frac{|\phi(\ell)|^2-|\phi(0)|^2}{2},
$$
we have
\begin{align*}
\Re (e^{i t\widehat H_\varkappa }\phi, \phi)&=\Re (W(t)\phi,\phi)=\Re(I+J)
\\&=1-\frac12
 \tau t +o(t)\quad \text{as}\quad t\downarrow 0.
\end{align*}
Here
\begin{align*}
\tau&=|\phi(0)|^2+|\phi(\ell)|^2 -2 \Re \left (\varkappa \phi(\ell)\overline{\phi(0)}\right )
\\&=|\phi(0)-\varkappa \phi(\ell)|^2+(1-|\varkappa|^2)|\phi(\ell)|^2,
\end{align*}
which proves \eqref{obschee}.

 \end{proof}

\begin{remark}
Notice that Theorems \ref{quanmon} and \ref{quanmondis} (with $c=\hbar =1$) are particular cases of the obtained result.
Indeed, if $|\varkappa|=1$ and therefore  the operator $\widehat H_\varkappa$ is self-adjoint, then
the expression for the decay rate $\tau$ in \eqref{obscheetau}
simplifies to
$$
\tau=|\phi(0)-\varkappa\phi(\ell)|^2=|\Delta_\varkappa \phi|^2
$$
(cf. Theorem \ref{quanmon}  ($\varkappa=1$)).

On the other hand, if
$$\phi\in\Dom(\widehat H_\varkappa),
$$
we have  that
$
\phi(0)=\varkappa \phi(\ell)
$ and therefore the general expression \eqref{obscheetau}
for the  decay rate  reduces to
\begin{equation}\label{decrate}
\tau=(1-|\varkappa|^2)|\phi(\ell)|^2=|\phi(\ell)|^2-|\phi(0)|^2=\Delta |\phi|^2
\end{equation}
(cf. Theorem \ref{quanmondis}).

\end{remark}

  Applying Corollary \ref{smes0} and taking into account that
   the self-adjoint dilation of the dissipative operator $\widehat H_\varkappa$  is now explicitly available, see Section 21,
  in view of Theorem \ref{general},
  we arrive at the following two results.  It is convenient to threat  the case $\varkappa\ne 0$ and the exceptional case $\varkappa=0$ separately.

 \begin{theorem}$(\varkappa\ne 0)$\label{classmoncor}
  Suppose that $\bbY$ is the  metric graph  in Case $(iii)$,
 $$\bbY=(-\infty,0)\sqcup(0,\ell)\sqcup(0,\infty).$$
 Given  $0<|\varkappa|<1$,
  in the Hilbert space $\cH=L^2(\bbY)$  denote by
 $ \bbH$ the differentiation operator
    \begin{equation}\label{}
 \bbH=ic\hbar \frac{d}{dx}
  \end{equation}
defined on the domain of functions satisfying the boundary conditions
$$
\begin{pmatrix}
f_\infty(0+)\\
f_\ell(0)
\end{pmatrix}=\begin{pmatrix}
|\varkappa|&- \sqrt{1-|\varkappa|^2} \frac{\varkappa}{|\varkappa|}\\
 \sqrt{1-|\varkappa|^2}&\varkappa
\end{pmatrix}
 \begin{pmatrix}
f_\infty(0-)\\
f_\ell(\ell)
\end{pmatrix}.
$$
Suppose that  $\widehat \phi=\begin{pmatrix}0\\\phi
\end{pmatrix}\in \fH$ is a state $(\|\phi \|=1)$ such that
\begin{equation}\label{groten}
 \phi\in W_2^1((0, \ell))\subset L^2(\bbY).
\end{equation}

 Then
 \begin{equation}\label{clastau2}
\lim_{n\to\infty}|(e^{-it/(\hbar n)\bbH }\widehat \phi,\widehat  \phi)|^{2n}=e^{-\tau |t|},\quad t\in \bbR,
 \end{equation}
 where
 \begin{equation}\label{obscheetaudil}
\tau=
c|\phi(0)-\varkappa \phi(\ell)|^2+c(1-|\varkappa|^2)|\phi(\ell)|^2.
\end{equation}
  \end{theorem}


 \begin{theorem}$(\varkappa= 0)$\label{kurioz}
  Suppose that $\bbY$ is the  metric graph  $\bbY=(-\infty,0)\sqcup(0,\ell)\sqcup(\ell,\infty)$ in Case (i).   In the Hilbert space $\cH=L^2(\bbY)$  denote by
 $ \bbH$ the differentiation operator
    \begin{equation}\label{}
 \bbH=ic\hbar \frac{d}{dx}
  \end{equation}
defined on $W_2^1(\bbY)$.

Suppose that  $\phi\in W_2^1((0, \ell))\subset L^2(\bbY)$ and  $\|\phi \|=1$.

 Then
 \begin{equation}\label{QEDD}
\lim_{n\to\infty}|(e^{-it/(\hbar n)\bbH }\phi, \phi)|^{2n}=e^{-\tau |t|},\quad t\in \bbR,
 \end{equation}
 where
  \begin{equation}\label{ibrag}
\tau=
c(|\phi(0)|^2+|\phi(\ell)|^2).
\end{equation}

  \end{theorem}

Below  provide an independent (illustrative) proof of Theorem \ref{kurioz} based on a direct application of the Gnedenko-Kolmorogov limit theorem.

\begin{proof}
  Let $\widehat\phi $ denote the Fourier transform of the state $\phi$,
  $$
\widehat   \phi(\lambda)=\frac{1}{\sqrt{2\pi}}\int_{-\infty}^\infty e^{-i\lambda x} \phi(x) dx=\frac{1}{\sqrt{2\pi}}\int_{0}^\ell e^{-i\lambda x} \phi(x)dx.
  $$
Under the hypotheses of $\phi\in W_2^1((0, \ell))$, we integrate by parts  and obtain the asymptotic representation
$$\widehat   \phi(\lambda)=\frac{i}{\sqrt{2\pi}}\frac{\phi(\ell) e^{-i\ell\lambda}-\phi(0)}{\lambda}
+o\left (\frac{1}{|\lambda|}\right )\quad \text{as }\quad |\lambda|\to \infty.$$

We have
\begin{align*}
\int_\lambda^\infty |\widehat \phi(s)|^2ds&=\frac{1}{2\pi}\int_{\lambda}^\infty \frac{|\phi(\ell)|^2+|\phi(0)|^2}{s^2}ds
\\&-2 \Re\left [\phi(\ell)\overline{\phi(0)}\int_\lambda^\infty e^{-i\ell s}\frac{ds}{s^2}
\right ]+o\left (\frac{1}{|\lambda|}\right )
\\&=\frac{|\phi(\ell)|^2+|\phi(0)|^2}{2\pi }\frac{1}{\lambda}+o\left (\frac{1}{|\lambda|}\right )
\quad \text{as}\quad \lambda\to \infty.\end{align*}
In a completely similar way  one shows that
$$
\int^{\lambda}_{-\infty} |\widehat \phi(s)|^2ds=\frac{|\phi(\ell)|^2+|\phi(0)|^2}{2\pi }\frac{1}{|\lambda|}+o\left (\frac{1}{|\lambda|}\right )
\quad \text{as}\quad \lambda\to - \infty.
$$
Therefore, the distribution function $F(x)$ of the (absolutely continuous)  probability measure $|\widehat \phi(x)|^2 dx$ satisfies the hypotheses \eqref{as+} and \eqref{as-} of Theorem \ref{stthm} in Appendix H
with $\alpha=1$, $h(x)=1$ and
$$
c_1=c_2=\frac{|\phi(\ell)|^2+|\phi(0)|^2}{2\pi}\cdot \frac{\pi}{2}=\frac{ |\phi(\ell)|^2+|\phi(0)|^2}{4}.
$$
Therefore, the law $F(x)$ belongs to the normal domain of attraction of the symmetric $1$-stable law the characteristic function of which is given by
$$
e^{-\frac12|(|\phi(\ell)|^2+|\phi(0)|^2)|t|}.
$$
In particular,
$$
\lim_{n\to \infty}|\Phi(t/n)|^{2n}=e^{-|(|\phi(\ell)|^2+|\phi(0)|^2)|t|},
$$
where $$\Phi(t)=\int_\bbR e^{itx}dF(x)$$ is the characteristic function of the probability distribution $|\widehat \phi (x)|^2 dx$.
Since
$$
(e^{-it/\hbar H}\phi,\phi)=\Phi(-ct),
$$
we conclude that
$$\lim_{n\to \infty} |(e^{-it/(\hbar n) H} \phi,  \phi)|^{2n}=\lim_{n\to \infty}|\Phi(-ct/n)|^{2n}=e^{-\tau  |t|},
$$
where $\tau$ is given by \eqref{ibrag}.
\end{proof}


The main idea of the proof (with appropriate minimal adjustments) can be used to obtain  the following more general result.

Suppose that the state $\phi$ is a piecewise continuous function with discontinuity
points $a_1<a_2<\dots, a_N$ such that
$$
\phi \in W_2^1((-\infty, a_1))\oplus W_2^1((a_1,a_2))\oplus \dots \oplus W_2^1((a_{N-1},a_N))
\oplus W_2^1((a_N, \infty)).$$

 Then
 $$
 \lim_{n\to \infty} |(e^{-it/(\hbar n) H} \phi,  \phi)|^{2n}=e^{-\tau |t|},
 $$
 where
 $$
 \tau=c\sum_{k=1}^N|\Delta \phi(a_k)|^2
 $$
 and
 $$
 \Delta \phi(a_k)=\phi (a_k+0)-\phi(a_k-0), \quad k=1, 2, \dots, N,
 $$
 is the jump of the piecewise continuous representative
 of the state $\phi$ at the point $a_k$.

\subsection{Discussion}

The decay law \eqref{decdec} shows that the decay rate  \eqref{obscheetau}  splits
into   two terms,
$$\tau=\tau_{\text{excl}}+\tau_{\text{inter}},
$$
where
\begin{equation}\label{101}\tau_{\text{excl}}=c(1-|\varkappa|^2)
 |\phi(
\ell)|^2
 \end{equation}
 and
 \begin{equation}\label{102} \tau_{\text{inter}}=c|\phi(0)-\varkappa \phi(\ell)|^2.
  \end{equation}

 Recall that the configuration space of the open quantum system referred to in Theorem \ref{general}
  is the  ring $\mathbb{S}$ obtained from the interval $[0,\ell]$ by identifying its end-points. As the result of gluing the ends of the interval,
  a ``point defect'' occurs which can be
   perceived as  a kind of membrane which can be characterized by a quantum gate coefficient $\varkappa$,
    the amplitude that the particle  goes through the membrane. Arguing as a pure probabilist, one  can conclude  that the particle penetrates through the membrane with probability
   $| \varkappa \phi(\ell)|^2$ and therefore with probability proportional to  $(1-|\varkappa|^2)
 |\phi(
\ell)|^2$ it should be emitted. In other words,  the law of probabilities is
applicable
\begin{equation}\label{103}
|\phi(\ell)|^2=|\phi_{\text{em}}^{\text{excl}}|^2+|\varkappa|^2|\phi(\ell)|^2.
 \end{equation}

 However,  a secondary emission mechanism is available:  the amplitude $\varkappa \phi(\ell)$  that the particle can be found to the right from the membrane
 interferes with the amplitude $\phi(0)$ to stay on the ring $\mathbb{S}$. So  composition law for amplitudes
\begin{equation}\label{104}
\varkappa \phi(\ell)=\phi_{\text{em}}^{\text{inter}}+\phi(0)
 \end{equation}
should be take into account.

It remains to recall that the emission amplitudes and the corresponding decay constants are related as
\begin{equation}\label{105}
\tau_{\text{excl}}=c|\phi_{\text{em}}^{\text{excl}}|^2
 \end{equation}
and
\begin{equation}\label{106}
\tau_{\text{inter}}=c|\phi_{\text{em}}^{\text{inter}}|^2
 \end{equation}
 and then \eqref{101} and  \eqref{102} follow from  \eqref{103},  \eqref{105}  and  \eqref{104}, \eqref{106}, respectively.

 Summarizing, we arrive at the following descriptive understanding of the decay processes under continuous monitoring.
 A   $\lamed\gimel$-particle moves along the ring from $0$ to $\ell$ as a particle, hits the membrane and with some probability is emitted following the classical scheme  \eqref{classdecr} discussed in Section 20. After the collision with the membrane, the particle transforms into a wave, interferes with itself and experiences the secondary emission following the radiation szenario described in Section 19.

 The suggested interpretation allows one to enhance  the status of
informal reasoning that led to the decay law \eqref{nuivot}:
set
$$
\varkappa =\frac{\phi(0)}{\phi(\ell)}
$$
and observe that
 $$\tau_{\text{inter}}=|\phi(0)-\varkappa \phi(\ell)|^2=0$$
 whenever  the initial state
$\phi$ belongs to the domain
of the dissipative operator $ \widehat  H_\varkappa$.
 This is what  was implicitly assumed in the derivation of the law  \eqref{nuivot}.
Having  these remarks in mind, one can consider that derivation being an informal retelling  of the rigorous time-dependent proof of Theorem \ref{general} in the case where the initial state
$\phi$ belongs to the domain
of $ \widehat  H_\varkappa$ .

 \subsection{Random Phase method}

The traditional way of deriving the laws of probabilities \eqref{problaw} from  the law of amplitudes \eqref{amplaw1}, \eqref{amplaw2} is based on the hypothesis that
``the performance of the corresponding experiment will necessarily alter the phase''
 \cite{Hei}  of the wave function by an unknown amount which eventually, after averaging, yields the law of probabilities \eqref{problaw}.

 We suggest the following mnemonic rule for heuristic derivation of the decay law
when  dealing with the open system referred to in Theorem \ref{general}.

 Start with the Kirchhoff rule for the amplitudes (see \eqref{amplaw})
 $$
 \phi_{\text{em}}=\phi(\ell)-\phi(0)
 $$
and
rewrite the rule in the following equivalent form
$$
 \phi_{\text{em}}+\varkappa \phi(\ell)=\phi(\ell)+(\varkappa \phi(\ell)-\phi(0)).
 $$
The performance of the corresponding experiment assumes counting incoming and emitted particles by observing readings of the corresponding detectors.

Rewrite as
$$
 \Theta_{\text{em}}\phi_{\text{em}}+\varkappa \phi(\ell)=\Theta_\ell\phi(\ell)+(\varkappa \phi(\ell)-\phi(0)),
 $$
 where $ \Theta_{\text{em}}$ and $\Theta_\ell$ are unimodular independent random variables
with zero mean.
After  the corresponding averaging we get
$$
\EE |\Theta_{\text{em}}\phi_{\text{em}}+\varkappa \phi(\ell)|^2=\EE|\Theta_\ell\phi(\ell)+(\varkappa \phi(\ell)-\phi(0))|^2,
 $$
where $\EE$ denotes the corresponding  mathematical expectation. Since  $ \Theta_{\text{em}}$ and $\Theta_\ell$ are independent with zero mean,
we arrive
at the law
$$
 |\phi_{\text{em}}|^2+|\varkappa|^2 \phi(\ell)|^2=|\phi(\ell)|^2+|\varkappa \phi(\ell)-\phi(0)|^2.
$$
Hence
$$
|\phi_{\text{em}}|^2=(1-|\varkappa|^2)| \phi(\ell)|^2+|\varkappa \phi(\ell)-\phi(0)|^2,
$$
which coincides with the decay rate $\tau$ given by \eqref{obscheetau}.

\section{Operator coupling  limit theorems}\label{limtem}

We say that  a  dissipative operator $\widehat A$ belongs to class $\fD(\cH)$ if
 $\widehat A$  is a   quasi-selfadjoint  extension of  a symmetric operators $\dot A$ with deficiency indices
$(1,1)$ (see Appendix G).  Recall that in this case the symmetric operator $\dot A$ can be recovered from $\widehat A$ as
\begin{equation}\label{pravilo}
\dot A=\widehat A|_{\dom (\widehat A)\cap \dom ((\widehat A))^*}.
\end{equation}

\begin{definition}

 We will say that two dissipative operators $\widehat A_1\in \fD(\cH_1)$  and $\widehat A_2\in \fD(\cH_2)$ coincide in distribution (are equally distributed),
in writing
$$
\widehat A_{1}\overset{\substack{\text{ d}}}{=}\widehat A_{2}
$$
if there  are appropriate  self-adjoint reference operators    $A_{1,2}$  such that the characteristic functions of the triples $(\dot A_1,\widehat A_1, A_1) $ and $(\dot A_2,\widehat A_2, A_2)  $ coincide. That is,
 $$
S_{(\dot A_1,\widehat A_1, A_1)}(z)=S_{(\dot A_2, \widehat A_2, A_2)}(z), \quad z\in \bbC_+.
$$
Here the symmetric operators $\dot A_{1,2} $ $(\dot A_{1,2}\subset  A_{1,2})$ are defined by \eqref{pravilo}.
\end{definition}

Recall that an  operator $\widehat A$ in the Hilbert space $\cH$  is completely non-self-adjoint if $\cH $ cannot be represented  in the form of an orthogonal sum  of two subspaces $\cH_1$ and $\cH_0\ne 0$ with the following properties, see, e.g.,  \cite{Brod}:
\begin{itemize}
\item[1)] $\cH_1$ and $\cH_0$ are invariant  relative to $\widehat A$;
\item[2)]
$\widehat A $
induces in $\cH_0$    a self-adjoint operator.
\end{itemize}

It is  worth mentioning that if  completely non-selfadjoint operators $\widehat A_1\in \fD(\cH_1)$ and $\widehat A_2\in \fD(\cH_2)$ are equally distributed, then they are unitarily equivalent. Indeed, in this case the corresponding symmetric operators $\dot A_{1,2} $ are prime and then the claim follows from  the Uniqueness Theorem \ref{unitar} in Appendix C.

\begin{definition} We say that a sequence of dissipative operators $\{\widehat A_n\}_{n=1}^\infty$, $\widehat A_n\in  \fD(\cH_n)$ converges in distribution   to a dissipative operator $\widehat A \in  \fD(\cH)$,
in writing,
$$
\lim_{n\to \infty} \widehat A \overset{\substack{\text{ d}}}{=}\widehat A,
$$ if one can find a sequence of reference self-adjoint operators $A_n=A^*_n$ in $\cH_n$ and $A=A^*$ in $\cH$ such that
the corresponding characteristic functions converge pointwise. That is,
$$
\lim_{n\to \infty} S_{(\dot A_n, \widehat A_n, A_n)}(z)=S_{(\dot A,\widehat A, A)}(z), \quad z\in \bbC_+.
$$
\end{definition}

As long as convergence in distribution for a sequence of operators is introduced, the next natural question is
to understand the behavior of the $n$-fold coupling of an operator with itself as $n\to \infty$.

The following limit theorem sheds some light on  the typical behavior of   $n$-fold couplings of operators from the class $ \fD(\cH) $.

 \begin{theorem}\label{limthm}
 Suppose that $\widehat A$ is a maximal dissipative operator from the class  $ \fD(\cH) $. Assume  that  the   characteristic function $S(z)=S_{(\dot A, \widehat A,A)}(z)$ associated with the triple  $(\dot A, \widehat A,A)$ for  some  (and therefore with all) reference self-adjoint extension $  A$  admits an analytic continuation to the lower  half-plane in a neighborhood of a point $\mu\in \bbR$.

\begin{itemize}
\item[(i)]
If $|S(\mu)|=1$, then
 $$
 \lim_{n\to \infty} n \, \underbrace{(\widehat A-\mu I)\uplus (\widehat A-\mu I)\uplus \dots\uplus (\widehat A-\mu I)}_{n \text{ times}} \overset{\substack{\text{ d}}}{=} \widehat D_{II}(\ell),
 $$
 where   $\widehat D_{II}(\ell)$ is  the dissipative differentiation operator on  the  finite interval $[0, \ell]$
with
 $$
\dom ( \widehat D_{II}(\ell))=\{f\in W_2^1((0, \ell))\, |\, f(0)=0\}
 $$
 and
\begin{equation}\label{lgiven}
\ell=\frac1i\frac{d}{d\lambda} \log S_{(\widehat A, A)}(\mu)>0.
 \end{equation}
\item[(ii)]
If
 $|S(\mu+i0)|<1$, then
 \begin{equation}\label{clcl}
  \lim_{n\to \infty} n \, \underbrace{(\widehat A-\mu I)\uplus (\widehat A-\mu I)\uplus \dots\uplus (\widehat A-\mu I)}_{n \text{ times}} \overset{\substack{\text{ d}}}{=} \widehat D_{I}(0),
 \end{equation}
where   $\widehat D_{I}(0)$ is  the dissipative differentiation operator on  the  real axis
on
 $$
\dom ( \widehat D_{I}(0))=\{f\in W_2^1((-\infty, 0))\oplus W_2^1(( 0, \infty ))\, |\, f(0+)=0\}.
 $$
\end{itemize}

 \end{theorem}
 \begin{proof}
(i).
 Introduce the $n$-fold coupling
 $$
 \widehat B_n=  \biguplus_{k=1}^n n (\widehat A-\mu I).$$
 Combining the invariance principle   (see Theorem   \ref{coinvv} in Appendix F)  and  the Multiplication Theorem \ref{opcoup} in Appendix G,
one ensures the existence of the corresponding reference operators $B_n$ and
 unimodular factors $\Theta_n$ such that
\begin{align*}
 S_{(\dot B_n, \widehat B_n, B_n)}(z)&=\Theta_n \left (S_{(\dot A, \widehat A, A)}\left (\frac zn+\mu \right )\right )^n
 \\&=\Theta_n \left (S_{(\dot A, \widehat A, A)}(\mu)+ i\frac{d}{d\lambda} S_{(\dot A, \widehat A, A)}(\mu)\frac zn+o(n^{-1})\right )^n\quad \text{as} \quad n\to \infty.
 \end{align*}
 By choosing a possibly different sequence of reference operators $B_n'$ one obtains that
 $$
  S_{(\dot B_n, \widehat B_n, B_n')}(z)=\left (1+i\ell  \frac zn  + o(n^{-1})\right )^n \quad \text{ as} \quad n\to \infty,
 $$
 where $\ell $ is given by \eqref{lgiven}.

 Therefore, $$
  \lim_{n\to \infty}S_{(\dot B_n, \widehat B_n, B_n')}(z)=e^{i\ell z},  \quad z\in \bbC_+.
  $$
 To  complete the proof it remains to notice that
 $$
 S_{(\dot D_{II}, \widehat D_{II}(\ell), D_{II})}(z)=e^{i\ell z},\quad z\in \bbC_+,
 $$
 where  $(\dot D_{II}, \widehat D_{II}(\ell), D_{II})$ is the triple  $(\dot D, \widehat D(\ell), D)$ referred to in Lemma \ref{key} in Case (ii).

(ii).  From Theorem \ref{coinvv} in Appendix F it follows that there exists a sequence of unimodular  factors $|\Theta_n|=1$ such that
 $$
 S_{(n (\dot A-\mu I), (n (\widehat A-\mu I), n (A-\mu I))}(z)= \Theta_n S_{(\dot A,  \widehat A, A)}\left (\frac zn+\mu \right ).
 $$
 Since
 $$
 \lim_{n\to \infty}  S_{(\dot A,  \widehat A, A)}\left (\frac zn+\mu\right )=S(\mu+i0), \quad z\in \bbC_+,
 $$
 one can always choose reference operators $ A_n'$ such that
\begin{equation}\label{vykk}
\lim_{n\to \infty}S_{((n (\dot A-\mu I), n (\widehat A-\mu I), n (A'-\mu I))}(z)=|S(\mu+i0)|=k.
 \end{equation}
  Since $k<1$,  we have that
 \begin{equation}\label{law}
  \lim_{n\to \infty} \left ( S_{(n (\dot A-\mu I), (n (\widehat A-\mu I), n (A-\mu I))}(z)\right )^n=0 \quad \text{for all }\quad z\in \bbC_+.
 \end{equation}
However,  by Lemma \ref{key}, we have that
 $$
 S_{(\dot D_I(0), \widehat D_I(0), D_I(0))}(z)=0, \quad z\in \bbC_+,
 $$
 and hence \eqref{law} means that the sequence of operators
 $$\widehat B_n=\underbrace{(\widehat A-\mu I)\uplus (\widehat A-\mu I)\uplus \dots\uplus (\widehat A-\mu I)}_{n \text{ times}} $$
 converges to $\widehat  D_I(0) $ in  distribution.

 \end{proof}

\begin{remark}\label{limrem} A closer look at the proof shows that (ii) is a simple consequence of the limit relation   (combine  \eqref{vykk} and Lemma \ref{key})
 \begin{equation}\label{elele}
 \lim_{n\to \infty} (n (\widehat A-\mu) I) \overset{\substack{ \text{d}}}{=} \widehat D_I(k),
 \end{equation}
 where
 $
 k=|S(\mu+i0)|$ and  $\widehat D_I(k)$
 is  the dissipative differentiation operator on  the  real axis
 on
 $$
\dom ( \widehat D_{I}(0))=\{f\in W_2^1((-\infty, 0))\oplus W_2^1(( 0, \infty ))\, |\, f(0+)=kf(0-)\}.
 $$
\end{remark}

 \begin{remark}
 Borrowing the terminology from probability theory, we will  say that an operator $\widehat A$ has a (strictly) stable distribution if
for arbitrary constants $c_1$ and $c_2$ there exists a constant $c$ such that
$$
c_1\widehat A_1\uplus  c_2\widehat A_2 \overset{\substack{\text{ d}}}{=}c\widehat A,
$$
whenever
$$
\widehat A_1  \overset{\substack{\text{ d}}}{=}\widehat A_2 \overset{\substack{\text{ d}}}{=}\widehat A.
$$

  In particular,  the operator $\widehat  D_{II}(\ell)$ has a stable distribution since
\begin{equation}\label{coconder}
c_1 \widehat D_{II}(\ell)\uplus c_2 \widehat D_{II}(\ell)\overset{\substack{\text{ d}}}{=}c\widehat D_{II}(\ell)
\end{equation}
with
\begin{equation}\label{conder}
\frac{1}{c}=\frac{1}{c_1}+\frac{1}{c_2}
\end{equation}

 The stability laws \eqref{coconder} and \eqref{conder} in particular imply
   \begin{equation}\label{rel44}
n \cdot \underbrace{\widehat D_{II}(\ell)\uplus \widehat D_{II}(\ell) \uplus \dots \uplus \widehat D_{II}(\ell)}_{n \, \text{times}} \overset{\substack{\text{ d }}}{=}\widehat D_{II}(\ell)
 \end{equation}
and
  \begin{equation}\label{rel45}
\frac 1n\cdot  \underbrace{V_\ell\uplus V_\ell\uplus \dots \uplus V_\ell}_{n \, \text{times}} \overset{\substack{\text{ d }}}{=}V_\ell,
 \end{equation}
 which is in  a good agreement with limit Theorems \ref{limthm} and Theorem  \ref{limthmbdd}  (see below).

Although the distribution laws for the operators $\widehat D_{I}(k)$, $0<k<1$, and $\widehat D_{III}(k,\ell)$ are not stable, nevertheless corresponding laws are infinitely divisible in the sense that
$$
n\cdot  \underbrace{\widehat D_{I}(k^{1/n})\uplus \widehat D_{I}(k^{1/n}) \uplus \dots \uplus  \widehat D_{I}(k^{1/n})}_{n \, \text{times}} \overset{\substack{\text{ d }}}{=}\widehat D_{I}(k),
$$
where  $\widehat D_I(k)$
 is  the dissipative differentiation operator on  the  real axis
 on
 $$
\dom ( \widehat D_{I}(k))=\{f\in W_2^1((-\infty, 0))\oplus W_2^1(( 0, \infty ))\, |\, f(0+)=kf(0-)\},
 $$
and
$$
n \cdot \underbrace{\widehat D_{III}\left (k^{1/n}, n^{-1}\ell \right ) \uplus \dots \uplus  \widehat D_{III}\left (k^{1/n}, n^{-1}\ell \right ) }_{n \, \text{times}} \overset{\substack{\text{ d }}}{=}\widehat D_{III}(k, \ell),
$$
with  $ \widehat D_{III}(k, \ell)$  the dissipative differentiation operator on the metric graph
 $$\bbY=(-\infty, 0]\sqcup [0, \infty) \sqcup[0, \ell]$$
 on
 $$
\Dom( \widehat D_{III}(k, \ell))=\left \{ f_\infty\oplus f_\ell\in  W_2^1(\bbY)\, |\,  \begin{cases}
f_\infty(0+)&=k f_\infty(0-)
\\
f_\ell(0+)&=\sqrt{1-k^2} f_\infty(0-)\\
\end{cases}
\right \}.
$$

It is also worth mentioning that by  Lemma \ref{key} and the Multiplication Theorem  \ref{opcoup} in Appendix G,
 the basic differentiation  operators $\widehat D_I(k)$, $ \widehat D_{II}(\ell)$ and   $\widehat D_{III}(k,\ell)$ have the  ``addition" laws with respect to the operator coupling
\begin{equation}\label{ppz1}
 \widehat  D_I(k)\uplus \widehat D_I(k')\overset{\substack{\text{ d}}}{=} \widehat  D_I(kk'),
\end{equation}
 while
  \begin{equation}\label{ppz2}
 \widehat  D_{II}(\ell)\uplus \widehat D_{II}(\ell')\overset{\substack{\text{ d}}}{=} \widehat  D_{II}(\ell+\ell').
  \end{equation}
  One also gets the following spectral synthesis rule
  \begin{equation}\label{ppz3}
 \widehat  D_{III}(k,\ell)\uplus \widehat D_{III}(k',\ell')\overset{\substack{\text{ d}}}{=} \widehat  D_{III}(kk',\ell+\ell').
    \end{equation}
 Notice that  in the  first two cases one can even state that the equalities in distribution imply the unitary equivalence of the corresponding operators.

\end{remark}

 The Limit Theorem \ref{limthm} (i)  has its natural counterpart for the class $\fD_0(\cH)$ (see Appendix G for the definition of the class).

 \begin{theorem}\label{limthmbdd}
 Suppose that $\widehat A\in  \fD_0(\cH)$ is a maximal  bounded dissipative operator. Then
 \begin{equation}\label{LLN}
 \lim_{n\to \infty}\frac1n \underbrace{ \widehat A \uplus   \widehat A\uplus \dots\uplus   \widehat A}_{n \text{ times} }\overset{\substack{\text{ d }}}{=}V_\ell,
\end{equation}
 where $V_\ell$ is the Volterra operator in $L^2((0, \ell))$
 $$
 (V_\ell f)(x)=i\int_0^x f(t)dt, \quad  f\in L^2((0, \ell)),
 $$
 and
 $$
 \ell=2\tr\left (\Im (\widehat A)\right ) .
 $$

 \end{theorem}
 \begin{proof}
 Since
 $$
\left (\frac{1}{2i}( V_\ell-V_\ell^*)f\right )(x)=\frac12\int_0^\ell f(t)dt \cdot \chi_{[0, \ell]}(x), \quad f\in L^2((0,\ell)),
 $$where
$$
\chi_{[0,\ell]}(x)=1 \quad \text{ on } \quad [0,  \ell],
$$
the imaginary part of $V_\ell$ is one-dimensional and therefore  $V_\ell\in  \fD_0( L^2((0,\ell))$.

Next, we evaluate the characteristic function  of the (bounded)  operator  $V_\ell$ (see \eqref{chfbd} in Appendix A).
We have
$$
S_{V_\ell}(z)=1+i(V_\ell^*-zI)^{-1}\chi_{[0,\ell]},\chi_{[0,\ell]}).
$$

Since
$$
\left ((V_\ell^*-zI)^{-1}\chi_{[0,\ell]}\right )(x)=\frac{(-1)}{z}e^{\frac{i}{z}(x-\ell)}, \quad x\in [0, \ell],
$$
we finally get that
$$
S_{V_\ell}(z)=1-\frac{1}{iz}\int_0^\ell e^{\frac{i}{z}(x-\ell)}dx=\exp\left (-i\frac{\ell}{z}\right ),\quad z\in \bbC_+.
$$

 Next,
 since $\widehat A\in  \fD_0(\cH)$, we have the representation
 $$
 \widehat A=A+i(\cdot g)g
 $$
  for some $g\in \cH$, $g\ne 0$.
 Therefore, the characteristic function $S_{\widehat A} $ has the following asymptotics  $$
  S_{\widehat A}(z)= 1+2i  ((\widehat A)^*-zI)^{-1}g, g)=1- 2i \frac{\tr (\Im (\widehat A))}{z}+o\left (\frac1z\right )\quad \text{as} \quad z\to \infty.
  $$

Set $$B_n=\frac{  \widehat A \uplus   \widehat A\uplus \dots\uplus   \widehat A}{n}.
 $$
Using  the invariance principle, Theorem \ref{coinvv} in Appendix F,  and the Multiplication Theorem \ref{opcoupbdd}   in Appendix G, one concludes that
 \begin{align*}
 \lim_{n\to \infty}S_{B_n}(z)&=  \lim_{n\to \infty}S_{\widehat A}(nz)
 = \lim_{n\to \infty}\left (1- 2i \frac{\tr (\Im (\widehat A))}{nz}+o\left (\frac1n\right )\right )^n
 \\&=\exp\left (-\frac{i\ell}{z}\right ), \quad z\in \bbC_+.
 \end{align*}
 Therefore,   $ \lim_{n\to \infty}S_{B_n}(z)$ coincides with the characteristic function of the Volterra operator $V_\ell$
which completes the proof.
 \end{proof}

 \begin{remark}
 Notice that the addition laws with respect to operator coupling  \eqref{ppz1}-\eqref{ppz3} can be extended by the following rule
$$
V_\ell\uplus V_{\ell'}\overset{\substack{\text{ d}}}{=}V_{\ell+\ell'},
$$
which can be deduced from \eqref{ppz2}  applying the invariance principe combined with the  observation that
$$
V_\ell=-D_{II}(\ell)^{-1}.
$$
In particular,  the operator $V_\ell$ has a stable distribution:
$$
(c_1' V_\ell) \uplus (c_2' V_\ell)\overset{\substack{\text{ d}}}{=}c'V_\ell ,
$$
with
$$
c'=c_1'+c_2'.
$$
\end{remark}

We conclude this section by the discussion of the   limit distribution  {\it universality} of the real part of $n$-fold couplings  (as $n\to\infty$)  of bounded dissipative operators from the class $\fD_0(\cH)$.

 \begin{corollary}{\it
  Suppose that $\widehat A\in  \fD_0(\cH)$ is a maximal  bounded dissipative operator and
 $
 \ell=2\tr\left (\Im (\widehat A)\right )
 $.
Let
  $$
  B_n=\frac1n \underbrace{ \widehat A \uplus   \widehat A\uplus \dots\uplus   \widehat A}_{n \text{ times} }
  $$
  denote  the ``averaged'' $n$-fold coupling of the operator $\widehat A$ with itself.

 Let  $M_n(z)$  be the Weyl-Titschmarsh function of the self-adjoint operator $\Re (\widehat B_n)$,
$$
M_n(z)=\frac{1}{\tr (\Im ( \widehat B_n))} \tr \left ((\Re(\widehat B_n)-zI)^{-1} \Im( \widehat B_n)\right ).
$$
Suppose that  $\mu_n(d\lambda)$ is    the probability measure from the representation theorem
$$
M_n(z)=\int_\bbR\frac{d\mu_n(\lambda)}{\lambda-z}, \quad z\in \bbC_+.
$$

Then the sequence of the measure  $\{\mu_n(d\lambda)\}_{n=1}^{\infty}$ converges weakly to
 the  pure point probability measure $\mu(d\lambda)$ }given by
\begin{equation}\label{numi}
\mu (d\lambda) =\frac4{\pi^2}\sum_{k\in \bbZ} \frac{1}{(2k+1)^2}\delta_{z_k}(d\lambda),
\end{equation}
with the ``Dirac masses" $ \delta_{z_k} (d\lambda)$ at the points
 $$
 z_k=\frac1\ell \frac{2}{(k+\frac 12)\pi}, \quad k\in \bbZ.
  $$
 That is,
  $$\lim_{n\to \infty}\mu_n((-\infty, \lambda))=\mu((-\infty, \lambda))\
  $$
at every point of continuity of the function $F(\lambda)=\mu((-\infty, \lambda))$.

 \end{corollary}

\begin{proof}
From  Lemma \ref{lemma318} in Appendix A it follows that
$$
M_n(z)=\frac{1}{it_n}\frac{S_{\widehat B_n}(z)-1}{S_{\widehat B_n}(z)+1},
$$
where $S_{\widehat B_n}(z)$ is the characteristic function of $\widehat B_n$ and
$$
t_n=\tr(\Im (\widehat B_n))=\tr(\Im (\widehat A))=\frac\ell2.
$$
Applying Theorem \ref{limthmbdd} we get that $$
\lim_{n\to \infty}S_{\widehat B_n}(z)=\exp\left (-\frac{i\ell}{z}\right )
$$
and therefore
$$
\lim_{n\to \infty}M_n(z)=\frac{2}{i\ell}\frac{\exp\left (-\frac{i\ell}{z}\right )-1}{\exp\left (-\frac{i\ell}{z}\right )+1}=-\frac{2}{\ell}\tan \frac\ell2 z, \quad z\in \bbC_+.
$$
Since
$$ \tan(z)=-2\sum _{k=0}^{\infty }\left(\frac {z}{z^2-(k+\frac 12)^2\pi^2}\right )
$$
and hence
  $$
  -\tan\frac1z= 2\sum _{k=0}^{\infty }\left(\frac {z}{1-(k+\frac 12)\pi z) (1+(k+\frac 12)\pi z)}\right ),
  $$
  it is easy to see that
  $$
  -\frac{2}{\ell}\tan \frac\ell2 z=\int_\bbR \frac{d\mu(\lambda)}{\lambda-z}, \quad z\in \bbC_+,
  $$
  with $\mu (d\lambda)$ given by \eqref{numi}.
  Since $\mu_n$ and $\mu$ are probability measures and
  \begin{equation}\label{stst}
  \lim_{n\to \infty}M_n(z)=\lim_{n\to \infty}\int_\bbR \frac{d\mu_n(\lambda)}{\lambda-z}=\int_\bbR \frac{d\mu(\lambda)}{\lambda-z} , \quad z\in \bbC_+,
  \end{equation}
 the pointwise convergence of the Stieltjes transforms \eqref{stst}  ensures the weak convergence of  $\mu_n$ to   $\mu$  by an analog of the L\'evy continuity theorem \cite{hill}.

\end{proof}

\appendix

\section{The characteristic function for rank-one perturbations}

In this section we introduce a characteristic function of a maximal
dissipative operator $\widehat A$ in the case where the domains of the operator  and its adjoint coincide \cite{ABT,Brod, Kuzhel,Lv1,LJ,MT-S,PavDrog,Str60}, that is,
\begin{equation}\label{domm}
\Dom(\widehat A)=\Dom ((\widehat  A)^*).
\end{equation}

For instance, if the operator $\widehat A$ is bounded, condition \eqref{domm} holds automatically.

We will treat the simplest case where   $\widehat  A$ has a  rank-one
imaginary part of the form
$$
\Im ( \widehat  A)= \frac{1}{2i}(\widehat A-(\widehat A)^*)=tP,
$$
where $t >0$ and
 $P$ is   a rank-one
 orthogonal projection,
$$
P =(\cdot , g)g, \quad \|g\|=1.
$$

Denote by $A$  the real part of $\widehat A$,
$$ A=\Re(\widehat  A)=\frac12(\widehat A+(\widehat A)^*), \quad \Dom (A)=\Dom(\widehat A)=\Dom ((\widehat  A)^*),$$
so that
$$
\widehat A= A+it P.
$$

The resolvent of $\widehat A$ is described in  the following lemma.

\begin{lemma}[cf. \cite{Simon}]\label{resf}
Let   $\widehat  A$ be a maximal dissipative operator with   a rank-one
imaginary part,
 $$  \widehat A= A+it P,$$
  $A=A^*$, $P =(\cdot , g)g$, $\|g\|=1$, and $t>0$.

 Denote by
  $$
M(z)=(( A-zI)^{-1}g,g), \quad z\in \rho (A),
$$
  the $M$-function associated with the real part
 $\Re(\widehat A)$ of the operator $\widehat A$
and the unit vector $g$.

Then the following resolvent formula
\begin{equation}\label{resform}
(\widehat A -zI)^{-1}=(A -zI)^{-1}-
p(z)
(A -zI)^{-1}P( A -zI)^{-1},
\end{equation}
$$
z\in \rho(\widehat A)\cap  \rho(A),
$$holds, where
$$
p(z)=\frac{1}{M(z)+\frac{1}{it}}.
$$

In particular,
\begin{equation}\label{kreinaron}
((\widehat A -zI)^{-1}g, g)=\frac{M(z)}{1+it M(z)}.
\end{equation}

Moreover,
\begin{equation}\label{vkluch}
\spec (\widehat A)\subset \{z\, : 0\le \Im(z)\le t\}.
\end{equation}

\end{lemma}

\begin{proof}
To prove the resolvent formula \eqref{resform}, one observes that
\begin{equation}\label{poraby}
(\widehat A -zI)^{-1}=(A -zI)^{-1}-
it (\widehat A -zI)^{-1}P( A -zI)^{-1},
\end{equation} and hence
$$
(\widehat A -zI)^{-1}g=(A -zI)^{-1}g-
it (( A -zI)^{-1}g,g)(\widehat A -zI)^{-1}g,
$$
 which yields the representation
 $$
 (\widehat A -zI)^{-1}g=(1+it M(z))^{-1}(A -zI)^{-1}g.
 $$
 Substituting this equality  back to
 \eqref{poraby}, one obtains
 $$
 (\widehat A -zI)^{-1}=(A -zI)^{-1}-
\frac{it}{1+it M(z)} (A -zI)^{-1}P( A -zI)^{-1},
$$
which proves \eqref{resform}.

 Now, it is easy to see  that the non-real spectrum of $\widehat A$
coincides with those $z$  that satisfy the equation
$$
\frac{1}{it}+M(z)=0, \quad \Im (z) \ne 0.
$$
To complete the proof, we use the inequality
$$
|M(z)|\le \frac{1}{\Im (z)}, \quad z\in \bbC_+.
$$
Therefore,
$
\{z: \, \Im(z)>t\}\subset \rho(\widehat A),
$ which proves \eqref{vkluch}.
\end{proof}

\begin{theorem}\label{cycl} Let   $\widehat  A$ be a maximal dissipative operator with   a rank-one
imaginary part
$$  \widehat A= A+it P,$$ $A=A^*$, $P =(\cdot , g)g$, $\|g\|=1$, and $t>0$.
Then $\widehat A$  is completely non-self-adjoint if and only  if  the element $g$ is generating for
 the self-adjoint operator $A= \Re (\widehat A)$.
\end{theorem}

\begin{proof}
Introduce the subspace,
$$
\cH_g=\Span_{\delta \in \cB(\bbR)}\{ E_{A}(\delta)g\}, \quad\text{with }
\cB(\bbR)\text{ the Borel $\sigma$-algebra}.
$$

 {\it Only If Part}.  Suppose  that  $\widehat A$  is completely non-self-adjoint.
Assume that $g$ is   not a generating element  for the self-adjoint operator
 $ A=\Re(\widehat A)$.  Then the orthogonal complement $\cH_g$ reduces the real part
$ A$ and since
the element  $g$ is orthogonal to $\cH_g^\perp$ of $\cH$, one obtains that
$$
\widehat Ah=Ah \quad \text{for } h\in \cH_g^\perp.
$$
Since $\cH_g$ reduces $ A$, the subspace   $\cH_g^\perp$ reduces $\widehat A$ as well. Furthermore,
the part of $\widehat A$ on $\cH_g^\perp$ is a self-adjoint operator.
Therefore, $\widehat A$ is not completely non-self-adjoint. We get a contradiction.

 {\it If Part}.
Suppose that $g$ is    a generating element  for the self-adjoint operator
 $A=\Re(\widehat A)$. Assume that $\widehat A$ is not completely non-self-adjoint and let
$\cH_0$ be its reducing subspace such that the part of $\widehat A$
in this subspace is a self-adjoint operator.

Given $0\ne h\in \cH_0$,
using  the resolvent formula
\eqref{resform},
 one obtains that
\begin{align*}
((\widehat A -iyI)^{-1}h,h)&=(( A -ityI)^{-1}h,h)
\\& -
\frac{
((A -ityI)^{-1}h,g)(( -ityI)^{-1}g,h)}{\frac{1}{it}
+M( iy)},
\\|y|&>t, \quad y\in \bbR.
\end{align*}
Therefore,
\begin{equation}\label{prox}
f_1(y)=f_2(y)-\frac{f_3(y)}{f_4(y)}, \quad |y|>t,
\end{equation}
where the functions $f_k$, $k=1,2,3,4$
are given by
\begin{align*}
f_1(y)&=((\widehat A -iyI)^{-1}h,h),
\\f_2(y)&=(( A -ityI)^{-1}h,h),\\
f_3(y)&=(( A -ityI)^{-1}h,g)(( A -ityI)^{-1}g,h)),
\\f_4(y)&=\frac{1}{it}
+M( iy).
\end{align*}
One observes that
\begin{equation}\label{F44}
f_k(y)=\overline{f_k(-y)}, \quad y\in \bbR,\quad  |y|>t, \quad k=1,2,3.
\end{equation}
However,
\begin{equation}\label{F4}
f_4(y)\ne \overline{f_4(-y)}
\end{equation}
for
$$
\overline{f_4(-y)}=\overline{\frac{1}{it}
+M( -iy)}=\frac{(-1)}{it}
+M( iy)\ne \frac{1}{it}
+M( iy)=f_4(y).
$$
Inequality \eqref{F4} together with \eqref{F44} is inconsistent with \eqref{prox}, provided that  $f_3(y)\ne 0$.
Finally,  the fact that  the function
$f_3(y)
$  is not identically zero easily follows from the assumption that
 the element $g$ is generating  for $A$.

The obtained contradiction shows that there is no reducing subspace such that the part of $\widehat A $ in this subspace is self-adjoint. Therefore,
$\widehat A$ is completely non-self-adjoint.

\end{proof}

\begin{theorem}[cf. \cite{BL60}]\label{osnov}
Assume that   $\widehat  A$ is a maximal dissipative operator with   a rank-one
imaginary part, $  \widehat A= A+it P,$ $A=A^*$, $P =(\cdot , g)g$, $\|g\|=1$, and $t>0$.
Assume, in addition, that  $\widehat A$ is completely
non-self-adjoint.

Let  $\mu(d\lambda)$ be  the   probability   measure from the representation
$$
(( A-zI)^{-1}g,g)=\int_\bbR\frac{d\mu(\lambda)}{\lambda-z},
\quad z\in \bbC_+.
$$

Then the operator
$\widehat A$ is unitarily equivalent to the operator $\widehat \cB$ of the form
$$
(\widehat \cB f)(\lambda)=\lambda f(\lambda)+it(f,\mathbb{1}) \mathbb{1}(\lambda),
$$
$$
\Dom(\widehat \cB)=\left \{f\in \,L^2(\bbR;d\mu) \,\bigg | \, \int_\bbR \lambda^2
| f(\lambda)|^2d
\mu(\lambda)<\infty \right \},
$$ in the Hilbert space $L^2(\bbR; d\mu)$,
where
\begin{equation}\label{defee}
\mathbb{1}(\lambda)=1\quad \text{for $\mu$-a.e. } \lambda\in \bbR.
\end{equation}
\end{theorem}

\begin{proof}  By hypothesis $\widehat A$ is completely
non-self-adjoint and therefore, by Theorem
 \ref{cycl}, the real part  $ A$ has  simple spectrum and the vector $g$
 is   generating for $ A$.

It follows that  the spectral measures
$$\nu(d\lambda)=(E_{ A}(d\lambda)g,g)$$
and
$$\mu(d\lambda)=(E_{ \cB}(d\lambda)\mathbb{1},\mathbb{1}),$$ where $\cB$ is the self-adjoint operator of multiplication  by  independent variable in $L^2(\bbR;d\mu)$ and   $\mathbb{1}(\lambda)$ is given by \eqref{defee}, coincide.
Since  $ \cB$ has  simple spectrum and $\mathbb{1}(\lambda)$ is a generating  vector for the self-adjoint operator
$ \cB$,  the Spectral Theorem for self-adjoint operators
yields the existence of   a unitary operator  $\cU:\cH\to L^2(\bbR;d\mu)$ such that
$$
\cU A\cU^{-1}=\cB
\quad \text{ and} \quad \cU g=\mathbb{1}.
$$
Hence
$$
\cU\widehat  A\cU^{-1}=\widehat \cB,
$$
which completes the proof.
\end{proof}

Given  a non-self-adjoint dissipative
operator $\widehat A$ with a  rank-one  imaginary part,  $  \widehat A= A+it P,$ $A=A^*$, $P =(\cdot , g)g$, $\|g\|=1$, and $t>0$,
denote by $S(z)$ the characteristic function of operator $\widehat A$ following  \cite{Lv1}
\begin{equation}\label{chfbd}
S(z)=1+2 i t\left (((\widehat A)^*-zI)^{-1}g,g\right ), \quad z\in \bbC_+.
\end{equation}

\begin{lemma}\label{lemma318}
 Let   $\widehat  A$ be a maximal dissipative operator with   a rank-one
imaginary part,
 $  \widehat A= A+it P,$ $A=A^*$, $P =(\cdot , g)g$, $\|g\|=1$, and $t>0$.

Then the  characteristic function $S(z)$ of
 $\widehat A$ admits the representation
$$
S(z)=\frac{1+itM(z)}{1-itM(z)}, \quad z\in \bbC_+,
$$
where $M(z)$ is the $M$-function of $\Re(\widehat A)=A$
given by
$$
M(z)=(( A-zI)^{-1}g,g), \quad z\in \rho (A).
$$
\end{lemma}

\begin{proof} Introduce the function
$$
M_t(z)=((\widehat  A-zI)^{-1}g,g)=(( A+itP -zI)^{-1}g,g),\quad z\in
\begin{cases} \bbC_-, \text{ if } t >0\\
\bbC_+, \text{ if }  t< 0
\end{cases}.
$$
From the resolvent formula \eqref{resform} one gets that
$$
M_t(z)=\frac{M(z)}{1+itM(z)},\quad z\in
\begin{cases} \bbC_-, \text{ if } t >0\\
\bbC_+, \text{ if }  t< 0
\end{cases},
$$
and hence
$$
S(z)=1+2 i tM_{-t}(z)=\frac{1+itM(z)}{1-itM(z)}, \quad z\in \bbC_+,
$$
completing the proof.
\end{proof}

\begin{theorem}[\cite{Lv1}]\label{lifunique}
The characteristic function  of a completely non-self-adjoint
operator with a rank-one imaginary  part uniquely determines the operator up to unitary equivalence.
\end{theorem}

\begin{proof}
Suppose that
$\widehat A$ is a completely non-self-adjoint
operator with  a rank-one imaginary part so that
$$
\widehat A=A+it(\cdot, g)g
$$
for some $\|g\|=1$, $  A=A^*$ and $t>0$.

 In view of Theorem \ref{osnov}, it suffices to show that the characteristic function
uniquely determines the parameter $t$ and the (probability) measure $\mu$ from the representation
for the $M$-function (in the case in question, $\mu(\bbR)=\|g\|^2=1$)
\begin{equation}\label{MMM}
M(z)=((A-zI)^{-1}g,g)=\int_\bbR\frac{d\mu(\lambda)}{\lambda-z},
\quad z\in \bbC_+.
\end{equation}
Indeed, by Lemma \ref{lemma318},
$$
S(z)=\frac{1+itM(z)}{1-itM(z)}=
\frac{1-itz^{-1}+o(z^{-1})}{1+itz^{-1}+o(z^{-1})}=1-\frac{2it}{z} +o(z^{-1}), \quad z\to \infty,
$$
and hence
$$
\frac{i}{2}\lim_{z\to \infty}z(S(z)-1)=t,
$$
which uniquely determines the perturbation parameter $t$.
Since
$$
M(z)=\frac{1}{it}\cdot \frac{S(z)-1}{S(z)+1},
$$
the knowledge of the characteristic function  $S(z)$ also  uniquely determines  the probability measure $\mu(d\lambda)$ in \eqref{MMM} by the Stieltjes inversion formula.
\end{proof}

\section{Prime symmetric operators}\label{A1}

Recall that a  densely defined linear operator $\dot A$ in a Hilbert space $\cH$ is called symmetric if
$$
(\dot Ax, y)=(x, \dot Ay) \quad \text{for all}\quad x,y\in \Dom(\dot A).
$$

\begin{definition} A symmetric operator
 $\dot A$ is called  a prime operator
if there does not exist a (non-trivial) subspace invariant under $\dot A$
such that the restriction of $\dot A$ on this subspace is self-adjoint.
\end{definition}

If a symmetric operator is not prime, it is useful to separate its self-adjoint
part from its prime part.


\begin{theorem}\label{prostota0}
Let $\dot A$ be a
 closed symmetric operator with equal  deficiency indices
in a Hilbert space $\cH$.  Then the Hilbert space splits into the orthogonal sum of two subspaces
\begin{equation}\label{plm1}
\cH=\cK\oplus \cL,
\end{equation}
where $$
\cK=\overline{\Span_{\Im(z)\ne 0}\Ker ((\dot A)^*-zI)}
$$
and
$$
\cL=\bigcap_{\Im(z)\ne 0} \Ran(\dot A-zI).
$$

 Both of the subspaces  $\cK$ and $\cL$ reduce the symmetric operator $\dot A$. Moreover,
the part $\dot A|_\cL$ of $\dot A$ in $\cL$ is a self-adjoint operator
and the part $\dot A|_\cK$  of $\dot A$ in $\cK$ is a prime symmetric operator.

\end{theorem}


\begin{proof}
Since by hypothesis the operator
$\dot A$ is a closed symmetric operator, $\Ran(\dot A-zI)$,
$\Im (z)\ne 0$,
 is a closed subspace and hence
$\cL$  being the intersection of closed subspaces is a closed subspace itself.

Assume that $h\in \cL$ and hence
\begin{equation}\label{iadra1}
h\in \Ran(\dot A-zI)\quad \text{for   all } z\in \bbC\setminus \bbR.
\end{equation}
Since
\begin{equation}\label{iaki1}
\cH=\Ker((\dot A)^*-\overline{z}I)\oplus \Ran (\dot A-zI), \quad z\in \bbC\setminus \bbR,
\end{equation}
from \eqref{iadra1} and \eqref{iaki1} one concludes that $h$ is
 orthogonal to the subspace
$\Ker((\dot A)^*-zI)$ for any $z\in \bbC\setminus \bbR$.
Hence $h$ is orthogonal to $\cK$, which means that
\begin{equation}\label{eeff1}
\cL\subset \cK^\perp.
\end{equation}

Now, assume that an element $h$ is orthogonal to $\cK$.

 Since the linear set
$\cD=\Span_{z\in \bbC\setminus \bbR}\Ker ((\dot A)^*-\overline{z}I)$ is a dense subset in    $\cK$,
the element $h$ is orthogonal to $\Ker ((\dot A)^*-\overline{z}I)$ for any
$z\in \bbC\setminus \bbR$. Therefore, by \eqref{iaki1},
$$
h\in \Ran(\dot A-zI)\quad \text{for   all } z\in \bbC\setminus \bbR
$$
and hence
$
h\in \cL
$
which means that
\begin{equation}\label{ffee1}
 \cK^\perp\subset \cL.
\end{equation}
Combining \eqref{eeff1} and \eqref{ffee1}
completes the proof of \eqref{plm1}.

To prove the remaining assertion of the theorem,
 we show first that the subspace $\cL$ is $\dot A$-invariant.

Indeed, assume that  $f\in
\cL\cap \, \Dom (\dot A)$ and hence
\begin{equation}\label{rangi1}
f\in \Ran(\dot A-z I) \quad \text{ for all} \quad z\in \bbC\setminus \bbR.
\end{equation}
Then from \eqref{rangi1} follows that
$$
\dot A f=(\dot A-zI)f+zf\in \Ran (\dot A-zI) \text{ for all }
 \,\,\,z\in \bbC\setminus \bbR,
$$ and
hence $\dot A f\in \cL$ by the definition of the space $\cL$.

Next, we will show that
\begin{equation}\label{ssn}
(\dot A-zI)(\cL\cap \Dom(\dot A))
=\cL\quad \text{ for any } z\in \bbC\setminus \bbR.\
\end{equation}

To see this, take an $f\in \cL$. Then $f\in \Ran(\dot  A-zI)$ for any $z\in \bbC\setminus \bbR$,
and therefore, for any $z\in \bbC\setminus \bbR$
there exists an element $g_z\in \dom(\dot A)$ such that
$$
f=(\dot A-zI)g_z.
$$

We claim that  $$g_z\perp \Ker ((\dot A)^*-z  I)\quad \text{ for all
}\quad  z \in \bbC\setminus \bbR.
$$

Indeed, let $A$ be a(ny)
self-adjoint extension of $\dot A$ (recall that $\dot A$ has equal deficiency indices and therefore $\dot A$ admits
self-adjoint extensions).
Then $$f=(\dot A-zI)g_z=(  A-zI)g_z$$ and hence
\begin{equation}\label{sno}
g_z=(  A-zI)^{-1}f,\quad  z \in \bbC\setminus \bbR.
\end{equation}

Fix a $z$ with $\Im (z)\ne 0$.
For any element $f_\zeta$ such that
$$f_\zeta \in \Ker ((\dot A)^*-\zeta  I),
 \quad \zeta \in \bbC\setminus \bbR,$$
with $\zeta\ne \overline{z}$,
one gets that
\begin{align}
(g_z, f_\zeta)&=(( A-zI)^{-1}f, f_\zeta)=
(f,(  A-\overline{z}I)^{-1}f_\zeta)\nonumber
\\&=\frac{1}{z-\overline{\zeta}}\left ((f,
( A-\zeta I)(  A-\overline{z}I)^{-1}f_\zeta)
-(f, f_\zeta)\right )\label{nji}.
\end{align}

By assumption $ f\in \cL$
 and hence $f\perp \cK$ by the first part of the proof.
Since
$$(A-\zeta I)(
 A-\overline{z}I)^{-1}f_\zeta\in \Ker ((\dot A)^*-\overline{z}I)\subset \cK,$$
$$
f_\zeta\in \Ker ((\dot A)^*-\zeta I)\subset \cK,
$$ and $f\perp \cK$,
 from
\eqref{nji} follows that
$(g_z, f_\zeta)=0$, i.e.,
\begin{equation}\label{snov}
g_z\perp \Ker ((\dot A)^*-\zeta I),\quad \zeta\ne \overline{z},\quad
\Im(\zeta)\ne0.
\end{equation}

It remains to show that
\begin{equation}\label{snova1}
g_z\perp \Ker ((\dot A)^*-\overline{z} I).
\end{equation}

Indeed, for any
$f_{\overline{z}}
\in \Ker ((\dot A)^*-\overline{z} I)$
we have that
$$( A-\overline{z} I)
(  A-\zeta I)^{-1}f_{\overline{z}}\in \Ker ((\dot A)^*-\zeta I).$$
Therefore, by
\eqref{snov},
$$
(g_z, ( A-\overline{z} I)
( A-\zeta I)^{-1}f_{\overline{z}})=0, \quad \zeta\ne \overline{z},\quad
\Im(\zeta)\ne0.
$$
Hence,
$$
(g_z,f_{\overline{z}})=\lim_{\zeta \to \overline{z}}
(g_z, ( A-\overline{z} I)
( A-\zeta I)^{-1}f_{\overline{z}})=0, \quad \text{for  all} \quad f_{\overline{z}}
\in \Ker ((\dot A)^*-\overline{z} I),
$$
which proves \eqref{snova1}

From \eqref{plm1} follows that $ g_z\in \cL$ which justifies
\eqref{ssn} for  $g_z$ was chosen to be an element of $\dom(\dot A)$.

Denote by $B$ the restriction of $\dot A$ on
$$
\Dom (B)=\cL\cap \Dom(\dot A).
$$

Our next claim is that  $\Dom (B)$ is dense in $\cL$.

Indeed, let $f\in \cL$ and $f\perp \Dom(B)$, that is,
$$
(f, g)=0 \quad \text{for all } \quad g \in \Dom(B).
$$
From \eqref{ssn} follows that
 $f\in \Ran (B-iI)$ and hence $f=(B-iI)h$ for some
 $h\in \Dom (B)$. Thus, for all $g \in \Dom(B)$ one obtains that
\begin{equation}\label{ssnn}
(f, g)=((B-iI)h, g)=(h, (B+iI)g)=0.
\end{equation}
On the other hand, \eqref{ssn} yields
$$
\Ran(B+iI)=\cL
$$and therefore from \eqref{ssnn} follows that
 $h=0$ and hence $f=(B-iI)h=0$.

So, we have shown that the operator
  $B$ is a densely defined symmetric operator in the Hilbert space  $\cL$
such that  $\Ran (B\pm iI)=\cL$ which means that  $B$ is  self-adjoint.

Now suppose  that a subspace $\cK_0$ reduces $\dot A$ and that  the
 part $\dot A|_{\cK_0}$ is self-adjoint. Then
$$\Ran(\dot A|_{\cK_0}-zI_{\cK_0})=\cK_0, \quad z \in \bbC\setminus \bbR,$$
which means that $\cK_0\subset \cL$, proving that the part $\dot A$ in the subspace $\cK$
is a prime symmetric operator.

The proof is complete.

\end{proof}

\begin{remark}\label{primepart} In the situation of Theorem \ref{prostota0} the operator $\dot A|_{\cK}$ is called the {\it prime part} of the symmetric operator $\dot A$.
\end{remark}

Recall that an element $h\in \cH$ is said to be a generating element for a self-adjoint operator $H$ in the  Hilbert space $\cH$ if
$$
\overline{\text{span}_{\Im(z)\ne 0}\{(H-zI)^{-1}h\}}=\cH.
$$
We also say   that a self-adjoint operator has  {\it simple spectrum} if it has a generating element.

\begin{corollary}[cf. \cite{T87}]\label{den}
{\it  Let $\dot A$ be a
 closed symmetric operator with equal  deficiency indices
 in a Hilbert space $\cH$.

Then $\dot A$ is a prime operator if and only if
$$
\cH=\overline{\Span_{\Im(z)\ne 0}\Ker ((\dot A)^*-zI)}.
$$

 If, in addition,  $\dot A$ has deficiency indices $(1,1)$,
then $\dot A$ is a prime operator if and only if for
 any self-adjoint extension $A$
of  $\dot A$
a deficiency element $0\ne g_+\in\Ker ((\dot A)^*-iI) $ is generating, that is,
\begin{equation}\label{prosrez}
\cH=\overline{{\Span}_{\Im( z)\ne 0}(A-zI)^{-1}g_+}.
\end{equation}

In particular, in this case,   any self-adjoint extension of $\dot A$ has  simple spectrum.
}
\end{corollary}


\begin{proof}
The first assertion has already been proven in Theorem \ref{prostota0}.

To prove the remaining statement, one proceeds as follows.

Suppose that \eqref{prosrez} fails to hold and therefore the orthogonality condition
\begin{equation}\label{eqv1}
((A-zI)^{-1}g_+,h)=0\quad \text{for all} \quad z\in \bbC\setminus \bbR
\end{equation}
holds for some element $h\in \cH$, $h\ne 0$. Since
\begin{align*}
((A-zI)^{-1}g_+,h)&=((A-zI)^{-1}g_+,(A+iI)(A+iI)^{-1}h)
\\&=((A-iI)(A-zI)^{-1}g_+,(A+iI)^{-1}h),
\end{align*}
one concludes that
\begin{equation}\label{eqv2}
((A-iI)(A-zI)^{-1}g_+,g)=0\quad \text{for all} \quad z\in \bbC\setminus \bbR,
\end{equation}
where
$$
g=(A+iI)^{-1}h.
$$
Observing that
$$
(A-iI)(A-zI)^{-1}g_+\in \Ker ((\dot A)^*-zI),
$$
the hypothesis that $\dot A$ has deficiency indices $(1,1)$ yields the orthogonality condition
$$
g\perp \overline{\Span_{\Im( z)\ne 0}\Ker ((\dot A)^*-zI)},
$$
and  therefore $\dot A$ is not a prime operator by Theorem \ref{prostota0}.

Conversely, suppose that $\dot A$ is not a prime operator and therefore \eqref{eqv2} holds for some $0\ne g\in \cH$. In particular,
\begin{equation}\label{nulili}
(g_+, g)=0.
\end{equation}

Using the first resolvent identity,
$$
(A-zI)^{-1}-(A-iI)^{-1}=(z-i)(A-iI)^{-1}(A-zI)^{-1},
$$
one obtains
$$
(A-iI)(A-zI)^{-1}=I+(z-i)(A-zI)^{-1}.
$$
Therefore,
$$
((A-iI)(A-zI)^{-1}g_+, g)=(g_+, g)+((z-i)(A-zI)^{-1}g_+, g)
$$
for all $z\in \bbC\setminus \bbR$.
Using  \eqref{eqv2} and \eqref{nulili}, this equality  yields
$$
((A-zI)^{-1}g_+,g)=0\quad \text{for all} \quad z\in \bbC\setminus \bbR, \quad z\ne i,
$$ and therefore, by continuity,
$$
((A-zI)^{-1}g_+,g)=0\quad \text{for all} \quad z\in \bbC\setminus \bbR,
$$
which shows that
$$
\cH\ne \overline{{\Span}_{\Im( z)\ne 0}(A-zI)^{-1}g_+}.
$$
So, we have shown that $\dot A$ is a prime operator if and only if \eqref{prosrez} holds.

The proof is complete.
\end{proof}
We will also need a variant of the first part of  this corollary in case when
 $\dot A$ has deficiency indices $(0,1)$ or $(1,0)$.


\begin{lemma}\label{popopo}Let $\dot A$ be a
 symmetric operator with   deficiency indices $(0,1)$ in a Hilbert space $\cH$.
Then $\dot A$ is a prime operator if and only if
\begin{equation}\label{popo}
\cH=\overline{\Span_{\Im (z)<0}\Ker ((\dot A)^*-zI)}.
\end{equation}
\end{lemma}


\begin{proof} {\it ``Only if ''} Part. Without loss one may assume that $\dot A$ is a closed operator.
It is well known (see,  \cite[Theorem 2, Ch.VIII, Sec. 104]{AkG}) that $\dot A$
is unitarily equivalent to the differentiation operator on the positive semi-axis
with the Dirichlet boundary condition at the origin. So,
 without loss of generality, one may assume that
$ \cH=L^2(\bbR_+)
$
and
\begin{equation}
(\dot Af)(x)=-\frac{1}{i} \frac{d}{dx}f(x) \quad \text{a. e.} \,\,x\in
\bbR_+
\end{equation}
on
$$
\Dom( \dot A )=
\left \{f\in W_2^1(\bbR_+),\,\,f(0)=0\right \}.
$$Clearly, the functions
 $$
h_z(x)=e^{-izx}, \quad x\in (0, \infty),
$$
generate the subspaces $\Ker ((\dot A)^*-zI)$, $\Im (z)<0$.

Now, if $f\in L^2((0,\infty))$
 is orthogonal to $h_z$ for all $z\in \bbC_-$, then
$$
\int_0^\infty e^{-izx}\overline{f(x)}dx=0\quad \text{ for all } z\in \bbC_-.
$$
In particular,
$$
H(s)=\int_0^\infty e^{-sx}\overline{f(x)}dx=0,  \quad s>0,
$$
and hence $f=0$ by the uniqueness theorem for the Laplace transform (see, e.g.,
\cite[Th. 5.5]{Doet}.
Therefore,
\eqref{popo} holds which  completes the proof.

{\it ``If ''} Part. Suppose that $\dot A$ is not a prime operator. Thereofre, by Theorem \ref{prostota0},
$$
\cH\ne\overline{\Span_{\Im(z)\ne 0}\Ker ((\dot A)^*-zI)}.
$$
Since $\dot A$ has deficiency indices $(0,1)$,
$$\overline{\Span_{\Im(z)\ne 0}\Ker ((\dot A)^*-zI)}=\overline{\Span_{\Im(z)< 0}\Ker ((\dot A)^*-zI)},
$$
and hence
$$
\cH\ne\overline{\Span_{\Im(z)< 0}\Ker ((\dot A)^*-zI)}.
$$

Therefore, \eqref{popo} fails to hold, which completes the proof of the  {\it ``If ''} Part.
\end{proof}

In a similar way one proves the following statement.


\begin{lemma}\label{radik}Let $\dot A$ be a
 symmetric operator with   deficiency indices $(1,0)$ in a Hilbert space $\cH$.
Then $\dot A$ is a prime operator if and only if
\begin{equation}\label{popo1}
\cH=\overline{\Span_{\Im (z)>0}
\Ker ((\dot A)^*-zI)}.
\end{equation}
\end{lemma}

\section{A functional model of a triple}

Recall  the notion of
 a {\it functional model} of a prime
 dissipative triple parameterized by the characteristic
function  \cite{MT-S}.

Given a contractive analytic map $S(z)$,
\begin{equation}\label{chchch}
S(z)=\frac{s(z)-\varkappa} {\overline{ \varkappa }\,s(z)-1}, \quad z\in \bbC_+,
\end{equation}
where $|\varkappa|<1$ and $s(z)$ is
an  analytic, contractive function in $\bbC_+$
satisfying the Liv\v{s}ic criterion \cite{L46} (also see \cite[Theorem 1.2]{MT-S}),
that is,
\begin{equation}\label{vsea0}
s(i)=0\quad \text{and}\quad \lim_{z\to \infty}
z(s(z)-e^{2i\alpha})=\infty \quad \text{for all} \quad  \alpha\in
[0, \pi),
\end{equation}
$$
0< \varepsilon \le \text{arg} (z)\le \pi -\varepsilon,
$$
introduce the function
\begin{equation}\label{s&M}
M(z)=\frac1i\cdot\frac{s(z)+1}{s(z)-1},\quad z\in \bbC_+.
\end{equation}
One observes that
$$
M(z)=\int_\bbR \left
(\frac{1}{\lambda-z}-\frac{\lambda}{1+\lambda^2}\right )
d\mu(\lambda), \quad z\in \bbC_+,
$$
for some infinite Borel measure,
\begin{equation}\label{infm}
\mu(\bbR)=\infty,
\end{equation}
such that
\begin{equation}\label{normmu}
\int_\bbR\frac{d\mu(\lambda)}{1+\lambda^2}=1.
\end{equation}

In the Hilbert space $L^2(\bbR;d\mu)$ introduce
 the    (self-adjoint)
operator  $\cB$  of multiplication   by  independent variable
 on
\begin{equation}\label{nacha1}
\Dom(\cB)=\left \{f\in \,L^2(\bbR;d\mu) \,\bigg | \, \int_\bbR
\lambda^2 | f(\lambda)|^2d \mu(\lambda)<\infty \right \}.
\end{equation} Denote by  $\dot \cB$  its
 restriction
on
\begin{equation}\label{nacha2}
\Dom(\dot \cB)=\left \{f\in \Dom(\cB)\, \bigg | \, \int_\bbR
f(\lambda)d \mu(\lambda) =0\right \},
\end{equation}
and let
 $\widehat \cB$ be   the dissipative quasi-selfadjoint extension  of the  symmetric operator  $\dot \cB$
 on
\begin{equation}\label{nacha3}
\Dom(\widehat \cB)=\dom (\dot \cB)\dot +\linspan\left
\{\,\frac{1}{\lambda -i}- \varkappa\frac{1}{\lambda +i}\right \},
\end{equation}
where the von Neumann parameter $\varkappa$ of  the triple
$(\dot \cB, \widehat \cB, \cB)$  is given by
$$\varkappa=S(i).
$$

Notice that in this case
\begin{equation}\label{kone}
\Dom(\cB)=
\dom (\dot \cB)\dot +\linspan\left
\{\,\frac{1}{\lambda -i}- \frac{1}{\lambda +i}\right \}.
\end{equation}

We will refer to the triple  $(\dot \cB,   \widehat \cB,\cB)$ as
{\it  the model
 triple } in the Hilbert space $L^2(\bbR;d\mu)$.

Let  $\dot A$ be  a densely defined symmetric operator with deficiency indices $(1,1)$,   $ A$ its self-adjoint (reference) extension and
$\widehat A$ a maximal non-selfadjoint dissipative extension of $\dot A$.

 Denote by    $S_\fA(z)=S_{(\dot A,\widehat  A, A)}(z)$ the characteristic function of the triple $\fA=(\dot A,\widehat  A, A)$
as  introduced in Section \ref{s3} by eq. \eqref{obr}.
 Notice that  the characteristic function $S_{\fB}(z)$ of the model triple $\fM=(\dot \cB, \widehat \cB, \cB)$
is given by \eqref{chchch}.

The following uniqueness result  shows that the characteristic function of a triple  $(\dot A,\widehat  A, A)$
 is a complete unitary invariant of
the  triple whenever  the symmetric operator
$\dot A$ is prime, equivalently, the dissipative operator $\widehat  A$ is completely non-selfadjoint.

\begin{theorem}[{\cite[Theorems 1.4, 4.1]{MT-S}}]\label{unitar}
Suppose that $\dot A$ and $\dot B$
 are prime, closed, densely defined  symmetric operators
with deficiency indices $(1,1)$. Assume, in addition, that
  $A$ and $B$ are some self-adjoint
extensions of $\dot A$ and $\dot B$ and that
$\widehat A$ and $\widehat B$ are maximal dissipative
extensions of $\dot A$ and $\dot B$, respectively
$(\widehat A\ne (\widehat A)^*$,
$\widehat B\ne (\widehat B)^*)$.

Then
\begin{itemize}
\item[(i)] the triples  $\fA=(\dot A,\widehat A, A)$ and $\fB=(\dot
B,\widehat B,  B)$
 are mutually unitarily equivalent\footnote{We say that triples  of operators
$(\dot A, \widehat A, A)$ and $(\dot B,\widehat B,  B)$
 in Hilbert spaces $\cH_A$ and $\cH_B$
are mutually unitarily equivalent if there is
a unitary map $\cU$ from $\cH_A$ onto $\cH_B$ such that
$\dot B=\cU\dot A\cU^{-1}$, $\widehat  B=\cU\widehat  A\cU^{-1}$, and  $ B=\cU A\cU^{-1}$.}
if, and only if, the characteristic functions   $S_\fA(z)=S_{(\dot A,\widehat  A, A)}(z)$
and $S_\fB(z)=S_{(\dot B, \widehat B, B)}(z)$ of the triples coincide;

\item[(ii)] the triple  $(\dot A,\widehat A,  A)$ is mutually unitarily
equivalent to the model triple
$$\fM=(\dot \cB, \widehat\cB,  \cB )$$
in
the Hilbert space $L^2(\bbR;d\mu)$, where $\mu(d\lambda)$ is the
representing measure for the Weyl-Titchmarsh function $M(z)=M_{(\dot A,
A)}(z)$ associated with the pair $(\dot A, A)$.
\end{itemize}

In particular,
\begin{itemize}\item[(iii)] the pairs $(\dot A, A)$ and $ (\dot B, B)$ are mutually unitary equivalent if and only if $M_{(\dot A, A)}(z)=M_{(\dot B, B)}(z)$.
\end{itemize}

\end{theorem}

\begin{remark}
\label{snoska} Notice that in view of \eqref{kone}, if $\cU$ is the unitary operator from the Hilbert space $\cH$ onto $L^2(\bbR;d\mu)$ that implements  mutual unitary equivalence of the triples
 $(\dot A,\widehat A,  A)$ and  $(\dot \cB, \widehat\cB,  \cB )$, then
 $$
 (\cU g_\pm )(\lambda)=\frac{\Theta}{\lambda\mp i}\quad \text{for some}\quad |\Theta|=1,
 $$
 provided that $g_\pm \in \Ker ((\dot A)^*\mp iI)$ are normalized deficiency elements of $\dot A$ such that
 \begin{equation}\label{aassdd}
 g_+-g_-\in \Dom (A).
 \end{equation}
 Indeed,
for the normalized deficiency elements  $g_\pm$ and
$\cU g_\pm$ of the symmetric  operators $\dot A$ and $\dot \cB$, respectively, we have
 $$
 (\cU g_\pm )(\lambda)=\frac{\Theta_\pm}{\lambda\mp i}\quad \text{for some}\quad |\Theta_\pm|=1.
 $$
Since $\cB=\cU A \cU^{-1}$, from \eqref{aassdd} it follows that the function $h=\cU(g_+-g_-)$ given by
 $$
 h(\lambda)=\frac{\Theta_+}{\lambda- i}-\frac{\Theta_-}{\lambda+ i}, \quad \lambda\in \bbR,
 $$
 belongs to $\dom(\cB)=L^2(\bbR;(1+\lambda^2)d\mu(\lambda))$ and therefore
 $$
 \int_\bbR\Big|\frac{\Theta_+}{\lambda- i}-\frac{\Theta_-}{\lambda+ i}\Big|^2(1+\lambda^2)d\mu(\lambda)<\infty.
 $$
 Taking into account that  by \eqref{infm} the measure $\mu(d\lambda)$ is infinite, one necessarily gets that $\Theta_+=\Theta_-$ and the claim follows.
\end{remark}

\section{The spectral analysis of the model dissipative operator }


In the suggested functional model for a triple in  $L^2(\bbR;d\mu)$,
the eigenfunctions of the model dissipative operator
 $\widehat \cB$ from the model  triple $(\dot \cB, \widehat \cB,
 \cB)$ given by \eqref{nacha1}-\eqref{nacha3}
look exceptionally simple \cite{MT-S}.

\begin{lemma}\label{specpoint} Suppose that
$\fM=(\dot \cB, \widehat \cB, \cB)$ is the model triple  in $L^2(\bbR;d\mu)$ given by \eqref{nacha1}-\eqref{nacha3}.
 Then a point $z_0\in \bbC_+$ is an eigenvalue
of the dissipative operator  $\widehat \cB$ if
and only if

$$S_{\fM}(z_0)=S_{(\dot \cB, \widehat \cB, \cB)}(z_0)=0.$$

 In
this case, the corresponding eigenfunction $f$ is of the form
$$
f(\lambda)=\frac{1}{\lambda-z_0}\quad
 \text{for $\mu$-almost all  } \lambda\in \bbR.
$$
\end{lemma}

\begin{proof} Suppose that $z_0\in \bbC_+$ is an eigenvalue
of $\widehat \cB$ and
 $f\in L^2(\bbR;d\mu)$ is  the corresponding eigenvector,
$$
\widehat \cB f=z_0 f, \quad f\in\Dom( \widehat \cB ).$$ Since
$f\in\Dom( \widehat \cB)$, the element $f$ admits the
representation
$$
f(\lambda)=f_0(\lambda)+K\left (\frac{1}{\lambda-i}-\varkappa
\frac{1}{\lambda+i}\right ),
$$
where $f_0\in \Dom (\dot \cB)$, $K$ is some constant and $\varkappa= S_\fB (i)$.
Then
$$
0=((\widehat \cB-z_0I)f)(\lambda)=(\lambda-z_0)f_0(\lambda)+K\left
(\frac{i-z_0}{\lambda-i}+\varkappa \frac{i+z_0}{\lambda+i}\right ).
$$
Hence,
$$
f_0(\lambda)=-\frac{K}{\lambda-z_0}\left
(\frac{i-z_0}{\lambda-i}+\varkappa \frac{i+z_0}{\lambda+i}\right ) .$$
Since $ f_0\in \Dom (\dot \cB)$, one obtains that
$$
\int_\bbR f_0(\lambda) d\mu(\lambda)=0
$$
and hence
\begin{align*}
0&=\int_\bbR\frac{1}{\lambda-z_0}\left
(\frac{i-z_0}{\lambda-i}+\varkappa \frac{i+z_0}{\lambda+i}\right )
d\mu(\lambda)
\\&
\\&=-\int_\bbR\left (\frac{1}{\lambda-z_0}-\frac{1}{\lambda-i}\right )
d\mu(\lambda)
+\varkappa \int_\bbR\left
(\frac{1}{\lambda-z_0}-\frac{1}{\lambda+i}\right )d\mu(\lambda)
\\&
\\&=
-M(z_0)+M(i)+\varkappa( M(z_0)-M(-i))
\\&
\\&=-M(z_0)+i+\varkappa( M(z_0)+i)
,\end{align*}
where $M(z)=M_{(\dot \cB,
\cB)}(z)$ is the Weyl-Titchmarsh function associated with the model pair
 $(\dot \cB, \cB)$.

Therefore,
$$
\varkappa =\frac{M(z_0)-i}{M(z_0)+i}=s_{(\dot \cB, \cB)}(z_0),
$$
where  $s_{(\dot \cB, \cB)}(z)$ is the Liv\v{s}ic function associated with the pair  $(\dot \cB, \cB)$.
Hence,  the characteristic  function $S_\fB(z)$
vanishes at the point $z_0$,
$$
S_{\fB}(z_0)=\frac{s_{(\dot \cB,
\cB)}(z_0)-\varkappa} {\overline{\varkappa}s_{(\dot \cB,
\cB)}(z_0)-1}=0.
$$
In this case,
\begin{align*}
f(\lambda)&=K\left [\left (\frac{1}{\lambda-i}-\varkappa
\frac{1}{\lambda+i}\right ) -\frac{1}{\lambda-z_0}
\left (\frac{i-z_0}{\lambda-i}+\varkappa \frac{i+z_0}{\lambda+i}\right )
\right ]\\& =K\left [\frac{1}{\lambda-i}\left  (1-\frac{i-z_0}{\lambda-z_0}\right )
-\varkappa\frac{1}{\lambda+i}\left (1+\frac{i+z_0}{\lambda-z_0}\right ) \right ]
\\&
=K(1-\varkappa)\cdot \frac{1}{\lambda-z_0}.
\end{align*}
So, we have shown that if $ z_0$ is an eigenvalue of
$\widehat \cB$,  then
$$
S_\fM(z_0)=0,
$$
with the corresponding eigenelement $f$  of the form
\begin{equation}\label{eigenf}
f(\lambda)=\frac{1}{\lambda-z_0}.
\end{equation}

 Repeating the same reasoning in the reversed order, one shows that
if $S_\fM(z_0)=0$, then the function
$f$ given by \eqref{eigenf}
belongs to $\Dom (\widehat \cB)$ and $\widehat \cB f=z_0f $.
\end{proof}

For the resolvents of the model dissipative operator $\widehat\cB$
and the self-adjoint (reference)
operator $\cB$ from the model  triple $\fM=(\dot \cB, \widehat \cB, \cB)$
 one gets
 the following resolvent formula \cite{MT-S}.
\begin{theorem}\label{modeldef} Suppose that
$\fM=(\dot \cB, \widehat \cB, \cB)$ is the model triple in the Hilbert space
$L^2(\bbR;d\mu) $ given by \eqref{nacha1}-\eqref{nacha3}.

 Then the
resolvent
 of the model dissipative operator $\widehat \cB$  in
$L^2(\bbR;d\mu) $ has the form
\begin{equation}\label{rezfor1}
(\widehat \cB- zI )^{-1}=(\cB- zI )^{-1}-p(z)(\cdot\, ,
g_{\overline{z}})g_z ,
\end{equation}
 with
\begin{equation}\label{rezfor12}
p(z)=\left (M_{(\dot \cB,
\cB)}(z)+i\frac{\varkappa+1}{\varkappa-1}\right )^{-1},
\end{equation}
$$z\in\rho(\widehat \cB)\cap \rho(\cB).
$$
Here $M_{(\dot \cB,
\cB)}(z)$ is the Weyl-Titchmarsh function associated with the pair
 $(\dot \cB, \cB)$ continued to the lower half-plane by the Schwarz reflection
principle, $\varkappa$ is the von Neumann parameter of the  triple $\fM$,
and the deficiency elements $g_z$,
$$
g_z\in \Ker ((\dot \cB)^*-zI), \quad z\in \bbC\setminus \bbR,
$$
are given by
\begin{equation}\label{defelrep}
g_z(\lambda)=\frac{1}{\lambda-z} \quad
\text{for $\mu$-almost all  } \lambda\in \bbR.
\end{equation}

\end{theorem}

\begin{proof} Given $h\in L^2(\bbR; d\mu)$ and $z\in \rho(\widehat \cB)$,
 suppose that
\begin{equation}\label{jkl}
(\widehat \cB- zI)f=h \quad \text{ for  some } f\in \Dom(\widehat
\cB).
\end{equation}
Since $f\in \Dom(\widehat \cB)$, one gets the representation
\begin{equation}\label{eqff}
f(\lambda)=f_0(\lambda)+K\left (\frac{1}{\lambda-i}-
\varkappa\frac{1}{\lambda+i}\right )
\end{equation}
for some $f_0\in \Dom(\dot \cB)$ and $K\in \bbC$. Eq. \eqref{jkl}
yields
$$
(\lambda-z)f_0(\lambda)+K\left (\frac{i-z}{\lambda-i}+
\varkappa\frac{i+z}{\lambda+i}\right )=h(\lambda)
$$
and hence
\begin{equation}\label{eqff1}
f_0(\lambda)=\frac{h(\lambda)}{\lambda-z}-
\frac{K}{\lambda-z}\left (\frac{i-z}{\lambda-i}+
\varkappa\frac{i+z}{\lambda+i}\right ).
\end{equation}
Since $f_0\in \Dom(\dot \cB)$,
$$
\int_\bbR f_0(\lambda)d\mu(\lambda)=0.
$$
Integrating \eqref{eqff1} against  the measure $\mu(d\lambda)$, one obtains that
\begin{equation}\label{KK}
K \int_\bbR \frac{1}{\lambda-z}\left (\frac{i-z}{\lambda-i}+
\varkappa\frac{i+z}{\lambda+i}\right )d\mu(\lambda)= \int_\bbR
\frac{h(\lambda)}{\lambda-z}d\mu(\lambda).
\end{equation}
Observing that
$$
\int_\bbR \frac{1}{\lambda-z}\left (\frac{i-z}{\lambda-i}+
\varkappa\frac{i+z}{\lambda+i}\right
)d\mu(\lambda)=i-M(z)+\varkappa(M(z)+i),
$$ with
$
M(z)=M_{(\dot \cB, \widehat \cB)}(z)
$, and solving \eqref{KK} for $K$, one obtains
$$
K=\frac{1}{(\varkappa-1)M(z)+i(1+\varkappa)}\int_\bbR
\frac{h(\lambda)}{\lambda-z}d\mu(\lambda)
.$$ Combining \eqref{eqff} and \eqref{eqff1}, for the element $f$
we have the representation
\begin{align}
f(\lambda)&=\frac{h(\lambda)}{\lambda-z}+ \frac{K}{\lambda-z}\left
(\frac{\lambda-z}{\lambda-i}-
\varkappa\frac{\lambda-z}{\lambda+i}- \left
[\frac{i-z}{\lambda-i}+ \varkappa\frac{i+z}{\lambda+i}
\right
]\right )
\nonumber\\&
\nonumber\\&=
\frac{h(\lambda)}{\lambda-z} -K\frac{\varkappa-1}{\lambda-z}
\nonumber\\&
\label{rezfor3}\\&
=\frac{h(\lambda)}{\lambda-z}-\left ( M(z) +
i\frac{\varkappa+1}{\varkappa-1}
\right)^{-1}\frac{1}{\lambda-z}\int_\bbR
\frac{h(\lambda)}{\lambda-z}d\mu(\lambda),
\nonumber\\&
\nonumber\\&\quad \quad \quad \quad \quad z\in \rho(\widehat \cB)\cap\rho(\cB),\nonumber
\end{align}
where we have used \eqref{KK} on the last  step.

To complete the proof, it remains to recall that $f=(\widehat \cB-zI)^{-1}h$ and to compare \eqref{rezfor1} with \eqref{rezfor3}.

\end{proof}

\begin{remark}\label{balansir} Using \eqref{s&M}, it is
 easy to see that the poles of the function
$p(z)$, defined in \eqref{rezfor12}, in the upper half-plane coincide with the roots of the equation
$$
s_{(\dot \cB, \cB)}(z)=\varkappa, \quad z\in \bbC_+,
$$
provided that $\varkappa \ne 0$ and $M(z)\ne i$  identically
in the upper half-plane.
Therefore, the zeros of the characteristic function $S_\fM(z)=S_{(\dot \cB, \widehat \cB, \cB)}(z)$
in the upper half-pane determine the poles of the resolvent of
 the dissipative operator
$\widehat \cB$ (cf. Lemma \ref{specpoint}).

We also remark that if $\varkappa = 0$ and $M(z)= i$ for all $z\in \bbC_+$, then the point spectrum of
the dissipative operator $\widehat \cB$ fills in the whole
open upper half-plane $\bbC_+$.
\end{remark}

Given a triple $(\dot A, \widehat A, A)$
satisfying  \eqref{rss} and \eqref{parpar},  the following
corollary provides an analog of the
Krein formula for resolvents
 for  all quasi-selfadjoint
dissipative extensions of the symmetric operator $\dot A$
with deficiency indices $(1,1)$.

\begin{corollary}[{\cite{MT-S}}]
{\it
Let  $\fA=(\dot A, \widehat A, A)$ be a triple
satisfying  \eqref{rss} and \eqref{parpar}. Then
 the following resolvent formula
 \begin{equation}\label{krres}
(\widehat A-zI)^{-1}=(A-zI)^{-1}-p(z)(\cdot\, ,g_{\overline{z}})g_z,
\end{equation}
$$z\in \rho(\widehat A)\cap \rho(A),
 $$
holds.

Here
\begin{itemize}
\item[(a)] the function $p(z )$ is given by
\begin{align}
 p(z)&=\left (M_{(\dot A,A)}(z)+i\frac{\varkappa+1}{\varkappa-1}\right )^{-1}
\label{WTf}\\&
=i\left (\frac{s_{(\dot A,A)}(z)+1}{s_{(\dot A,A)}(z)-1}
-\frac{\varkappa+1}{\varkappa-1}\right )^{-1};
\label{Lf}\end{align}
\item[(b)]$M_{(\dot A, A)}(z) $ and $s_{(\dot A, A)}(z)$ are
 the Weyl-Titchmarsh
and   the Liv\v{s}ic function
of the pair $(\dot A, A)$, respectively;
\item[(c)]
$g_z$ are deficiency elements of $\dot A$,
$$
g_z\in \Ker((\dot A)^*-zI),
$$
 satisfying the normalization condition
\begin{equation}\label{norcon}
\|g_z\|=\left (\int_\bbR \frac{d\mu(\lambda)}{|\lambda-z|^2}
\right )^{1/2}
\end{equation}
$($the deficiency elements $g_z$ can be chosen to be analytic in
$z\in  \rho(\widehat A)\cap \rho(A)$$)$;
\item[(d)]
 $\mu(d\lambda)$ is the measure from the Herglotz-Nevanlinna representation
$$M_{(\dot A, A)}(z)=\int_\bbR \left
(\frac{1}{\lambda-z}-\frac{\lambda}{1+\lambda^2}\right )
d\mu(\lambda);
$$

\item[(e)] $\varkappa$ is the von Neumann
 parameter of the triple $(\dot A, \widehat A, A)$ which characterizes the domain of
 the dissipative extension $\widehat A$ in such a way that
\begin{equation}\label{vnp}
g_+-g_{-}\in \Dom (A) \quad \text{and}\quad g_+-\varkappa g_{-}\in
\Dom (\widehat A),
\end{equation}
where $g_\pm=g_{\pm i}$.
\end{itemize}

}
\end{corollary}

\begin{remark} We would like to stress
 that the von Neumann parameter
$\varkappa$ of the triple $\fA=(\dot A, \widehat A, A)$,
 the Liv\v{s}ic function $s_{(\dot A, A)}(z)$,
and
the Weyl-Titchmarsh function
$M_{(\dot A, A)}(z)$,
 can easily  be recovered from the knowledge of the  the
characteristic function $S(z)=S_\fA(z)$. (Recall that  the
characteristic function $S(z)$
 is a complete unitary invariant of the triple  $\fA=(\dot A, \widehat A, A)$, provided that $\dot A$ is a prime operator).

Indeed,
\begin{align*}
\varkappa  &=S(i),
\\s_{(\dot A, A)}(z)&=\frac{S(z)-\varkappa}{\overline{\varkappa} S(z)-1},
\\
 M_{(\dot A, A)}(z)&=\frac1i\cdot \frac{s_{(\dot A, A)}(z)+1}{s_{(\dot A, A)}(z)-1},
\end{align*}
$$z\in \bbC_+,
$$
with $ M_{(\dot A, A)}(z)$
continued to the lower half-plane by the Schwarz reflection principle
$$
M_{(\dot A, A)}(z)=\overline{M_{(\dot A, A)}(\overline{z})},\quad z\in \bbC_-.
$$
\end{remark}
\begin{remark}
The resolvent formula \eqref{krres}--\eqref{WTf}
also holds if $|\varkappa|=1$ and hence $\widehat A$ is self-adjoint. In this case, it
coincides with  the   Krein resolvent formula
for self-adjoint extensions of $\dot A$.
\end{remark}

\begin{remark}
If two triples
$(\dot A,\widehat A_1,A)$ and $(\dot A ,\widehat A_2,A)$ with the same reference operator $A$  satisfy
 \eqref{rss} and \eqref{parpar}
and have  the von Neumann parameters $\varkappa_1$ and $\varkappa_2$, respectively,
then one gets the following resolvent formula
 for the dissipative extensions $\widehat A_1$ and $\widehat A_2$
refining, in the rank-one setting,  a result in \cite{K}:
$$
(\widehat A_2-zI)^{-1}=(\widehat A_1-zI)^{-1}-q(z)(\cdot, g_{\overline{z}})g_z,
$$
where
$
q(z)=p_2(z)-p_1(z)
$
with
\begin{align*}
p_k(z)&=
\left (M_{(\dot A,
A)}(z)+i\frac{\varkappa_k+1}{\varkappa_k-1}\right )^{-1},
\\&=i\left (\frac{s_{(\dot A,A)}(z)+1}{s_{(\dot A,A)}(z)-1}
-\frac{\varkappa_k+1}{\varkappa_k-1}\right )^{-1},\quad k=1,2,
\end{align*}
$$
z\in \rho(A)\cap \rho(\widehat A_1)\cap \rho(\widehat A_2).
$$
We recall, see \eqref{obr},  that if $S_1(z)$ and $S_2(z)$ are the characteristic functions of
the triples $(\dot A,\widehat A_1,A)$ and $(\dot A, \widehat A_2,A)$,
respectively, then
$$
s_{(\dot A,A)}(z)=\frac{S_k(z)-\varkappa_k}{\overline{\varkappa_k} S_k(z)-1},
\quad k=1,2.$$
\end{remark}

\begin{corollary}\label{obrat1}
{\it Suppose that
$\fM=(\dot \cB, \widehat \cB, \cB)$ is the model triple  in $L^2(\bbR;d\mu)$ given by \eqref{nacha1}-\eqref{nacha3}. Assume  that $z=0$
is a regular point for both the
dissipative operator $\widehat \cB$ and the $($reference$)$ self-adjoint operator $\cB$.
Then the inverse $\widehat \cB^{-1}$ is a rank-one perturbation of
 the $($bounded$)$ self-adjoint operator $\cB^{-1}$. That is,
$$
\widehat \cB^{-1}=\cB^{-1}-p Q,
$$
where
\begin{equation}\label{pppp}
p=\left (M(0)+
i\frac{\varkappa+1}{\varkappa-1}\right)^{-1},
\end{equation}
$Q$ is a rank-one self-adjoint operator
$$
(Qf)(\lambda)=\frac{1}{\lambda}
\int_\bbR \frac{f(s)}{s}d\mu(s), \quad
\text{ $\mu$-a.e. } \lambda \in \bbR,
$$
and $M(0)$ is the value of the Weyl-Titchmarsh function associated with the pair $(\dot \cB, \cB)$
at the point zero.
}
\end{corollary}

\section{Transformation laws}\label{s4}

In this section we discuss
the dependence of the Liv\v{s}ic, Weyl-Titchmarsh and   characteristic functions upon the reference  operator.

\begin{lemma}\label{ind0}
Assume that   $\dot A$ is a symmetric, densely  defined, closed operator with
deficiency indices $(1,1)$ and $\widehat A$ its maximal dissipative extension such that $\widehat A\ne (\widehat A)^*$.
 Suppose that
 $g_\pm\in\Ker ((\dot A)^*\mp iI)$,  $||g_\pm||=1$.
  Given $\alpha\in [0, \pi)$,
denote by  $A_\alpha$ a unique self-adjoint extension
 of $\dot A$ such that
\begin{equation}\label{dom1}
g_+-e^{2i\alpha}g_-\in \Dom (A_\alpha).
\end{equation}
Let
$
s_\alpha(z)=s_{(\dot A, A_\alpha)}(z),
$
$M_\alpha(z)=M_{(\dot A, A_\alpha)}(z)$, and
$
S_\alpha(z)=S_{(\dot A, \widehat A, A_\alpha)}(z)
$
be  the Liv\v{s}ic  function,
the
 Weyl-Titchmarsh function associated with the pair $(\dot A, A_\alpha)$,
   and
the characteristic function   associated with the triple $(\dot A, \widehat A, A_\alpha)$,
respectively.

Then \begin{align}
s_\alpha(z)&=e^{2i(\beta-\alpha)}s_\beta(z) \label{vags},
\\
M_\beta(z)&=\frac{\cos(\beta-\alpha) M_\alpha(z)-\sin (\beta-\alpha)}{
\cos(\beta-\alpha) +\sin (\beta-\alpha)M_\alpha(z)}, \label{trans}
\\
S_\alpha(z)&=e^{2i(\beta-\alpha)}S_\beta(z). \label{vagS}
\end{align}
\end{lemma}
\begin{proof}
Suppose that the deficiency elements
 $g_\pm\in\Ker ((\dot A)^*\mp iI)$,  $||g_\pm||=1$, are such that
 $$
g_+-\varkappa g_-\in \Dom (\widehat A).
$$
for some   $\varkappa$, $|\varkappa| <1$.
Denote by $A_*$ the reference self-adjoint  extension of $\dot A$ such that
$$
g_+-g_-\in \Dom(A_*),
$$
and let $s_*(z)=s_{(\dot A, A_*)}(z)$, $M_*(z)=M_{ (\dot A, A_*)}(z)$ and $S_*(z)=S_{(\dot A,\widehat A, A_*)}(z)$ be the corresponding Liv\v{s}ic, Weyl-Titchmarsh and the characteristic functions, respectively.
From the definition of the Liv\v{s}ic function  $s_\alpha(z)$ it follows that
$$
s_\alpha(z)=e^{-2i\alpha}s_*(z).
$$
By \eqref{blog} we have
$$
s_*(z)=\frac{M_*(z)-i}{M_*(z)+i}\quad \text{and}\quad s_\alpha(z)=e^{-2i\alpha}s_*(z)=\frac{M_\alpha(z)-i}{M_\alpha(z)+i}\
$$
and hence
$$
M_\alpha(z) =\frac{\cos( \alpha) M_*(z)-\sin (\alpha)(z)}{
\cos(\alpha) +\sin (\alpha)M_*(z)}.
$$

From \eqref{obr} it follows
$$
S_*(z)=\frac{s_*(z)-\varkappa} {\overline{\varkappa}s_*(z)-1}.
$$
Therefore,
$$
e^{-2i\alpha}S_*(z)=\frac{e^{-2i\alpha}s_*(z)-e^{-2i\alpha}\varkappa} {\overline{(e^{-2i\alpha}\varkappa)}e^{-2i\alpha}s_*(z)-1}
=\frac{s_\alpha(z)-e^{-2i\alpha}\varkappa} {\overline{(e^{-2i\alpha}\varkappa)}s_\alpha(z)-1} =S_\alpha(z),
$$
which proves \eqref{vags}, \eqref{trans}, \eqref{vagS} first  for $\beta=0$ and hence for all $\beta$ taking into account that the transformations $\alpha\to s_\alpha, M_\alpha, S_\alpha$  are one-parameter groups.
\end{proof}

Our next result shows that the concept of a characteristic function of a triple is essentially determined by the corresponding dissipative operator only rather than by the triple itself (cf. \cite{AkG,L46}).

\begin{proposition}\label{concept} {\it Let $(\dot A, \widehat A, A)$ and $(\dot B, \widehat B, B)$ be two triples satisfying the hypothesis of Lemma \ref{ind0}.
Suppose that the dissipative operators $\widehat A$ and $\widehat B$ are unitarily  equivalent.

 Then the characteristic functions of the triples   $(\dot A, \widehat A, A)$ and $(\dot B, \widehat B, B)$ coincide up to a  constant unimodular factor.

In particular, the absolute values of the von Neumann parameters $ \varkappa_{(\dot A, \widehat A, A)}$ and $ \varkappa_{(\dot B, \widehat B, B)}$
 of the triples  $(\dot A, \widehat A, A)$ and $(\dot B, \widehat B, B)$ coincide,
 \begin{equation}\label{vonuni}
| \varkappa_{(\dot A, \widehat A, A)}|=|\varkappa_{(\dot B, \widehat B, B)}|.
 \end{equation}
}
\end{proposition}
\begin{proof}
To be more specific, assume that $\cU$ is a unitary operator such that
$$
\widehat B=\cU^{-1} \widehat A \cU.
$$ That is,
$$
\cU(\dom (\widehat B))= \dom (\widehat A)
$$
and
$$
\cU\widehat B f= \widehat A \cU\quad \text{for all } f\in \dom (\widehat B).
$$
Literally repeating  the  proof of Lemma \ref{nudada}  one shows that
$$
(\widehat B)^*=\cU^{-1} (\widehat A)^* \cU.
$$
Therefore, the symmetric operators $\dot A$ and $\dot B$ are unitarily equivalent
$$
\dot B=\cU^{-1}\dot  A \cU,$$
since
$$
\dot A=\widehat A|_{\Dom (\widehat A)\cap \dom((\widehat A)^*)}\quad\text{and}  \quad  \dot B=\widehat B|_{\Dom (\widehat B)\cap \dom((\widehat B)^*)}.
$$
Moreover, the operator $B'=\cU^{-1} A \cU$ is a self-adjoint extensions of $\dot B$.

By Lemma \ref{ind0}, the characteristic functions $S_{(\dot B, \widehat B, B)}(z) $ and  $S_{(\dot B, \widehat B, B')}(z)$ of the triples
$(\dot B, \widehat B, B) $ and  $(\dot B, \widehat B, B')$ are related as
$$
S_{(\dot B, \widehat B, B)}(z)=\Theta S_{(\dot B, \widehat B, B')}(z), \quad z\in \bbC_+,
$$
for some constant $\Theta$, $|\Theta|=1$. Since  the tripes $(\dot B, \widehat B, B')$ and  $(\dot A, \widehat A, A)$ are mutually unitarily equivalent by construction, we have
$$
S_{(\dot B, \widehat B, B')}(z)=S_{(\dot A, \widehat A, A)}(z),\quad z\in \bbC_+,
$$
by the uniqueness Theorem \ref{unitar}.
Therefore,
$$
S_{(\dot B, \widehat B, B)}(z)=\Theta S_{(\dot A, \widehat A, A)}(z),\quad z\in \bbC_+,
$$
which completes the proof of the first assertion of the proposition.

To prove \eqref{vonuni}, we use the relation \eqref{sac} to conclude that
$$
| \varkappa_{(\dot A, \widehat A, A)}|=
|S_{(\dot A, \widehat A, A)}(i)|=|S_{(\dot B, \widehat B, B)}(i)|=|\varkappa_{(\dot B, \widehat B, B)}|.
$$
\end{proof}

Our next  goal is
to  establish a transformation law
 for  the  Liv\v{s}ic function
 under the affine transformations of  the pair $(\dot A, A)$,
$$
(\dot A, A)\longrightarrow (a\dot A+bI , aA+bI),
\quad a, b\in \bbR, \quad a>0.
$$


\begin{lemma}\label{masha} Let $\dot A$ be a symmetric operator
 with deficiency indices $(1,1)$ and $A$ its   self-adjoint
extension.
Suppose that
$
f(z)=az+b
$
with $a, b\in \bbR$, $a>0$ is an affine transformation.

Then the Liv\v{s}ic
function  associated with the pair
$(f(\dot A), f(A))$ admits the representation
\begin{equation}\label{vtcher}
s_{(f(\dot A), f(A))}(z)=\frac{m(z)-m(i)}{m(z)-\overline{m(i)}},
\end{equation}
where
$$
m(z)=M(f^{-1}(z))
$$
and   $M(z)=M_{(\dot A, A)}(z)$ is the Weyl-Titchmarsh function associated with the pair $(\dot A, A)$.

\end{lemma}


\begin{proof} By Remark \ref{mnogopros}, without loss of generality one may assume that $\dot A$ is a prime symmetric operator.
Let $M(z)=M_{(\dot A, A)}(z)$ be the Weyl-Titchmarsh function associated with the pair $(\dot A, A)$.
Next,  we may assume that
$A$ is the multiplication operator by independent variable in $L^2(\bbR, d\mu)$ and $\dot A$ is its restriction on
$$
\Dom (\dot A)=\left \{f\in \Dom (A) \,\bigg  |\int_\bbR f(\lambda)d\mu(\lambda)=
0\right \},
$$
where $\mu(d\lambda)$ is the representing measure for $M(z)$ (see Theorem  \ref{unitar}).

Introduce the  family of  functions
$$
G_z(\lambda)=\frac{1}{\lambda-f^{-1}(z)}, \quad \Im(z)\ne 0.
$$
Clearly,
$$
G_z\in \Ker((f(\dot A)^*)-  zI),\quad \Im(z)\ne 0
.$$
Set
$$
G_+=G_{f^{-1}(i)}\quad\text{and}\quad G_-= G_{f^{-1}(-i)}.
$$
One easily checks that
$$
\|G_+\|=\|G_-\|, \quad  G_\pm\in \Ker ((f(\dot A)\mp i I),
$$
and that
$$
G_+-G_-\in \Dom (A)=\Dom (f(A)).
$$

Therefore,
the Liv\v{s}ic  function $s_{(f(\dot A), f(A))}(z)$ has  the representation
\begin{equation}\label{hjk1}
s_{(f(\dot A), f(A))}(z)=\frac{z-i}{z+i}
\frac{(G_z,G_-)}{(G_z,G_+)}.
\end{equation}
We have
$$
(G_z,G_-)=\int_\bbR \frac{d\mu(\lambda)}{(\lambda-f^{-1}(z))
\overline{(\lambda-f^{-1}(-i))}}
$$
and
$$
(G_z,G_+)=\int_\bbR \frac{d\mu(\lambda)}
{(\lambda-f^{-1}(z))\overline{(\lambda-f^{-1}(i))}},
$$
where  $\mu (d\lambda)$ is the measure from the representation
$$
M(z)=\int_\bbR \left (\frac{1}{\lambda-z}-\frac{\lambda}{1+\lambda^2}\right )
d\mu(\lambda), \quad \Im(z)\ne 0.
$$
Therefore,
\begin{align}
(z-i)(G_z,G_-)&=(z-i)\int_\bbR \frac{d\mu(\lambda)}
{(\lambda-f^{-1}(z))(\lambda-f^{-1}(i))}
\label{hjk2}
\\&\nonumber\\&
=\frac{z-i}{f^{-1}(i)-f^{-1}(z)}\left (M(f^{-1}(z))-M(f^{-1}(i))
\right )\nonumber
\\&\nonumber\\&
=-a\left (M(f^{-1}(z))-M(f^{-1}(i))
\right )\nonumber
\end{align}
and
\begin{align}
(z+i)(G_z,G_+)&=(z+i)\int_\bbR \frac{d\mu(\lambda)}
{(\lambda-f^{-1}(z))(\lambda-f^{-1}(-i))}
\label{hjk3}
\\& \nonumber\\&
=\frac{z+i}{f^{-1}(-i)-f^{-1}(z)}\left (M(f^{-1}(z))-M(f^{-1}(-i))
\right )
\nonumber \\& \nonumber\\&
=-a\left (M(f^{-1}(z))-M(f^{-1}(-i))
\right ).\nonumber
\end{align}

Now  \eqref{hjk1}, \eqref{hjk2} and \eqref{hjk3}
 yield
$$s_{(f(\dot A), f(A))}(z)=
\frac{M(f^{-1}(z))-M(f^{-1}(i))}{M(f^{-1}(z))-M(f^{-1}(-i))}
=\frac{m(z)-m(i)}{m(z)-\overline{m(i)}},
$$
which completes the proof.

\end{proof}

Next, we discuss the transformation law under the  affine transformation of the pair
$(\dot A, A)$ given by
$$
(\dot A, A)\longrightarrow (-\dot A, -A).
$$


\begin{lemma}\label{minus}
If $\dot A$ is a closed symmetric operator with deficiency indices
$(1,1)$ and $A$ its self-adjoint extension, then the
Weyl-Titchmarsh functions $M_\pm(z)$ associated with the pairs
$(\pm\dot A, \pm A)$ are related as follows
\begin{equation}\label{xiloC}
M_-(z)=-\overline{M_+(-\overline{z})}, \quad z\in \rho(A).
\end{equation}

In particular, for  the Li\v{s}ic functions associated with the pairs
$(\dot A, A)$ and  $(-\dot A, -A)$  we have
\begin{equation}\label{xiloCC}
s_{(-\dot A,- A)}(z)=\overline{s_{(\dot A, A)}(-\overline{z})}, \quad z\in \bbC_+.
\end{equation}

\begin{proof}
Let $n$  be  a unit vector in $\Ker ((\dot A)^*-iI)$ and
\begin{equation}\label{keliC}
m=(A-iI)(A+iI)^{-1}n.
\end{equation}
Then $m\in \Ker ((\dot A)^*+iI)=\Ker ((-\dot
A)^*-iI).$ By the definition of the Weyl-Titchmarsh function one
obtains that
$$
M_-(z)=\left ((-Az+I)(-A-zI)^{-1}m, m\right ).
$$
Therefore,
\begin{align}
-M_-(-\overline{z})&=((A\overline{z}+I)(A-\overline{z}I)^{-1}m,m)
\nonumber\\
&=((A\overline{z}+I)(A-\overline{z}I)^{-1}(A-iI)(A+iI)^{-1}n,(A-iI)(A+iI)^{-1}n)
\nonumber\\
&=((A-iI)(A+iI)^{-1}(A\overline{z}+I)(A-\overline{z}I)^{-1}n,(A-iI)(A+iI)^{-1}n)
\nonumber \\
&=((A\overline{z}+I)(A-\overline{z}I)^{-1}n,n)=M_+(\overline{z})=\overline{M_+(z)}.\label{xilC}
\end{align}
Here we have used  the    Schwarz symmetry  principle for the
Weyl-Titchmarsh function
$$
M_+(z)=\overline{M_+(\overline{z})}, \quad z\in \rho(A),
$$
and the observation  that the Cayley
transform $(A-iI)(A+iI)^{-1}$ is a unitary operator commuting with
the operator $A$.
Finally, \eqref{xiloC} follows from \eqref{xilC} by the substitution
$z\to -\overline{z}$.

To prove the last assertion, we use the relation  \eqref{blog} to conclude that
$$
s_{(-\dot A, -A)}(z)=\frac{M_-(z)-i}{M_-(z)+i}
=\frac{-\overline{M_+(-\overline{z})}-i}{-\overline{M_+(-\overline{z})}+i}
=\overline{s_{(\dot A, A)}(-\overline{z})},
$$
completing the proof.
\end{proof}
\end{lemma}

\section{The invariance principle }
 The main goal of this section is to establish  an  invariance principle for  the characteristic function  of a triple of operators under linear transformations of the operators from  the triple.

Introduce    the class  $\fD(\cH)$ of
 maximal dissipative  unbounded densely defined operators $\widehat A $,
($\widehat A \ne (\widehat A)^* $),
 in the Hilbert space $\cH$  such that
$$\dot A= \widehat A |_{\Dom (\widehat A)\cap \Dom(\widehat A^*)}$$   is a densely defined symmetric operator with deficiency indices $(1,1)$. In this case,
$$
\dot A\subset \widehat A\subset  (\dot A)^*
$$
and therefore $\widehat A$ is automatically a quasi-selfadjoint extension of $\dot A$ (see, e.g., \cite{MT-S}).

If $f(z)$ is the affine transformation  $
f(z)=az+b
$,
introduce the triple  $f(\fA)$ as
 $$
 f(\fA)=(f(\dot A),
f(\widehat A), f(A)).
  $$

\begin{theorem}\label{coinvv}
Let $\fA=(\dot A, \widehat A, A)$ be a triple such that $\widehat A\in \fD(\cH)$.
Suppose that
$
f(z)=az+b
$
with $a, b\in \bbR$, $a>0$, is an affine transformation.

Let $M(z)$ be the Weyl-Titchmarsh function associated with the pair $(\dot A, A)$.
Then  the von Neumann parameters $\varkappa$ and $  \varkappa'$ of the triples $\fA$ and $f(\fA)$ are related as
\begin{equation}\label{kkllD}
\frac{1+\varkappa}{1-\varkappa}
=\frac{m-\varkappa'\overline{m}}{i(1-\varkappa')},
\end{equation}
where
$$
m=M(f^{-1}(i)).
$$

Moreover,  the characteristic functions   $ S_{f(\fA)}(z)$ and $ S_{\fA}(z)$ are related as
\begin{equation}\label{inveq1}
S_{f(\fA)}(f(z))=\Theta_f S_{\fA}
(z), \quad z\in \bbC_+,
\end{equation}
where
 $$
\Theta_f= \left (\frac{1-\varkappa}{1-\overline{\varkappa}}\right )^{-1}\cdot \frac{1-\varkappa'}{1-\overline{\varkappa'}}.
 $$ is   a unimodular factor.
 In particular,  $\Theta_f $ continuously depends on $f$.

\end{theorem}


\begin{proof}
As in the proof of Lemma \ref{masha},  from the very beginning one can assume that $\dot A$ is a prime symmetric operator.

Let $\mu(d\lambda)$ denote the representing measure for the   Weyl-Titchmarsh function $M(z)$.
Without loss of generality (see Theorem \ref{unitar}) one may assume that
$A$ is the multiplication operator by independent variable in $L^2(\bbR, d\mu)$ and $\dot A$ coincides with  its restriction on
$$
\Dom (\dot A)=\left \{h\in \Dom (A) \,\bigg  |\int_\bbR h(\lambda)d\mu(\lambda)=
0\right \}.
$$

In this case, from Theorem \ref{modeldef} (see \eqref{defelrep}), we know that  the functions
$$
g_+(\lambda)=\frac{1}{\lambda-i} \quad \text{and }\quad g_-(\lambda)=\frac{1}{\lambda+i}
$$
form a basis in the deficiency subspace,
$$
g_\pm\in \Ker((\dot A)^*\mp i I), \quad \|g_\pm\|=1.
$$
From \eqref{kone} is also follows that
\begin{equation}\label{ifff}
g_+-g_-\in \dom (A).
\end{equation}

Clearly,  the functions
$$
G_\pm(\lambda)=\frac{1}{\lambda-f^{-1}(\pm i)}
$$
have the properties
$$
G_\pm \in \Ker((f(\dot A)^*)\mp I),\quad
\|G_+\|=\|G_-\|,$$
 and
\begin{equation}\label{gpm11}
G_+-G_-\in \dom (A)=\Dom(f(A)).
\end{equation}

From the definition of the  von Neumann parameters  $\varkappa, \varkappa'\in \bbD $ for the triples $\fA$  and $f(\fA)$ it follows that
\begin{equation}\label{initiD}
g_+-\varkappa g_-\in \dom (\widehat A)
\end{equation}
and
 \begin{equation}\label{initiiD}
G_+-\varkappa' G_-\in \dom (f(\widehat A))=\dom(\widehat A).
\end{equation}

Introduce the function
\begin{equation}\label{danu}
m(z)=M(f^{-1}(z)), \quad \Im(z)\ne 0.
\end{equation}

In order to
establish the relationship  \eqref{kkllD} between the von Neumann parameters,
 notice that
\begin{equation}\label{starstar}
G_\pm-\left [m(\pm i)\frac{g_+-g_-}{2i}+\frac{g_++g_-}{2}\right ]
\in \Dom(f(\dot A)),
\end{equation}
that is,
$$
\int_\bbR \left (G_\pm(\lambda)-\left [m(\pm i)\frac{g_+(\lambda)-g_-(\lambda)}{2i}+
\frac{g_+(\lambda)+g_-(\lambda)}{2}\right ]\right )d\mu(\lambda)=0.
$$
$$
$$

Indeed, since
$$
\frac{g_+(\lambda)-g_-(\lambda)}{2i}=\frac{1}{\lambda^2+1},
\quad
\frac{g_+(\lambda)+g_-(\lambda)}{2}=\frac{\lambda}{\lambda^2+1},
$$
and
$$
G_\pm(\lambda)=\frac{1}{\lambda-f^{-1}(\pm i)},
$$
one needs to verify the equality
$$\int_\bbR \left (\frac{1}{\lambda-f^{-1}(\pm i)}-\left [\frac{M(f^{-1}(\pm i))}
{\lambda^2+1}+
\frac{\lambda}{\lambda^2+1}\right ]\right )d\mu(\lambda)=0,
$$
which simply follows from the observations that
$$
\int_\bbR\frac{M(f^{-1}(\pm i))}
{\lambda^2+1}d\mu(\lambda)=M(f^{-1}(\pm i))
$$
and
$$
\int_\bbR \left (\frac{1}{\lambda-f^{-1}(\pm i)}-
\frac{\lambda}{\lambda^2+1}\right )d\mu(\lambda)=M(f^{-1}(\pm i)).
$$

Combining \eqref{initiiD} and \eqref{starstar} we get that
\begin{align*}
h&=\left (m\frac{g_+-g_-}{2i}+\frac{g_++g_-}{2}\right )
-\varkappa'\left ( \overline{m}\frac{g_+-g_-}{2i}+\frac{g_++g_-}{2}\right )
\\&=
\frac{m-\varkappa'\overline{m}+i -i\varkappa'}{2i}g_+
+\frac{i-i\varkappa'-m+\varkappa'\overline{m} }{2i}g_- \in  \Dom(\widehat A).
\end{align*}

Therefore, in view of \eqref{initiD},
\begin{equation}\label{varkap}
\varkappa=\frac{m-\varkappa'\overline{m}-i+i\varkappa'}
{m-\varkappa'\overline{m}+i -i\varkappa'}
=\frac{m-i-\varkappa'(\overline{m}-i)}
{m+i-\varkappa'(\overline{m} +i)}
\end{equation}
and
\eqref{kkllD} follows.

From   \eqref{initiD} and \eqref{initiiD} it follows that
 the characteristic function $S_{\fA}(z)$ associated with the triple $\fA=(\dot A, \widehat A,A)$
(see \eqref{ch12}) can be evaluated as
$$
S_{\fA}(z)=\frac{s_{(\dot A, A)}(z)-\varkappa}
{\overline{ \varkappa}s_{(\dot A, A)}(z) -1}.
$$

Representing $s_{(\dot A,A)}(z)$ via the
Weyl-Titchmarsh function $M(z)$,
$$
s_{(\dot A,A)}(z)=\frac{M(z)-i}{M(z)+i},
$$
one concludes that
$$
S_{\fA}(z)=\frac{\frac{M(z)-i}{M(z)+i}-\varkappa}
{\overline{ \varkappa}\frac{M(z)-i}{M(z)+i} -1}.
$$
Therefore, taking into account \eqref{danu}, one obtains
\begin{equation}\label{zxcvD}
S_{\fA}(f^{-1}(z))=\frac{\frac{m(z)-i}{m(z)+i}-\varkappa}
{\overline{ \varkappa}\frac{m(z)-i}{m(z)+i} -1}.
\end{equation}

In a similar way, using that  $$
G_+-G_-\in \Dom(f(A))
\quad \text{and}\quad
G_+-\varkappa G_+\in \Dom (f(\widehat A)),
$$
one also gets
$$
S_{f(\fA)}(z)=\frac{s_{(f(\dot A),f( A))}(z)-\varkappa'}
{\overline{ \varkappa'}s_{(f((\dot A), f(A))}(z) -1}.
$$

By Lemma \ref{masha},
$$
s_{(f(\dot A),f( A))}(z)=\frac{m(z)-m}{m(z)-\overline{m}},
$$
so that
\begin{equation}\label{vcxzD}
S_{f(\fA)}(z)=
\frac{\frac{m(z)-m}{m(z)-\overline{m}}-\varkappa'}
{\overline{ \varkappa'} \frac{m(z)-m}{m(z)-\overline{m}}-1}.
\end{equation}

From \eqref{zxcvD} one gets that
\begin{align}
S_{\fA}(f^{-1}(z))&=
\frac{\frac{m(z)-i}{m(z)+i}-\varkappa}
{\overline{ \varkappa}\frac{m(z)-i}{m(z)+i} -1}
=
\frac{{m(z)-i}-\varkappa (m(z)+i)}
{\overline{ \varkappa}(m(z)-i)-(m(z)+i) }
\label{ssddD}
\\&\nonumber \\&
=\frac{(1-\varkappa)m(z)-i(1+\varkappa )}
{(\overline{ \varkappa}-1)m(z)-i(1+\overline{\varkappa} )}
=\frac{1-\varkappa}{\overline{\varkappa}-1}\cdot
\frac{m(z)-i\frac{1+\varkappa}{1-\varkappa}}
{m(z)+i\frac{1+\overline{\varkappa}}{1-\overline{\varkappa}}}.
\nonumber\end{align}

A similar computation for the right hand side of \eqref{vcxzD} yields
\begin{align}
S_{f(\fA)}(z)&=
\frac{\frac{m(z)-m}{m(z)-\overline{m}}-\varkappa'}
{\overline{ \varkappa'} \frac{m(z)-m}{m(z)-\overline{m}}-1}
=
\frac{{m(z)-m}-\varkappa' (m(z)-\overline{m})}
{\overline{ \varkappa'}(m(z)-m)-(m(z)-\overline{m}) }
\label{ddssD}
 \\&
=\frac{(1-\varkappa')m(z)-(m-\varkappa'\overline{m})}
{(\overline{ \varkappa'}-1)m(z)+(\overline{m}-\overline{\varkappa'}m )}
=\frac{1-\varkappa'}{\overline{\varkappa'}-1}\cdot
\frac{m(z)-\frac{m-\varkappa'\overline{m}}{1-\varkappa'}}
{m(z)+\frac{\overline{m}-\overline{\varkappa'}m}
{\overline{\varkappa'}-1}}
\nonumber
\\&\nonumber\\&
=\frac{1-\varkappa'}{\overline{\varkappa'}-1}\cdot
\frac{m(z)-i\frac{1+\varkappa}{1-\varkappa}}
{m(z)+i\frac{1+\overline{\varkappa}}{1-\overline{\varkappa}}},
\nonumber\end{align}
where we used the relation \eqref{kkllD} on the last step.

Comparing  \eqref{ssddD} and \eqref{ddssD}, we obtain
$$
S_{f(\fA)}(f(z))=\left (\frac{1-\varkappa}{1-\overline{\varkappa}}\right )^{-1}\cdot \frac{1-\varkappa'}{1-\overline{\varkappa'}}
\cdot S_{\fA}(z),$$
which proves \eqref{inveq1}.

\end{proof}


We conclude this section by establishing an invariance principle under the  anti-holomorphic  transformation (involution) of the   triple
$$
\fA=(\dot A, \widehat A, A)\longrightarrow -\fA^*= (-\dot A, -(\widehat A
)^*, -A).$$

\begin{theorem}\label{comin}
Let $\dot A$ be a densely defined,  closed symmetric operator
  with deficiency indices $(1,1)$, $A$ its   self-adjoint
extension and $\widehat A$  quasi-selfadjoint
dissipative extension of $\dot A$.

Then the characteristic functions associated with the triples $\fA=(\dot A,\widehat A, A)$
and $-\fA^*=(-\dot A, -(\widehat A)^*, -A)$  are related as follows
\begin{equation}\label{charinv}
S_{-\fA^*}(z)=\overline{S_{\fA}(-\overline{z})}.
\end{equation}

\end{theorem}

\begin{proof}

Let $g_\pm$ be normalized deficiency elements of $\dot A$,
$$
g_\pm\in \Ker ((\dot A)^*\mp I),
$$
such that
$$
g_+-g_-\in \Dom (A)
\quad \text{and}\quad
g_+-\varkappa g_-\in \Dom (\widehat A).
$$
Clearly,
$$
g_\pm \in \Ker ((-\dot A)^*\pm I),
$$

$$
g_--g_+\in \Dom (A)=\Dom(-A),
$$and
$$
g_--\overline{\varkappa} g_+\in \Dom ((\widehat A)^*)=\Dom ((-\widehat A)^*).
$$
Hence, using Lemma \ref{minus}, one obtains
\begin{align*}
S_{- \fA^*}(z)&=\frac{s_{(-\dot A, -A)}(z)-\overline{\varkappa}}
{\varkappa s_{(-\dot A, -A)}(z)-1}
=\frac{\overline{s(\widehat A, A)(-\overline{z})}-\overline{\varkappa}}
{\varkappa\overline{s_{(\dot  A, A)}(-\overline{z})} -1}
\\&
=\overline{S_{\fA}(-\overline{z})}.
\end{align*}
The proof is complete.
\end{proof}

\section{The operator coupling and the multiplication theorem}\label{littem}

Introduce the class $\fD_0(\cH)$  of maximal dissipative   densely defined operators $\widehat A $ in the Hilbert space $\cH$    of the form
$$ \widehat A =A+itP,$$
where $A=\Re( \widehat A)$ is a self-adjoint operator,  $t>0$, and $P$ is  a rank-one orthogonal projection \cite{ABT,Brod,Lv1}.

Introduce the concept of an operator coupling of two operators from the classes  $\fD_0(\cH_1)$ and $\fD_0(\cH_2)$.

\begin{definition}\label{opcupd0}
Suppose   that
  $\widehat A_1\in \fD_0(\cH_1)$   and $\widehat A_2\in \fD_0(\cH_2)$ are maximal dissipative  operators acting in the Hilbert spaces $\cH_1$ and $\cH_2$, respectively.

We say that a maximal dissipative  operator  $\widehat A$ from the class  $\ \fD_0(\cH_1\oplus \cH_2)$ is an operator coupling   of $\widehat A_1$ and $\widehat A_2$,
in writing, $$\widehat A=\widehat A_1\uplus \widehat A_2,$$   if
$ \widehat A-(\widehat A_1\oplus \widehat A_2)$ is a rank-one operator,
 the Hilbert space $\cH_1$ is invariant for $\widehat A$, and
   the restriction of $\widehat  A$ on $\cH_1$ coincides with the dissipative operator $\widehat A_1$. That is,
   $$
   \Dom(\widehat A)\cap \cH_1=\Dom (\widehat A_1)
   $$
   and
$$
\widehat A|_{\cH_1\cap \Dom({\widehat A_1)}}=\widehat A_1.
$$

\end{definition}

\begin{theorem} [{cf. \cite{ABT,BL60,Brod,LP}}]
\label{opcoupbdd}   Let
$\widehat A=\widehat A_1\uplus \widehat A_2$ be  an operator coupling  of two  maximal dissipative operators
  $\widehat A_k \in \fD_0(\cH_k)$, $k=1,2$.
  Then the characteristic function of an operator coupling  $\widehat A_1\uplus \widehat A_2$
  coincides with the product of the ones
  of  $ \widehat A_1$ and $\widehat A_2$,
  \begin{equation}\label{proizvas}S_{\widehat A_1\uplus \widehat A_2 }(z)=S_{\widehat A_1}(z)\cdot S_{\widehat A_2}(z), \quad z\in \bbC_+.
\end{equation}

\end{theorem}
\begin{proof}
Suppose that
$$
\widehat A_k= A_k+i(\cdot, g_k )g_k, \quad k=1,2,
$$
where $A_k=A_k^*$, $g_k\in \cH_k$.  Denote by $P_k$ ($k=1, 2$) the orthogonal projections onto the subspaces $\cH_k$, respectively.
From the definition of an operator coupling it follows that
$$
\widehat A=\widehat A_1P_1+\widehat A_2P_2+(\cdot ,\widetilde g) g
$$
for some $g, \widetilde g\in \cH_1\oplus \cH_2$ and that
$$
\widehat AP_1=\widehat A_1P_1.
$$
In particular,
$$
(\widehat A)^*P_2=(\widehat A_2)^*P_2
$$
and therefore
\begin{equation}\label{otto}
\widetilde g\in \cH_2 \quad \text{and}\quad  g\in \cH_1.
\end{equation}

First we show that
$$
\Im (\widehat A)=(\cdot ,\phi)\phi,
$$
where
\begin{equation}\label{predd}
\phi=(\Theta_1g_1)\oplus (\Theta_2g_2)
\end{equation}
for some $|\Theta_k|=1$, $k=1,2$.  Indeed, we have
\begin{equation}\label{dlinno}
(\cdot, g_1)g_1+(\cdot, g_2)g_2+\frac{1}{2i }\left ((\cdot,\widetilde  g)g-(\cdot, g)\widetilde g\right )=(\cdot ,\phi)\phi.
\end{equation}
Introducing
$$
\phi_k=P_k \phi, \quad k=1,2,
$$
from \eqref{otto} and \eqref{dlinno}  it follows that
$$
|(g_k,g_k)|=|(g_k, \phi)|=|(g_k, \phi_k)|, \quad k=1,2,
$$
and then  we get \eqref{predd}.

Rewrite the equality \eqref{dlinno} one more time
$$
(\cdot, g_1)g_1+(\cdot, g_2)g_2+\frac{1}{2i }\left ((\cdot,\widetilde  g)g-(\cdot, g)\widetilde g\right )=(\cdot ,(\Theta_1g_1)\oplus (\Theta_2g_2)
)(\Theta_1g_1)\oplus (\Theta_2g_2).
$$
We get
$$
\frac{1}{2i }\left ((\cdot,\widetilde  g)g-(\cdot, g)\widetilde g\right )=\overline{\Theta_1}\Theta_2(\cdot ,g_1)
g_2+ \Theta_1\overline{\Theta_2}(\cdot ,g_2)
g_1
$$and therefore
$$
(\cdot,\widetilde  g)g=2i(\cdot,\phi_2)\phi_1.
$$
In particular, we have that
$$
\widehat A=( A_1+i(\cdot, \phi_1)\phi_1)P_1+( A_2+i(\cdot, \phi_2)\phi_2)P_2+2i(\cdot,\phi_2)\phi_1,
$$
$$
\Im(\widehat A )=(\cdot , (\phi_1+\phi_2))(\phi_1+\phi_2)
$$
and we arrive at the definition of an operator coupling as presented in  \cite[eq. (2.1)]{Br}.  Literally  repeating step by step the proof of the Multiplication Theorem
\cite[Theorem 2.1]{Brod} one justifies \eqref{proizvas}.
\end{proof}

\begin{remark}  We remark that an operator coupling of two dissipative operators from the classes $\fD_0(\cH_1)$  and $\fD_0(\cH_2)$ is not unique. In fact, we have shown that
an  operator coupling    $\widehat A_1\uplus \widehat  A_2$  of two dissipative operators
$$\widehat A_k= A_k+i(\cdot, g_k )g_k, \quad k=1,2,$$ is necessarily of the form
 \begin{equation}\label{kokol}
\widehat A_1\uplus \widehat A_2=( A_1+i(\cdot, g_1)g_1)P_1+( A_2+i(\cdot, g_2)g_2)P_2+2i\Theta (\cdot,g_2)g_1,
\end{equation}
for some $|\Theta|$=1.
Moreover, for any choice of $\Theta$ such that $|\Theta|=1$
the right hand side of \eqref{kokol} meets the requirements to be an operator coupling of $\widehat A_1$ and $\widehat A_2$.
\end{remark}

Recall that    the class  $\fD(\cH)$ consists of all
 maximal dissipative  unbounded densely defined operators $\widehat A $,
  \textcolor{red}{($\widehat A \ne (\widehat A)^* $),
 } in the Hilbert space $\cH$  such that
$$\dot A= \widehat A |_{\Dom (\widehat A)\cap \Dom((\widehat A)^*)}$$   is a densely defined symmetric operator with deficiency indices $(1,1)$
(see Appendix  F).

\begin{definition} [\cite{MT-MAT}] \label{defcoup}
Suppose   that
  $\widehat A_1\in \fD(\cH_1)$ and   $\widehat A_2\in \fD(\cH_2)$.
We say that a   $\widehat A\in \fD(\cH_1\oplus\cH_2)$ is an operator coupling   of $\widehat A_1$ and $\widehat A_2$,
in writing, $$\widehat A=\widehat A_1\uplus \widehat A_2,$$   if
\begin{itemize}
\item[(i)] the Hilbert space $\cH_1$ is invariant for $\widehat A$ and
   the restriction of $\widehat  A$ on $\cH_1$ coincides with the dissipative operator $\widehat A_1$, that is,
   $$
   \Dom(\widehat A)\cap \cH_1=\Dom (\widehat A_1),
   $$
$$
\widehat A|_{\cH_1\cap \Dom({\widehat A_1)}}=\widehat A_1,
$$
\end{itemize}
and
\begin{itemize}
\item[(ii)]the symmetric operator  $\dot A=   \widehat A|_{\Dom(\widehat A)\cap \Dom ((\widehat A)^*)}$ has the property
$$
\dot A\subset \widehat A_1\oplus (\widehat  A_2)^*.
 $$
 \end{itemize}
\end{definition}

 The corresponding multiplication theorem   for the class $ \fD(\cH)$  can be formulated as follows (see  \cite[Theorem 6.1, cf. Theorem 5.4]{MT-MAT}).

\begin{theorem}\label{opcoup}
Suppose   that $\widehat A=\widehat A_1\uplus \widehat A_2$ is an operator coupling  of two  maximal dissipative operators
  $\widehat A_k \in \fD(\cH_k)$, $k=1,2$. Denote by  $\dot A $, $\dot A_1$ and $\dot A_2$  the corresponding  symmetric operators with deficiency indices $(1,1)$, respectively.
That is,  $$\dot A=\widehat A|_{\Dom(\widehat A)\cap\Dom((\widehat A)^*)}
$$
and
$$
\dot A_k=\widehat A_k|_{\Dom(\widehat A_k)\cap\Dom((\widehat A_k)^*)}, \quad k=1,2.
$$

Then there exist self-adjoint reference operators  $A$, $A_1$, and $A_2$,  extending   $\dot A$, $\dot A_1$ and $\dot A_2$, respectively,  such that
\begin{equation}\label{proizvasunbdd}
S_{(\dot A,\widehat A_1\uplus \widehat A_2,  A)}(z)=S_{(\dot A_1,\widehat A_1,A_1)}(z)\cdot S_{(\dot A_2,\widehat A_2, A_2)}(z), \quad z\in \bbC_+.
\end{equation}

Moreover, for any operator coupling $\widehat A$ of $\widehat A_1$ and $\widehat A_2$, the multiplication rule
\begin{equation}\label{multkiappa1}\widehat \kappa(\widehat A)=\widehat \kappa (\widehat A_1)\cdot \widehat \kappa(\widehat A_2)
\end{equation}
holds.
Here $\widehat \kappa(\cdot )$ stands for  the absolute value of the von Neumann parameter of a dissipative  operator defined by \eqref{defmoduli}.

\end{theorem}

\begin{corollary} {\it Assume the hypotheses of Theorem \ref{opcoup}. Then the von Neumann logarithmic potential $\Gamma_{\widehat A}(z)$
$($see Definition \ref{logpotdef}$)$
 is  an additive functional in the sense that
$$
\Gamma_{\widehat A_1\uplus \widehat A_2}(z)=\Gamma_{\widehat A_1}(z)+\Gamma_{ \widehat A_2}(z),\quad z\in \rho_{\widehat A_1}\cap
\rho_{\widehat A_2}\cap \rho_{\widehat A_1\uplus \widehat A_2}\cap\bbC_+.
$$
}
\end{corollary}

\section{Stable Laws}\label{gnedkol}
 Recall (see, e.g., \cite{F,IL,Z}) that  a distribution $G$ (of a random variable)  is said to be stable, if a linear combination of two independent random variables with this distribution
  has the same distribution, up to location and scale parameters. That is,
for any $b_1$, $b_2>0$, there exist a $b>0$ and $ a\in \bbR $ such that
$$
G\left (\frac{x}{b_1}\right )\star G\left (\frac{x}{b_2}\right )=G\left (\frac{x-a}{b}\right ),
$$
where $\star$ denotes the convolution of distributions (see \cite[Ch. V.4]{F}).

It turns out that a (non-degenerated) law  $G$ is  stable if and only if  the logarithm of its characteristic function  has the representation \cite[Theorem B.2]{Z}
\begin{equation}\label{stlaw}
\log g(t)=\sigma\left  (it\gamma-|t|^\alpha\left (1- i \beta \frac{t}{|t|}\omega (t, \alpha)\right )\right )
\end{equation}
for some $\sigma >0$, $-\infty <\gamma<\infty$,
\begin{align}
0<&\alpha \le 2,\quad  \text{(the index of stability)}\nonumber
\\
-1\le &\beta\le 1
, \quad \text{(the skew parameter)}. \label{data}
\end{align}
Here
\begin{equation}\label{omegaalpha}
\omega(t, \alpha)=\begin{cases}
 \tan \left (\frac\pi2 \alpha\right ), & \alpha \ne 1, \\
- \frac2\pi \log |t|, & \alpha=1.
\end{cases}
\end{equation}
The skew  parameter $
\beta $ is irrelevant when $\alpha =2$.

Recall (see, e.g., \cite{Z})  that   a distribution $F$ is said to belong to
{\it the domain of attraction of a law}  if
  there are constants $A_n$ and $B_n>0$ such that the following non-zero  limit
 $$
\lim_{n\to \infty} \log \left [f(t/B_n)\right ]^ne^{iA_nt}
$$
exists, where
 $f(t)$   is the characteristic function of the  distribution $F$,
$$
f(t)=\int_\bbR e^{itx}dF(x).
$$
In this case   the limit coincides with the logarithm of a  stable law \eqref{stlaw} for an appropriate choice of the parameters $\alpha$, $\beta$, $\gamma$ and  $\sigma$.

Recall that a positive function $h(x)$, defined for $x\ge 0$, is said to be slowly varying if, for all $t>0$,
$$
\lim_{x\to\infty}\frac{h(tx)}{h(x)}=1.
$$
Also,  by the Karamata theorem (see, e.g., \cite[Appendix 1]{IL} for an exposition of the Karamata theory), a slowly varying function $h$ which is integrable on any finite interval can be represented in the form
$$
h(x)=c(x)\exp\left \{ \int _{x_0}^x\frac{\epsilon(t)}{t}dt \right \},\quad x_0>0,
$$
where
$$
\lim_{x\to \infty}c(x)=c\ne0\quad \text{ and }\quad \lim_{x\to \infty}\epsilon(x)=0.
$$

A key result in this area is the following Gnedenko-Kolmogorov limit theorem.

\begin{theorem} [{\cite[Theorem 2.6.1]{IL}}]\label{stthm}
A distribution $F$ belongs to the domain of attraction of   a stable law  \eqref{stlaw}   with exponent $\alpha$, $0<\alpha\le 2$,  and parameters
$\sigma$, $\beta$ and $\gamma$ if and only if
\begin{align}
1-F(x)&=\frac{c_1+o(1)}{x^\alpha}h(x), \quad x>0\label{as+}\\
F(x)&=\frac{c_2+o(1)}{(-x)^\alpha}h(-x), \quad x<0\label{as-}
\end{align}
as $|x|\to \infty$, where   $c_1, c_2\ge 0$, $c_1+c_2>0$ and $h$ is slowly varying in the sense of Karamata.

In this case,
\begin{equation}\label{kaksigma}
\sigma=(c_1+c_2)d(\alpha),
\end{equation}
where
\begin{equation}\label{dalpha}
d(\alpha)=\begin{cases}
\Gamma (1-\alpha)\cos \left (\frac12 \pi \alpha\right ),& \alpha \ne 1\\
\frac\pi2,&\alpha =1
\end{cases},
\end{equation}
and
\begin{equation}\label{kakbeta}
\beta =\frac{c_1-c_2}{c_1+c_2}.
\end{equation}

\end{theorem}

\begin{remark}\label{strem}  The (tauberian type) relationship between the set of data $(c_1, c_2, h)$ and $(\alpha, \beta, \gamma, \sigma) $  referred to in Theorem \ref{stthm}
(also see \eqref{data})
 can be described as follows:
if a distribution  $F$ belongs to the domain of attraction of the stable law \eqref{stlaw}, that is,
$$
\lim_{n\to \infty} \log \left [f(t/B_n)\right ]^ne^{iA_nt}=\sigma\left  (it\gamma-|t|^\alpha +i \beta  \frac{t}{|t|}\omega (t, \alpha)\right )
$$
for   some  constants $A_n$ and $B_n>0$,
then (see, e.g., \cite[Theorem 2.6.5]{IL})
$$
\log f(t)=i\tilde \gamma t -\sigma |t|^\alpha  h(1/t) \left (1-i\beta \frac{t}{|t|} \omega(t,\alpha)\right )(1+o(1)) \quad \text{as}\quad t\to 0,
$$
where $\tilde \gamma$ is  in general not necessarily the  same as $\gamma$.
Recall that in this case the norming constants $B_n$ necessarily satisfy the relation
 $$
  \lim_{n\to \infty}nB_n^{-\alpha}h(B_n)=1.
  $$

If in the hypothesis of Theorem \ref{stthm} the slowly varying  function $h(x)$ has   the property that
$\displaystyle{\lim_{x\to \infty} h(x)=1}$, then  the scaling factors $B_n$  can be given by
$$
B_n=n^{1/\alpha}.
$$
Under this hypothesis, the probability distribution $F$ is said to belong to the domain of  {\it normal}  attraction of a stable law. In particular, every stable law belongs to the  normal of its own normal attraction.

\end{remark}

\end{document}